\documentclass[letterpaper,12pt,titlepage,oneside,final]{book}

\input{Qcircuit}
\newcommand{\href}[1]{#1} % does nothing, but defines the command so the
\usepackage{ifthen}
\newboolean{PrintVersion}
\setboolean{PrintVersion}{false}
\usepackage{amsmath,amssymb,amstext,amsthm} % Lots of math symbols and environments
\usepackage[pdftex]{graphicx} % For including graphics N.B. pdftex graphics driver
\usepackage{paralist,bm,dcolumn}

\usepackage[pdftex,letterpaper=true,pagebackref=false]{hyperref} % with basic options
\hypersetup{
    plainpages=false,       % needed if Roman numbers in frontpages
    pdfpagelabels=true,     % adds page number as label in Acrobat's page count
    bookmarks=true,         % show bookmarks bar?
    unicode=false,          % non-Latin characters in Acrobat’s bookmarks
    pdftoolbar=true,        % show Acrobat’s toolbar?
    pdfmenubar=true,        % show Acrobat’s menu?
    pdffitwindow=false,     % window fit to page when opened
    pdfstartview={FitH},    % fits the width of the page to the window
    pdftitle={Approximation, Proof Systems, and Correlations in a Quantum World},    % title: CHANGE THIS TEXT!
    pdfauthor={Sevag Gharibian},    % author: CHANGE THIS TEXT! and uncomment this line
    pdfsubject={Computer Science},  % subject: CHANGE THIS TEXT! and uncomment this line
    pdfnewwindow=true,      % links in new window
    colorlinks=true,        % false: boxed links; true: colored links
    linkcolor=blue,         % color of internal links
    citecolor=green,        % color of links to bibliography
    filecolor=magenta,      % color of file links
    urlcolor=cyan           % color of external links
}
\ifthenelse{\boolean{PrintVersion}}{   % for improved print quality, change some hyperref options
\hypersetup{	% override some previously defined hyperref options
    citecolor=black,%
    filecolor=black,%
    linkcolor=black,%
    urlcolor=black}
}{} % end of ifthenelse (no else)

\let\origdoublepage\cleardoublepage
\newcommand{\clearemptydoublepage}{%
  \clearpage{\pagestyle{empty}\origdoublepage}}
\let\cleardoublepage\clearemptydoublepage

\newtheorem{theorem}{Theorem}
\numberwithin{theorem}{chapter}
\newtheorem{corollary}[theorem]{Corollary}
\newtheorem{lemma}[theorem]{Lemma}

\newtheorem{definition}[theorem]{Definition}

\newtheorem{alg}[theorem]{Algorithm}

\newtheorem{problem}[theorem]{Problem}

\newcommand{\myparagraph}[1]{\paragraph{#1.}}

\newcommand{\eps}{\varepsilon}
\newcommand{\epssdp}{\varepsilon_{\rm sdp}}
\newcommand{\braket}[2]{\langle #1|#2\rangle}
\newcommand{\ketbra}[2]{\ket{#1}{\bra{#2}}}
\newcommand{\lmin}[1] {\lambda_{\operatorname{min}}(#1)}

\newcommand{\CQ}{\mathcal{CQ}}
\newcommand{\lh}{\operatorname{LH}}
\newcommand{\flh}{\operatorname{5-LH}}
\newcommand{\klhh}{\operatorname{k-LH}}
\newcommand{\qma}{\operatorname{QMA}}
\newcommand{\enc}[1]{\left<#1\right>}

\newcommand{\beq}{\begin{equation}}
\newcommand{\eeq}{\end{equation}}

\newcommand{\trace}{{\rm Tr}}

\newcommand{\norm}[1]{\left\|\,#1\,\right\|}       % norm
\newcommand{\pnorm}[1]{\left\|\,#1\,\right\|_p}       % norm
\newcommand{\onorm}[1]{\norm{#1}_{\mathrm{1}}}      % Euclidean norm for vectors
\newcommand{\enorm}[1]{\norm{#1}_{\mathrm{2}}}      % Euclidean norm for vectors
\newcommand{\trnorm}[1]{\norm{#1}_{\mathrm {tr}}}  % trace norm
\newcommand{\fnorm}[1]{\norm{#1}_{\mathrm {F}}}    % frobenius norm
\newcommand{\snorm}[1]{\norm{#1}_{\mathrm {\infty}}}    % spectral norm

\newcommand{\set}[1]{{\left\{#1\right\}}}    % braces for set notation
\newcommand{\ve}[1]{\mathbf{#1}}
\newcommand{\abs}[1]{\left\lvert #1 \right\rvert}
\newcommand{\optprod}{\OPT_P}

\newcommand{\optt}{\operatorname{OPT_2}}

\newcommand{\poly}{\operatorname{poly}}
\newcommand{\cc}{d^{\frac{k}{2}}}
\newcommand{\OPT}{{\rm OPT}}
\newcommand{\QMA}{{\rm QMA}}

\newcommand{\NP}{{\rm NP}}

\newcommand{\PH}{{\rm PH}}
\newcommand{\BPP}{{\rm BPP}}
\newcommand{\BQP}{{\rm BQP}}

\newcommand{\complex}{{\mathbb C}}
\newcommand{\reals}{{\mathbb R}}
\newcommand{\ints}{{\mathbb Z}}
\newcommand{\nats}{{\mathbb N}}

\newcommand{\spa}[1]{\mathcal{#1}}
\newcommand{\dens}{\mathcal{D}(\spa{A}\otimes\spa{B})}

\newcommand{\LL}{\mathcal{L}}
\newcommand{\DD}{\mathcal{D}}
\newcommand{\HH}{\mathcal{H}}
\newcommand{\UU}{\mathcal{U}}

\newcommand{\klh}{MAX-$k$-local Hamiltonian}

\mathchardef\mhyphen="2D

\newcommand{\ayes}{A_{\rm yes}} %CHECK
\newcommand{\ano}{A_{\rm no}} %CHECK
\newcommand{\nl} {\mathcal{L}_1}
\newcommand{\nll} {\mathcal{L}_2}

\newcommand{\SSC}{{\rm SUCCINCT~SET~COVER}}
\newcommand{\QSSC}{{\rm QUANTUM~SUCCINCT~SET~COVER}}
\newcommand{\IRR}{{\rm IRREDUNDANT}}
\newcommand{\QIRR}{{\rm QUANTUM~IRREDUNDANT}}

\newcommand{\QMMWW}{{\rm QUANTUM~MONOTONE~MINIMUM~WEIGHT~WORD}}
\newcommand{\cqsn}[1]{{\rm cq}\mhyphen\Sigma_{#1}}

\newcommand{\qcsn}[1]{{\rm qc}\mhyphen\Sigma_{#1}}

\newcommand{\cqs}{\cqsn{2}}
\newcommand{\cqslh}[1]{{\rm cq}\mhyphen\Sigma_2{\rm LH}}
\newcommand{\ssat}[1]{{\Sigma_{#1} {\rm SAT}}}

\newcommand{\qcs}{\qcsn{2}}

\newcommand{\ccs}{{\rm cc}\mhyphen\Sigma_2}
\newcommand{\qqs}{{\rm qq}\mhyphen\Sigma_2}
\newcommand{\st}{\s{2}}
\newcommand{\s}[1]{\Sigma_{#1}^p}

\newcommand{\XX}{\rm cQMA}
\newcommand{\hin}{H_{\rm in}}
\newcommand{\hprop}{H_{\rm prop}}

\newcommand{\hout}{H_{\rm out}}
\newcommand{\hstab}{H_{\rm stab}}
\newcommand{\hist}{\Pi_{\rm hist}}
\newcommand{\shist}{\spa{S}_{\rm hist}}

\newcommand{\class}[1]{\textup{#1}}
\newcommand{\ppoly}{\class{P}_{\rm/poly}}
\newcommand{\histstate}{\ket{\psi}_{\rm hist}}
\newcommand{\histstatecq}[2]{\ket{#1,#2}_{\rm hist}}
\newcommand{\histstateketbra}{\ketbra{\psi}{\psi}_{\rm hist}}

\newcommand{\histstatecqketbra}[2]{\ketbra{#1,#2}{#1,#2}_{\rm hist}}
\newcommand{\cqma}{\rm cQMA}
\newcommand{\cqslhmin}{{\rm cq}\mhyphen\Sigma_2{\rm LH}\mhyphen{\rm HW}}

\newenvironment{mylist}[1]{\begin{list}{}{
	\setlength{\leftmargin}{#1}
	\setlength{\rightmargin}{0mm}
	\setlength{\labelsep}{2mm}
	\setlength{\labelwidth}{8mm}
	\setlength{\itemsep}{0mm}}}
	{\end{list}}

\newcommand{\ip}[2]{\left\langle #1 , #2\right\rangle}
\newcommand{\tr}{\trace}
\newcommand{\setft}[1]{\mathrm{#1}}

\newcommand{\density}[1]{\setft{D}\left(#1\right)}

\newcommand{\pos}[1]{\setft{Pos}\left(#1\right)}
\newcommand{\sep}[1]{\setft{Sep}\left(#1\right)}

\newcommand{\inner}[2]{\langle #1, #2 \rangle}

\newcommand{\calA}{\mathcal{A}}

\newcommand{\R}{\mathbb{R}}

\newcommand{\Span}{\mathrm{Span}}

\newcommand{\interior}{\mathrm{int}}
\newcommand{\Pos}{\mathrm{Pos}}

\newcommand{\Sep}{\mathrm{Sep}}

\newcommand{\E}{\mathbb{E}}

\def\I{I}
\def\({\left(}
\def\){\right)}
\def\X{\mathcal{X}}
\def\Y{\mathcal{Y}}
\def\Z{\mathcal{Z}}

\def\yes{\text{yes}}
\def\no{\text{no}}
\def\poly{\textup{poly}}
\def\blog{\textup{log}}

\newcommand{\be}{\begin{equation}}
\newcommand{\ee}{\end{equation}}
\newcommand{\ben}{\begin{eqnarray}}
\newcommand{\een}{\end{eqnarray}}
\newcommand{\bes}{\begin{subequations}}
\newcommand{\ees}{\end{subequations}}
\newcommand{\dg}{\dagger}

\newcommand{\pmsmt}{\{\Pi^A_j\}}
\newcommand{\ro}{\rho}
\newcommand{\rof}{\ro_f}

\newcommand{\dqc}{\rho_{DQC1}}
\newcommand{\disc}{\delta}
\newcommand{\discm}{\delta(\ro)}
\newcommand{\gdisc}{\delta_G(\ro)}
\newcommand{\rdisc}{\delta_R(\ro)}

\newcommand{\fu}{\operatorname{d}(\ro,U_A)}
\newcommand{\fua}[2]{\operatorname{d}(#1,#2)}
\newcommand{\fum}{\operatorname{d}_{\operatorname{max}}(\ro)}
\newcommand{\fuma}[1]{\operatorname{d}_{\operatorname{max}}(#1)}
\newcommand{\fumdqc}{\operatorname{d}_{\operatorname{max}}(\dqc)}
\newcommand{\proj}[1]{\mbox{$|#1\rangle \!\langle #1 |$}}

\newcommand{\pr}{\operatorname{Pr}}
\newcommand{\IM}{\mathcal{I}}
\newcommand{\JM}{\mathcal{J}}

\newcommand{\ox}{\otimes}

\newcommand{\II}{I}

\newcommand{\wek}[1]{{{#1}}}
\newcommand{\ot}[0]{\otimes}
\newcommand{\Tr}{{\rm Tr}}

\newcommand{\cB}{b}
\newcommand{\cBB}{\mathcal{B}}

\newcommand{\cN}{\mathcal{N}}

\newcommand{\CC}{{{\mathbb C}}}
\newcommand{\EE}{{{\mathbb E}}}

\newcommand{\A}{\spa{A}}
\newcommand{\B}{\spa{B}}
\newcommand{\UA}{U_A}
\newcommand{\UAv}{U^A_{\ve{v}}}
\newcommand{\denstn}{\mathcal{D}(\complex^2\otimes\complex^N)}
\newcommand{\denstt}{\mathcal{D}(\complex^2\otimes\complex^2)}
\newcommand{\densdd}{\mathcal{D}(\complex^d\otimes\complex^d)}
\newcommand{\HERM}{\mathcal{H}(\A\otimes\B)}

\newcommand{\proja}{\set{\Pi_j^A}}
\newcommand{\projv}{\set{\Pi_j^A}_{\ve{v}}}

\newcommand{\rous}{{\rm RU}(\A)}

\begin{document}

% For a large document, it is a good idea to divide your thesis
% into several files, each one containing one chapter.
% To illustrate this idea, the "front pages" (i.e., title page,
% declaration, borrowers' page, abstract, acknowledgements,
% dedication, table of contents, list of tables, list of figures,
% nomenclature) are contained within the file "uw-ethesis-frontpgs.tex" which is
% included into the document by the following statement.
%----------------------------------------------------------------------
% FRONT MATERIAL
%----------------------------------------------------------------------
% T I T L E   P A G E
% -------------------
% Last updated May 24, 2011, by Stephen Carr, IST-Client Services
% The title page is counted as page `i' but we need to suppress the
% page number.  We also don't want any headers or footers.
\pagestyle{empty}
\pagenumbering{roman}

% The contents of the title page are specified in the "titlepage"
% environment.
\begin{titlepage}
        \begin{center}
        \vspace*{1.0cm}

        \Huge
        {\bf Approximation, Proof Systems, and Correlations in a Quantum World }

        \vspace*{1.0cm}

        \normalsize
        by \\

        \vspace*{1.0cm}

        \Large
        Sevag Gharibian \\

        \vspace*{3.0cm}

        \normalsize
        A thesis \\
        presented to the University of Waterloo \\
        in fulfillment of the \\
        thesis requirement for the degree of \\
        Doctor of Philosophy \\
        in \\
        Computer Science \\

        \vspace*{2.0cm}

        Waterloo, Ontario, Canada, 2012 \\

        \vspace*{1.0cm}

        \copyright\ Sevag Gharibian 2012 \\
        \end{center}
\end{titlepage}

% The rest of the front pages should contain no headers and be numbered using Roman numerals starting with `ii'
\pagestyle{plain}
\setcounter{page}{2}

\cleardoublepage % Ends the current page and causes all figures and tables that have so far appeared in the input to be printed.
% In a two-sided printing style, it also makes the next page a right-hand (odd-numbered) page, producing a blank page if necessary.

% D E C L A R A T I O N   P A G E
% -------------------------------
  % The following is the sample Delaration Page as provided by the GSO
  % December 13th, 2006.  It is designed for an electronic thesis.
  \noindent
I hereby declare that I am the sole author of this thesis. This is a true copy of the thesis, including any required final revisions, as accepted by my examiners.

  \bigskip

  \noindent
I understand that my thesis may be made electronically available to the public.

\cleardoublepage
%\newpage

% A B S T R A C T
% ---------------

\begin{center}\textbf{Abstract}\end{center}

This thesis studies three topics in quantum computation and information: (1) The approximability of ``inherently quantum'' problems, (2) quantum proof systems, and (3) non-classical correlations in quantum systems. Our results in each area are summarized as follows.

Our first area of study concerns the approximability of computational problems which are complete for quantum complexity classes. In the classical setting, the study of approximation algorithms and hardness of approximation is one of the main research areas of theoretical computer science. Yet, little is known regarding approximability in the setting of quantum computational complexity. Our first result (joint work with Julia Kempe) is a polynomial-time approximation algorithm for dense instances of the canonical QMA-complete quantum constraint satisfaction problem, the local Hamiltonian problem. Our second result (joint work with Julia Kempe) goes in the opposite direction by first introducing a quantum generalization of the polynomial-time hierarchy. We then introduce problems which are not only complete for the second level of this hierarchy, but are in fact hard to approximate.

Our second area of study concerns quantum proof systems. Here, an interesting question which remains open despite much effort is whether a proof system with multiple unentangled quantum provers is equal in expressive power to a proof system with a single quantum prover (i.e.\ is QMA(poly) equal to QMA?). Our results here (joint work with Jamie Sikora and Sarvagya Upadhyay) study variants of this question. We first show that if each unentangled prover has logarithmic size proofs, then this is equivalent to having a single quantum prover which sends a classical proof. We then show that a variant of the class BellQMA(poly) collapses to QMA. Finally, we give an alternate proof of the fact [Harrow and Montanaro, FOCS, p. 633--642 (2010)] that the class SepQMA(m) (which is equivalent to QMA(m)) admits perfect parallel repetition. Our alternate proof is novel in that it is based on cone programming duality.

Our final area of study concerns non-classical correlations in quantum systems. Specifically, in recent years it has come to light that there appear to be genuinely quantum correlations in mixed quantum states beyond entanglement which may nevertheless prove useful from a computing and information theoretic perspective. Our first result in this area (joint work with Animesh Datta) motivates the study of such correlations by exploring possible connections to the quantum task of locking of classical correlations [DiVincenzo \emph{et al.}, PRL 92, 067902 (2004)] and the DQC1 model of mixed-state quantum computing [Knill and Laflamme, PRL 81, 5672 (1998)]. Our second result in this area introduces a novel scheme for quantifying non-classical correlations based on the use of local unitary operations. Our third result (joint work with Marco Piani, Gerardo Adesso, John Calsamiglia, Pawe\l~Horodecki, and Andreas Winter) introduces and studies a protocol through which non-classical correlations in a starting system can be ``activated'' into distillable entanglement with an ancilla system. Surprisingly, we find that, according to the non-classicality measures derived from our protocol, mixed entangled states can be ``more non-classical'' than pure entangled states. Finally, our last result (joint work with Marco Piani, Gerardo Adesso, John Calsamiglia and Pawe\l~Horodecki) continues the study of the activation protocol above by determining when the entanglement generated with the ancilla can be mapped back onto the starting state via entanglement swapping.

\cleardoublepage
%\newpage

% A C K N O W L E D G E M E N T S
% -------------------------------

\begin{center}\textbf{Acknowledgements}\end{center}

\begin{quote}\emph{I'd like to congratulate myself, and thank myself, and give myself a big pat on the back.}\\--- Dee Dee Ramone, Rock and Roll Hall of Fame induction ceremony, 2002~\cite{DeeDee}.
\end{quote}

There are many greats in this world who have the ability to inspire and support us, whether they be artists, academics, or those we hold dear. I am indebted to the following people who have played such a role during the course of my graduate studies, without whom this thesis would not have been possible.

First, I would like to thank the readers of my thesis: Richard Cleve, Debbie Leung, Ashwin Nayak, Barbara Terhal, and John Watrous. Thank you for agreeing to take on this task; I hope it does not prove too painful.

I would like to thank my thesis advisory committee, Richard Cleve, Ashwin Nayak, and John Watrous, for their guidance and feedback, particularly in times when I have been wrong, and stubbornly so at that. I have always appreciated their constructive comments, and contrary to popular belief, feel that the more embarassing the mistake revealed by their criticism, the less likely I am to repeat the blunder in the future.

I am indebted to my supervisor, Richard Cleve, for his unfailing support over the years, whether in terms of research or at a personal level. His demand for research excellence, precision, and moral steadfastness has greatly inspired and helped guide me over the years. I may (hopefully) be leaving Waterloo having gained a Ph.D., but I will be missing a good friend.

I am also ever grateful to Julia Kempe, who has in many ways acted as a second unofficial advisor for me. Her unwavering belief in me and constant push for success has had a profound effect on my development. Coupled with her sincere hospitality, I could not imagine asking for a better host for a student on exchange. In this vein, I must also thank Oded Regev, who has also played the great host and conversation partner; his input into research projects and conference talk preparations has proven invaluable.

Though neither official nor unofficial supervisors of mine, I am also indebted to Marco Piani and John Watrous. I cannot recall any instance in which either of them has turned down an opportunity to answer one of my many questions; in this and other ways, their perspectives on research have been a significant influence on me.

I would like to thank my co-authors who have been a part of the research behind this thesis: Gerardo Adesso, John Calsamiglia, Animesh Datta, Pawe\l~Horodecki, Julia Kempe, Marco Piani, Jamie Sikora, Sarvagya Upadhyay, and Andreas Winter. It has been an honor working with and learning from you.

Over my time at Waterloo, I have been lucky enough to have had a circle of great friends. At some point it was decided that, having used the words ``Hamiltonian'' and ``ground state energy'' one time too many, that I had become a physicist, and a doodle of ``photon Sev'' mysteriously appeared on my office wall. Thank you for the great times, they will be sorely missed.

I am always grateful to my family, who has tirelessly supported and believed in me. Without their love and care, I would not and could not be where I am today.

Finally, words cannot express my gratitude to my wife, Mareike M\"{u}ller. Together we lived in a ``rabbit box'' on campus for four years working on our Ph.D.'s. With any other person in such constantly close proximity, I think I would have lost my mind. But with her, it was a joy. Thank you for the wonderful experience, love, and support.

\vspace{2mm}
\noindent\emph{Financial support.} I would like to thank the following agencies and programs for their funding support over the course of my Ph.D. studies: Natural Sciences and Engineering Research Council of Canada (NSERC), NSERC Michael Smith Foreign Study Supplement program, David R. Cheriton Scholarship program, EU-Canada Exchange program, the Institute for Quantum Computing at the University of Waterloo, and the Graduate Studies Office at the University of Waterloo.

\vspace{2mm}
\noindent\emph{The reader is referred to the end of each chapter for chapter-specific acknowledgements.}

\cleardoublepage
%\newpage

% D E D I C A T I O N
% -------------------

\begin{center}\textbf{Dedication}\end{center}

To my family for their love and support, the foundation upon which all other success can be built.

\cleardoublepage
%\newpage

% T A B L E   O F   C O N T E N T S
% ---------------------------------
\renewcommand\contentsname{Table of Contents}
\tableofcontents
\cleardoublepage
\phantomsection
%\newpage

% L I S T   O F   F I G U R E S
% -----------------------------
%\addcontentsline{toc}{chapter}{List of Figures}
%\listoffigures
%\cleardoublepage
%\phantomsection		% allows hyperref to link to the correct page
%\newpage

% L I S T   O F   T A B L E S
% ---------------------------
%\addcontentsline{toc}{chapter}{List of Tables}
%\listoftables
%\cleardoublepage
%\phantomsection		% allows hyperref to link to the correct page
%\newpage

% L I S T   O F   S Y M B O L S
% -----------------------------
% To include a Nomenclature section
% \addcontentsline{toc}{chapter}{\textbf{Nomenclature}}
% \renewcommand{\nomname}{Nomenclature}

% \printglossary
% \cleardoublepage
% \phantomsection % allows hyperref to link to the correct page
% \newpage

% Change page numbering back to Arabic numerals
\pagenumbering{arabic}

%----------------------------------------------------------------------
% MAIN BODY
%----------------------------------------------------------------------
% Because this is a short document, and to reduce the number of files
% needed for this template, the chapters are not separate
% documents as suggested above, but you get the idea. If they were
% separate documents, they would each start with the \chapter command, i.e,
% do not contain \documentclass or \begin{document} and \end{document} commands.
%======================================================================
\chapter{Introduction}\label{chap:intro}
%======================================================================

\begin{quote}
    \emph{The ``paradox'' is only a conflict between reality and your feeling of what reality ``ought to be.''} --- Richard Feynman, 1964~\cite{F64}.
\end{quote}

From its earliest days, the theory of quantum mechanics puzzled its inventors. In 1935, for example, Einstein, Podolsky, and Rosen published their now famous paper rejecting quantum mechanics as a complete physical theory~\cite{EPR35}. The problem? The mathematical theory of quantum mechanics predicts certain physical phenomena which are completely at odds with our everyday understanding of the world around us. To put this into everyday language, in 1935 Schr\"{o}dinger proposed~\cite{S35} a thought experiment now known as \emph{Schr\"{o}dinger's cat}, in which under certain circumstances, a cat in a closed box is predicted by quantum mechanics to be both alive and dead, \emph{at the same time}. What could this mean? And how much did it trouble the discoverers of quantum mechanics, if it led them to ask questions such as:
\begin{quote}
    \emph{I recall that during one walk Einstein suddenly stopped, turned to me and asked whether I really believed that the moon exists only when I look at it.}\\--- Abraham Pais~\cite{einstein}.
\end{quote}
Clearly, quantum mechanics was not an easy pill to swallow, even for the fathers of the theory, many of whom rejected their beautiful child at the time.

Fast forwarding to the end of the 20th century, however, physicists and computer scientists came to a startling realization: As strange as quantum mechanics may seem, if its peculiarities could somehow be computationally harnessed, then the possibility of outperforming classical computers with so-called \emph{quantum} computers may indeed exist. In 1982, for example, physicist Richard Feynman proposed~\cite{F82} the notion of building a quantum computer in order to simulate physical quantum systems faster then apparently possible with a classical computer (see also Benioff~\cite{B80,B82_1,B82_2}). On the computer science side, in 1985 David Deutsch demonstrated a quantum algorithm which outperformed the best possible classical deterministic algorithms for what is now referred to as \emph{Deutsch's problem}~\cite{D85}. Thus, the roots of the field of quantum computation were sown. Two and a half decades later, we now have a number of good reasons for seriously devoting research effort to the field of quantum computing, which we now discuss.

\paragraph{Relevance.} We now state three reasons which, in our opinion, justify the study of quantum computation and information. The first is from an engineering-oriented perspective. Up until 2005, the speed of microprocessors increased rapidly, primarily through the brute force approach of increasing the number of transistors able to fit on a single microchip. Indeed, Intel's original Pentium P5 processor, released in 1993, had a clock speed of 60 MHz, and consisted of 3.1 million transistors~\cite{intel}. By 2005, Intel's Pentium 4E Prescott processor was up to 3.8 GHz, and packed in a whopping 169 million transistors. Yet, in 2005, something curious happened: Intel introduced its first \emph{dual-core} chip, the Pentium D Smithfield, which clocked in not at 3.8 GHz, but at a \emph{slower} 3.2 Ghz. What happened? It turns out that the brute force approach to building faster processors has a number of seemingly fundamental problems, such as excess heat production and energy loss~\cite{wierd}; however, the primary problem of interest in this thesis is that at the scale current microchip components are approaching, the pertinent laws of physics are no longer those of classical mechanics, but rather those of \emph{quantum} mechanics~\cite{berk}. This raises the natural question: \emph{Why not just build a computer which works based on the laws of quantum mechanics to begin with, i.e.\ a \emph{quantum} computer?}

The second motivation for studying quantum computing, and perhaps the most commonly cited one, came with a startling discovery: Peter Shor's quantum factoring algorithm of 1994~\cite{S97}. As whether the question of whether factoring large integers can be done efficiently on a classical computer has long been open, Shor's algorithm is in itself arguably a strong indication that the quantum computational model is indeed one deserving of study. Further, since the algorithm's inception, a number of other instances of quantum speedup have been uncovered, from Grover's algorithm for unstructured search~\cite{G96} (which yields a square root speedup for NP-complete problems over the brute force approach) to the evaluation of NAND trees~\cite{FGG08,ACRSZ07,CCJY09} to estimating quantities related to solving systems of linear equations~\cite{HHL09}, among others.

%this drew widespread attention and acted as a catalyst for process of transforming the field from a mathematical curiosity to serious endeavor based on rigorous foundation.

The reasons stated thus far, however, are rather ``selfish'', aiming to exploit quantum mechanics to serve the purpose of the computer science community. There is another view regarding the study of quantum computing which follows the converse mantra: \emph{Ask not what quantum mechanics can do for you, but what you can do for quantum mechanics.} Indeed, as computation is inherently physical, it follows that understanding the limits of quantum computation yields new tools for studying the properties of quantum mechanics itself. A primary example of this, discussed further in Section~\ref{0_sscn:LH}, is that via quantum complexity theory, one can give a rigorous proof that a significant problem in quantum mechanics, that of estimating the ground state energy of a given local Hamiltonian, cannot be solved efficiently (modulo standard complexity theoretic conjectures). Thus, the third reason for studying quantum computation is that it not only allows us to learn about the limits of computing, but also of physics itself. Moreover, there has even been a \emph{pedagogical} benefit to physics from quantum computing; apparently, there is a growing movement to replace the teaching of introductory quantum mechanics using, say, the model of the hydrogen atom, with the simpler model of quantum bits and quantum computation~\cite{AB90} (see Chapter Notes and History for Chapter 10 therein).

In closing, we have provided three motivations for studying quantum computing from engineering, computer science, and physics standpoints. In practice, however, it is of course not until a thorough study of quantum computing is undertaken that we will know the precise extent to which the field will prove relevant, particularly from a practical technological perspective. Such uncertainty lies unfortunately (or fortunately, for the adventurous type) at the very heart of the nature of our work as researchers. In the words of one of our greats:
\begin{quote}
    \emph{If we knew what it was we were doing, it wouldn't be called `research', would it?} --- Albert Einstein~\cite{einstein}.
\end{quote}
%The final reason we offer is based on the following example. Newcomers to the field of complexity theory are often puzzled as to where the class NP derives its name (as was I as an undergraduate), despite its intuitive interpretation regarding efficient proof verification. The answer is that the original definition of NP did \emph{not} have proof verification in mind, but rather \emph{non-deterministic polynomial time computation}~\cite{AB90}, a concept whose relevance is not immediately clear \emph{a priori}. One can make a similar argument regarding quantum computation: That even for skeptics of its importance \emph{a priori}, it is not until a thorough study is undertaken that we will know the extent to which quantum computing is indeed relevant or not. Although this may at first sound like a weak argument to end on, it is in fact the very embodiment of the nature of our work as researchers. In the words of one of our greats:
\paragraph{Focus of this thesis.} The field of quantum computation and information nowadays covers a broad expanse of topics, with research areas ranging from computer-science-motivated topics such as quantum algorithms and quantum proof systems, to engineering or experimental physics-oriented topics such as how to actually build a quantum computer in a lab, to theoretical-physics-motivated topics such as the limits of physical theories and the correlations between systems they allow. In this thesis, we focus on three particular areas of interest: Approximation of quantum problems, quantum proof systems, and quantum correlations. We briefly describe each area below. As each (research) chapter is intended to be as self-contained as possible, we defer more in-depth introductions to the beginning of each relevant chapter.

Our first area of interest is that of approximating quantum problems. Here, by a \emph{quantum} problem, we are referring to a computational problem which is in some sense intrinsically related to physical quantum systems in nature. From a complexity theoretic perspective, we define such problems as those which are complete for quantum complexity classes. (Relevant quantum complexity classes are defined in Section~\ref{0_scn:complexity}.) In particular, the canonical quantum problem generalizing classical constraint satisfaction which we are interested in here is called the local Hamiltonian problem, and it is complete for a quantum generalization of NP. (This problem is important from both a quantum complexity theoretic and physics point of view, and as such is given a thorough treatment in Section~\ref{0_def:localhamiltonian}.) The primary aim of our research in this area is to ask how well such problems can be \emph{approximated rigorously}, in the well-studied classical sense of \emph{approximation algorithms} and \emph{hardness of approximation}~\cite{V01}. In the quantum complexity theoretic setting, this approach to approximating physically relevant quantum problems is very much in its infancy, and it complements decades of effort by the physics community on similar problems using different tools involving heuristics (see e.g.~\cite{O11} for a brief survey). Based on joint work with Julia Kempe, Chapters~\ref{chap:approx} and~\ref{chap:hardnessapprox} discuss our results in this area, the first of which is a positive result regarding approximation algorithms for the local Hamiltonian problem, and the second of which is a negative result involving hardness of approximation for a new quantum complexity class generalizing the second level of the well-known polynomial-time hierarchy, $\st$.

Our second area of interest deals with quantum proof systems. In the classical setting, proof systems are one of the cornerstones of complexity theory, with wide-ranging impact from the theory of NP-completeness~\cite{C72,L73} to the stunning PCP theorem~\cite{AS98,ALMSS98} of the early 1990's. It is thus natural to consider studying \emph{quantum} proof systems, beginning with a quantum generalization of NP called Quantum Merlin Arthur (QMA). However, just as quantum mechanics offers new quantum phenomena to be harnessed for the purpose of computation, such phenomena now play intriguing roles in quantum proof systems. In particular, their presence can turn trivial questions in the classical setting into highly non-trivial questions in the quantum setting. For example, in the classical setting, modifying NP to allow multiple provers is straightforwardly equivalent in expressive power to the original definition of NP, since a single prover can straightforwardly simulate multiple provers. However, the question of whether QMA with multiple provers is equal to QMA is very challenging, due to the possible presence of strong correlations between quantum systems known as \emph{quantum entanglement}. In joint work with Jamie Sikora and Sarvagya Upadhyay, Chapter~\ref{chap:qmapoly} studies variants of this stubbornly open question.

Our final area of interest is the study of quantum correlations. As mentioned when discussing quantum proof systems above, a pair of quantum systems can display very strong correlations known as entanglement, which is a purely quantum phenomenon; such correlations are not possible in the classical setting. As a testament to the mysterious nature of quantum mechanics, however, after nearly a century of study, it has only been in recent years that a new type of purely quantum correlation has been identified, known simply as \emph{non-classical} correlations. Some of the biggest questions in this area are how to quantify and provide operational interpretations for such correlations, as well as to understand whether and how they may be exploited for computational gain. In joint work with Animesh Datta, Chapter~\ref{chap:dqc} studies the role of such correlations in quantum computation. Chapter~\ref{chap:localunitary} then proposes and studies a novel approach for quantifying such non-classical correlations. Finally, Chapters~\ref{chap:activation} (joint work with Marco Piani, Gerardo
Adesso, John Calsamiglia, Pawe\l~Horodecki, and Andreas Winter) and~\ref{chap:entanglementswap} (joint work with Marco Piani, Gerardo
Adesso, John Calsamiglia, and Pawe\l~Horodecki) introduce and study a new protocol which provides an operational interpretation for non-classical correlations by \emph{activating} them into entanglement.

\section{Organization} This thesis is organized as follows. In the remainder of this section, we provide background on the basics of quantum computation and information (Section~\ref{0_scn:basics}), and follow with brief technical expositions of the various topics studied in this thesis: Quantum computational complexity theory (Section~\ref{0_scn:complexity}) and quantum entanglement and non-classical correlations (Section~\ref{0_scn:nonclassicality}). %quantum algorithms (Section~\ref{0_scn:algorithms}),
%and approximation algorithms (Section~\ref{0_scn:approx}).

The remaining chapters are focused as follows. Chapters~\ref{chap:approx} and~\ref{chap:hardnessapprox} study the approximability of quantum complexity theoretic problems, such as the local Hamiltonian problem and its variants. Specifically, Chapter~\ref{chap:approx} presents our approximation algorithm for the local Hamiltonian problem. Chapter~\ref{chap:hardnessapprox} then introduces our quantum generalization of $\st$, and shows completeness and hardness of approximation for it with respect to new local Hamiltonian-like quantum covering problems we define.

Chapter~\ref{chap:qmapoly} discusses our results regarding multi-prover quantum proof systems, showing that in a certain setting, multiple quantum provers are no more powerful than a single prover.

Chapters~\ref{chap:dqc},~\ref{chap:localunitary},~\ref{chap:activation},~\ref{chap:entanglementswap} discuss non-classical correlations in quantum systems beyond entanglement. Specifically, Chapter~\ref{chap:dqc} first motivates this direction of work by studying models of quantum computing and communication where entanglement does not seem to explain the advantage gained in the quantum setting over classical computation. Chapter~\ref{chap:localunitary} then presents a novel approach for quantifying non-classical correlations in quantum systems based on local unitary operations. Chapter~\ref{chap:activation} gives an operational interpretation to such non-classical correlations by demonstrating an explicit protocol through which such correlations can be ``activated'' into entanglement. Chapter~\ref{chap:entanglementswap} further studies and attempts to extend the framework of the activation protocol of Chapter~\ref{chap:activation}.

We now begin in Section~\ref{0_scn:notation} by collecting common notation used throughout this thesis.

\section{Notation}\label{0_scn:notation}
The following notation is assumed throughout this thesis. The symbols $\complex$, $\reals$, $\ints$, and $\nats$ denote the sets of complex, real, integer, and natural numbers, respectively. For $m$ a positive integer, the notation $[m]$ indicates the set $\{1, \ldots, m\}$. The terms $\LL(\spa{X})$, $\HH(\spa{X})$, $\pos{\X}$, and $\DD(\spa{X})$ denote the sets of linear, Hermitian, positive semidefinite, and density operators acting on complex Euclidean space $\spa{X}$, respectively. The projector onto space $\spa{X}$ is denoted $\Pi_{\spa{X}}$. We sometimes use the shorthand $\B:=\complex^2$. The notation $A\succeq B$ means operator $A-B$ is positive semidefinite. The smallest (largest) eigenvalue of $A\in \HH(\X)$ is given by $\lambda_{\operatorname{min}}(A)$ ($\lambda_{\operatorname{max}}(A)$). The trace, Frobenius, and spectral (or operator) norms of $A\in \LL(\spa{X})$ are defined as
\begin{equation}
 \trnorm{A}:=\trace\left(\sqrt{A^\dagger A}\right),\quad\quad
 \fnorm{A}:=\sqrt{\trace(A^\dagger A)},\quad\quad
 \snorm{A}:=\max_{\ket{x}\in\spa{X}\mbox{ s.t. } \enorm{x}= 1}\enorm{A\ket{x}},
\end{equation}
respectively, where $:=$ denotes a definition. The $(m,n)$th entry of matrix $A$ is given by $A(m,n)$. We define the \emph{encoding} or \emph{description} of a matrix $A$ as a classical description of the entries of $A$. Specifically, let $\enc{\Delta}$ denote the number of bits used to encode $\Delta\in\complex$ to some desired precision. Then, we define the length of the encoding of $A$ by $\enc{A} := \sum_{m,n} \enc{A(m,n)}$. We extend this straightforwardly to sums of matrices; for example, $\enc{\sum_i A_i}=\sum_i\enc{A_i}$. The notation $\ve{v}$ denotes a vector. Unless otherwise noted, all logarithms are taken to base two. We sometimes use the shorthand $\poly(n)$ to mean $p(n)$ for some {fixed} polynomial $p$.

\section{Linear algebra}\label{0_scn:linalg}

We now briefly review basic concepts from linear algebra crucial to the content of this thesis. Parts of this section follow the course notes of Watrous~\cite{W08_2,W08_3}; the reader is also referred to the text of Horn and Johnson~\cite{HJ90} for further details. Those familiar with basic linear algebra can safely skim over this section or refer to it as needed.

\paragraph{Complex Euclidean spaces.} The setting in which all the excitement takes place is that of a complex Euclidean space $\spa{X}$, defined as follows. Let $\Sigma$ be a finite, non-empty set. Consider the set of all functions from $\Sigma$ to the complex numbers $\complex$, denoted $\complex^\Sigma$. Then, define for any $\ve{u},\ve{v}\in\complex^\Sigma$ and $\alpha\in\complex$ the addition and scalar multiplication operations in the standard way: The addition $\ve{u}+\ve{v}\in\complex^\Sigma$ obeys $(\ve{u}+\ve{v})(i)=\ve{u}(i)+\ve{v}(i)$ for all $i\in \Sigma$, and scalar multiplication $\alpha\ve{u}\in\complex^\Sigma$ obeys $(\alpha \ve{u})(i)=\alpha\ve{u}(i)$ for all $i\in\Sigma$. Then, the set $\complex^\Sigma$ along with these operations is known as a complex Euclidean space, which we denote as $\spa{X}$. The dimension of $\spa{X}$ is given by $\abs{{\Sigma}}$, the cardinality of $\Sigma$. For concreteness, we henceforth assume $\Sigma=[d]$ for $[d]:=\set{1,2,\ldots,d}$, and use the simplified notation $\complex^\Sigma=\complex^d$.

We think of (column) vectors $\ve{v}\in\spa{X}$ as $d$-tuples, i.e.\
\begin{equation}
    \ve{v}=\left(
             \begin{array}{c}
               v(1) \\
               \vdots \\
               v(d) \\
             \end{array}
           \right)
\end{equation}
for $v(i)\in\complex$. In quantum computation, $\ve{v}$ is commonly denoted using $\ket{v}$. Here, $\ket{\cdot}$ is called Dirac notation, also sometimes affectionately known as ``dog-houses'' for vectors~\cite{Wj99}. A remark about vector notation: Generally, our choice of notation $\ve{v}$ or $\ket{v}$ will be dictated by context. For example, when a vector is to be interpreted as a quantum state, we shall use Dirac notation $\ket{v}$; otherwise, we typically revert to the notation $\ve{v}$. An exception to this rule, even in purely linearly algebraic contexts, is when it is more convenient to use Dirac notation, such as when vectors are to be labeled by complicated expressions. In much of the introductory discussion on linear algebra that follows, we assume $\ve{v}=\ket{v}$ holds for the pedagogic purpose of familiarizing the reader with Dirac notation. However, in general this equality is not assumed to hold; for example, the zero vector $\ve{0}$ is not equal to $\ket{0}=(1,0)^T$. We hope the distinction will be clear from context.

% Throughout this thesis, we use Dirac notation if a vector is to be interpreted as a quantum state; otherwise, we revert to the notation $\ve{v}$. An exception to this rule is in this section, where we sometimes intentionally use both $\ve{v}$ and $\ket{v}$ to denote the same vector; this for pedagogical reasons to facilitate the adoption of Dirac notation. In general, whether $\ve{v}$ and $\ket{v}$ are referring to the same vector is to be dictated by context.

Continuing, the conjugate transpose of $\ve{v}$ is denoted $\ve{v}^\dagger$, or $\bra{v}$ in Dirac notation, and is the row vector%i.e.\ $\bra{\psi}:=\overline{\ve{v}}^T$ for $\overline{\ve{v}}$ the entry-wise complex conjugate of $\ve{v}$ and $T$ the transpose, such that
\begin{equation}
    \ve{v}^\dagger = \bra{v} = \left( \overline{v(1)}, \overline{v(2)},\ldots, \overline{v(d)} \right),\label{0_eqn:bra}
\end{equation}
for $\overline{a}$ the complex conjugate of $a\in\complex$.

\paragraph{Vector norms.} For any two vectors $\ve{v},\ve{w}\in\spa{X}$, we define their inner product as
\begin{equation}
    \langle\ve{v},\ve{w}\rangle=\ve{v}^\dagger\ve{w}=\braket{v}{w}=\sum_{i=1}^d \overline{v(i)}w(i).
\end{equation}
Then, we measure the length of $\ve{v}\in\complex^d$ via the \emph{Euclidean norm}, defined as $\enorm{\ve{v}}=\sqrt{\langle\ve{v},\ve{v}\rangle}$. The Euclidean norm is just one of an entire class of norms known as \emph{p}-norms, defined for $p\in[1,\infty)$ such that
\begin{equation}
    \norm{\ve{v}}_p := \left(\sum_{i=1}^d \abs{v(i)}^p\right)^{\frac{1}{p}},
\end{equation}
and for $p=\infty$ as
$
    \norm{\ve{v}}_\infty := \left(\max_{i\in[d]}\abs{v(i)}\right).
$
Note that setting $p=2$ yields the Euclidean norm. The $p$-norms have the following properties:
\begin{enumerate}
    \item (Positive scalability) $\norm{a\ve{v}}_p= \abs{a}\pnorm{\ve{v}}$ for $a\in\complex$.
    \item (Triangle inequality) For any $\ve{v},\ve{w}\in\spa{X}$, $\norm{\ve{v}+\ve{w}}_p\leq \norm{\ve{v}}_p+\norm{\ve{w}}_q$.
    \item For $\ve{v}\in\spa{X}$, if $\norm{\ve{v}}=0$, then $\ve{v}=\ve{0}$, where $\ve{0}$ denotes the zero vector whose entries are all zero.
\end{enumerate}
From the first two properties, we conclude that for all $\ve{v}\in\spa{X}$, $\pnorm{\ve{v}}\geq 0$, since
\begin{equation}
    0 = \abs{0}\pnorm{\ve{0}}=\pnorm{0\cdot\ve{0}}=\pnorm{\ve{0}}=\norm{\ve{v}-\ve{v}}_p \leq \norm{\ve{v}}_p+\pnorm{-\ve{v}}\leq 2\pnorm{\ve{v}}.
\end{equation}

A useful inequality regarding inner products is the H\"{o}lder inequality, which states that for any $\ve{v},\ve{w}\in\X$,
\begin{equation}
    \abs{\langle\ve{v},\ve{w}\rangle}\leq \pnorm{\ve{v}}\norm{\ve{w}}_q
\end{equation}
for $\frac{1}{p} + \frac{1}{q}=1$. (For $p=1$, $q=\infty$.) When $p=q=2$, we recover the Cauchy-Schwarz inequality. As a testament to the applicability of the latter, we show that $\onorm{v}\leq\sqrt{d}\enorm{v}$, a frequently useful inequality. Let $\ve{j}$ be the $d$-dimensional all-ones vector and $\abs{\ve{v}}$ the entry-wise absolute value of $\ve{v}$. Then:
\begin{equation}
    \onorm{v}=\langle\ve{j},\abs{\ve{v}}\rangle\leq \abs{\langle\ve{j},\abs{\ve{v}}\rangle}\leq \enorm{\ve{j}}\enorm{\abs{\ve{v}}}=\sqrt{d}\enorm{\ve{v}}.
\end{equation}
It also holds that $\enorm{\ve{v}}\leq\sqrt{d}\snorm{\ve{v}}$, and conversely that $\onorm{\ve{v}}\geq \enorm{\ve{v}}\geq\snorm{\ve{v}}$.

\paragraph{Orthonormal bases.} A set of vectors $\set{\ve{v}_i}\subseteq\spa{X}$ is \emph{orthogonal} if for all $i\neq j$, $\langle\ve{v}_i,\ve{w}_j\rangle=0$, and \emph{orthonormal} if $\langle\ve{v}_i,\ve{w}_j\rangle=\delta_{ij}$. Here, $\delta_{ij}$ is the Kroenecker delta, whose value is $1$ if $i=j$ and $0$ otherwise. Every complex Euclidean space $\spa{X}$ of dimension $d$ has an orthonormal \emph{basis} consisting of $d$ elements, where a basis is a set of vectors $\set{\ve{v}_i}\subseteq\spa{X}$ such that any $\ve{w}\in\spa{X}$ can be expressed as
\begin{equation}
    \ve{w}=\sum_{i=1}^d \alpha_i\ve{v}_i
\end{equation}
for some $\set{\alpha_i}\subseteq \complex$. A common basis for $\spa{X}$ is the \emph{computational} or \emph{standard} basis $\set{\ve{e}_i}$, defined such that $\ve{e}_i(j)=\delta_{ij}$. In Dirac notation, we frequently denote this basis simply as $\set{\ket{i}}_{i=1}^d$.

\paragraph{Linear operators and matrices.} Given two complex Euclidean spaces $\spa{X}$ and $\spa{Y}$, a \emph{linear operator} or \emph{linear map} from $\spa{X}$ to $\spa{Y}$ is a map $\Phi:\spa{X}\mapsto\spa{Y}$ with the property that
\begin{equation}
\Phi\left(\sum_i\alpha_i\ve{v}_i\right)=\sum_i\alpha_i \Phi(\ve{v}_i),
\end{equation}
where $\set{\ve{v_i}}\subseteq\spa{X}$. The set of all such linear maps from $\spa{X}$ to $\spa{Y}$ is denoted $\mathcal{L}(\spa{X},\spa{Y})$, which when coupled with operations for addition and scalar multiplication in the standard way, yields a vector space of dimension $\dim(\spa{X})\dim(\spa{Y})$. Here, $\dim(\spa{X})$ is the dimension of $\spa{X}$. For brevity, we use the shorthand $\mathcal{L}(\spa{X})$ to mean $\mathcal{L}(\spa{X},\spa{X})$.

A convenient way to represent and study linear maps is via their matrix representation. Here, an $m\times n$ \emph{matrix} $A$ is a two-dimensional array of complex numbers whose $(i,j)$th entry is denoted $A(i,j)\in\complex$ for $i\in[m]$, $j\in[n]$. To represent a linear map $\Phi:\complex^n\mapsto\complex^m$ as an $m\times n$ matrix $A_\Phi$, recall that the action of a map is completely specified by its action on a basis. Specifically, the $i$th column of $A_\Phi$ is given by $\Phi(\ve{e}_i)$ for $\set{\ve{e}_i}$ the standard basis for $\complex^n$, or
\begin{equation}
    A_\Phi = \left[
               \begin{array}{cccc}
                 \Phi(\ve{e}_1), \Phi(\ve{e}_2), \ldots, \Phi(\ve{e}_m) \\
               \end{array}
             \right].
\end{equation}
Recovering $\Phi$ from $A_\Phi$ thus also follows immediately from this view. When we henceforth discuss $A\in\LL(\spa{X})$, we are implicitly referring to the matrix representation of map $A$.

The product $AB$ of two $d\times d$ matrices $A$ and $B$ is defined such that
\begin{equation}
    AB(i,j)=\langle\overline{\ve{r}}^A_{i},\ve{c}^{B}_{j}\rangle
\end{equation}
for $\ve{r}^A_{i}$ the $i$th row of $A$ and $\ve{c}^B_{j}$ the $j$th column of $B$. In general, it is not true that $AB=BA$. The difference $AB-BA$ is called the \emph{commutator} $[A,B]$ of $A$ and $B$, and the \emph{anti-commutator} is $\set{A,B}=AB+BA$.

The \emph{rank} of $A\in\LL(\X,\Y)$ is the dimension of its \emph{image}, where the latter is defined as $\operatorname{Im}(A):=\set{\ve{y}\in\Y\mid \ve{y}=A\ve{x}\mbox{ for some }\ve{x}\in\X}$. The rank satisfies
\begin{equation}
    \operatorname{rank}(AB)\leq\min\set{\operatorname{rank}(A),\operatorname{rank}(B)}.
\end{equation}
Defining the \emph{null space} or \emph{kernel} of $A\in\LL(\X)$ as $\operatorname{Ker}(A):=\set{\ve{v}\in\X\mid A\ve{v}=0}$, it holds that $\operatorname{dim}(\operatorname{Ker}(A))+\operatorname{dim}(\operatorname{Im}(A))=d$.

\paragraph{Eigenvalues and eigenvectors.} For any $A\in \LL(\spa{X})$, we say $\ve{v}$ is an \emph{eigenvector} of $A$ with \emph{eigenvalue} $\lambda$ if $\ve{v}\neq\ve{0}$ and $A\ve{v}=\lambda A$. The multiset of eigenvalues of $A$ (with multiplicity) is known as its \emph{spectrum}. The eigenvalues of $A$ arise as the roots of the degree-$d$ \emph{characteristic polynomial} of $A$, $p_A$, defined such that
\begin{equation}
    p_A(x) := \det(x I-A),
\end{equation}
where $I(i,j):=\delta_{ij}$ is the \emph{Identity} matrix and $\det$ is the \emph{determinant}. One way to define the latter, known as the Laplace expansion, is via the recursive definition
\begin{equation}
    \det(A) = \sum_{j=1}^d (-1)^{i+j}A(i,j)\operatorname{det}(A_{ij}).
\end{equation}
Here, $A_{ij}$ is the matrix obtained from $A$ by deleting row $i$ and column $j$, and we define the base case of this recursion (i.e.\ a $1\times 1$ matrix $[c]$) as $\operatorname{det}([c])=c$. This equation holds for any $i\in[d]$.

\paragraph{Matrix operations.} A number of operations on matrices $A\in\spa{X}$ arise repeatedly in quantum computing. First, the \emph{complex conjugate}, \emph{transpose} and \emph{adjoint} operations are respectively defined via
\begin{equation}
    \overline{A}(i,j):=\overline{(A(i,j))}\quad\quad\quad A^T(i,j):=A(j,i)\quad\quad\quad A^\dagger:=(\overline{A})^T.
\end{equation}
These operations apply to vectors as well so that $\bra{v}$, defined in Equation~(\ref{0_eqn:bra}), is simply $\ket{v}^\dagger$.

The $\emph{trace}$ of $A$ is a linear function defined as $\trace(A):=\sum_{i=1}^dA(i,i)=\sum_{i=1}^d\lambda_i(A)$, where $\set{\lambda_i(A)}\subseteq\complex$ are the eigenvalues of $A$. Henceforth, when clear from context, we simply write $\lambda_i$ for the latter. The trace has the useful property of being \emph{cyclic}, i.e.\ $\trace(ABC)=\trace(CAB)$. With the trace in hand, we can define an inner product on $\LL(\X)$ as $\langle A,B\rangle = \trace(A^\dagger B)$.

The \emph{tensor product} is an important operation through which joint quantum systems can be described. Specifically, for complex Euclidean spaces $\X$ and $\Y$, their tensor product is $\X\otimes \Y=\complex^{d_x\times d_y}$. For vectors $\ve{u}\in\X$ and $\ve{y}\in\Y$, we define for all $i\in[d_x]$ and $j\in[d_y]$
\begin{equation}
    (\ve{u}\otimes \ve{v})(i,j):=u(i)v(j).
\end{equation}
For linear operators $A\in\LL(\X)$, $B\in\LL(\Y)$, $A\otimes B$ yields a complex matrix whose index sets are given by $([d_x]\times [d_y],[d_x]\times [d_y])$, such that
\begin{equation}
    (A\otimes B)((i_1,j_1),(i_2,j_2)):=A(i_1,i_2)B(j_1,j_2)
\end{equation}
for all $i_1,i_2\in[d_x]$ and $j_1,j_2\in[d_y]$. The tensor product has the following properties for any $A,C\in \X$, $B,D\in\Y$, $c\in\complex$:
\begin{eqnarray}
    (A+C)\otimes B &=& A\otimes B + C\otimes B\\
    A\otimes (B+D) &=& A\otimes B + A\otimes D\\
    c(A\otimes B) &=& (cA)\otimes B = A\otimes (cB)\\
    (A\otimes B)(C\otimes D) &=& AC\otimes BD\\
    \trace(A\otimes B) &=& \trace(A)\trace(B)\\
    (A\otimes B)^\dagger &=& A^\dagger\otimes B^\dagger.
\end{eqnarray}
These properties hold analogously in the vector setting.

Given the composition of two spaces $\X$ and $\Y$ via the tensor product, we also require an operation in the reverse direction for removing one of these spaces. For this, we define the linear \emph{partial trace} map. Specifically, for $A\otimes B\in \LL(\X\otimes \Y)$, the partial trace $\trace_\X (A\otimes B)\in\Y$ is defined as
\begin{equation}
    \trace_\X (A\otimes B):=\trace(A)B.
\end{equation}
Alternatively, for any orthonormal basis $\set{\ve{v}_i}_{i=1}^d$ for $\X$, we can write for $A\in\LL(\X\otimes \Y)$
\begin{equation}
    \trace_\X(A) = \sum_{i=1}^d \left(\ve{v}_i^\dagger \otimes I\right) A \left(\ve{v}_i \otimes I\right).
\end{equation}

\paragraph{Special classes of operators.} A few classes of linear operators play important roles in quantum computing. The first of these is the class of \emph{Hermitian} operators $\HH(\X)\subseteq\LL(\X)$, defined as the set of $A\in\LL(\X)$ satisfying $A^\dagger =A$. As the set of Hermitian operators is closed under addition and real scalar multiplication, and since $\langle A,B\rangle\in\reals$ for all $A,B\in\HH(\X)$, it follows that $\HH(\X)$ forms a real inner product space of dimension $d^2$.

The eigenvalues of Hermitian operators are real. If the eigenvalues of Hermitian $A$ are in $\set{0,1}$, then equivalently $A^2=A$, and $A$ is called an (orthogonal) projection. (Non-Hermitian $A$ satisfying $A^2=A$ are called \emph{oblique} projections, and are not used here.)

More generally, a Hermitian matrix $A\in\HH(\X)$ whose eigenvalues are all non-negative is called  \emph{positive semidefinite}, denoted $A\succeq 0$ (more generally, the notation $A\succeq B$ means $A-B\succeq 0$). Positive semidefinite matrices $A\in\HH(\X)$ can equivalently be characterized as follows:
\begin{itemize}
    \item $\ve{x}^\dagger A\ve{x}\geq 0$ for all $\ve{x}\in\X$.
    \item $A = B^\dagger B$ for some $B\in\LL(\X)$.
\end{itemize}
The set of positive semidefinite operators acting on $\X$ is denoted $\pos{\X}$.

Next, a \emph{unitary} operator $U\in\UU(\X)$ is defined as satisfying $UU^\dagger=U^\dagger U=I$. The eigenvalues of $U$ are complex numbers of modulus $1$. All unitary operators preserve the length of any vector $\ve{v}$, i.e.\ $\langle U\ve{v} ,U \ve{v}\rangle =\langle \ve{v} , \ve{v}\rangle$. More generally, any $U\in \LL(\X,\Y)$ with $U^\dagger U=I_{\X}$ is called an \emph{isometry}.

Hermitian, positive semidefinite, and unitary matrices are in fact all special cases of \emph{normal} matrices $A$, defined such that $AA^\dagger=A^\dagger A$. Normal matrices are important due to the Spectral Decomposition theorem, which we discuss next.

\paragraph{Matrix decompositions.} An extremely useful property of normal matrices $A$ acting on $\spa{X}$ is that they can be written in terms of their \emph{spectral decomposition}, i.e.\
\begin{equation}
    A = \sum_{i=1}^d\lambda_i\ket{\lambda_i}\bra{\lambda_i}= UDU^\dagger,
\end{equation}
where recall $\lambda_i$ are the eigenvalues of $A$, the set $\set{\ket{\lambda_i}}_{i=1}^d$ is a corresponding orthonormal set of eigenvectors of $A$, $D=\operatorname{diag}(\set{\lambda_i})$ is a diagonal operator with entries $D(i,i)=\lambda_i$, and $U$ is a unitary matrix whose $i$th column is $\ket{\lambda_i}$. Here we have switched to Dirac notation to highlight, in our opinion, one of its strengths --- the ability to label vectors easily by complicated expressions. Note that if $\lambda_i\neq \lambda_j$ for all $i,j$, then the set of eigenvectors above is unique.

A common problem in quantum mechanics is to analyze the spectrum of a sum of two matrices $A,B\in\LL(\X)$. In general, this is a difficult problem. However, if the matrices are normal and they \emph{commute}, i.e.\ $[A,B]=0$, then this task is made easier by the fact that $A$ and $B$ must \emph{simultaneously diagonalize}. In other words for normal $A$ and $B$, $[A,B]=0$ if and only if there exists an orthonormal basis $\set{\ket{b_i}}\subseteq\X$ such that
\begin{equation}
    A = \sum_{i=1}^d \lambda_i(A)\ketbra{b_i}{b_i},\quad\quad\quad\quad B = \sum_{i=1}^d \lambda_i(B)\ketbra{b_i}{b_i}.
\end{equation}

While the spectral decomposition holds only for normal matrices, a more general decomposition known as the \emph{singular value} decomposition exists even for non-square matrices. The latter says that for any $d_y\times d_x$ matrix $A\in\LL(\X,\Y)$, we have
\begin{equation}
    A = U D V^\dagger
\end{equation}
for $d_y\times d_y$ unitary $U$, $d_x\times d_x$ unitary $V$, and $d_y\times d_x$ diagonal matrix $D$ whose entries $D(i,i)$ are non-negative real numbers called the \emph{singular values} of $A$.

\paragraph{Operator functions.} With the spectral decomposition in hand, we can now apply functions $f:\complex\mapsto\complex$ to normal operators $A\in\X$ as follows. Let $A$ have spectral decomposition $A = \sum_{i=1}^d\lambda_i\ket{\lambda_i}\bra{\lambda_i}$. Then, assuming $\set{\lambda_i}$ is a subset of the domain of $f$,
\begin{equation}
   f(A) := \sum_{i=1}^df(\lambda_i)\ket{\lambda_i}\bra{\lambda_i}.
\end{equation}

Three common functions $f$ encountered in this thesis are $f(x)=e^x$, $f(x)=\log x$, and $f(x)=\sqrt{x}$, the operator functions of which are denoted as $e^A$, $\log A$, and $\sqrt{A}$, respectively. Here, the logarithm is taken to base two.

\paragraph{Operator norms.} Similar to the $p$-norms we defined for vectors, a useful class of norms for measuring the ``length'' or ``magnitude'' of a matrix are the Schatten $p$-norms. Their definition is simple: For any $p\in[1,\infty]$, let $\ve{\sigma}(A)$ denote the vector of singular values of $A\in\X$. Then,
\begin{equation}
    \pnorm{A}:=\pnorm{\ve{\sigma}(A)}.
\end{equation}
A particularly nice aspect of this definition is that for Hermitian operators, $\sigma_i(A)=\abs{\lambda_i(A)}$. Moreover, properties of the vector $p$-norms carry over straightforwardly to the Schatten $p$-norms, such as the H\"{o}lder inequality, positive scalability, and the triangle inequality.

Some further important properties of the $p$-norms for any $A\in\LL(\X)$ are:
\begin{enumerate}
    \item $\pnorm{A}=\pnorm{\overline{A}}=\pnorm{A^T}$, from which also $\pnorm{A}=\pnorm{A^\dagger}$.
    \item (Invariance under isometries) $\pnorm{UAV^\dagger}=\pnorm{A}$ for any isometries $U$ and $V$ for which $UAV^\dagger$ is well-defined.
    \item $\pnorm{ABC}\leq\snorm{A}\pnorm{B}\snorm{C}$.
    \item (Submultiplicativity) $\pnorm{AB}\leq\pnorm{A}\pnorm{B}$. This follows from Property 3.
\end{enumerate}

There are three specific values of $p$ of interest here: $p=1$, $p=2$, and $p=\infty$. They correspond to the \emph{trace}, \emph{Frobenius}, and \emph {spectral} (or \emph{operator}) norms, respectively, and can alternatively be defined as
\begin{equation}
 \trnorm{A}:=\trace\left(\sqrt{A^\dagger A}\right),\quad\quad
 \fnorm{A}:=\sqrt{\trace(A^\dagger A)},\quad\quad
 \snorm{A}:=\max_{\ket{x}\in\spa{X}\mbox{ s.t. } \enorm{x}= 1}\enorm{A\ket{x}}.
\end{equation}
The trace norm has two further properties of interest: First, it is non-increasing under the partial trace, meaning that for $A\in \LL(\X\otimes \Y)$, $\trnorm{\trace_{\Y}(A)}\leq\trnorm{A}$. Second, for unit vectors $\ve{u},\ve{v}\in\X$ we have
\begin{equation}\label{0_eqn:traceeuclid}
    \trnorm{\ve{u}\ve{u}^\dagger - \ve{v}\ve{v}^\dagger}=2\sqrt{1-\abs{\langle\ve{u},\ve{v}\rangle}^2}\leq 2\enorm{\ve{u}-\ve{v}}.
\end{equation}
The second inequality follows by expanding the definition of the Euclidean norm and applying the identity $1-x^2\leq 2(1-x)$. The first equality follows~\cite{W08_2} by noting that $A:=\ve{u}\ve{u}^\dagger - \ve{v}\ve{v}^\dagger$ is Hermitian, and so its trace norm is a function of the absolute values of its eigenvalues, which we now analyze. Since $\operatorname{rank}(A)\leq 2$ and $\trace(A)=0$, its spectrum must be $\set{\lambda,-\lambda, 0,\ldots,0}$ for some $\lambda \in\reals$. Thus, $\trace(A^2)=2\lambda^2$. However, a direct evaluation of $\trace(A^2)$ from the definition of $A$ also reveals $\trace(A^2)=2-2\abs{\langle\ve{u},\ve{v}\rangle}^2$. Combining these two expressions for $\trace(A^2)$, the claim follows.

\paragraph{Linear super-operators.} We have discussed (linear) operators $\Phi:\X\mapsto\X$ and $\Phi:\Y\mapsto\Y$. Moving a step up the ladder, we can also discuss linear operators $\Phi:\LL(\X)\mapsto\LL(\Y)$. Such maps are called linear \emph{super-operators}. Bestowed with the standard definitions of addition and scalar multiplication, the set of super-operators, denoted $T(\X,\Y)$, forms a linear space. The tensor product operation applies analogously to super-operators as it did to operators.

The \emph{adjoint} of super-operator $\Phi\in T(\X,\Y)$, $\Phi^*\in T(\Y,\X)$, is uniquely defined by the equation
\begin{equation}
    \langle A,\Phi(B)\rangle=\langle \Phi^*(A),B\rangle,
\end{equation}
which holds for all $B\in\X$ and $A\in\Y$.

\paragraph{Special classes of super-operators.} From a quantum computing perspective, we are most interested in super-operators which are {trace-preserving} and {completely positive} (TPCP). A \emph{trace-preserving} super-operator $\Phi\in T(\X,\Y)$ is defined as satisfying
\begin{equation}
    \trace(A)=\trace(\Phi(A))
\end{equation}
for any $A\in\LL(\X)$. To define a completely positive map, we first define a \emph{positive} map $\Phi\in T(\X,\Y)$ as satisfying $\Phi(A)\succeq 0$ for any $A\in\LL(\X)$ such that $A\succeq 0$. Then, a map $\Phi\in T(\X,\Y)$ is called \emph{completely positive} if $I_{\LL(\X)}\otimes \Phi$ is a positive map. Intuitively, a completely positive map $\Phi$ sends positive semidefinite operators to positive semidefinite operators, even if $\Phi$ acts on only part of a larger composite system.

\paragraph{Matrix representations of super-operators.} Just as we discussed a matrix representation for linear operators, there are a number of useful matrix representations for linear super-operators. (See the notes of Watrous~\cite{W08_5} for an excellent exposition.) Here, we discuss two particular representations used in this thesis, known as the \emph{Stinespring} and \emph{Kraus} representations.

The Strinespring representation lends a nice interpretation to admissible quantum maps later. Specifically, it says that the action of any TPCP map $\Phi\in T(\X,\Y)$ on arbitrary $X\in \LL(\X)$ can be written as
\begin{equation}
    \Phi(X)=\trace_{\Z}(AXA^\dagger),
\end{equation}
for some complex Euclidean space $\Z$ and some linear isometry $A\in \LL(\X,\Y\otimes \Z)$. Moreover, $\dim(\Z)$ can be taken as $\dim(\Z)\leq\dim(\X)\dim(\Y)$. In the context of quantum computation, it will be particularly useful to note that this is equivalent~\cite{AKN98} to saying $\Phi(X)$ can be written as, for $\Y=\Y_1=\Y_2$,
\begin{equation}
    \Phi(X)=\trace_{\X\otimes\Y_2}\left[U(X_\X\otimes\ketbra{0}{0}_{\Y_1\otimes \Y_2})U^\dagger\right],\label{0_eqn:church}
\end{equation}
for some unitary $U\in\UU(\X\otimes \Y_1\otimes\Y_2)$.

We now define the Kraus representation, which is sometimes also known as the \emph{operator-sum} representation~\cite{NC00}. The Kraus representation says that any TPCP map $\Phi\in T(\X,\Y)$ can be expressed in terms of a set of \emph{Kraus operators} $\set{K_i}_{i=1}^k\subseteq \LL(\X,\Y)$ such that
\begin{equation}
    \Phi(X)=\sum_{i=1}^k K_i X K_i^\dagger,
\end{equation}
where $\sum_{i=1}^k K_i^\dagger K_i=I_{\X}$ and $k\leq \dim(\X)\dim(\Y)$.

\section{Basics of quantum computation}\label{0_scn:basics}

We now introduce the basics of quantum computation. For further details, the interested reader is referred to the texts of Nielsen and Chuang~\cite{NC00}, Kitaev, Shen, and Vyalyi~\cite{KSV02}, and Kaye, Laflamme, and Mosca~\cite{KLM07}. From a computer scientist's perspective, note that the primary background required is \emph{not} quantum physics, but rather linear algebra~\cite{HJ90}. This is because, just as with any (say) sports game, in order to play the game, you simply have to learn the \emph{rules} of the game. Quantum mechanics, in particular, has four simple rules, and they are all based on linear algebra. These rules govern the following four intuitively logical concepts: How a quantum state is described, how does one ``read'' or measure a quantum state, what operations can be performed on a quantum state, and finally, how does one describe multiple quantum systems jointly.

\subsection{Describing quantum states}\label{0_sscn:describe}
Let $\spa{X}$ denote a complex Euclidean space. Then, in a nutshell, any $\rho\in\pos{\spa{X}}$ with trace $1$ describes a valid quantum state. Let us now provide some intuition as to how this statement comes about.

In classical computing, the basic unit of information is a bit, which takes on values in the set $\set{0,1}$. One can equivalently encode a bit using the set $\set{\ket{0},\ket{1}}$, where $\set{\ket{0},\ket{1}}\subseteq\complex^2$ is the standard basis for $\complex^2$, i.e.\ $\ket{0}=(1,0)^T$ and $\ket{1}=(0,1)^T$. The key difference between classical bits and qubits is that in the quantum world, one can interpolate between the two discrete values $\ket{0}$ and $\ket{1}$ by taking a \emph{superposition}, i.e.\ the vector
\begin{equation}
    \ket{\psi}=\alpha\ket{0}+\beta\ket{1}
\end{equation}
describes a valid quantum state if $\abs{\alpha}^2+\abs{\beta}^2=1$. In other words, any unit vector in $\complex^2$ describes a quantum bit, or \emph{qubit}.

More generally, assume $\spa{X}$ has dimension $d$. Then, any unit vector $\ket{\psi}\in\spa{X}$ describes a $d$-dimensional quantum state, sometimes dubbed a qu\emph{d}it. Such vectors are called \emph{pure} states, and do not yet capture the set of all possible $d$-dimensional quantum states. To complete the picture, we simply allow probabilistic mixtures of such pure states, more generally referred to as \emph{mixed} states. Such probabilistic mixtures are described in the following straightforward manner, known as the \emph{density matrix} formalism.

Associated with any probabilistic mixture is an \emph{ensemble},
\begin{equation}
\set{\set{p_i}_{i=1}^k,\set{\ketbra{\psi_i}{\psi_i}}_{i=1}^k},
\end{equation}
where $\set{p_i}_{i=1}^k$ forms a probability distribution and $\set{\ket{\psi_i}}\subseteq\spa{X}$ is a set of unit vectors. The corresponding mixed quantum state $\rho$ is thus:
\begin{equation}
    \rho = \sum_{i=1}^k p_i \ketbra{\psi_i}{\psi_i}.\label{0_eqn:density}
\end{equation}
Here, $\rho$ is called the \emph{density matrix} describing the underlying quantum state. We denote the set of density operators acting on $\X$ as $\DD(\X)$.

Let us now tie this back into the statement made at the beginning of this subsection. Note that since in Equation~(\ref{0_eqn:density}), $\rho$ is a non-negative sum of positive semidefinite operators, we must have $\rho\succeq 0$. Moreover, by applying the cyclic property of the trace, we have $\trace(\rho)=1$, as claimed. Indeed, based on the exposition above, we can now intuitively see why any $\rho\in\mathcal{X}$ with $\rho\succeq 0$ and $\trace(\rho)=1$ describes a valid quantum state --- simply take the spectral decomposition of $\rho$ to recover an ensemble $\set{\set{p_i}_{i=1}^k,\set{\ketbra{\psi_i}{\psi_i}}_{i=1}^k}$.

We remark that although here we have attempted to present a simple exposition of how quantum states are classically described, in reality the precise interpretation of what such a classical description means is highly non-trivial and continues to be debated after decades of research.

\subsection{Measuring quantum states}\label{0_sscn:measurement}
Now that we have a mathematical description of quantum states, we require a formalism for modeling how a quantum state is ``observed'', or measured. For this, let $\rho\in\DD(\X)$ be a density matrix. Then, a quantum \emph{measurement} is formalized by a set of operators $\Pi:=\set{M_i}\subseteq\LL(\X)$ satisfying
\begin{equation}
    \sum_i M_i^\dagger M_i=I,
\end{equation}
where the latter is called the \emph{completeness relation}. The act of measuring $\rho$ with $\Pi$ is in general an inherently probabilistic process, \emph{even if} $\rho$ corresponds to a pure state (unlike in the classical case of bits). Specifically, when measuring $\rho$ with respect to $\Pi$, we obtain outcome $i$ with probability given by
\begin{equation}
    \pr(\mbox{outcome }i | \rho)=\trace(M_i\rho M_i^\dagger).
\end{equation}
Once a particular outcome $i$ is observed, the state $\rho$ ``collapses'' to a new state $\rho'$ consistent with this outcome, i.e.\
\begin{equation}
    \rho' = \frac{M_i\rho M_i^\dagger}{\pr(\mbox{outcome }i | \rho)}.
\end{equation}
Note that the denominator above serves the role of renormalizing $\rho'$ so that $\trace(\rho')=1$.

We have thus far described general measurements. Often, we are interested in the special case when each $M_i$ is an orthogonal projection operator (not necessarily of rank one), such that $M_iM_j=\delta_{ij}M_i$. Such measurements are called \emph{projective} or \emph{von Neumann} measurements. A common way to represent a projective measurement is via an \emph{observable} $M\in\HH(\X)$. Via the spectral decomposition, we can write $M=\sum_i\lambda_i\Pi_i$, where $\lambda_i\neq \lambda_j$ for $i\neq j$ and each $\Pi_i$ is a projection operator (of rank possibly greater than one). Then, each eigenvalue $\lambda_i$ corresponds to a distinct label for a measurement outcome, and the measurement operators are $M_i=\Pi_i$. An advantage of using observables is that the expected value of the measurement, denoted $\mathbb{E}_M$, takes a very simple form:
\begin{equation}
    \mathbb{E}_M(\rho)=\sum_i\lambda_i\pr(\mbox{outcome }i | \rho)=\sum_i\lambda_i \trace(\Pi_i\rho\Pi_i^\dagger)=\sum_i\lambda_i \trace(\Pi_i\rho)=\trace(M\rho).
\end{equation}

Finally, note that the framework above for general measurements $\Pi=\set{M_i}$ allows one to determine both the probability of outcome $i$, as well as the output state of the measurement process once $i$ is read. If we only care about the former, as is the case in situations where the quantum system is only to be measured once and subsequently discarded, then this formalism is often simplified by defining positive semidefinite $E_i:=M_i^\dagger M_i$ with $\sum_i E_i = I$. We hence have:
\begin{equation}
    \pr(\mbox{outcome }i | \rho)=\trace(M_i\rho M_i^\dagger)=\trace(M_i^\dagger M_i\rho)=\trace(E_i\rho).
\end{equation}
The set $\set{E_i}$ is called a \emph{Positive Operator-Valued Measure (POVM)}. An advantage of using POVMs, for example, is that since the POVM elements $E_i$ are positive semidefinite, optimizations over the set of all POVMs can be handled via semidefinite programming techniques.

\subsection{Evolution of quantum states}\label{0_sscn:evolution}

We now know how to describe a quantum state $\rho\in\DD(\X)$, as well how to model a measurement or observation of $\rho$. The next question we ask is: What kind of operations can we perform on $\rho$? For example, to a classical bit, we can apply a NOT gate to flip its value. What can we do to a \emph{qu}bit?

In the quantum setting, the set of valid operations on a \emph{closed} (defined shortly) quantum system with state $\rho\in\DD(\X)$ is the set of unitary operators $U\in\UU(\X)$. Specifically, $U$ maps $\rho$ to
\begin{equation}
    \rho' := U\rho U^\dagger.
\end{equation}

For example, for $\rho\in\DD(\complex^2)$, i.e.\ a single qubit, a frequently used set of unitary operators are the \emph{Pauli} operators (where $i:=\sqrt{-1}\in\complex$)
\begin{equation}
    X=\left(
        \begin{array}{cc}
          0 & 1 \\
          1 & 0 \\
        \end{array}
      \right)\quad\quad\quad
    Y=\left(
        \begin{array}{cc}
          0 & -i \\
          i & 0 \\
        \end{array}
      \right)\quad\quad\quad
    Z=\left(
        \begin{array}{cc}
          1 & 0 \\
          0 & -1 \\
        \end{array}
      \right).
\end{equation}
Note, for example, that the Pauli $X$ plays the role of a quantum NOT gate, i.e.\ $X\ket{0}=\ket{1}$ and $X\ket{1}=\ket{0}$.

We said that unitary operations describe the evolution of {closed} quantum systems above --- let us elaborate on this further. A \emph{closed} quantum system is one which does not interact with its environment. Conversely, if a system is not closed, it is called \emph{open}. In this latter case, the set of allowed operations strictly contains $\UU(\X)$, and is in fact the set of TPCP maps, which we henceforth refer to as \emph{admissible} maps or operations. Despite this, there is a sense in which discussing unitary operations is without loss of generality --- this is implied by the Stinespring representation of super-operators and specifically Equation~(\ref{0_eqn:church}), which states that any valid TPCP operation on a quantum system $A$ can be simulated by moving to a larger joint system $AB$, evolving $AB$ via a unitary operator, and subsequently tracing out part of $AB$. (We discuss joint systems $AB$ further in Section~\ref{0_sscn:composite}.)

For example, let us consider the process of performing a measurement on $A$. In order to measure or observe a quantum state in $A$, one introduces a measurement apparatus, which we think of as system $B$. To complete the actual measurement, $B$ must interact with $A$, implying $A$ is an open system. Thus, if we look at $A$ alone, the action of the measurement on $A$ is not described by a unitary operator, but by a TPCP map. However, if we instead look at $AB$ as a whole, this joint system is now closed, and hence its evolution is described by a unitary operator.

\paragraph{Hamiltonians, and the connection to unitary operations.} We said above that the evolution of a (closed) quantum system is described by a unitary operator. Although this is a great abstract description for mathematicians and computer scientists to work with, one should ask the question: Why \emph{unitary} operations? The answer lies, not surprisingly, in physics. Here we define the notion of a \emph{Hamiltonian}, which will play an important role in later chapters such as those involving Hamiltonian complexity.

First, note that any unitary $U\in\UU(\X)$ can be written as $U=\exp(i H)$ for some $H\in\HH(\X)$. This is easily seen by taking the spectral decomposition $U=\sum_j e^{i\theta_j}\ketbra{\psi_j}{\psi_j}$, and observing that defining
\begin{equation}
    H=\sum_j \theta_j\ketbra{\psi_j}{\psi_j}
\end{equation}
yields $U=e^{iH}$ (see the discussion on operator functions in Section~\ref{0_scn:linalg}). The operator $H$ is called a \emph{Hamiltonian}.

Thus, corresponding to each $U\in\UU(\X)$, there exists an $H\in\HH(\X)$. Where does $H$ then come from? It turns out that the time evolution of a closed quantum system $\ket{\psi}$ according to $H$ is given by the famous Schr\"{o}dinger equation,
\begin{equation}
    i\hbar\frac{d\ket{\psi}}{dt}=H\ket{\psi},
\end{equation}
where $\hbar$ denotes Planck's constant (whose value is not of interest here). For a quantum system evolving from time $t_1$ to $t_2$, the solution to this equation is given by
\begin{equation}\label{0_eqn:schroedinger}
    \ket{\psi(t_2)}=\exp\left(i\frac{t_1-t_2}{\hbar}H\right)\ket{\psi(t_1)},
\end{equation}
from which we now see the connection to unitary operators directly.

For this reason, Hamiltonians have been the object of intense study, and there is nowadays an entire field devoted to Hamiltonian complexity (see Section~\ref{0_scn:complexity}). The eigenstates $\set{\ket{\lambda}}$ of a Hamiltonian are referred to as its \emph{energy eigenstates}, and the eigenvalue $\lambda$ corresponding to $\ket{\lambda}$ is the \emph{energy} of state $\ket{\lambda}$. The smallest eigenvalue $\lambda_{\min}$ of $H$ is called the \emph{ground state energy}, and $\ket{\lambda_{\min}}$ the \emph{ground state} of $H$. Determining the ground state energy of a given $H$ is in general a very difficult problem, as we shall soon see in Section~\ref{0_scn:complexity}.

Before closing, we make two final remarks. First, there is another interpretation of the Hamiltonian versus unitary pictures of time evolution presented here which is of interest. The application of any fixed unitary $U$ can be thought of as a \emph{discrete-time} evolution, since by Equation~(\ref{0_eqn:schroedinger}) it corresponds to evolution by some fixed time $t$. In the Hamiltonian picture, however, for any fixed Hamiltonian $H$, one can in principle vary the time of evolution $t$ as desired, resulting in a notion of \emph{continuous-time} evolution.

Finally, in our discussion here we have focused on time-\emph{independent} Hamiltonians. More generally, one can also consider evolution under time-\emph{dependent} Hamiltonians which are allowed to change with time.

\subsection{Composite quantum systems}\label{0_sscn:composite}

Thus far, we have discussed the basics of how to mathematically discuss single quantum systems. Suppose now we have two quantum systems $A$ and $B$ --- how do we describe their joint state $AB$? It turns out that if $A$ and $B$ correspond to complex Euclidean spaces $\X$ and $\Y$, then the joint system $AB$ corresponds to the space $\X\otimes \Y$. In other words, if, for example, $\X=\Y=\complex^2$, then any $\rho\in\DD(\X\otimes \Y)$ defines a valid two-qubit quantum system.

The simplest examples of two-party systems $AB$ are given by \emph{product states}, which for any given $\rho_A\in\DD(\X)$ and $\rho_B\in\DD(\Y)$, are given by $\rho_A\otimes \rho_B$. Such states are uncorrelated between systems $A$ and $B$. For example, two classical bits in state $00$ can be embedded in such a two-qubit quantum state as $\ket{0}\otimes\ket{0}$. For brevity, when discussing pure states, we simply denote this state as $\ket{0}\ket{0}$ or $\ket{00}$. More generally, one can also consider joint states $\ket{\phi}\in\complex^2\otimes\complex^2$ such as
\begin{equation}
    \ket{\phi^+}=\frac{1}{\sqrt{2}}\ket{00}+\frac{1}{\sqrt{2}}\ket{11}.
\end{equation}
This state is referred to as a \emph{Bell state}, and possesses a strong degree of \emph{quantum} correlations between systems $A$ and $B$ known as \emph{quantum entanglement}, as discussed further in Section~\ref{0_scn:nonclassicality}.

Given a description $\rho$ of the state of a joint system $AB$, we now require a method for describing the marginal state on $A$ (or $B$) alone. Specifically, given a composite system $\rho\in\DD(\X\otimes\Y)$, the \emph{reduced state} $\rho_A$ on $A$ (analogously, $\rho_B$ on $B$) is given by the partial trace operation described in Section~\ref{0_scn:linalg}. In other words,
\begin{equation}
    \rho_A = \trace_B(\rho).
\end{equation}
For example, $\trace_B(\rho_A\otimes\rho_B)$ is simply $\rho_A$, and $\trace_B(\ketbra{\phi^+}{\phi^+})=I/2$. The partial trace is employed here as it is the unique function which correctly produces the measurement statistics for arbitrary observables $M$ measured on $A$ alone.

We close by remarking that our description of two-party composite systems straightforwardly extends to multiple parties:  For systems $A_1$ through $A_n$ corresponding to complex Euclidean spaces $\X_1$ through $\X_n$, the corresponding joint space is given by $\bigotimes_{i=1}^n\X_i$.

\subsection{Quirks of quantum mechanics}\label{0_sscn:quirks}

Marking a drastic departure from the classical setting, a fundamental result in quantum mechanics is that an unknown quantum state $\ket{\psi}\in\X$ cannot be copied or \emph{cloned}. This is called the \emph{No-Cloning} Theorem~\cite{D82,WZ82}. To give a brief flavor of why this holds, we demonstrate a simple proof from Nielsen and Chuang~\cite{NC00} (Box 12.1) for the case regarding the non-existence of a unitary $U\in\UU(\X\otimes\X)$ achieving the mapping
\begin{equation}
    \ket{\psi}_{\X}\otimes\ket{s}_{\X}\mapsto  \ket{\psi}_{\X}\otimes\ket{\psi}_{\X},
\end{equation}
where $\ket{s}$ is some fixed starting state. For sake of contradiction, suppose such a $U$ does exist. Then for vectors $\ket{\psi_1}$,$ \ket{\psi_2}$, let
\begin{eqnarray}
\ket{\phi_1} &:=& U(\ket{\psi_1}\otimes\ket{s})=  \ket{\psi_1}\otimes\ket{\psi_1}\\
\ket{\phi_2} &:=& U(\ket{\psi_2}\otimes\ket{s})=  \ket{\psi_2}\otimes\ket{\psi_2}.
\end{eqnarray}
Then, $\braket{\phi_1}{\phi_2}=\braket{\psi_1}{\psi_2}=(\braket{\psi_1}{\psi_2})^2$. But the equation $x=x^2$ only has solutions $0$ and $1$, implying that for general $\ket{\psi_1}$ and $\ket{\psi_2}$, such a $U$ cannot exist. We remark that using the Stinespring representation, this proof is easily adapted to show that even TPCP maps cannot clone non-orthogonal states~\cite{Y86}.

\section{Quantum computational complexity}\label{0_scn:complexity}

With the basics of linear algebra and quantum computing under our belts, we can now begin discussing the first central area this thesis studies: Computational complexity theory. This field aims to rigorously classify computational problems based on the inherent difficulty of solving them. Specifically, the central idea here is to ask:

\emph{Given a set of resources $S$, such as a certain amount of space or time in which a computation is to run, what is the class of computational problems which can be solved?}

\noindent This approach has led to an entire zoo of such \emph{complexity classes} (literally, a zoo~\cite{AKG}), including the ubiquitous classes P and NP. In this section, we review the extension of some of these concepts to the quantum setting. This includes defining the standard quantum circuit model our work is based on, introducing relevant quantum complexity classes, and presenting an exposition of the quantum version of the Cook-Levin theorem~\cite{C72,L73}. The content of this section is based partly on the excellent surveys of Aharonov and Naveh~\cite{AN02} and Watrous~\cite{W09_2}, as well as the text of Nielsen and Chuang~\cite{NC00}. We assume background knowledge of basic (classical) computational complexity; the interested reader is referred to the text of Arora and Barak for an introduction~\cite{AB90}.

\paragraph{Notation and definitions specific to this section.} Throughout our discussion, we encode all computational problems over the binary alphabet $\Sigma:=\set{0,1}$. We say a function $f:\Sigma^*\mapsto\Sigma^*$ is \emph{polynomial-time computable} if there exists a polynomial time deterministic Turing machine which, given any input $x\in\Sigma^*$, outputs $f(x)$.  A function $f:\nats\mapsto\nats$ is called \emph{polynomially-bounded} if there exists a polynomial-time deterministic Turing machine which, on any input $x\in\nats$, outputs $1^{f(x)}$. A \emph{language} is a partitioning $\Sigma^*=\ayes\cup\ano$ such that $\ayes\cap\ano=\emptyset$, for $\emptyset$ the empty set. If, more generally, $\ayes\cup\ano\subseteq \Sigma^*$, then we have a \emph{promise problem}. In a promise problem, one assumes the input $x$ satisfies $x\in\ayes$ or $x\in\ano$; if an algorithm solving this promise problem is given input $x\not\in\ayes\cup\ano$, we adopt the convention that the algorithm is allowed to err. We remark that promise problems are particularly natural in the quantum setting, as quantum computations are inherently probabilistic processes, and as such, some ``margin of error'' appears to be needed separating $\ayes$ from $\ano$. This is clarified further when introducing our relevant quantum complexity classes.

\subsection{Quantum circuit model}\label{0_sscn:circuit}

In Section~\ref{0_sscn:evolution}, we discussed the general types of admissible operations on quantum systems. In the context of complexity theory, however, we require a formal model for specifying and analyzing such operations, for which we employ the standard \emph{quantum circuit model}.  To begin, suppose we have a quantum system consisting of $n$ qubits, whose associated complex Euclidean space is $\X = (\complex^2)^{\otimes n}$. A \emph{quantum circuit} can be thought of as a directed acyclic graph with $n$ \emph{input} nodes of in-degree zero and out-degree one (i.e.\ $n$ sources), $n$ \emph{output} nodes of in-degree one and out-degree zero (i.e.\ $n$ sinks), and a set of ``intermediate'' nodes or \emph{gates}, each of which has matching in- and out-degree $c$ for some $c\in\Theta(1)$ (where each gate can have a different value of $c$). Intuitively, the input (output) nodes are the $n$ input (output) qubits to the circuit, and the intermediate notes are unitary gates acting on $\Theta(1)$ qubits. The edges of the graph correspond to \emph{wires} in the circuit, the direction of which are indicative of the direction of data flow.

For example, three common single-qubit unitary gates mentioned in Section~\ref{0_sscn:evolution} are the Pauli $X$, $Y$, and $Z$ operators, which are specified in the circuit model as:
\begin{eqnarray}
\Qcircuit @C=1em @R1em {
    \lstick{\ket{b}} & \gate{X} & \rstick{\ket{b\oplus 1}} \qw \\
}
\\
\Qcircuit @C=1em @R1em {
    \lstick{\ket{b}} & \gate{Y} & \rstick{(-1)^bi\ket{b\oplus 1}} \qw \\
}
\\
\Qcircuit @C=1em @R1em {
    \lstick{\ket{b}} & \gate{Z} & \rstick{(-1)^b\ket{b}} \qw \\
}
\end{eqnarray}
Here, we assume $b\in\set{0,1}$; the action of each gate is extended to all single qubit states by linearity. The notation $\oplus$ denotes the XOR operation (i.e.\ addition modulo $2$). On the left of each gate is the input qubit, and on the right is the output qubit.

Two other single-qubit gates, whose importance is discussed shortly, are the \emph{Hadamard} and $T$ (also referred to as $\pi/8$) gates, defined below.
\begin{eqnarray}
H &=& \frac{1}{\sqrt{2}}\left(
      \begin{array}{cc}
        1 & 1 \\
        1 & -1 \\
      \end{array}
    \right)
\hspace{10mm}\equiv\hspace{15mm}\Qcircuit @C=1em @R1em {
    \lstick{\ket{b}} & \gate{H} & \rstick{\frac{1}{\sqrt{2}}(\ket{0}+(-1)^b\ket{1})} \qw \\
}\hspace{25mm}\\
T &=& \left(
      \begin{array}{cc}
        1 & 0 \\
        0 & e^{i\pi/4}\\
      \end{array}
    \right)\hspace{15mm}\equiv\hspace{15mm}
    \Qcircuit @C=1em @R1em {
    \lstick{\ket{b}} & \gate{T} & \rstick{e^{i\frac{\pi\cdot b}{4}}\ket{b}} \qw \\
}\hspace{25mm}
\end{eqnarray}
\noindent A ubiquitous two-qubit gate is the \emph{Controlled-NOT} gate, shown below.
\begin{equation}
    CNOT = \left(
             \begin{array}{cccc}
               1 & 0 & 0 & 0 \\
               0 & 1 & 0 & 0 \\
               0 & 0 & 0 & 1 \\
               0 & 0 & 1 & 0 \\
             \end{array}
           \right)\hspace{10mm}\equiv\hspace{20mm}
    \Qcircuit @C=1em @R=.7em {
    \lstick{\ket{b_1}} & \ctrl{1}& \rstick{\ket{b_1}}\qw\\
    \lstick{\ket{b_2}} & \targ & \rstick{\ket{b_2\oplus b_1}}\qw
    }\hspace{10mm}
\end{equation}
Finally, a measurement (in the computational basis) in this model is specified by the following.
\begin{equation}
\Qcircuit @C=1em @R1em {
    & \qw & \meter
}
\end{equation}

\paragraph{Universal gate sets.} When it comes to quantifying the \emph{cost} of a circuit, it seems \emph{a priori} that we are in a bind: How do we quantify the cost of an arbitrary gate if there is a continuum of unitary gates to choose from? It would be preferable to have a fixed finite set of gates, each of which is assigned unit cost, and with which we could simulate all other gates. This would yield a rigorous framework in which to quantify the cost of a circuit. Such a set of unitaries is called a \emph{universal} set, and indeed exists: The set $S:=\set{H, T, CNOT}$ is universal. To show this (see, e.g.,~\cite{NC00}), one first demonstrates that the $CNOT$ coupled with the set of all one-qubit unitaries is universal in an \emph{exact} sense --- any unitary $U\in\UU(\X)$ can be represented \emph{exactly} using CNOT and single-qubit gates. One then applies the Solovay-Kitaev theorem~\cite{K97}, which yields that for any $U\in\UU(\complex^2)$ and any $\epsilon>0$, there exists a $V\in\UU(\complex^2)$ consisting of the composition of $O(\log^c(1/\epsilon))$ gates from $\set{H, T}$ such that $\snorm{U-V}\leq \epsilon$ (here, $c\in\Theta(1)$).

What does such a bound on the spectral norm buy us? Suppose we can substitute the original unitaries $U=U_m\cdots U_1$ in a circuit with unitaries $V=V_m\cdots V_1$ with the promise that $\snorm{U_i-V_i}\leq \epsilon$ for all $i\in[m]$ and for $\epsilon$ to be chosen as needed. Since we are typically interested in running $U$ on some input $\ket{\psi}$, followed by a measurement according to some POVM, we would like the probability of obtaining any measurement outcome to deviate by at most $\delta$ when substituting $V$ for $U$, where $\delta>0$ can be chosen as desired. In other words, for all POVM elements $M$, pure states $\ket{\psi}$, and error parameters $\delta>0$, we would like that setting $\epsilon$ small enough yields that the probability of obtaining outcome $M$ when measuring $U\ket{\psi}$ versus $V\ket{\psi}$ differs by at most $\delta$. Indeed, this is achieved by setting $\epsilon =\delta/(2m)$ and combining the facts that
\begin{equation}
    \abs{\trace(M U\ketbra{\psi}{\psi}U^\dagger)-\trace(M V\ketbra{\psi}{\psi}V^\dagger)}\leq 2\snorm{U-V},
\end{equation}
and
\begin{equation}
    \snorm{U_m\cdots U_1 - V_m \cdots V_1}\leq \sum_{j=1}^m \snorm{U_j-V_j}.
\end{equation}
We refer the reader to~\cite{NC00} for further details.

We close this section by remarking that here we have assumed that quantum circuits are unitary and act on pure state inputs $\ket{\psi}\in\X$; recall from Section~\ref{0_sscn:evolution} that by the Stinespring representation and Equation~(\ref{0_eqn:church}), this is without loss of generality. We refer the reader to the work of Aharonov, Kitaev, and Nisan~\cite{AKN98} for a more general model of quantum circuits which directly operates on mixed states, and which explicitly harnesses this connection with the Stinespring representation.

\paragraph{Oracles.} A commonly used construct in the setting of quantum circuits is that of an \emph{oracle}. An oracle $Q_n$ (where we more precisely deal with a family of oracles $\set{Q_n}$) can be thought of as a black-box unitary operation encoding some predicate $f:\Sigma^n\mapsto\Sigma$. In the quantum circuit model, this is formalized via the action
\begin{equation}
    Q_n\ket{x}\ket{y} = \ket{x}\ket{y\oplus f(x)},
\end{equation}
for $x\in\Sigma^n$ and $y\in\Sigma$. Each such application of $Q_n$ is called a \emph{query} to the oracle, and we typically think of each query as having unit cost.

Suppose now that we wish to compute some property P of the predicate $f$; the number of queries to $Q_n$ required to do so is called the \emph{query complexity} of P (relative to $Q_n$). Perhaps the most well-known example of this in the quantum setting is Grover's algorithm~\cite{G93}, which shows how to compute the OR function $\bigvee_{i=1}^{2^n}f(i)$ with high probability using $O(\sqrt{2^n})$ queries to $Q_n$, a quadratic improvement over the classical setting. Although the query model may \emph{a priori} seem restricted, the model is nevertheless important; Shor conceived his factoring algorithm~\cite{S97}, for example, by studying Simon's algorithm~\cite{Si94,Si97} from the quantum query model.

\subsection{Standard quantum complexity classes: BQP and QMA}\label{0_sscn:classes}

Recall that in complexity theory, we classify computational problems into \emph{complexity classes} depending on the resources capable of solving them. The classes P and NP are two such classes, forming two cornerstones of classical complexity theory. We now discuss the natural quantum analogues of these classes, BQP and QMA. (More precisely, BQP and QMA are generalizations of BPP and MA.) For completeness, we recall the definitions of P and NP below.
\begin{definition}[P]
     A promise problem $A=(\ayes,\ano)$ is in P if and only if there exists a deterministic polynomial-time Turing machine $M$ which on input $x\in\ayes$, accepts, and on input $x\in\ano$, rejects.
\end{definition}

\begin{definition}[NP]
     A promise problem $A=(\ayes,\ano)$ is in NP if and only if there exists a deterministic polynomial-time Turing machine $M$ and a polynomial $p$, such that on input $x\in\Sigma^*$:
      \begin{itemize}
        \item If $x\in\ayes$, then there exists a proof $y\in\Sigma^{p(\abs{x})}$ such that $M$ accepts $(x,y)$.
        \item If $x\in\ano$, then for all proofs $y\in\Sigma^{p(\abs{x})}$, $M$ rejects $(x,y)$.
      \end{itemize}
\end{definition}

Now, since we have defined our complexity theoretic model for quantum computing based on the quantum circuit model, we next require the notion of a polynomial-time uniform family of quantum circuits. Specifically, since the length of input $x\in\Sigma^*$ to a computational problem is allowed to vary, whereas the input size to a given circuit is fixed, we require a method for ``scaling'' our circuits up to match the length of arbitrary input $x\in\Sigma^*$.

\begin{definition}[Polynomial-time uniform family of quantum circuits] A set of quantum circuits $\set{Q_n}$ is polynomial-time uniform if there exists a polynomial-time deterministic Turing machine, which on input $1^n$, outputs a description of $Q_n$.
\end{definition}

We now define BQP~\cite{BV97}, which stands for \emph{Bounded-Error Quantum Polynomial Time}, and which is intuitively the set of promise problems which can be efficiently solved with high probability on a quantum computer. For both BQP and QMA, we henceforth say a quantum circuit $Q$ \emph{accepts} input $x$ (where $x$ can be either a classical string or quantum state) if running $Q$ on input $x$ and subsequently measuring a designated output qubit of $Q$ in the computational basis yields outcome $1$.

\begin{definition}[BQP]
    A promise problem $A=(\ayes,\ano)$ is in BQP if and only if there exists a polynomial $q$ and a polynomial-time uniform family of quantum circuits $\set{Q_n}$, where $Q_n$ takes as input a string $x\in\Sigma^*$ with $\abs{x}=n$, and $q(n)$ ancilla qubits in state $\ket{0}^{\otimes q(n)}$, such that:
    \begin{itemize}
    \item (Completeness) If $x\in\ayes$, then $Q_{n}$ accepts input $x$ with probability at least $2/3$.
    \item (Soundness) If $x\in\ano$, then $Q_{n}$ accepts input $x$ with probability at most $1/3$.
    \end{itemize}
\end{definition}

Note that if we replace the uniform quantum circuit family above with a uniform \emph{classical} circuit family which takes as input both $x$ and a polynomial-size string $y$ chosen uniformly at random, then we are reduced to BPP. Like BPP, the completeness and soundness parameters $2/3$ and $1/3$ above can straightforwardly be amplified to values exponentially close to $1$ and $0$ simply by running the verification procedure $Q$ independently polynomially many times in parallel, accepting if and only if the majority of runs accepted, and applying the Chernoff bound. We remark that $\class{BPP}\subseteq\class{BQP}$ follows since probabilistic classical computations can be simulated with quantum circuits (see, e.g.~\cite{W09_2}). The decision versions of the factoring and discrete logarithm problems are, for example, not known to be in BPP, but are in BQP due to Shor's algorithm~\cite{S97}.

We next define QMA, or \emph{Quantum Merlin Arthur}, a quantum generalization of NP.

\begin{definition}[QMA]\label{0_def:QMA}
    A promise problem $A=(\ayes,\ano)$ is in QMA if and only if there exist polynomials $p$, $q$ and a polynomial-time uniform family of quantum circuits $\set{Q_n}$, where $Q_n$ takes as input a string $x\in\Sigma^*$ with $\abs{x}=n$, a quantum proof $\ket{y}\in (\complex^2)^{\otimes p(n)}$, and $q(n)$ ancilla qubits in state $\ket{0}^{\otimes q(n)}$, such that:
    \begin{itemize}
    \item (Completeness) If $x\in\ayes$, then there exists a proof $\ket{y}\in (\complex^2)^{\otimes p(n)}$ such that $Q_n$ accepts $(x,\ket{y})$ with probability at least $2/3$.
    \item (Soundness) If $x\in\ano$, then for all proofs $\ket{y}\in (\complex^2)^{\otimes p(n)}$, $Q_n$ accepts $(x,\ket{y})$ with probability at most $1/3$.
    \end{itemize}
\end{definition}

\noindent It is often helpful to think of $\ket{y}$ above as a proof sent by an all-powerful but untrustworthy prover Merlin, who claims $x\in\ayes$, and to correspondingly interpret $\set{Q_n}$ as an honest but computationally bounded verifier Arthur, whose job it is to verify the correctness of Merlin's proof. We are not overly fond of the names Merlin and Arthur, and as such, prefer to simply refer to both parties in this interpretation as being the prover and verifier, respectively. As an aside, we remark that QMA was originally known as Bounded Error Quantum NP (BQNP)~\cite{KSV02}.

Note now that if we instead ask in the definition of QMA that $y\in\Sigma^{p(\abs{x})}$, then the corresponding complexity class is known as quantum-classical Merlin-Arthur (QCMA)~\cite{AN02,JW06,A06,AK07,Beigi08,ABBS08,WY08}. (QCMA is also known by the name Merlin-Quantum-Arthur (MQA), as suggested by Watrous~\cite{W09_2}.) Finally, if $y$ is classical \emph{and} we replace $\set{Q_n}$ with a classical circuit family of the type used in defining BPP, then the class we obtain is Merlin-Arthur (MA)~\cite{B85}.

\paragraph{Error reduction for QMA.} Like BQP, the completeness and soundness parameters in the definition of QMA can be amplified to values exponentially close to $1$ and $0$, respectively. However, the arguments employed here are not as straightforward as in the case of BQP. For QMA, there are two approaches for achieving error reduction, which we refer to as \emph{weak} and \emph{strong} error reduction, and which we now discuss.

\emph{Weak} or \emph{standard} error reduction runs analogously to the case of BQP, i.e.\ by running the verification protocol some number of times $m$ in parallel and taking a majority vote. However, since from Section~\ref{0_sscn:quirks}, we know that unknown quantum states cannot be cloned, the verifier must ask the prover for multiple copies of the proof $\ket{y}$, one for each of the $m$ parallel runs of the protocol. If the verifier is honest, the proof sent for the new protocol is a product state $\ket{y'}=\ket{y}^{\otimes m}\in(\complex^2)^{\otimes p(\abs{x})\cdot m}$, in which case the $m$ runs of the verification protocol are independently and identically distributed Bernoulli trials, and the Chernoff bound can be applied. However, if we have a NO-instance, i.e.\ $x\in\ano$, then in a desperate attempt to trick the verifier into thinking $x\in\ayes$, the prover may elect to cheat by sending a proof $\ket{y'}$ which deviates from this product state structure. Can we still apply the Chernoff bound argument here?

It turns out the answer is \emph{yes}, the intuition for which we now sketch. (A detailed proof can be found in~\cite{AN02}.) Specifically, let $V$ denote the original verification protocol. Then, given any $\ket{y'}\in(\complex^2)^{\otimes p(\abs{x})\cdot m}$, we adopt the following view: On the first $p(\abs{x})$ proof qubits, we run the first copy of $V$, measure and read the output qubit, and subsequently discard these $p(\abs{x})$ qubits. Note that the reduced state of $\ket{y'}$ on these first $p(\abs{x})$ qubits before running $V$ is simply a convex mixture of proofs $\ket{y}\in(\complex^2)^{\otimes p(\abs{x})}$; thus, by the soundness property of the QMA protocol, the probability of acceptance in this first run is at most $1/3$. We can iterate this argument over each of the remaining $m-1$ copies of $V$, each time obtaining a probability of accepting of at most $1/3$. It follows that a majority vote, coupled with the Chernoff bound, yields the desired error reduction.

Finally, although weak error reduction is simple, its disadvantage is that it requires an increase in the proof size, since the prover must send multiple copies of the original proof. Is it possible to reduce the error \emph{without} increasing the proof length? Remarkably, Marriot and Watrous have shown~\cite{MW05} that the answer is \emph{yes}. The rough idea here is best illustrated in the case of a zero-error verifier $V$, i.e.\ where the completeness and soundness parameters are $1$ and $0$, respectively. Specifically, let $V$ be a zero-error verifier $V$, and $\ket{y}$ the prover's proof for some instance $x\in\ayes$. Then, if we run $V$ on $(x,\ket{y})$ and measure the output qubit, we will see outcome $1$ with certainty. Thus, the measurement does not alter the output state of $V$. Further, if we now run $V$ in reverse and measure the ancillary qubits of $V$, they should read all zeroes with certainty, implying this second measurement also does not alter the state being measured. In fact, we can repeat this back and forth process as many times as we like, each time obtaining the same ``good'' measurement outcomes.

What happens now if we do not have a zero-error QMA verifier $V$, and have a NO instance $x\in\ano$? In this case, the output qubit of $V\ket{x}\otimes\ket{y}$ must yield outcome $1$ with probability at most $1/3$ --- in other words, measuring this qubit now disturbs the state $V\ket{x}\otimes\ket{y}$. Moreover, when we next apply $V^\dagger$ and measure the ancilla qubits, since $V$ is unitary, the outcome cannot be the all-zeroes string with non-negligible probability, again disturbing the state. Intuitively, by repeating this back-and-forth procedure, we thus quickly amplify the likelihood of obtaining ``bad'' measurement outcomes in this process. In our opinion, the entire process can be thought of as analogous to a spinning top --- if the top wobbles badly enough to begin with (if $x\in\ano$), the spinning motion (the back and forth measurement process) quickly sends the top out of control.

\subsection{BQP and QMA in further depth}~\label{0_scn:qma}
As QMA plays an important role in this thesis, we now further discuss its properties, variants, and complete problems. Along the way, we also mention some further properties of BQP.

First, we have
\begin{equation}\label{0_eqn:temp1}
    \class{NP}\subseteq\class{MA}\subseteq\class{QCMA}\subseteq\class{QMA}\subseteq\class{PP}.
\end{equation}
Here, PP is defined analogously to BPP, except that when input $x\in\ayes$, then the verifier accepts with probability strictly larger then $1/2$; if $x\in\ano$, the verifier accepts with probability at most $1/2$. The second of the containments above follows since a QCMA verifier can choose to act classically. The third containment holds since a QMA verifier can force a given quantum proof to encode a classical string by preceding the verification procedure with a measurement in the computational basis. Finally, the fourth containment has an elegant proof via the strong error reduction technique of Marriott and Watrous~\cite{MW05}, and was originally proven by Kitaev and Watrous~\cite{KW00}.

Regarding BQP, we have that
\begin{equation}\label{0_eqn:temp2}
    \class{BPP}\subseteq\class{BQP}\subseteq\class{QMA},
\end{equation}
where the second containment follows since the verifier can simple flush the prover's proof down the toilet and run the BQP circuit instead. Combining Equations~(\ref{0_eqn:temp1}) and~(\ref{0_eqn:temp2}) yields $\class{BQP}\subseteq\class{PP}$; we remark that this containment was directly proven by Adleman, DeMarrais, and Huang~\cite{ADH97} and Fortnow and Rogers~\cite{FR99}. Marriott and Watrous have shown that $\class{BQP}=\class{QMA}_{\log}$~\cite{MW05}, where $\class{QMA}_{\log}$ is QMA with a logarithmic size proof. The classical version of this equality might be written $\class{P}=\class{NP}_{\log}$, i.e.\ NP with logarithmic size proofs is contained in P. Finally, it is well-known that in the classical setting, $\BPP\subseteq\st$~\cite{S83,L83}, for $\st$ the second level of the polynomial hierarchy $\PH$. Whether $\BQP\subseteq\PH$, however, remains a major open question~\cite{Aa10,FU10,Aa11}.

\paragraph{One-sided error.} Next, we discuss the one-sided error versions of MA, QCMA, and QMA. Specifically, let $\class{MA}_1$, $\class{QCMA}_1$, and $\class{QMA}_1$ be defined as MA, QCMA, and QMA, respectively, except with completeness $1$ in each case. In other words, if $x\in\ayes$, the verifier for the new classes accepts with certainty. Zachos and Furer have shown that $\class{MA}=\class{MA}_1$~\cite{ZF87} (see also Goldreich and Zuckerman~\cite{GZ97}), and more recently, Jordan, Kobayashi, Nagaj, and Nishimura have proven that $\class{QCMA}_1=\class{QCMA}$~\cite{JKNN12}. Whether $\class{QMA}_1=\class{QMA}$, however, remains an interesting open question, particularly since both QCMA and $\class{QIP}(3)$ in the chain $\class{QCMA}\subseteq\class{QMA}\subseteq\class{QIP}(3)$ allow one-sided error~\cite{KW00}. Here, $\class{QIP}(k)$ is the class of promise problems having \emph{Quantum Interactive Proofs} with $k$ rounds, meaning it is a generalized version of QMA in which $k$ quantum messages are passed back and forth between prover and verifier. For example, $\class{BQP}=\class{QIP(0)}$, $\class{QMA}=\class{QIP}(1)$, and $\class{QIP}(3)$ consists of a message from prover to verifier, followed by a message from verifier to prover, and a final message back from the  prover to the verifier. Aaronson has demonstrated a quantum oracle relative to which $\class{QCMA}_1\subset \class{QCMA}$ and $\class{QMA}_1\subset \class{QMA}$~\cite{A09}.

\paragraph{Complete problems.} We now move to arguably one of the most important questions for any complexity class: \emph{What problems {characterize}, or are complete for QMA?} In general, the set of QMA-complete problems is not yet nearly as rich as that for its classical cousin, NP. The historically first QMA-complete problem was the local Hamiltonian problem (first presented by Kitaev at~\cite{K99}, and later written up in~\cite{KSV02}), which is a natural generalization of the NP-complete problem of classical constraint satisfaction, and relevant from a physics perspective. In fact, we devote Section~\ref{0_sscn:LH} entirely to this problem and its variants, and thus do not discuss it further here.

Perhaps the second-most studied and natural QMA-complete problem is the Consistency problem for local density matrices of Liu~\cite{L06}. In this problem, one is given a classical description of a set of density matrices $\rho_S$, each acting on a subset $S\subseteq [n]$ qubits for $\abs{S}=k$ and $k\in\Theta(1)$. The question is whether there exists a globally consistent $n$-qubit state $\rho$ such that $\trace_{[n]\backslash S}(\rho) = \rho_S$ for all $S$. The proof of QMA-hardness for $k=2$ follows via a polynomial-time Turing or Cook reduction involving convex programming from the $2$-local Hamiltonian problem~\cite{L06}; the reduction in the reverse direction was later given by Liu in~\cite{L07}, and goes via a strong theorem of alternatives in semidefinite programming. Other physically motivated variants of the Consistency problem have also been shown to be QMA-complete: The variant involving fermions, known as the N-representability problem, was shown QMA-complete by Liu, Christandl, and Verstraete~\cite{LCV07}, as well as its bosonic counterpart by Wei, Mosca, and Nayak~\cite{WMN10}.

What other QMA-complete problems are known? Given a classical description of a quantum circuit, the problem of determining whether it is ``close'' to the identity, known as the Identity Check problem, was shown QMA-complete by Janzing, Wocjan, and Beth~\cite{JWB03}. Rosgen~\cite{R09} has shown that a similar problem where one is asked whether a given quantum circuit is close to a linear isometry is QMA-complete. Finally, Beigi and Shor~\cite{BS07} have proposed a QMA-complete quantum generalization of the Clique problem, which asks: Given an (entanglement-breaking) channel $\Phi$, do there exist $k$ quantum states $\rho$ which are distinguishable without error after passing through the channel?

\paragraph{Multiple provers.} QMA is a proof system with a single prover and verifier. A curiosity emerges when we ask the question: What happens to the power of the proof system if we introduce a \emph{second} prover? In other words, what if there are two provers, $P_1$ and $P_2$, who send a joint proof of the form $\ket{\psi_P}\otimes\ket{\psi_{Q}}$ to the verifier? Interestingly, unlike the classical setting where having two provers is trivially equivalent to having a single prover, in the quantum setting, the possibility of \emph{entanglement} between the two proofs (entanglement is introduced in Section~\ref{0_scn:nonclassicality}) makes this a non-trivial question. This class is called $\class{QMA}(2)$~\cite{KMY03}. Why should it be of any interest? Perhaps surprisingly, Blier and Tapp~\cite{BT10} have shown that all languages in NP have very short proofs in this model; specifically, it suffices for $P_1$ and $P_2$ to send proofs $\ket{\psi_{P_1}}$ and $\ket{\psi_{P_2}}$, respectively, consisting of just $O(\log n)$ qubits each. The reader is referred to Chapter~\ref{chap:qmapoly} for formal definitions and details regarding this model, where it is studied in further depth.

\subsection{Local Hamiltonian complexity: An overview}\label{0_sscn:LH}

In Section~\ref{0_scn:qma}, we initiated our discussion of QMA-complete problems, and stated that the first known such problem was the local Hamiltonian problem. As this problem features heavily in Chapters~\ref{chap:approx} and~\ref{chap:hardnessapprox}, we now discuss it in further depth. We begin by defining the problem, and follow by demonstrating how it generalizes the canonical NP-complete problem MAX-SAT. We then discuss some of its variants and its history with respect to the field of complexity theory. Later in Section~\ref{0_sscn:5LH}, we give Kitaev's proof that the $5$-local Hamiltonian problem is QMA-complete.

Beginning with definitions, the local Hamiltonian problem ($\lh$) was introduced by Alexei Kitaev~\cite{K99,KSV02}, and can intuitively be thought of as follows: Given a ``succint'' representation of a ``large'' Hamiltonian $H$, what is $H$'s smallest eigenvalue? Of course, the obvious approach to answering this question is to diagonalize $H$ --- however, the catch is that while $H$ is a $2^n\times 2^n$-dimensional matrix, the succinct encoding we are given of $H$ consists of $\poly(n)$ bits. In other words, a simple diagonalization approach would take time exponential in the input size.

Let us now define $\lh$ more formally. To do so, we first define the term \emph{$k$-local Hamiltonian}. %Here, given a linear operator $H:\B^{\otimes n}\mapsto\B^{\otimes n}$ (recall $\B := \complex^2$) and $S\subseteq[n]$, we shall use the convention $H[S]$ to mean that $H$ acts non-trivially only on the qubits indexed by set $S$, i.e.\ $H\in \LL(\B^{\otimes \abs{S}})$.

\begin{definition}\label{0_def:localhamiltonian}
    An operator $H\in \HH(\B^{\otimes n})$ is called a $k$-local Hamiltonian if it can be written
    \begin{equation}\label{0_eqn:temp3}
        H = \sum_{j=1}^r H_j,
    \end{equation}
    where $\set{H_j}_{j=1}^r\subseteq \HH(\B^{\otimes k})$ is a collection of local Hamiltonian terms, such that each $H_j$ acts non-trivially on some subset ${S_j}\subseteq[n]$ of at most $k$ qubits and satisfies $0	\preceq H_j \preceq I$. Note: In Equation~(\ref{0_eqn:temp3}), we adopt the convention that each $H_j$ acts as the identity on all qubits in the set $[n]\backslash S_j$.
\end{definition}
\noindent Note that although we define $H$ as acting on qubits above, the definition extends straightforwardly to the case of higher-dimensional local systems. Intuitively, the definition above says that a $k$-local Hamiltonian $H$ can be expressed as a sum of ``smaller'' Hermitian operators $H_j$, each of which is restricted to act non-trivially on at most $k$ out of $n$ qubits. %Note that we allow a \emph{multiset} of local terms $H_j$ above --- this is useful shortly when discussing spectral gaps in the definition of $\klhh$~\cite{W09_2}.

We now phrase the $\klhh$ problem. We remark that later, in Chapter~\ref{chap:approx}, we shall formulate $\klhh$ in a slightly different manner; the definition below is, however, arguably more natural and thus better suited to an introductory section.

\begin{problem}[$k$-Local Hamiltonian ($\klhh$)~\cite{KSV02}] \label{0_def:localhamiltonianproblem}Given as input:
    \begin{compactenum}
        \item A $k$-local Hamiltonian $H$ acting on $n$ qubits, specified as a collection of local Hamiltonian terms $\set{H_j}_{j=1}^r\subseteq \HH(\B^{\otimes k})$ (i.e. as a collection of $(2^k\times 2^k)$-dimensional matrices $H_j$) where $k\in \Theta(1)$,
        \item Threshold parameters $a,b\in\reals$, such that $0\leq a < b$ and $(b-a)\geq 1$,
    \end{compactenum}
    decide, with respect to the complexity measure $\enc{H} + \enc{a} + \enc{b}$:
    \begin{compactenum}
        \item If $\lmin{H}\leq a$, output YES.
        \item If $\lmin{H} \geq b$, output NO.
    \end{compactenum}
\end{problem}
\noindent Note that often $\klhh$ is phrased with $(b-a)\geq 1/p(n)$ for some polynomial $p$; such an inverse polynomial gap can straightforwardly be boosted to the constant $1$ above by defining $H$ to have $p(n)$ many copies of each local term $H_j$~\cite{W09_2}.% (this is where the multiset in Definition~\ref{0_def:localhamiltonian} is useful).

Although it may not be \emph{a priori} obvious, $\klhh$ generalizes the canonical NP-complete problem MAX-k-CSP, where CSP stands for \emph{Constraint Satisfaction Problem} (of which a special case is the more familiar problem MAX-k-SAT). To see this, recall that in MAX-k-CSP, one is given a set of Boolean functions, $c_i:\set{0,1}^k\mapsto\set{0,1}$ (note the $c_i$ are not restricted to be of any particular form such as conjunctive normal form), where each $c_i$ acts on $k$ out of $n$ possible bits. We then ask: What is the largest number of clauses $c_i$ we can satisfy with a Boolean assignment to the $n$ bits? To embed this problem into $\klhh$, we design a $k$-local Hamiltonian $H$ acting on $n$ qubits as follows. For each clause $c_i$, define a $2^k\times 2^k$-dimensional diagonal matrix $H_{c_i}\in\HH(\B^{\otimes k})$ such that $H_{c_i}(m,m)=0$ if the binary representation of $m$ is a satisfying assignment for clause $c_i$; otherwise, $H_{c_i}(m,m)=1$. In other words, for $x\in\set{0,1}^k$, $\trace(H_{c_i}\ketbra{x}{x})=0$ if $x$ satisfies $c_i$, and $\trace(H_{c_i}\ketbra{x}{x})=1$ otherwise, i.e.\ failing assignments are given an \emph{energy penalty}. To now see that the optimal value of our MAX-$k$-CSP instance corresponds to the smallest eigenvalue of $H=\sum_c H_{c_i}$, we use the fact that since all the $H_{c_i}$ are diagonal, they commute and thus simultaneously diagonalize. Hence, $H$ has integer eigenvalues. Moreover, since the $H_{c_i}$ are simultaneously diagonal in the computational basis, the smallest eigenvalue of $H$ equals the minimum number of unsatisfied clauses over all $n$-qubit computational basis states. It follows that  $\klhh$ generalizes MAX-k-CSP, and thus $\klhh$ is NP-hard. This raises the natural question: \emph{Could $\klhh$ be a canonical QMA-complete quantum constraint satisfaction problem?}

\paragraph{Variants of $\klhh$ and a brief history.} It turns out that $\klhh$ is indeed QMA-complete; Kitaev~\cite{K99,KSV02} showed the problem to be in QMA for $k\geq 1$ and QMA-hard for $k\geq 5$. The proof of QMA-hardness was inspired by earlier ideas of Feynman~\cite{KSV02,F85}, and can be thought of as exploiting Feynman's ideas to adapt the classical Cook-Levin theorem in a non-trivial fashion to the quantum setting. The fact that $3$-$\lh$ is also QMA-complete was shown subsequently by Kempe and Regev~\cite{KR03} (an alternate proof was later also given by Nagaj and Mozes~\cite{NM06}). Finally, Kempe, Kitaev, and Regev showed~\cite{KKR06} that even $2$-$\lh$ is QMA-complete. Note that $1$-$\lh$ is in P, since one can simply optimize for each $1$-local term independently. Although these results are interesting from a complexity theoretic perspective, a more natural question from a physics perspective is whether such QMA-hardness results can be shown even if the QMA-hard classes of local Hamiltonians arising in the reductions employed correspond to \emph{physical} quantum systems in nature~\cite{OT05,AGIK09,N08,SV09}. Along these lines, Oliveira and Terhal next showed~\cite{OT05} that $2$-$\lh$ with the Hamiltonians restricted to nearest-neighbor interactions on a 2D grid is still QMA-complete. Furthermore, in stark contrast to the classical case of MAX-2-CSP on the line (which is in P), Aharanov, Gottesman, Irani and Kempe~\cite{AGIK09} showed that $2$-LH with nearest-neighbor interactions on the line is also QMA-complete if the local systems have dimension at least $12$ (Nagaj later improved this to $11$ states per particle~\cite{N08}).

Although this thesis focuses on the general local Hamiltonian problem as defined in Definition~\ref{0_def:localhamiltonianproblem}, for completeness, we now mention a few interesting variants of LH which have also been studied. First, Bravyi and Vyalyi showed that the variant of $2$-$\lh$ (with local systems of arbitrary, but constant, dimension) in which all local Hamiltonian terms $H_j$ pairwise commute is in NP. This result was extended to the case of $3$-$\lh$ on qubits by Aharonov and Eldar~\cite{AE11}. Bravyi~\cite{B06} introduced a variant of $\klhh$ known as Quantum $k$-SAT, in which each local Hamiltonian term $H_j$ is a projector, and in which the threshold $a$ is set to $0$. We remark that in the YES case of such a setup, the local Hamiltonian is referred to as \emph{frustration-free}, since the optimal assignment lies in the null space of every interaction term. Bravyi then showed that, like classical 2-SAT, Quantum 2-SAT is in P (whereas recall 2-$\lh$ is QMA-complete)~\cite{B06}. In contrast, Quantum 4-SAT is $\class{QMA}_1$-complete (recall $\class{QMA}_1$ is the one-sided error analog of QMA)~\cite{B06}. Whether Quantum 3-SAT on qubits is $\class{QMA}_1$-complete remains an intriguing open question (see Reference~\cite{NM06}). Next, there has been a line of work on so-called \emph{stoquastic} local Hamiltonians~\cite{BBT06,BDOT08,BT09,L07,JGL10}. Specifically, the Stoquastic $k$-SAT problem, defined the same as Quantum $k$-SAT except that all local projectors have real non-negative matrix elements when expressed in the computational basis, was shown to be in MA for $k\geq 1$, and MA-complete for $k\geq 6$~\cite{BBT06,BT09}. (Incidentally, this was the first non-trivial example of an MA-complete promise problem.) The problem Stoquastic LH-MIN, defined as $\klhh$ except where each local Hamiltonian constraint $H_j$ has real non-positive off-diagonal matrix elements in the computational basis, was shown complete for the class StoqMA~\cite{BBT06} for $k\geq 2$. Here, StoqMA is a variant of QMA in which the verifier is restricted to preparing qubits in the states $\ket{0}$ and $\ket{+}$, performing classical reversible gates, and measuring in the Hadamard (i.e.\ $\ket{+},\ket{-}$) basis. Note that $\class{MA}\subseteq\class{StoqMA}\subseteq\class{QMA}$. Finally, variations of LH with symmetry constraints have been studied from a complexity theoretic perspective in, for example,~\cite{GI09,K07}.

\paragraph{Connection to physics.} Although we have primarily discussed LH from a complexity theoretic viewpoint involving quantum constraint satisfaction, the initial motivation for studying LH comes of course from physics. Indeed, the study of the local Hamiltonian problem is part of the more general field of Hamiltonian Complexity, whose aim is to understand how difficult it is to simulate physical systems. In particular, LH can be phrased as a special case of the more general Simulation Problem~\cite{O11}, which roughly asks the following: Given a description of a Hamiltonian $H$, an initial state $\rho$, an observable $M$, and a time $t\in\complex$, estimate the expectation
\begin{equation}
    \trace\left[M \frac{(e^{iHt})^\dagger \rho e^{iHt}}{\trace\left((e^{iHt})^\dagger \rho e^{iHt}\right)}\right].
\end{equation}
The local Hamiltonian problem is recovered by choosing $H$ as a local Hamiltonian, setting $M=H$, $\rho=I/\trace(I)$, and considering $t=i\beta$ for $\beta\in\reals$ and $\beta\rightarrow\infty$. We refer the reader to the survey of Osborne for further details~\cite{O11}.

\subsection{Kitaev's quantum Cook-Levin theorem}\label{0_sscn:5LH}

In Section~\ref{0_sscn:LH}, we discussed the local Hamiltonian problem (LH) and its variants. As Chapter~\ref{chap:hardnessapprox} heavily exploits the structure and details of Kitaev's quantum version of the Cook-Levin theorem, i.e.\ his proof that $5$-LH is QMA-complete, we present the latter here. This requires two steps: One first shows that $\klhh\in\class{QMA}$ for $k\geq 1$. One then shows that $\klhh$ is QMA-hard for $k\geq 5$. Our discussion is based on a project completed by the present author for a graduate course on quantum complexity theory at the University of Waterloo~\cite{course09}, and follows the text of Kitaev, Shen, and Vyalyi~\cite{KSV02} closely. The reader is referred to the survey of Aharonov and Naveh for an alternate exposition~\cite{AN02}.

\subsubsection{Local Hamiltonian is in QMA}%\label{scn:inqma}
We begin by showing that $\klhh\in\qma$ for any constant $k$. Specifically, for any YES-instance $(H,a,b)$ of $\klhh$ with $k$-local Hamiltonian $H=\sum_{j=1}^r H_j\in \LL(\B^{\otimes n})$, we show that there exists a poly-size quantum proof $\ket{\psi}$ and a poly-size quantum verification circuit $V$, such that a single-qubit measurement on $V\ket{\psi}$ yields $1$ with high probability.

First, the quantum proof is constructed as $\ket{\psi}\in\complex^r\otimes\B^{\otimes n}\otimes\B$ as:
\begin{equation}
    \ket{\psi}=\left(\frac{1}{\sqrt{r}}\sum_{j=1}^{r}\ket{j}\right)\otimes\ket{\eta}\otimes\ket{0},
\end{equation}
for $\set{\ket{j}}_{j=1}^r$ an orthonormal basis for $\complex^r$, and $\ket{\eta}$ an eigenvector corresponding to some eigenvalue $\lambda$ of $H$. We call the first register of $\ket{\psi}$ the \emph{index} register, the second the \emph{proof} register, and the last the \emph{answer} register.

To define the verification procedure $V$, recall that $H=\sum_{j=1}^{r}H_j$, where each $H_j$ acts on the set of qubits denoted by $S_j$. Suppose $H_j$ has spectral decomposition $H_j=\sum_{s}\lambda_s \ketbra{\lambda_s}{\lambda_s}$. Then, define unitary $W_j$ acting on the proof and answer registers, i.e.\ $W_j\in \UU (\B^n\otimes\B)$, such that
\begin{equation}
    W_j\left(\ket{\lambda_s}\otimes\ket{0}\right) = \ket{\lambda_s}\otimes\left(\sqrt{\lambda_s}\ket{0} + \sqrt{1-\lambda_s} \ket{1}\right).
\end{equation}
Observe that one can implement this operation as follows. First, run phase estimation on $\exp(iH_j)$ to extract $\lambda_s$ to some ancilla register. Despite the fact that simulating $\exp(iH_j)$ can in general be costly, in our case, since $\abs{S_j}$ is constant, the simulation can be done efficiently. Conditioned on the value of the ancilla, we then rotate the answer register to obtain the desired superposition, and finally uncompute $\lambda_s$ in the ancilla. Define now unitary $V:=\sum_{j=1}^r\ketbra{j}{j}\otimes W_j$.

Having defined $\ket{\psi}$ and $V$, the verification procedure now proceeds as follows:
\begin{enumerate}
    \item Apply $V$ to $\ket{\psi}$.
    \item Measure the answer register and return the result.
\end{enumerate}

Let us analyze the probability of measuring $1$ in the answer register with this procedure. If we assume the index register is implicitly measured at the end of the verification, then we can think of Step 1 above as using the index register to choose an index $j$ uniformly at random, followed by applying $W_j$ to the proof register. Then, we can analyze the probability that this procedure returns $1$ as follows:
\begin{equation}
    \pr(\text{output 1}) = \sum_{j=1}^r\frac{1}{r}\pr(\text{output 1}\mid W_j\text{ is applied}),\label{eqn:2}
\end{equation}
where one has
\begin{eqnarray}
    \pr(\text{output 1}\mid W_j\text{ is applied}) &=& \trace\left[(I_{\B^{\otimes n}}\otimes \ketbra{1}{1})W_j(\ketbra{\eta}{\eta}\otimes\ketbra{0}{0})W_j^\dagger\right]\\
        &=&(\bra{\eta}\otimes\bra{0}) W_j^\dagger (I_{\B^{\otimes n}}\otimes \ketbra{1}{1})W_j (\ket{\eta}\otimes\ket{0}).\label{eqn:1}
\end{eqnarray}
The projector $\ketbra{1}{1}$ above acts on the answer register. To simplify this, rewrite $\ket{\eta}$ in the eigenbasis of $H_j$, i.e.\ $\ket{\eta}=\sum_s \alpha_s\ket{\lambda_s}$, and observe that%use the fact that we know the action of $W_j$ on $\set{\ket{\psi_s}}_s$ to obtain:
\begin{eqnarray}
    (I_{\B^{\otimes n}}\otimes \bra{1})W_j (\ket{\eta}\otimes\ket{0})&=&
        (I_{\B^{\otimes n}}\otimes \bra{1})W_j \left(\sum_s \alpha_s \ket{\lambda_s}\otimes\ket{0}\right)\\
        &=&(I_{\B^{\otimes n}}\otimes \bra{1}) \left[\sum_s \alpha_s \ket{\lambda_s}\otimes\left(\sqrt{\lambda_s}\ket{0}+\sqrt{1-\lambda_s}\ket{1}\right)\right]\nonumber\\
        &=&\sum_s \alpha_s \left(\sqrt{1-\lambda_s}\right)\ket{\lambda_s}.
\end{eqnarray}
Substituting this into Equation~(\ref{eqn:1}), we obtain:
\begin{eqnarray}
    \pr(\text{output 1}\mid W_j) &=& \left(\sum_t \alpha_t^\ast \left(\sqrt{1-\lambda_t}\right)\bra{\lambda_t}\right)\left(\sum_s \alpha_s \left(\sqrt{1-\lambda_s}\right)\ket{\lambda_s}\right)\nonumber\\
    &=& \sum_s(1-\lambda_s)\abs{\alpha_s}^2\\
    &=& 1- \sum_s\lambda_s\abs{\alpha_s}^2\\
    &=&1- \bra{\eta}H_j\ket{\eta},
\end{eqnarray}
where we have used the fact that $\sum_s\abs{\alpha_s}^2=1$. Substituting this into Equation~(\ref{eqn:2}) finally yields:
\begin{equation}
    \pr(\text{output 1}) = \sum_{j=1}^r\frac{1}{r}1- \bra{\eta}H_j\ket{\eta}
                = 1- \frac{1}{r}\bra{\eta}\left(\sum_{j=1}^rH_j\right)\ket{\eta}
                = 1- \frac{1}{r}\bra{\eta} H\ket{\eta}.
\end{equation}
Recalling that we chose $\eta$ to be an eigenvector of $H$ with some eigenvalue $\lambda$, we have that if $H$ corresponds to a YES instance (i.e.\ there exists $\lambda\leq a$), it follows that we can choose $\eta$ such that our verification procedure returns $1$ with probability $1-r^{-1}\lambda\geq 1-r^{-1}a$. On the other hand, if $H$ corresponds to a NO instance (i.e.\ for all $\lambda$, we have $\lambda \geq b$), we have $\pr(\text{output 1})\leq 1-r^{-1}b$. Since the probabilities in the YES and NO cases differ by an inverse polynomial in the input size, we can apply the error reduction techniques for QMA discussed in Section~\ref{0_sscn:classes} to conclude that $\lh\in\qma$.

\subsubsection{5-local Hamiltonian is hard for QMA}%\label{scn:solveqma}
We next show that $5$-local Hamiltonian is QMA-hard. To do so, we show a polynomial-time many-one or Karp reduction from an arbitrary problem in QMA to $\flh$.

To begin, let $P$ be a promise problem in $\qma$, and let $V=V_LV_{L-1}\dots V_1$ be a verification circuit for $P$ composed of unitaries $V_k$. Without loss of generality, we assume each $V_k$ acts on pairs of qubits. We assume $V\in \UU(\B^{\otimes m}\otimes \B^{\otimes N-m})$, where the $m$-qubit register contains the proof $V$ verifies, and the remaining qubits are ancilla qubits.

Our goal is to define a $5$-local Hamiltonian $H$ that will have a small eigenvalue if and only if there exists a proof $\ket{\psi}\in\B^{\otimes m}$ causing $V$ to accept with high probability. Kitaev's idea~\cite{KSV02} was to exploit the structure of $V$ by forcing the minimizing eigenvector of $H$ to ``simulate'' the action of $V$. To do so, let $H$ act on $\B^{\otimes m}\otimes \B^{\otimes N-m}\otimes\complex^{L+1}$, which is simply the initial space $V$ acts on, tensored with an $(L+1)$-dimensional \emph{counter} or \emph{clock} register. This clock register will ``keep track of time'' in the simulation, i.e.\ a value of $k$ in the register will correspond to having ``applied'' $V_1\dots V_k$. For clarity of exposition, where necessary, we label the three registers $H$ acts on as $p$ for proof, $a$ for ancilla, and $c$ for clock, respectively.

Having defined the space $H$ acts on, we now define $H$ itself:
\begin{equation}
    H := \hin + \hprop + \hout,
\end{equation}
with the terms $\hin$, $\hprop$, and $\hout$ defined as follows (intuitive explanations to follow). Let
\begin{equation}
    \hin:=I_p\otimes
    \left(I_a - \ketbra{0\dots0}{0\dots0}_a\right)\otimes \ketbra{0}{0}_c.
\end{equation}
Note that the projector $(I_a - \ketbra{0\dots0}{0\dots0}_a)$ is used here for simplicity of exposition; the same analysis holds if we instead use the $1$-local constraint $\sum_{i=1}^{N-m}(\ketbra{1}{1}_{i})_c$ (where the $i$th projector acts on the $i$th ancilla qubit) --- hence, we do not violate the constraint that $H$ be $5$-local. Next, $\hout$ is defined as
\begin{equation}
    \hout:=\left(\ketbra{0}{0}\otimes I_{\B^{\otimes m-1}}\right)_p\otimes
     I_a\otimes \ketbra{L}{L}_c.
\end{equation}
Finally, define $\hprop$ as
\begin{eqnarray}
    \hprop &:=& \sum_{j=1}^{L} H_j \text{,~~~~where}\nonumber\\
    H_j &:=& -\frac{1}{2}V_j\otimes\ketbra{j}{j-1}_c -\frac{1}{2}V_j^\dagger\otimes\ketbra{j-1}{j}_c +\\&&\frac{1}{2}I\otimes(\ketbra{j}{j}+\ketbra{j-1}{j-1})_c.\label{eqn:H_j}
\end{eqnarray}

\noindent Each of the terms $\hin$, $\hout$, and $\hprop$ allow us to ``force'' the minimizing eigenvector of $H$ to ``simulate'' V as follows. Recall that our goal is to have $\bra{\eta}H\ket{\eta}$ for some $\ket{\eta}\in\B^{\otimes m}\otimes \B^{\otimes N-m}\otimes\complex^{L+1}$ be small if and only if $V$ outputs $1$ with high probability on some proof $\ket{\psi}\in\B^{\otimes m}$. Suppose such a $\ket{\psi}$ exists. Then, for $\hin$, note that when one runs $V$ on $\ket{\psi}$, the initial state should be $\ket{\psi}_p\otimes\ket{0}^{\otimes N-m}_a$, i.e.\ all ancilla qubits should be set to $0$, with the purported proof in the proof register. But $\hin$ enforces \emph{precisely} this constraint for any $\ket{\eta}$. In particular, if the clock register of $\ket{\eta}$ is in state $\ketbra{0}{0}_c$ and the ancilla register is \emph{not} all zeroes, then we have $\bra{\eta}\hin\ket{\eta}>0$, i.e.\ $\ket{\eta}$ incurs an energy penalty. In other words, if $\ket{\eta}$ does not simulate the initial state of the verification procedure $V$, $\hin$ penalizes $\ket{\eta}$. Next, for $\hout$, note that after running $V$ on $\ket{\psi}$, we expect the first qubit in the proof register to be a $1$ with high probability. Again, observe that $\hout$ enforces exactly this constraint on $\ket{\eta}$ --- if the clock register is in state $\ketbra{L}{L}_c$ and the first qubit reads $0$, we again have $\bra{\eta}\hout\ket{\eta}>0$. Finally, $\hprop$ follows the same idea by forcing $\ket{\eta}$ to encode in superposition a simulation of each step of the verification procedure $V$. It follows that the minimizing vector $\ket{\eta}$ is of the following form, often called a \emph{history state}:
\begin{equation}\label{eqn:eta}
    \ket{\eta} := \frac{1}{\sqrt{L+1}}\sum_{j=0}^{L}\left(V_j\dots V_1 \ket{\psi}_p\otimes\ket{0}^{\otimes N-m}_a\right)\otimes \ket{j}_c.
\end{equation}

\noindent To recap, if there exists a $\ket{\psi}$ such that $V$ accepts with high probability, then the history state $\ket{\eta}$ corresponds to a small eigenvalue of $H$. On the other hand, if no such $\ket{\psi}$ exists, either $\ket{\eta}$ will be of the form in Equation~(\ref{eqn:eta}) (i.e.\ will faithfully simulate $V$), in which case we are hit with a large penalty by $\hout$ since the answer qubit cannot be $1$ with high probability, or $\ket{\eta}$ ``cheats'' by deviating from either the initial conditions or the intermediate steps of the protocol, in which case the terms $\hin$ and $\hout$ hit $\ket{\eta}$ with an energy penalty, respectively. Thus, the corresponding energy of $\ket{\eta}$ would be large. Of course, it remains to show that this intuition is indeed correct!

Before we begin, we first apply the following change of basis operator to $\hprop$, which greatly simplifies the analysis (intuition to follow):
\begin{equation}
    W=\sum_{j=0}^{L} V_j\dots V_1\otimes\ketbra{j}{j}_c.
\end{equation}
Thus, instead of $\ket{\eta}$ and $H$, we consider $\ket{\hat{\eta}}:=W\ket{\eta}$ and $\hat{H}:=W^\dagger H W$. To see what $\hat{H}$ looks like, we analyze the action of $W$ on each of $\hin$, $\hout$, and $\hprop$ separately. Observe first that $\hat{H}_{\rm in}:=W^\dagger \hin W=\hin$, since at time $0$, $W$ implicitly applies the identity to the proof and ancilla registers. Second, for $\hout$, we have
\begin{equation}
    \hat{H}_{\rm out}:=W^\dagger \hout W=V^\dagger \left[\left(\ketbra{0}{0}\otimes I_{\B^{\otimes m-1}}\right)_p\otimes
     I_a \right]V\otimes\ketbra{L}{L}_c= (V^\dagger\otimes I_c)\hout(V\otimes I_c),
\end{equation}
since at time $L$, $W$ applies the entire circuit $V$. Finally, for $\hprop$, considering the effect of $W$ on each component of $H_j$ in Equation~(\ref{eqn:H_j}) separately and using simple algebra, one finds
\begin{equation}
    W^\dagger H_j W = I_{p,a}\otimes\frac{1}{2}(\ketbra{j-1}{j-1}-\ketbra{j-1}{j}-\ketbra{j}{j-1}+\ketbra{j}{j})_c=I_{p,a}\otimes\frac{1}{2}\left(
                                                                                                     \begin{array}{cc}
                                                                                                       1 &-1 \\
                                                                                                       -1 &1 \\
                                                                                                     \end{array}
                                                                                                   \right)_c.
\end{equation}
It follows that $\hat{H}_{\rm prop}=\sum_j W^\dagger H_j W$ is tridiagonal and of the form
\begin{equation}
    \hat{H}_{\rm prop}=I_{p}\otimes I_a\otimes\left(
      \begin{array}{cccccc}
        \frac{1}{2} & -\frac{1}{2} & 0 & 0 & 0 & \dots \\
        -\frac{1}{2} & 1 & -\frac{1}{2} & 0 & 0 & \dots \\
        0 & -\frac{1}{2} & 1 & -\frac{1}{2} & 0 & \dots \\
        0 & 0 & -\frac{1}{2} & 1 & -\frac{1}{2} & \dots \\
        0 & 0 & 0 & -\frac{1}{2} & \ddots & \ddots \\
        \vdots & \vdots & \vdots & \vdots & \ddots & \ddots \\
      \end{array}
    \right)=:I_{p}\otimes I_a\otimes E_c,
\end{equation}
where we have let $E$ denote the tridiagonal matrix acting on the clock register for later reference. Intuitively, one can think of the change of basis $W$ as ``flushing out'' the computation $V$, so that it is pushed to the very end to time step $L$ (hence $V$ only appears in $\hout$). This has the effect of simplifying $\hat{H}_{\rm prop}$ to a nice tridiagonal form, since it no longer needs to keep track of the unitaries $V_i$.

Finally, observe that since $W$ is unitary, $\hat{H}$ and $H$ have precisely the same set of eigenvalues. We can thus work with $\hat{H}$ instead of $H$ in our eigenvalue analysis. Hence, for the remainder of this section, by $\ket{\eta}$ we shall mean $\ket{\hat{\eta}}$, and by $H$, we mean $\hat{H}$. We now show that $H$ has the correct spectral properties for both YES and NO instances of $5$-LH.

\subsubsection{YES case: $H$ has a small eigenvalue}%\label{scn:small}
We have thus far set up a Hamiltonian $H\in \HH(\B^{\otimes m}\otimes \B^{\otimes N-m}\otimes\complex^{L+1})$ corresponding to the verification procedure $V\in \UU(\B^{\otimes m}\otimes \B^{\otimes N-m})$. We now show that if there exists such a $\ket{\psi}\in\B^{\otimes m}$ which causes $V$ to output $1$ with high probability, then $H$ must have a small eigenvalue.

Suppose there exists $\ket{\psi}$ such that a measurement of the first qubit of $V\ket{\psi}$ yields $1$ with probability at least $1-\epsilon$. To demonstrate that $H$ has a small eigenvalue, we explicitly construct a vector $\ket{\eta}\in \B^{\otimes m}\otimes \B^{\otimes N-m}\otimes\complex^{L+1}$ such that $\bra{\eta}H\ket{\eta}$ is small. Let \begin{equation}
\ket{\eta}=\ket{\psi}_p\otimes\ket{0}^{\otimes N-m}_a\otimes\ket{\gamma}_c,
\end{equation}
where
\begin{equation}
    \ket{\gamma} := \frac{1}{\sqrt{L+1}}\sum_{j=0}^L\ket{j}.\label{eqn:gamma}
\end{equation}

We analyze $\bra{\eta}H\ket{\eta}$ by considering $\hin$, $\hprop$, and $\hout$ separately. First, observe that $\bra{\eta}\hin\ket{\eta}=0$, since the ancilla register of $\ket{\eta}$ is in the all zeroes state. For $\hprop$, we have that
\begin{equation}
    \bra{\eta}\hprop\ket{\eta} = \bra{\eta}I_{p,a}\otimes E_c\ket{\eta} = \bra{\gamma}E\ket{\gamma}=0,
\end{equation}
where in the last equality we have used the fact that the sum of each row and column of $E$ is $0$, implying $\ket{\gamma}$ is a $0$-eigenvector of $E$. Note that we have not used the probability of $V$ answering $1$ yet --- this now comes in handy for $\hout$, where
\begin{eqnarray}
    \bra{\eta}\hout\ket{\eta} &=& \bra{\eta}\left(V^\dagger \left[\left(\ketbra{0}{0}\otimes I_{\B^{\otimes m-1}}\right)_p\otimes
     I_a \right]V\otimes\ketbra{L}{L}_c\right)\ket{\eta}\\
        &=& \frac{1}{L+1}\left[\bra{\psi}_p\otimes\bra{0}^{\otimes N-m}_aV^\dagger\right] \left[\left(\ketbra{0}{0}\otimes I_{\B^{\otimes m-1}}\right)_p\otimes
     I_a \right]\left[V\ket{\psi}_p\otimes\ket{0}^{\otimes N-m}_a \right].\nonumber
\end{eqnarray}
Observe, however, that this expression corresponds to the probability that we begin with the proof $\ket{\psi}_p\otimes\ket{0}^{\otimes N-m}_a$, apply the verification $V$, and then measure the first qubit and obtain $0$. By our assumption at the beginning of this section, this probability is at most $\epsilon$. Hence,
\begin{equation}
    \bra{\eta}\hout\ket{\eta} \leq \frac{1}{L+1}\epsilon,
\end{equation}
implying there must exist an eigenvalue for $H$ of value at most $\epsilon/(L+1)$. Thus, if we have a YES-instance of our QMA problem $P$, then $H$ has a small eigenvalue, as required.

\subsubsection{NO case: $H$ has no small eigenvalues}%\label{scn:nosmall}
We now show that if there does not exist such a proof $\ket{\psi}$ which causes verification procedure $V$ to output $1$ with high probability, then $H$ must have no small eigenvalues.

Suppose that for all proofs $\ket{\psi}$, $V$ does not output $1$ with probability more than $\epsilon$. To lower bound the eigenvalues of $H$, we play a game of divide-and-conquer by letting $H=A_1+A_2$, where $A_1 := \hin + \hout$, and $A_2 := \hprop$, and analyzing the eigenvalues of $A_1$ and $A_2$ separately. The challenge arises in combining these separate eigenvalue estimates into eigenvalue estimates for $H$, since unfortunately, $[A_1,A_2]\neq0$, implying that $A_1$ and $A_2$ do not diagonalize in a common basis. To surmount this obstacle, Kitaev uses the following approach~\cite{KSV02}:

\begin{enumerate}
    \item We first prove Kitaev's Geometric Lemma (Lemma~\ref{l:pluslemma}), which takes as input operators $B$ and $C$, as well as a set of parameters $S$ dependent on $B$ and $C$, and outputs a lower bound on the eigenvalues of $B+C$.
    \item We compute the parameters $S$ relevant to our specific operators $A_1$ and $A_2$, and plug them into Lemma~\ref{l:pluslemma} to show that $H=A_1+A_2$ has no small eigenvalues.
\end{enumerate}

We now state and prove Kitaev's Geometric Lemma.

\begin{lemma}[Kitaev, Shen, Vyalyi~\cite{KSV02}, Geometric Lemma, Lemma 14.4]\label{l:pluslemma}
    Let $A_1,A_2\succeq 0$, such that the minimum \emph{non-zero} eigenvalue of both operators is lower bounded by $v$. Assume that the null spaces $\nl$ and $\nll$ of $A_1$ and $A_2$, respectively, have trivial intersection, i.e.\ $\nl\cap\nll=\set{\ve{0}}$. Then
    \begin{equation}
        A_1+A_2 \succeq 2v\sin^2\frac{\alpha(\nl,\nll)}{2}I\enspace,
    \end{equation}
    where the \emph{angle} $\alpha(\spa{X},\spa{Y})$ between $\spa{X}$ and $\spa{Y}$ is defined over unit vectors $\ket{x}$ and $\ket{y}$ as $\cos\left[\angle(\spa{X},\spa{Y})\right] := \max_{\ket{x}\in\spa{X},\ket{y}\in\spa{Y}}\abs{\braket{x}{y}}$.
\end{lemma}

%\begin{lemma}[Kitaev, Shen, Vyalyi~\cite{KSV02}, Lemma 14.4]\label{l:pluslemma}
%    Let $A_1,A_2\succeq 0$, such that the minimum \emph{non-zero} eigenvalue of both operators is lower bounded by $v$. Assume that the null spaces $\nl$ and $\nll$ of $A_1$ and $A_2$, respectively, have trivial intersection, i.e.\ $\nl\cap\nll=\set{\ve{0}}$. Then
%    \begin{equation}
%        A_1+A_2 \succeq 2v\sin^2\frac{\alpha(\nl,\nll)}{2}I,
%    \end{equation}
%    where $\alpha(\spa{X},\spa{Y})$ is the \emph{angle} between spaces $\spa{X}$ and $\spa{Y}$.
%\end{lemma}
%Here, the angle $\alpha(\spa{X},\spa{Y})$ is defined as follows, where the maximization is over unit vectors:
%\begin{equation}
%    \cos\left[\alpha(\spa{X},\spa{Y})\right] := \operatorname{max}_{\ket{x}\in\spa{X},\ket{y}\in\spa{Y}}\abs{\braket{x}{y}}.
%\end{equation}
Note that if $\spa{X}$ and $\spa{Y}$ have non-trivial intersection, i.e.\ there exists $\ket{x}\neq \ve{0}$ such that $\ket{x}\in\spa{X}$ and $\ket{x}\in\spa{Y}$, then $\alpha(\spa{X},\spa{Y})$ is trivially $0$. Also, note that demanding  $\spa{X}\cap\spa{Y}=\set{\ve{0}}$ is not equivalent to demanding $\spa{X}$ and $\spa{Y}$ be orthogonal --- for example, the spaces $\spa{X}=\operatorname{span}\set{\ket{0}}$ and $\spa{Y}=\operatorname{span}\set{\ket{+}}$ contain elements which have non-zero overlap, but the sets have trivial intersection.

We now tackle step 1 of Kitaev's approach by proving Lemma~\ref{l:pluslemma}.
\begin{proof}
    By the definition of $v$, we have $A_1\succeq v(I-\Pi_{\nl})$ and $A_2\succeq v(I-\Pi_{\nll})$, where $\Pi_{\spa{X}}$ denotes the projector onto $\spa{X}$. Combining the latter two, it follows that it suffices to show $v(I-\Pi_{\nl})+v(I-\Pi_{\nll})\succeq 2v\sin^2(\alpha(\nl,\nll)/2)$. By rearranging terms and using the identity $\cos(2\theta)=1-2\sin^2\theta$, this is the equivalent of showing
    \begin{equation}
        \Pi_{\nl}+\Pi_{\nll} \preceq \left[1+\cos(\alpha(\nl,\nll))\right]I.\label{eqn:5}
    \end{equation}

    To upper bound the eigenvalues of $\Pi_{\nl}+\Pi_{\nll}$, suppose we have some eigenvector $\ket{\zeta}$ with corresponding eigenvalue $\lambda > 0$. Let $\ket{x_1}\in\nl$ and $\ket{x_2}\in\nll$ be unit vectors such that $\Pi_{\nl}\ket{\zeta}=u_1\ket{x_1}$ and     $\Pi_{\nll}\ket{\zeta}=u_2\ket{x_2}$ for some real $u_1,u_2>0$. Then:
    \begin{equation}
        \lambda = \bra{\zeta}(\Pi_{\nl}+\Pi_{\nll})\ket{\zeta}=u_1\braket{\zeta}{x_1}+u_2\braket{\zeta}{x_2}=u_1^2+u_2^2.\label{eqn:3}
    \end{equation}
    Further, since $\lambda\ket{\zeta} = (\Pi_{\nl}+\Pi_{\nll})\ket{\zeta} = u_1\ket{x_1}+u_2\ket{x_2}$, we can also derive a non-equivalent expression for $\lambda^2$, i.e.\
    \begin{equation}
        \lambda^2 = [\bra{\zeta}\lambda][\lambda\ket{\zeta}]=(u_1\bra{x_1}+u_2\bra{x_2})(u_1\ket{x_1}+u_2\ket{x_2})=u_1^2 + u_2^2 + 2u_1u_2\operatorname{Re}\braket{x_1}{x_2},\label{eqn:4}
    \end{equation}
    where $\operatorname{Re}(x)$ denotes the real part of $x\in\complex$. Combining Eqns.~(\ref{eqn:3}) and~(\ref{eqn:4}) by taking the following linear combination, we have:
    \begin{eqnarray}
        (1+\abs{\operatorname{Re}\braket{x_1}{x_2}})\lambda-\lambda^2 &=& (1+\abs{\operatorname{Re}\braket{x_1}{x_2}})(u_1^2+u_2^2) - (u_1^2 + u_2^2 + 2u_1u_2\operatorname{Re}\braket{x_1}{x_2})\nonumber\\
        &=&u_1^2\abs{\operatorname{Re}\braket{x_1}{x_2}} + u_2^2\abs{\operatorname{Re}\braket{x_1}{x_2}} - 2u_1u_2\operatorname{Re}\braket{x_1}{x_2}\\
        &=&\abs{\operatorname{Re}\braket{x_1}{x_2}} \left(u_1^2 + u_2^2 \pm 2u_1u_2\right)\\
        &=&\abs{\operatorname{Re}\braket{x_1}{x_2}} \left(u_1\pm u_2\right)^2\\
        &\geq& 0.
    \end{eqnarray}
    Moving $\lambda^2$ to the right side of the last inequality and dividing through by $\lambda$ hence gives
    \begin{equation}
    \lambda \leq (1+\abs{\operatorname{Re}\braket{x_1}{x_2}})\leq 1+\cos(\alpha(\nl,\nll)),
    \end{equation}
    where the latter inequality follows straightforwardly from the definition of $\alpha(\nl,\nll)$. We thus have that all eigenvalues of $\Pi_{\nl}+\Pi_{\nll}$ are upper bounded by $1+\cos(\alpha(\nl,\nll))$, which by Equation~(\ref{eqn:5}) implies the desired lower bound on $A_1 + A_2$.
\end{proof}

We now move to step 2 of Kitaev's approach, i.e.\ we now use Lemma~\ref{l:pluslemma} to lower bound the eigenvalues of $H=A_1+A_2$. To do so, we must determine the values of parameters $v$ and $\alpha(\nl,\nll)$ for $A_1$ and $A_2$ used in Lemma~\ref{l:pluslemma}. Recall that Lemma~\ref{l:pluslemma} also requires $\nl\cap\nll=\set{\ve{0}}$ --- we handle this constraint at the end of the section (at which point it will be obvious, given the analysis to come).

We start with $v$, which is the lower bound on the positive eigenvalues of both $A_1$ and $A_2$. Note that since $A_1=\hin+\hout$ is simply a sum of commuting projectors, its eigenvalues must be non-negative integers. In particular, its smallest \emph{positive} eigenvalue is at least $1$. For $A_2$, since $A_2=\hprop=I_{p,a}\otimes E_c$, its eigenvalues will be determined by those of $E_c$. The eigenvalues of the latter are~\cite{KSV02} $\lambda_k=1-\cos[\pi k/(L+1)]$ for $0\leq k\leq  L$. This expression is clearly minimized when $k=1$ (note $k=0$ would yield a zero eigenvalue), implying the smallest positive eigenvalue of $A_2$ is at least
\begin{equation}
1-\cos(\pi/(L+1))\geq c/L^2
\end{equation}
for some constant $c$. To see why this inequality holds, use the Taylor series expansion for $\cos x$ to show that whenever $x\leq 1$, one has
\begin{equation}
    \cos x = 1 - \frac{x^2}{2!} + \frac{x^4}{4!} - \frac{x^6}{6!} +\ldots \leq 1-\frac{x^2}{2!} + \frac{x^4}{4!}\leq 1 - \left(\frac{1}{2!} - \frac{1}{4!}\right)x^2 = 1-c x^2.
\end{equation}
Taking the minimum of our lower bounds for $A_1$ and $A_2$ thus yields that $v\in\Omega(1/L^2)$.

We next estimate the angle $\alpha(\nl,\nll)$ between the null spaces $\nl$ and $\nll$ of $A_1$ and $A_2$, respectively. This can be done by exploiting the structure of $\nl$ and $\nll$. In particular, we have that
\begin{eqnarray}
    \nl &=& \left[(\B^{\otimes m})_p\otimes \ket{0}^{\otimes N-m}_a\otimes \ket{0}_c\right]\oplus
    \nonumber\\
                    &&\left[ (\B^{\otimes N})_{p,a}\otimes \operatorname{span}(\ket{1},\ldots \ket{L-1})_c\right]\oplus\nonumber\\
                    && \left[
                    V^\dagger (\ket{1}\otimes\B^{\otimes N-1})_{p,a}\otimes \ket{L}_c\right],\label{eqn:nl}
\end{eqnarray}
where each of the three terms in this expression follow directly from the definitions of $\hin$ and $\hout$ (e.g.\ any state with the clock register set to $0$ and all zeroes in the ancilla is a $0$-eigenvector of both $\hin$ and $\hout$). Similarly, we have
\begin{equation}
    \nll = (\B^{\otimes N})_{p,a}\otimes \ket{\gamma}_c,\label{eqn:nll}
\end{equation}
which follows straightforwardly if we recall that $\hprop=I_{p,a}\otimes E_c$ and $E\ket{\gamma}=\ve{0}$, for $\ket{\gamma}$ defined in Equation~(\ref{eqn:gamma}).

To exploit this structure, instead of estimating $\alpha(\nl,\nll)$, we estimate $\cos^2\alpha(\nl,\nll)$, which can be rewritten in the form (where the maximization is over unit vectors):
\begin{equation}
    \cos^2\alpha(\nl,\nll)=\operatorname{max}_{\ket{x}\in\nl,\ket{y}\in\nll}\abs{\braket{x}{y}}^2
    =\operatorname{max}_{\ket{x}\in\nl,\ket{y}\in\nll}\braket{y}{x}\braket{x}{y}
    =\operatorname{max}_{\ket{y}\in\nll}\bra{y}\Pi_{\nl}\ket{y}.\label{eqn:6}
\end{equation}
The last equality holds without loss of generality since the maximum for $\bra{y}\Pi_{\nl}\ket{y}$ is achieved by projecting onto a pure state $\ketbra{x}{x}$ for some $\ket{x}\in\nl$. Let us upper bound the rightmost term in the equation above. Observe that by Equation~(\ref{eqn:nll}), any $\ket{y}\in\nll$ has the form $\ket{y}=\ket{\zeta}_{p,a}\otimes\ket{\gamma}_c$ for some $\ket{\zeta}\in\B^{\otimes m}\otimes\B^{\otimes {N-m}}$. Since by Equation~(\ref{eqn:nl}), $\Pi_{\nl}$ breaks down into a sum of three projections, we can bound $\bra{y}\Pi_{\nl}\ket{y}$ by determining the contribution of each projector separately when sandwiched by $\ket{y}$.

The contribution of the second projection is easiest to see --- it is simply $(L-1)/(L+1)$, since every term in $\ket{\gamma}$ except $\ket{0}$ and $\ket{L}$ contribute $1/(L+1)$ to the sum.

As for the first and third projections, let $\mathcal{K}_1=\B^{\otimes m}\otimes \ket{0}^{\otimes N-m}$ and $\mathcal{K}_2=V^\dagger \ket{1}\otimes\B^{\otimes N-1}$. Then the contribution of the first and third projections is given by:
\begin{equation}
    \bra{y}(\Pi_{\mathcal{K}_1}\otimes\ketbra{0}{0}_c + \Pi_{\mathcal{K}_2}\otimes\ketbra{L}{L}_c)\ket{y}=\frac{1}{L+1}\bra{\zeta}(\Pi_{\mathcal{K}_1} + \Pi_{\mathcal{K}_2})\ket{\zeta}.\label{eqn:8}
\end{equation}
If we let $\varphi(\mathcal{K}_1,\mathcal{K}_2)$ denote the angle between $\mathcal{K}_1$ and $\mathcal{K}_2$, we can straightforwardly use Equation~(\ref{eqn:5}) to bound the quantity above by
\begin{equation}
    \frac{1}{L+1}\bra{\zeta}(\Pi_{\mathcal{K}_1} + \Pi_{\mathcal{K}_2})\ket{\zeta} \leq \frac{1}{L+1}\left(1+\cos\varphi(\mathcal{K}_1,\mathcal{K}_2)\right).
\end{equation}
Observe, however, that
\begin{equation}
    \cos^2\varphi(\mathcal{K}_1,\mathcal{K}_2)=\operatorname{max}_{\ket{k}\in\mathcal{K}_1,\ket{l}\in\mathcal{K}_2}\abs{\braket{k}{l}}^2,
\end{equation}
where $\mathcal{K}_1$ is just the set of initial states with all-zero ancilla for the verification procedure $V$, and $\mathcal{K}_2$ is the set of initial states for which applying $V$ yields a $1$ in the first qubit with certainty. Hence, the maximum overlap between vectors in $\mathcal{K}_1$ and $\mathcal{K}_2$ is directly tied to the maximum probability with which we can obtain outcome $1$ with an initial state with all-zero ancilla. In particular, we have $\cos^2\varphi(\mathcal{K}_1,\mathcal{K}_2)$ equals the maximum probability of outputting $1$. Since in this section we are dealing with the NO case, however, meaning no proof can cause an output of $1$ with probability greater than $\epsilon$, we have $\cos^2\varphi(\mathcal{K}_1,\mathcal{K}_2)\leq\epsilon$, implying:
\begin{equation}
    \frac{1}{L+1}\left(1+\cos\varphi(\mathcal{K}_1,\mathcal{K}_2)\right)\leq \frac{1}{L+1}\left(1+\sqrt{\epsilon}\right).
\end{equation}
Adding the contributions of all three projections thus yields:
\begin{equation}
    \cos^2\alpha(\nl,\nll)
    =\operatorname{max}_{\ket{y}\in\nll}\bra{y}\Pi_{\nl}\ket{y} \leq \left(\frac{L-1}{L+1}\right) + \left(\frac{1+\sqrt{\epsilon}}{L+1}\right)=1-\frac{1-\sqrt{\epsilon}}{L+1}.
\end{equation}
Using the identity $\sin^2 x + \cos^2 x=1$, this implies $\sin^2\alpha(\nl,\nll)\geq (1-\sqrt{\epsilon})/(L+1)$. Then, since $\sin^2 \frac{x}{2}\geq \frac{1}{4}\sin^2 x$ (shown using the identity $\sin (2x)=2\sin x\cos x$), we have
\begin{equation}
    \sin^2\frac{\alpha(\nl,\nll)}{2}\geq \frac{1}{4}\sin^2\alpha(\nl,\nll)\geq \frac{1-\sqrt{\epsilon}}{4(L+1)}.
\end{equation}

Finally, we have all estimates required to use Lemma~\ref{l:pluslemma}: $v=\Delta/L^2$ for some constant $\Delta$ and $\sin^2[\alpha(\nl,\nll)/2]\geq (1-\sqrt{\epsilon})/[4(L+1)]$. In addition, given Equations~(\ref{eqn:nl}) and~(\ref{eqn:nll}), it is now easy to see that $\nl\cap\nll=\set{\ve{0}}$ (as required by Lemma~\ref{l:pluslemma}), since any state of the tensor product form $\ket{\psi}_{p,a}\otimes\ket{\gamma}_c$ cannot live in $\nl$. Plugging everything into Lemma~\ref{l:pluslemma}, we conclude that in the NO case, the minimum eigenvalue of $H$ is of the order $\Omega((1-\sqrt{\epsilon})/L^3)$ (i.e.\ $H$ has no ``small'' eigenvalues). As required by Definition~\ref{0_def:localhamiltonianproblem}, note that this lower bound is inverse polynomially separated from the upper bound on the smallest eigenvalue of $H$ from the YES case if we first apply error reduction to $V$ to bring $\epsilon$ inverse polynomially close to $0$.

\subsubsection{Is the Hamiltonian $H$ $5$-local?}%\label{sscn:counter}

We have so far set up a Hamiltonian $H$ whose eigenvalues are small or large, depending on whether we have a YES or NO instance of our QMA problem $P$, respectively. We now ask: Is $H$ $5$-local?

The answer is \emph{almost}. Recall that $H\in \HH(\B^{\otimes m}\otimes \B^{\otimes N-m}\otimes\complex^{L+1})$, where the counter register is $\complex^{L+1}$. If we implement the counter straightforwardly using $O(\log{L})$ qubits, the resulting operations on it, such as incrementing the counter, could require updating all $O(\log{L})$ qubits, making $H$ $(\log {L})$-local at best. In order to circumvent this, Kitaev~\cite{KSV02} uses a different representation for the counter for which any operation requires acting on at most $3$ qubits of the counter. Specifically, we let $H$ act on $\B^{\otimes m}\otimes \B^{\otimes N-m}\otimes\B^{L}$, where the counter register is now given in \emph{unary}, i.e.\ $\ket{j}\in\complex^{L+1}$ is represented as
\begin{equation}
    |\underbrace{1,\ldots,1}_{j},0,\ldots,0\rangle.\label{eqn:7}
\end{equation}

The operator basis $\ketbra{i}{j}$ for $\LL(\complex^{L+1})$ translates to this new representation as follows. Operator $\ketbra{j}{j}\in\LL(\complex^{L+1})$ is mapped to $\ketbra{1}{1}_j\otimes\ketbra{0}{0}_{j+1}$ in the new space, i.e.\ being in state $\ket{j}$ in the old encoding is equivalent to having the $j$th qubit set to 1 and the $(j+1)$-th qubit set to $0$ in the new encoding. Similarly, operator $\ketbra{j-1}{j}$ is mapped to $\ketbra{1}{1}_{j-1}\otimes\ketbra{0}{1}_j\otimes\ketbra{0}{0}_{j+1}$, i.e.\ if we think of $\ketbra{j-1}{j}$ as moving us from state $\ket{j}$ to $\ket{j-1}$, this is equivalent in the new encoding to flipping the $j$th bit to $0$, followed by a safety check that qubits $j-1$ and $j+1$ are $1$ and $0$, respectively. The remaining basis elements are defined analogously. These operations are at most $3$-local. Combined with the fact that $H$ is based on the verification circuit $V$, which itself is composed of $2$-qubit unitaries $V_i$, we have that $H$ is $5$-local Hamiltonian, as desired.

With $H$ being $5$-local, there is one final issue to be addressed --- since the counter is now represented using a larger space, one must deal with the possibility of \emph{invalid} settings to the counter register. To discourage such behavior, a fourth penalty term is added to $H$ acting only on the counter space, namely
\begin{equation}
    \hstab := I_{p,a}\otimes \sum_{j=1}^{L-1}\ketbra{0}{0}_j\otimes\ketbra{1}{1}_{j+1}.
\end{equation}
Hence, the new $H$ is given by $H=\hin+\hprop+\hout+\hstab$. Note that $\hstab$ discourages counter states which are not of the form in Equation~(\ref{eqn:7}), i.e.\ states containing the subsequence $01$ are given an energy penalty.

Does the previous analysis of the smallest eigenvalue of $H$ still hold when $\hstab$ is added to the picture? The answer is \emph{yes}. The YES case is easy to see, since all valid counter states are in the null space of $\hstab$. Thus, an honest proof receives no energy penalty from $\hstab$, as desired.

For the NO case, let $\spa{S}= \B^{\otimes m}\otimes \B^{\otimes N-m}\otimes\complex^{L+1}$ (the original space we had defined $H$ as acting on). Observe that $\hin+\hprop+\hout$ and $\hstab$ both act invariantly on $\spa{S}$, meaning they map operators in $\spa{S}$ to operators in $\spa{S}$. Thus, we can split our analysis into two independent cases: when $H$ acts on $\spa{S}$, and when $H$ acts on the orthogonal complement of $\spa{S}$, denoted $\spa{S}^\perp$. In the former case, $\hstab$ is just the zero operator with respect to $\spa{S}$; thus, the previous eigenvalue analysis goes through unscathed, yielding an eigenvalue lower bound on $H$ of $\Omega((1-\sqrt{\epsilon})/L^3)$. As for the second case when $H$ is restricted to $\spa{S}^\perp$, observe that $\hstab$ always administers an energy penalty, since $\spa{S}^\perp$ contains only invalid counter states. Since $\hstab$ is a sum of commuting projectors, its eigenvalues will be non-negative integers --- in particular, its smallest non-zero eigenvalue is at least $1$. Since $\hin+\hprop+\hout\succeq 0$, it follows that when restricted to $\spa{S}^\perp$, we have $H\succeq 1$. Taking the minimum of the estimates for the two cases of $\spa{S}$ and $\spa{S}^\perp$ yields the desired bound that the smallest eigenvalue of $H$ is still in $\Omega((1-\sqrt{\epsilon})/L^3)$, despite the new representation for the counter. This concludes Kitaev's proof that $5$-local Hamiltonian is complete for QMA.

%\subsection{Simulating quantum systems}\label{0_sscn:simulation}

\section{Quantum correlations}\label{0_scn:nonclassicality}

As mentioned earlier, the growing field of quantum computation and information has positively impacted both computer science and physics. The next area this thesis studies has in particular benefited greatly from this cross-fertilization, and is the study of \emph{quantum correlations}. Here, we are interested in understanding correlations between individual quantum subsystems of a larger composite system. Specifically, we shall introduce and discuss two notions of quantum correlations: \emph{quantum entanglement} and \emph{non-classical correlations}.

\paragraph{Motivation.} We mention two reasons why the study of quantum correlations is important. The first is that the existence of certain correlations predicted by quantum theory, specifically quantum entanglement, has long troubled physicists. In a letter to Max Born in 1947, for example, Einstein dubs entanglement as ``spukhafte Fernwirkung'', or ``spooky action at a distance''~\cite{BE71}. This mentality was moreover the basis for the rejection of quantum mechanics as a complete physical theory a decade earlier by the famous Einstein, Podolsky, and Rosen (EPR) paper of 1935~\cite{EPR35}. Thus, a better understanding of quantum correlations appears to be key to understanding both the nature of our world around us, as well as our theories describing this world. The second reason is that quantum correlations are generally believed to be required for quantum computers to outperform their classical counterparts. It has been rigorously shown, for example, that in the pure-state setting, the amount of entanglement present in a quantum system must grow with the problem size if a quantum computation is to achieve an exponential speedup over classical computers~\cite{jozsa03a}. Thus, a better understanding of quantum correlations may prove advantageous for designing quantum algorithms, as well as for uncovering the boundary between classical and quantum computing.

\subsection{Quantum entanglement}\label{0_sscn:entanglement}

The canonical notion of quantum correlations between quantum systems dates back to the EPR paper of 1935~\cite{EPR35}, and is called \emph{quantum entanglement}. The name ``entanglement'' was coined by physicist Erwin Schr\"{o}dinger, who used the term ``Vershr\"{a}nkung'' in 1935~\cite{S35}, which in colloquial ``non-physicist'' German means ``folding of the arms''~\cite{B01}. Much has been discovered in the field of \emph{entanglement theory} over the last two decades, from its quantification and characterization, to its manipulation and use for quantum computational and information theoretic tasks. In particular, what was once considered ``spooky action at a distance'' is now regarded as a valuable resource in quantum information (see, e.g.~\cite{HHHH09}). In this thesis, entanglement is not a primary focus, but rather has important connections to non-classical correlations in the results of Chapters~\ref{chap:activation} and~\ref{chap:entanglementswap}. We give a brief introduction to entanglement here; the reader is referred to the surveys of Bru\ss~\cite{B01} and Horodecki$^{\otimes 4}$~\cite{HHHH09} for further details.

To begin, the canonical example of an entangled state is the two-qubit EPR pair,
\begin{equation}
    \ket{\phi^+}=\frac{1}{\sqrt{2}}(\ket{00}+\ket{11}).
\end{equation}
By observing that $\trace_1(\ketbra{\phi^+}{\phi^+})=\trace_2(\ketbra{\phi^+}{\phi^+})=I/2$, we have one of the characteristic traits of quantum mechanics --- that for quantum systems, knowledge of the whole quantum system does not imply knowledge of its parts. Since entangled (pure) states, such as the EPR pair, cannot be written as a \emph{product state} $\ket{\psi_1}\otimes\ket{\psi_2}$ of single qubit states $\ket{\psi_1},\ket{\psi_2}$, a primary area of study in quantum information has been the quantification of ``how far'' an entangled state is from product form. (Note that all classical states, by which we mean bit strings, are of product form.)

The answer to this question varies greatly depending on context. For bipartite \emph{pure} states $\ket{\psi_{AB}}\in\DD(\complex^m\otimes \complex^n)$, the canonical measure of entanglement is given by the \emph{entropy of entanglement}~\cite{HHHH09},
\begin{equation}
    E(\ket{\psi_{AB}})=S(\trace_B{\ketbra{\psi_{AB}}{\psi_{AB}}})=S(\trace_A{\ketbra{\psi_{AB}}{\psi_{AB}}}),
\end{equation}
where $S(\rho):=-\trace(\rho\log(\rho))$ is the von Neumann entropy of $\rho$. It holds that $0\leq E(\ket{\psi_{AB}})\leq \log (\min(m,n))$, where the lower bound is achieved if and only if a state is of product form, and the upper bound is achieved if and only if a state is \emph{maximally entangled}, such as the EPR pair.

The definition of $E(\ket{\psi_{AB}})$ is perhaps better motivated by the fact that any bipartite $\ket{\psi_{AB}}\in\complex^m\otimes \complex^n$ can be written in terms of the \emph{Schmidt decomposition}, such that
\begin{equation}
    \ket{\psi_{AB}}=\sum_{i=1}^{\min(m,n)}\alpha_i\ket{\psi_i}\otimes\ket{\phi_i}.
\end{equation}
Here, the real $\alpha_i\geq 0$ are called \emph{Schmidt coefficients}, and the sets $\set{\ket{\psi_i}}$ and $\set{\ket{\phi_i}}$ are orthonormal bases for $\complex^m$ and $\complex^n$, respectively, known as the \emph{Schmidt bases}. The Schmidt decomposition is extremely useful in quantum information; some of our results in Chapter~\ref{chap:approx}, for example, depend heavily on it. A proof of existence for the Schmidt decomposition is straightforward, and makes use of the $\operatorname{vec}$ mapping (defined in the proof of Corollary~\ref{5_cor:neg} here) and singular value decomposition for operators; we refer the reader to~\cite{W08_2} for details. To now see the connection between $E$ and the Schmidt decomposition, let $\ve{p}\in\reals^{\min(m,n)}$ with $\ve{p}(i)=\alpha_i^2$. Then, $E(\ket{\psi})=H(\ve{p})$, where $H(\ve{p}):=-\sum_i p(i)\log p(i)$ is the Shannon entropy of probability distribution $\ve{p}$. In the other words, the more ``tightly concentrated'' the Schmidt coefficients of $\ket{\psi_{AB}}$ are, the less entangled $\ket{\psi_{AB}}$ is. Note that a state is product if and only if it has a Schmidt coefficient $\alpha_i=1$, and a state is maximally entangled if and only if all its Schmidt coefficients are $1/\sqrt{d}$ for $d=\min(m,n)$.

Moving to the mixed state case, the quantification of entanglement becomes much more complex. Most generally, we say operator $\rho\in\pos{\X\otimes \Y}$ is \emph{separable} (i.e.\ unentangled) if and only if it can be written~\cite{W08_14}
\begin{equation}\label{0_eqn:sep}
    \rho = \sum_i A_i\otimes B_i
\end{equation}
for $A_i\in\pos{\X}$ and $B_i\in\pos{\Y}$. This definition of separability (with the added trace one constraint) was first given by Werner~\cite{W89}. We denote the set of separable operators acting on $\X\otimes \Y$ as $\sep{\X,\Y}$. Note that $\sep{\X,\Y}$ is a convex cone; this property is vital to the results of Section~\ref{3_scn:parrep}. Here, a \emph{cone} is a set $S\subseteq\X$ such that $\lambda x\in S$ for all $x\in S$ and all $\lambda\geq 0$. If  additionally $u+v\in S$ for $u,v\in S$, then $S$ is called a \emph{convex cone}. When we restrict ourselves to the set of separable \emph{density} operators in $\sep{\X,\Y}$ (i.e.\ we impose the trace one constraint), we obtain a \emph{convex set}. (A set $S\subset \X$ is called \emph{convex} if $p x + (1-p) y\in S$ for all $x,y\in S$ and $0\leq p\leq 1$.) The set of separable density operators has the following properties, which prove useful in Section~\ref{3_scn:parrep}: It is compact and contains a ball around the maximally mixed state (which is, of course, separable)~\cite{GB02,GB03,GB05}.

The problem of determining whether a given density operator $\rho\in\DD(\X\otimes\Y)$ is in $\sep{\X,\Y}$ (where one is allowed to work in time polynomial in the dimension), known as the \emph{Quantum Separability Problem}, was shown NP-hard to solve within inverse exponential precision by Gurvits~\cite{G03} (see also the work of Ioannou~\cite{I07}). This was later extended to inverse polynomial precision by the present author~\cite{G10}, and shortly thereafter independently by Beigi~\cite{Beigi08}. Recently, a breakthrough result of Christandl, Brand{\~{a}}o, and Yard~\cite{BCY11} has shown that the problem is \emph{quasi-polynomial}-time solvable for the case of \emph{constant} precision; the result goes via a powerful new de Finetti-type theorem for the Frobenius (and LOCC, where LOCC stands for \emph{local operations and classical correlations}) norms.

Thus, as suggested by the NP-hardness of Quantum Separability Problem, in the mixed-state case there is no known efficient test for separability, unlike the pure-state case. To this end, there have been many mixed state entanglement measures proposed to date; the reader is referred to the survey of Horodecki$^{\otimes 4}$~\cite{HHHH09} for an in-depth look.

Here, we mention two entanglement detection schemes used in this thesis. The first is the popular approach proposed by Peres~\cite{peresPT} known as the \emph{positive partial transpose} (PPT) test, which plays a role in Chapters~\ref{chap:activation} and~\ref{chap:entanglementswap}. Specifically, consider the super-operator $I\otimes T$ acting on space $\LL(\X\otimes\Y)$, where $T$ denotes the transpose map. Then, given any $\rho\in\DD(\X\otimes \Y)$, if $(I\otimes T)(\rho)\not\succeq 0$, then $\rho$ is not separable. This follows since for any separable operator $\sum_i A_i\otimes B_i$,
\begin{equation}
    (I\otimes T)\left(\sum_i A_i\otimes B_i\right)=\sum_i A_i\otimes T(B_i)\succeq 0.
\end{equation}
Above, we have used the fact that the transpose map does not change the spectrum of an operator. The PPT test is known to be necessary and sufficient for pure states of all dimensions, and for mixed states of $(2\times 2)$ and $(2\times 3)$-dimensional systems~\cite{peresPT,HorodeckiPT}. In higher dimensions, however, we remark that there exist mixed entangled states which nevertheless have a positive partial transpose; such states are called \emph{bound entangled}~\cite{H97_2,HHH98}. Bound entangled states have the property that they cannot be \emph{distilled}, meaning roughly that in the asymptotic limit, given many copies of a bound entangled state $\rho$, there does not exist an LOCC (local operations and classical communication) protocol which can extract the entanglement present in the copies of $\rho$ into pure EPR pairs. The quantification of just how much entanglement can be distilled in this sense is given by another entanglement measure, the \emph{distillable entanglement}~\cite{reviewplenio}; this makes a brief appearance in Chapter~\ref{chap:activation}.

Finally, there is an easy way to compute the partial transpose given a matrix representation of state $\rho\in\DD(\complex^m\otimes\complex^n)$: Namely, partition the matrix into $(m\times n)$-dimensional blocks, and take the tranpose of each block individually. For example, for the EPR pair $\ket{\phi^+}=(\ket{00}+\ket{11})/\sqrt{2}$, we have
\begin{equation}
    (I\otimes T)(\ketbra{\phi^+}{\phi^+})=(I\otimes T)\left(
                  \begin{array}{cccc}
                    \frac{1}{2} & 0 & 0 & \frac{1}{2} \\
                    0 & 0 & 0 & 0 \\
                    0 & 0 & 0 & 0 \\
                    \frac{1}{2} & 0 & 0 & \frac{1}{2} \\
                  \end{array}
                \right)=
                \left(
                  \begin{array}{cccc}
                    \frac{1}{2} & 0 & 0 & 0 \\
                    0 & 0 & \frac{1}{2} & 0 \\
                    0 & \frac{1}{2} & 0 & 0 \\
                    0 & 0 & 0 & \frac{1}{2} \\
                  \end{array}
                \right)\not\succeq 0.
\end{equation}

The second entanglement detection scheme we define here is the \emph{relative entropy of entanglement}~\cite{PhysRevA.57.1619,HHHOSSS05}. Specifically, define for $\rho,\sigma\in\DD(\X)$ the \emph{relative entropy} as
\begin{equation}
    S(\rho||\sigma):=-\trace(\rho\log\sigma)-S(\rho).
\end{equation}
Then, for $\rho \in\DD(\X\otimes \Y)$, the relative entropy of entanglement is defined as
\begin{equation}\label{0_eqn:REE}
    E_R(\rho) = \min_{\sigma\in \Sep(\X,\Y)}S(\rho||\sigma).
\end{equation}
The following properties regarding $E_R$ hold~\cite{PhysRevA.57.1619}: It takes value $0$ if and only if $\rho\in\Sep(\X,\Y)$, is invariant under local unitary operations, is convex, reduces to the entropy of entanglement for pure states, and is an upper bound on the distillable entanglement (see also~\cite{R01}). It is further non-increasing under LOCC, which follows since $S(\rho||\sigma)\geq S(\Phi(\rho)||\Phi(\sigma))$ for any TPCP map $\Phi$~\cite{V02}. In fact, a stronger and physically more relevant statement holds --- that even if we allow \emph{post-selection} after performing an LOCC measurement, the value of $E_R$ does not increase on average~\cite{PhysRevA.57.1619}. In other words, let $\set{K_i}$ be a complete set of Kraus operators for a TPCP map, i.e.\ $\sum_i K_i^\dagger K_i=I$. Then, letting $\rho_i := K_i\rho K_i^\dagger$, it holds that
\begin{equation}
    E_R(\rho)\geq \sum_i \trace(\rho_i) E_R\left(\frac{\rho_i}{\trace(\rho_i)}\right).
\end{equation}

We close this section by noting that the definition of separability presented here extends straightforwardly to the multipartite setting. The structure of multipartite entanglement, however, is markedly more daunting than in the bipartite case.

\subsection{Non-classical correlations}\label{0_sscn:nonclassical}

Having discussed quantum entanglement, we now turn our attention to another form of quantum correlations, called simply \emph{non-classical correlations}. Such correlations have attracted much attention in the last decade or so, both in terms of their characterization and quantification, as well as with respect to their use as a resource in quantum information. We begin by motivating the study of non-classical correlations, and follow with definitions. We then discuss the role of such correlations in quantum information processing tasks, and close by surveying a number of known non-classicality measures. The reader is referred to Modi \emph{et al.}~\cite{MBCPV11} for a more comprehensive survey of the topic.

\paragraph{Motivation.} As mentioned earlier, it is known that in the case of pure-state quantum computation, entanglement is a necessary resource for exponential speedup over classical computers~\cite{jozsa03a}. What happens, however, if we instead consider \emph{mixed}-state quantum computing? This is a particularly relevant question, as typically one deals with mixed states in a laboratory setting due to noise from the environment. In 1998, Knill and Laflamme~\cite{kl98} proposed a model of computing known as \emph{Deterministic Quantum Computing with one clean qubit (DQC1)} (see Chapter~\ref{chap:dqc}), wherein all but one qubit of the computation are initialized to the maximally mixed state --- in other words, the quantum computation acts on a highly mixed state. (Note that this model is motivated experimentally by nuclear-magnetic resonance (NMR) information processing, in which states are highly mixed.) Yet, this model can perform the task of (normalized) trace estimation of a given unitary exponentially faster than the best known classical algorithm. This raises the question: Is entanglement also the root of the believed speedup in DQC1? (This is a natural question since ``very highly mixed'' states are separable due to a ball around the maximally mixed state in the set of separable quantum states~\cite{GB02,GB03,GB05}.) Or are there other correlations possibly at play? Recent work has suggested that although Erwin Schr\"{o}dinger once wrote that entanglement is \emph{``not just one of many traits, but
the characteristic trait of quantum physics''}~\cite{S35} (as quoted in~\cite{MBCPV11}), that between purely classical correlations and entanglement, there lies another form of quantum correlations whose nature is only now beginning to be understood. Such correlations are known simply as \emph{non-classical correlations}.

\paragraph{Defining non-classical correlations.} We now define what we mean by non-classical correlations. To begin, we say that a quantum state $\rho_{AB}\in\DD(\X\otimes \Y)$, henceforth denoted as $\rho$ to avoid clutter, is \emph{strictly classically correlated} or \emph{classical} if it can be diagonalized in a local product basis. In other words, $\rho$ is classical if there exist local orthonormal bases $\set{\ket{\psi_i}},\set{\ket{\phi_i}}$ for $\X$ and $\Y$, respectively, such that
\begin{equation}\label{eqn:NCdef1}
    \rho =\sum_{ij}\lambda_i\ketbra{\psi_i}{\psi_i}\otimes\ketbra{\phi_j}{\phi_j},
\end{equation}
for $\set{\lambda_i}$ the eigenvalues of $\rho$. Note that such a state is simply an embedding of a classical bipartite distribution into the quantum formalism. Any state not satisfying this definition is called \emph{non-classical}. We remark that this definition of classicality extends straightforwardly to the multipartite setting.

Continuing in the bipartite setting, a particularly interesting class of states which subsume the classical states are the so-called \emph{classical-quantum (CQ)} states, which are only classical in system A. Specifically, a state $\rho\in\DD (\X\otimes\Y)$ is CQ if there exists a local orthonormal basis $\set{\ket{\psi_i}}$ for $\X$ such that
\begin{equation}
    \rho =\sum_{i}p_i\ketbra{\psi_i}{\psi_i}\otimes\rho_i,
\end{equation}
for $\set{p_i}$ a probability distribution and for arbitrary $\rho_i\in\DD(\Y)$. Note that system A in $\rho$ simply plays the role of a classical label: Upon measuring it in basis $\set{\ket{\psi_i}}$ and obtaining outcome $i$, we know the induced state $\rho_i$ in B. Also, observe that CQ states are separable. An analogous definition straightforwardly yields the similar class of quantum-classical (QC) states. As an aside, note that neither classical nor CQ states form a convex set, unlike the set of separable quantum states.

\paragraph{Non-classical correlations and quantum information processing.}

A number of connections are known between non-classical correlations and quantum information processing tasks, involving for example local broadcasting~\cite{pianietal2008nolocalbrodcast,pianietal2009broadcastcopies}, extended state merging~\cite{CABMPW10}, the locking of classical correlations~\cite{dhlst04,DG08,WPM09,BACMW11} (see Chapter~\ref{chap:dqc}), assisted optimal state discrimination~\cite{RRA11,BFWF12}, remote state preparation~\cite{DLMRKBPVZBW12},  entanglement distribution~\cite{SKB12,CMMPPP12}, and activation of non-classical correlations into entanglement~\cite{PGACHW11} (see Chapter~\ref{chap:activation}, and also related work by Streltsov, Kampermann and Bru\ss~\cite{streltsov2010}). We now discuss two of these tasks: local broadcasting and entanglement distribution.

We begin with the task of \emph{local broadcasting}. Specifically, generalizing the no-cloning theorem of Section~\ref{0_sscn:quirks} is the following statement. Given a state $\rho\in\DD(\X)$, we say $\rho'\in\DD(\X\otimes \X)$ is a \emph{broadcast state} for $\rho$ if \begin{equation}
    \trace_1(\rho')=\trace_2(\rho')=\rho,
\end{equation}
where the $1:2$ split is across the two copies of $\X$. Now, suppose we are given a set of density operators $\set{\rho_{i}}\subseteq\DD(\X)$, and some arbitrary starting state $\sigma$. Then, the statement we are interested in is that there exists a TPCP map $\Lambda\in T(\X\otimes\X)$ which, for all $i$, achieves the mapping $\rho_i\otimes\sigma\mapsto \rho'_i\in \DD(\X\otimes\X)$ for $\rho'_i$ a broadcast state for $\rho_i$ if and only if the $\rho_i$ pairwise commute. This is called the \emph{no-broadcasting theorem}~\cite{BCFJZ96,BBLW07}. With respect to non-classical correlations, a variant of this theorem is the \emph{no-local-broadcasting} theorem of Piani \emph{et al.}~\cite{pianietal2008nolocalbrodcast,pianietal2009broadcastcopies}, which states that for any bipartite state $\rho_{AB}\in\DD(\X\otimes \Y)$, there exist local TPCP maps $\Theta_A\in T(\X,\X\otimes \X)$ and $\Theta_B\in T(\Y,\Y\otimes \Y)$ such that $\Theta_A\otimes\Theta_B(\rho_{AB})$ is a broadcast state if and only if $\rho_{AB}$ is strictly classical. Thus, the classicality of correlations in $\rho$ is strongly tied to how well one can carry out the information theoretic task of local broadcasting.

We next discuss the task of entanglement distribution. Consider a tripartite system ABC consisting of Alice, Bob, and a carrier system C. Roughly, the goal of entanglement distribution is for Alice and Bob to increase the entanglement between their systems A and B by having Alice send Bob the {carrier} system C. More specifically, we imagine Alice holds systems A and C to start, and Bob holds system B. Alice applies some encoding operation jointly to A and C. She then sends C to Bob. Bob finally applies some decoding operation to B and C. We now ask: Is the entanglement in the $AC:B$ cut before the protocol was run strictly smaller than the entanglement in the $A:BC$ cut after Bob receives the carrier $C$? What is perhaps most surprising about this task is that the answer to this question can be yes even if the carrier C is \emph{not} entangled with A and B throughout the protocol~\cite{CVDC03}! Motivated by the question of whether non-classical correlations could be the resource behind this phenomenon, Streltsov \emph{et al.}~\cite{SKB12} and Chuan \emph{et al.}~\cite{CMMPPP12} (both works appeared concurrently and independently) showed that (definitions to follow)
\begin{equation}
    \abs{{E}_R^{AC|B}(\rho_{ABC})-{E}_R^{A|BC}(\sigma_{ABC})}\leq \disc_R^{AB|C}(\sigma_{ABC}),
\end{equation}
where $\rho_{ABC}$ is the state before the protocol is run, $\sigma_{ABC}$ is the state once Bob receives $C$ from Alice, and where we measure non-classicality by the relative entropy of discord (RED) $\disc_R^{AB|C}$ of Equation~(\ref{0_eqn:RED}) (to be defined shortly) across the $AB:C$ cut, and we measure entanglement by the relative entropy of entanglement ${E}_R^{AC|B}$ (${E}_R^{A|BC}$) across the $AC:B$ ($A:BC$) cut . In other words, the amount of entanglement which can be transferred from Alice to Bob is bounded by the amount of non-classical correlations between the carrier C and AB (after Alice has applied her encoding operation). Note thus that this upper bound can be non-zero even if C is unentangled with A and B throughout the protocol (and in fact must be non-zero for the example of Cubitt \emph{et al.}~\cite{CVDC03} mentioned above).

\paragraph{Quantifying non-classical correlations.} Finally, we close this section by discussing a number of known non-classicality measures.

The formal notion of CQ states first arose with the works of Ollivier and Zurek~\cite{ollivier01a} and Henderson and Vedral~\cite{henderson01a}, where a measure of quantum correlations dubbed the \emph{quantum discord} was proposed. The aim of this measure is to quantify purely \emph{quantum} correlations in a bipartite state $\rho$. To define the discord, recall first that the (classical) \emph{mutual information} is a measure of correlation between (classical) random variables $A$ and $B$, i.e.\
\begin{equation}
    \IM(A:B)=H(A)+H(B)-H(A,B),
\end{equation}
where $H$ is the Shannon entropy defined in Section~\ref{0_sscn:entanglement} and $H(A,B)=-\sum_{a,b}\pr(A=a \cap B=b)\log\pr(A=a \cap B=b)$. Using the fact that $\pr(B|A)=\pr(A\cap B)/\pr(A)$, one can straightforwardly also express the mutual information as
\begin{equation}
    \JM(A:B)=H(B)-H(B|A),
\end{equation}
where $H(B|A):=\sum_a \pr(A=a) H(B|A=a)$. Although $\IM$ and $\JM$ are equivalent in the classical setting, their quantum counterparts no longer share the same relationship. Specifically, the quantum mutual information can be defined as
\begin{equation}
    \IM(\rho_{AB})=S(\rho_A)+S(\rho_B)-S(\rho_{AB}),
\end{equation}
where recall $\rho_A=\trace_B(\rho_{AB})$. However, a quantum variant of $\JM$ is non-trivial to define, since it requires specifying a value for $B$ for the conditional entropy $H(A|B)$ --- in particular, unlike the classical setting, quantumly the choice of measurement basis is non-trivial. To this end, for rank-one projective measurement $\set{\Pi_j^A}$, one defines~\cite{ollivier01a} a quantum conditional entropy
\begin{equation}
        S\left(\rho_{B|\set{\Pi_j^A}}\right) := \sum_jp_jS\left((\Pi_j^A\otimes I^B)\rho(\Pi_j^A\otimes I^B)\Big/p_j\right),
\end{equation}
where $p_j = \tr(\Pi_j^A\otimes I^B\rho)$. Then, a quantum version of $\JM$ for given measurement basis $\set{\Pi_j^A}$ can be defined as
 \be
\mathcal{J}_{\set{\Pi_j^A}}(\rho) =
S(\rho_B)-S\left(\rho_{B|\set{\Pi_j^A}}\right).
 \ee
Note that $\JM_{\set{\Pi_j^A}}(\rho)$ quantifies the amount of classical correlations which can be extracted from $\rho_{AB}$ via a projective measurement on one party; since we are in the end interested in purely quantum correlations, intuitively one would thus choose the \emph{optimum} measurement $\set{\Pi_j^A}$ here so as to extract all purely classical correlations, leaving only quantum correlations behind. With this in mind, the quantum discord is now defined as
 \ben \label{4_eqn:discord_def}
    \discm := \mathcal{I}(\rho)-\max_{\set{\Pi_j^A}}\mathcal{J}_{\set{\Pi_j^A}}(\rho)
           =
           S(\rho_A)-S(\rho_{AB})+\min_{\set{\Pi_j^A}}S\left(\rho_{B|\set{\Pi_j^A}}\right).
 \een
The discord is~\cite{MBCPV11} non-negative, non-symmetric with respect to exchange of systems $A$ and $B$, invariant under local unitaries, and most importantly for our discussion here, takes value zero if and only if $\rho$ is CQ~\cite{ollivier01a,dattathesis}. Moreover, there exist \emph{separable} states, such as the two-qubit state
\begin{equation}
    \frac{1}{2}\ketbra{0}{0}\otimes\ketbra{0}{0} + \frac{1}{2}\ketbra{+}{+}\otimes\ketbra{1}{1},
\end{equation}
which have non-zero discord, thus showing that discord quantifies correlations beyond entanglement. (Aside: The state above is studied further in Chapters~\ref{chap:localunitary},~\ref{chap:activation}, and~\ref{chap:entanglementswap}.)

The next measure of non-classical correlations we discuss is the \emph{geometric quantum discord}~\cite{DVB10}. Let $\CQ\subseteq\DD(\X\otimes\Y)$ denote the set of classical-quantum states. Then for $\rho\in\DD(\X\otimes\Y)$ the geometric discord is defined as
\begin{equation}\label{0_eqn:gdisc}
    \gdisc := \min_{\sigma\in\CQ}\fnorm{\rho-\sigma}^2 = \min_{\set{\Pi_j^A}}\fnorm{\rho-\sum_j\Pi^A_j \rho\Pi^A_j}^2,
\end{equation}
where $\fnorm{\cdot}$ is the Frobenius norm and the second equality was shown by Luo and Fu~\cite{LF10}. The name \emph{geometric} derives from the fact that the measure attempts to quantify distance from $\CQ$ via a metric. We have included the right-most expression in Equation~(\ref{0_eqn:gdisc}) as it offers another intuitive interpretation of non-classical correlations involving \emph{disturbance under measurement}. Namely, recall that in the classical world, there always exists a choice of measurement basis $\set{\Pi_j^A}$ (the computational basis) leaving the target state undisturbed. In the quantum setting, however, this is in general not the case. For example, this is an intuitive reason why CQ states are considered classical in A; there exists a measurement basis acting invariantly on A. The second expression for $\gdisc$ in Equation~(\ref{0_eqn:gdisc}) thus attempts to understand how much $\rho$ must be disturbed in a (rank one projective) measurement, regardless of the choice of local measurement basis for A.

The next non-classicality measure we discuss is similar to the geometric discord, but replaces the Frobenius norm with the \emph{relative entropy}. We thus arrive at the \emph{relative entropy of discord (RED)}~\cite{MPSVW10},
\begin{equation}\label{0_eqn:RED}
    \rdisc = \min_{\sigma\in \CQ}S(\rho||\sigma).
\end{equation}
An analogous definition for the case of general strictly classically correlated states goes under the name of the \emph{relative entropy of quantumness (REQ)}~\cite{Bravyi2003,PhysRevA.71.062307,groismanquantumness,PhysRevA.77.052101,MPSVW10}; this is studied further in Chapters~\ref{chap:activation} and~\ref{chap:entanglementswap}.

Interestingly, the RED turns out to be equal to (a variant of) another measure of non-classical correlations we discuss next, the \emph{quantum deficit}~\cite{HHHOSSS05}. The latter's definition is motivated by work extraction from quantum systems coupled to a heat bath. Roughly, the idea here is that a state is strictly classically correlated if and only if the same amount of work can be drawn from the global state versus from the local subsystems after allowing a suitably restricted subset of local operations and classical communication (LOCC) known as \emph{closed LOCC}. The variant of the deficit which is equal~\cite{HHHOSSS05} to the RED is the \emph{one-way deficit} $\Delta^\rightarrow$, given by (simplified from the original definition):
\begin{equation}
    \Delta^\rightarrow := \min_{\set{\Pi_j^A}} S\left(\sum_j\Pi^A_j \rho\Pi^A_j\right)-S(\rho_{AB}).
\end{equation}
Here, $\set{\Pi_j^A}$ again denotes a rank-one projective measurement. The correspondence between RED and the deficit does not stop here, however; the two-sided analogue of the RED, the REQ, is equal~\cite{HHHOSSS05} to the so-called \emph{zero-way deficit} $\Delta^{\emptyset}$:
\begin{equation}
    \Delta^\emptyset := \min_{\set{\Pi_i^A},\set{\Pi_j^B}} S\left(\sum_{ij}\Pi^A_i\otimes\Pi^B_j \rho\Pi^A_i\otimes\Pi^B_j\right)-S(\rho_{AB}).
\end{equation}

We have discussed a number of non-classicality measures here. Later in Chapter~\ref{chap:localunitary}, we introduce a novel measure of non-classical correlations based on local unitary operations, which for $(2\times N)$-dimensional quantum states turns out to coincide with the geometric discord. Chapters~\ref{chap:activation} and~\ref{chap:entanglementswap} then introduce and study a protocol for ``activating'' non-classical correlations into entanglement, while also providing an operational interpretation for the REQ.

\chapter{Approximation algorithms for QMA-complete problems}\label{chap:approx}
%======================================================================

%{Approximation algorithms for QMA-complete problems}
%\author{Sevag Gharibian\footnote{David R. Cheriton School of Computer Science and Institute for Quantum Computing, University of Waterloo, Waterloo N2L 3G1, Canada. Supported by a Natural Sciences and Engineering Research Council of Canada (NSERC) Canada Graduate Scholarship (CGS) and Michael Smith Foreign Study Supplement (MSFSS), David R. Cheriton Graduate Scholarship, EU-Canada Transatlantic Exchange Partnership programme, and the Canadian Institute for Advanced Research (CIFAR).}
%\and
% Julia Kempe\footnote{Blavatnik School of Computer Science, Tel Aviv University, Tel Aviv 69978, Israel and CNRS \& LIAFA,
% Universit\'{e} Paris 7, Paris, France. Supported
%by an Individual Research Grant of the Israeli Science Foundation, by European Research Council (ERC) Starting Grant
%QUCO and by the Wolfson Family Charitable Trust.}}
%\date{}
%\maketitle
%
%\begin{abstract}

%\end{abstract}

%\vspace{105mm}
\emph{This chapter is based on~\cite{GK11}:}\\

\vspace{-4mm}
\noindent S. Gharibian and J. Kempe. Approximation algorithms for QMA-complete problems.
In \emph{Proceedings of 26th IEEE Conference on Computational Complexity},
pages 178-–188, 2011, DOI: 10.1109/CCC.2011.15, \copyright~2011 IEEE, ieeexplore.ieee.org.

\vspace{3mm}
\noindent Approximation algorithms for classical constraint satisfaction problems are one of the main research areas in
theoretical computer science. In this chapter, we define a natural approximation version of the QMA-complete local Hamiltonian
problem and initiate its study. We present two main results. The first shows that a non-trivial approximation ratio can be obtained in the class NP using product states. The second result (which builds on the first one), gives a
polynomial time (classical) algorithm providing a similar approximation ratio  for dense instances of the problem. The
latter result is based on an adaptation of the ``exhaustive sampling method'' by Arora \emph{et al.}~\cite{AKK99} to the quantum setting, and might be of independent interest.

\section{Introduction and results}\label{1_scn:intro}
In the last few years, the quantum analog of the class NP, the class QMA~\cite{KSV02}, has been extensively studied,
and several QMA-complete problems have been found~\cite{L06,B06,LCV07,BS07,R09,JGL10,SV09,WMN10}. Arguably the
most important (and historically first) QMA-complete problem is the $k$-local Hamiltonian problem~\cite{KSV02,KR03,OT05,KKR06,AGIK09}. Recall from Section~\ref{0_sscn:LH} that here, the
input is a set of Hamiltonians (Hermitian matrices), each acting on at most $k$-qubits each. The task is to determine
the largest eigenvalue of the sum of these Hamiltonians. This problem generalizes  the central NP-hard problem
MAX-$k$-CSP, where we are given a set of Boolean constraints on $k$ variables each, with the goal to satisfy as many
constraints as possible. The local Hamiltonian problem is of significant interest to complexity theorists and to
physicists studying properties of physical systems alike (e.g.~\cite{BV05,ADKLLR07,BDOT08,AALV09,CV09,LLMSS10,SC10}).

Moving to the classical scenario, the theory of NP-completeness is one of the great success stories of classical
computational complexity~\cite{AB90}. It was soon realized that many natural optimization problems are NP-hard, and are hence unlikely to have polynomial time algorithms. A natural question (both in theory and in practice) is to look for polynomial time algorithms that produce
solutions that are close to optimum. More precisely, one says that an algorithm achieves an \emph{approximation ratio}
of $c \in [0,1]$ for a certain maximization problem if on all inputs, the value of the algorithm's output is at least
$c$ times that of the optimum solution (the output value should also be at most the optimal solution). The closer $c$ is to $1$, the better the approximation. The
investigation of approximation algorithms is, after decades of heavy research, still a very active area (e.g.,~\cite{H97,V01}). For many central NP-hard problems, tight polynomial time approximation algorithms
are known.

In the context of QMA-complete problems, it is thus natural to search for approximation algorithms for these
problems, and in particular for the local Hamiltonian problem. The question we address here is: \emph{How well can one efficiently approximate the $k$-local Hamiltonian problem?}

It should be noted that a large host of heuristics has been developed in the physics community to
approximate properties of local Hamiltonian systems (see, e.g.,~\cite{CV09} for a survey) and this area
is extremely important in the study of physical systems. However, the systematic complexity theoretic study of
approximation algorithms for QMA-complete problems is still very much in its infancy, and our work is one of the first
steps in this research direction. We note that there has been a lot of interest in recent years~\cite{AALV09,A06} in
establishing a so-called quantum PCP theorem~\cite{AS98,ALMSS98}, which amounts to showing that for some
constant $c<1$ close enough to $1$, approximating the $k$-local Hamiltonian (or related problems) to within $c$ is
QMA-hard. Our results can also be seen as a natural continuation of that investigation.

\paragraph{Our results:}

Let us start by precisely defining the optimization version of the local
Hamiltonian problem, which is parameterized by two integers $k$ and $d$, which we always think of as constants. Note that the definition below differs slightly from that given in Section~\ref{0_sscn:LH}, Definition~\ref{0_def:localhamiltonianproblem}; we discuss the differences after stating the definition.

\begin{definition}[MAX-$k$-local Hamiltonian problem on $d$-level systems (qudits)]  An instance of the problem
consists of a collection of $\binom{n}{k}$ Hermitian matrices, one for each subset of $k$ qudits. The matrix
$H_{i_1,\ldots,i_k}$ corresponding to some $1 \leq i_1 \le \cdots \le i_k \le n$ is assumed to act on those
qudits (terms acting on less than $k$ qudits can be
incorporated by tensoring them with the identity), to be positive semidefinite, and to have operator norm at most $1$.
We call any pure or mixed state $\rho$ on $n$ qudits an \emph{assignment} and define its {\em value} to be $\trace (H \rho)$
where $H=\sum_{i_1, \ldots ,i_k} H_{i_1,\ldots,i_k}$. The goal is to find the largest eigenvalue of $H$ (denoted
$\OPT$), or equivalently, the maximum value obtained by an assignment. We say that an algorithm provides an
\emph{approximation ratio} of $c \in [0,1]$ if for all instances, it outputs a value that is between $c \cdot \OPT$ and
$\OPT$.
\end{definition}

This definition, we believe, is the natural quantum analog of the MAX-$k$-CSP problem. We note that it differs slightly
from the usual definition of the $k$-local Hamiltonian problem. Namely, we consider maximization (as opposed to
minimization), and also restrict the terms of $H$ to be positive semidefinite, and have norm at most $1$ (the latter two contraints are also common to Definition~\ref{0_def:localhamiltonianproblem}; more generally, the local terms of $H$ can be arbitrary Hermitian operators). As long as
one considers the \emph{exact} problem, these assumptions are without loss of generality, and do not affect the
definition, as seen by simply scaling the Hamiltonians and adding multiples of identity as necessary. However,
when dealing with the \emph{approximation} version, these assumptions are important for the problem to make sense; for
instance, one cannot meaningfully talk about approximation ratios if the optimum can take both negative and positive
values. That is why we require the terms to be positive semidefinite. The requirement that the terms have operator norm
at most $1$ does not affect the problem and later allows us to conveniently define dense instances.
Finally, changing the maximization to a minimization would lead to an entirely different approximation problem: the
quantum analogue of MIN-CSP (e.g.~\cite{KSTW01}). Minimization problems
are, generally speaking, harder than maximization problems, and we leave this research direction for future work. %We remark that our definition of \klh contains the special case of MAX-$k$-quantum SAT~\cite{B06} (when analogously defined).\snote{Added this sentence}

Before stating our results, we state a trivial way to get a $d^{-k}$-approximation for MAX-$k$-local
Hamiltonian. Observe that the maximally mixed state has at least $d^{-k}$ overlap with the reduced
density matrix of the optimal assignment on any $k$ particles. A similar property holds classically, where a random
assignment gives (in expectation) a $d^{-k}$ approximation of MAX-$k$-CSP. We now describe our two main results.
\myparagraph{Approximation by product states} One inherently quantum property of the local Hamiltonian problem is the
fact that the optimal state might in general be highly entangled (and hence not efficiently describable in polynomial
time or space). This is why we do not require  outputting the assignment itself in the above definition. If, however, the
optimal assignment (or some other good assignment) was guaranteed to be a {\em product state}, then we could describe it efficiently. The following theorem shows just that.

\begin{theorem}\label{1_thm:approxprodstate}
For an instance of MAX-$k$-local Hamiltonian with optimal value $\OPT$, there is a (pure) product state assignment that has
value at least $\OPT/d^{k-1}$.
\end{theorem}
This result is {\em tight} for product states in the case of $2$-local Hamiltonians (we remark that $2$-local Hamiltonians are often the most relevant case from a physics perspective). For example, consider the Hamiltonian on $2$-qubits
that projects onto the EPR state $\frac{1}{\sqrt{2}}(\ket{00}+\ket{11})$. It is easy to see that no product state
achieves value more than $1/2$. For general $d$ and $k$, we can only show that product states cannot achieve an approximation
ratio greater than $1/d^{\lfloor k/2 \rfloor}$ (see Section~\ref{1_scn:prodratio}, where better bounds in more specific cases are also discussed).

%This result is {\em tight} for product states in the case of $2$-local Hamiltonians, seen by considering the Hamiltonian on $2$-qubits
%that projects onto the EPR state $\frac{1}{\sqrt{2}}(\ket{00}+\ket{11})$. We remark that $2$-local Hamiltonians are often the most relevant case from a physics standpoint. For general $k$, we can only show that product states cannot achieve an approximation ratio greater than $1/d^{\lfloor k/2 \rfloor}$ (see Section~\ref{1_scn:sepopt}).

If we could efficiently find the best product state assignment, we would obtain an algorithm achieving a non-trivial
$d^{-k+1}$ approximation ratio. Unfortunately, this problem is NP-complete, since it would allow one to solve (e.g.)
the special case of MAX-$k$-SAT (as discussed in Section~\ref{0_sscn:LH}, for each clause $C$ acting on variables $\set{i_1,\ldots,i_k}$ in an instance of MAX-$k$-SAT, define the corresponding Hamiltonian term $H_{i_1,\ldots,i_k}$ diagonal in the computational basis and projecting onto the satisfying assignments for $C$. Then, without loss of generality, the optimal product state assignment can be taken to be a computational basis state), implying such an algorithm cannot exist unless $\class{P}=\class{NP}$. Still, the theorem has the following
interesting implication: It shows that unless $\NP=\QMA$, approximating the local Hamiltonian problem to within a factor less than $d^{-k+1}$ is not QMA-hard. This follows simply because product states have polynomial size
classical descriptions. (More accurately, since one uses a polynomial number of classical bits to approximately specify a product state in NP, the ratio in the implication above is $d^{-k+1}-f(m)$ for some function $f$ which scales inverse exponentially in the input size $m$.)

%SEV: Commented out MPS stuff
%Using the technique developed in the proof of Theorem~\ref{1_thm:approxprodstate}, we then show the following simple result in Section~\ref{1_scn:prodratio} regarding \emph{Matrix Product State (MPS)} assignments~(see e.g.~\cite{V03}).
%
%\begin{corollary} \label{1_cor:mps}For an instance $H=\sum_{i=1}^{n-1} H_{i,i+1}$ of MAX-$2$-local Hamiltonian on a line of $n$ $d$-dimensional qudits with optimal value $\OPT$, there is an MPS assignment of bond dimension $d$ that has value at least $\frac{1}{2}(1+\frac{1}{d^2})\OPT$.
%\end{corollary}

\myparagraph{A polynomial time approximation algorithm for dense instances}

Our second result gives a classical polynomial time approximation algorithm for {\em dense} instances of the local
Hamiltonian problem. This result is perhaps our technically most challenging one, and we hope the techniques we develop
might turn out useful elsewhere.

Dense instances of classical constraint satisfaction problems have been studied in depth \cite{V96,FK96,GGR98,AKK99,VK00,AVKK02,BVK03,VKKV05}. Our result is inspired by work of Arora \emph{et al.}~\cite{AKK99}
who provide a polynomial time approximation scheme, or PTAS (i.e., an efficient $1-\eps$ approximation algorithm for
any fixed $\eps>0$), for several types of dense constraint satisfaction problems. In the classical case, dense (for
$2$-local constraints) simply means that the average degree in the constraint graph is $\Omega(n)$, or equivalently,
that the optimum is $\Omega(n^2)$.  In
analogy, we define an instance of MAX-$k$-local Hamiltonian to be {\em dense} if $\OPT=\Omega(n^k)$, or equivalently,
if $\trace(H \frac{I}{d^n})=\Omega(n^k)$ (the equivalence follows from the fact that the mixed state assignment $I/d^n$ has value between $\OPT$ and $\OPT/d^k$).

It is not hard to see that the (exact) dense local Hamiltonian problem remains QMA-hard (see Section~\ref{1_sscn:finalptas}).
We hope the
dense case might be of practical interest to physicists who study systems of particles by incorporating all possible
interactions between them. Our second main result is the following:

\begin{theorem}\label{1_thm:approxdense}
For all $\eps >0$ there is a polynomial time $(1/d^{k-1} -\eps)$-approximation algorithm for the dense MAX-$k$-local
Hamiltonian problem over qudits.
\end{theorem}
%In particular for the $2$-local problem on qubits our algorithm gives a $(\frac{1}{2}-\eps)$-approximation.
Theorem~\ref{1_thm:approxdense} follows immediately by combining Theorem~\ref{1_thm:approxprodstate} with the following
theorem, which gives an approximation scheme for the problem of optimizing over the set of product states.

\begin{theorem}\label{1_thm:ptasdense}
Let $\optprod$ denote the value of the optimal product state assignment for an instance of MAX-$k$-local Hamiltonian $H$. Then, for all $\eps>0$, there is a polynomial time algorithm which outputs a product state assignment attaining value at least $\optprod-\epsilon n^k$. For all $\eps>0$, this yields an efficient $(1-\epsilon)$-approximation algorithm for computing $\optprod$ for dense MAX-$k$-local Hamiltonian.
\end{theorem}

We remark that the algorithm of Theorem~\ref{1_thm:ptasdense} also applies in the \emph{minimization} setting, in which one is interested in computing the \emph{smallest} eigenvalue of $k$-local Hamiltonian $H$. Here, our algorithm outputs a value at most $\optprod+\epsilon n^k$.

\paragraph{Proof ideas and new tools:}
The proofs of Theorem~\ref{1_thm:approxprodstate} and Theorem~\ref{1_thm:ptasdense} are independent and employ different
techniques. To show the product state approximation guarantee, we show a slightly stronger statement: For {\em any}
assignment $\ket{\Psi}$, there is a way to construct a product assignment of at least $d^{-k+1}$ its value. The proof is constructive (given $\ket{\Psi}$): we use a type of recursive Schmidt decomposition of $\ket{\Psi}$ to obtain a mixture
of product states whose value is off by at most the desired approximation factor (see Section~\ref{1_scn:prodratio}).

Our second result is technically more challenging and introduces a few new ideas to this problem, inspired by work of
Arora \emph{et al.}~\cite{AKK99} in the classical setting. We illustrate the main ideas for MAX-$2$-local Hamiltonian on $n$
qubits. Recall that our goal is to find a PTAS for the local Hamiltonian problem {\em over product states}. The value
of the optimal product state assignment, $\optprod$, can be written
\begin{equation}
      \optprod\hspace{2mm} = \hspace{2mm}\max \quad\sum_{i=1}^n \sum_{j \in N(i)} \trace (H_{i,j} (\rho_i \otimes \rho_j))
    \quad\mbox{s.t.}\quad \rho_i\succeq 0 \mbox{ and }\trace(\rho_i)=1\quad\mbox{for } 1\leq i\leq n,\label{1_eqn:introprogram}
\end{equation}
%\begin{align}
%    \optprod\hspace{2mm} = \hspace{2mm}\max \hspace{7mm}& \sum_{i=1}^n \sum_{j \in N(i)} \trace (H_{i,j} (\rho_i \otimes \rho_j))\label{1_eqn:introprogram}\\
%    \mbox{s.t.} \hspace{7mm}& \rho_i\succeq 0 \hspace{16mm} \mbox{for } 1\leq i\leq n,\nonumber\\
%    &\trace(\rho_i)=1 \hspace{9mm} \mbox{for } 1\leq i\leq n,\nonumber
%\end{align}
where $N(i)$ is the set of indices $j$ for which a local Hamiltonian term $H_{i,j}$ is present. We might call this a
{\em quadratic} semidefinite program, as the maximization is quadratic in the $\rho_i$ (and as such not efficiently
solvable in general). Note, however, that if the terms in the maximization were {\em linear}, then we would obtain a semidefinite
program (SDP), which is efficiently solvable~\cite{GLS93}. To ``linearize'' our optimization, we use the ``exhaustive
sampling method" developed by Arora \emph{et al.}~\cite{AKK99} (a method which was later key in many developments in property
testing, e.g.~\cite{GGR98}). We write each Hamiltonian term in a basis that separates its two qubits, for instance the
Pauli basis $\{\sigma_0, \sigma_1, \sigma_2, \sigma_3\}$, $H_{i,j}=\sum_{k,l=0}^3 \alpha^{ij}_{kl} \sigma_k \otimes
\sigma_l$. For $i=1,\ldots,n$ and $k=0,1,2,3$, define
\begin{equation}
    c_k^i := \sum_{j \in N(i)} \sum_{l=0}^3 \alpha_{kl}^{ij} \trace (\sigma_l \rho_j).
\end{equation}
If we knew the values of $c_k^i$ for the optimal $\rho_i$, then solving the SDP below would yield the optimal $\rho_i$:
\begin{align}
    \max \quad \sum_{i=1}^n \sum_{k=0}^3 c_k^i \trace (\sigma_k \rho_i)\quad
    \mbox{s.t.}\quad& \rho_i\succeq 0\mbox{ and } \trace(\rho_i)=1\hspace{11mm} \mbox{for } 1\leq i\leq n,\label{1_eqn:introdecomp}\\
    &\hspace{-2mm}\sum_{j \in N(i)} \sum_{l=0}^3 \alpha_{kl}^{ij} \trace (\sigma_l \rho_j)=c_k^i\hspace{5mm} \mbox{for } 1\leq i\leq n\mbox{ and } 0\leq k\leq 3.\nonumber
\end{align}
%\begin{align}
%    \max \hspace{7mm}& \sum_{i=1}^n \sum_k c_k^i \trace (\sigma_k \rho_i)\label{1_eqn:introdecomp}\\
%    \mbox{s.t.} \hspace{7mm}& \rho_i\succeq 0 \hspace{41mm} \mbox{for } 1\leq i\leq n,\nonumber\\
%    &\trace(\rho_i)=1 \hspace{34mm} \mbox{for } 1\leq i\leq n,\nonumber\\
%    &\sum_{j \in N(i)} \sum_l \alpha_{kl}^{ij} \trace (\sigma_l \rho_j)=c_k^i \hspace{10mm}\mbox{for } 1\leq i\leq n\mbox{ and } 0\leq k\leq 3.\nonumber
%\end{align}
Of course, this reasoning is circular, as in order to obtain the $c_k^i$ we need the optimal $\rho_i$. The crucial idea
is now to use {\em sampling} to {\em estimate} the $c_k^i$. More precisely, assume for a second that we could sample
$O(\log n)$ of the $\rho_i$ randomly from the optimal assignment. Then, by standard sampling bounds, with high
probability over the choice of the sampled qubits we can estimate all the $c_k^i$ to within an additive error $\pm \eps
n$ for some $\eps$. If we had these estimates $a_k^i$ for the $c_k^i$, we could solve the SDP above with the slight
modification that the last constraint should be $a_k^i -\eps n \leq \sum_{j \in N(i)} \sum_l \alpha_{kl}^{ij} \trace
(\sigma_l \rho_j)\leq a_k^i +\eps n$. With high probability over the sampled qubits, this SDP will give a solution that
is within an additive $\eps n^2$ of the optimal one (more subtle technicalities and all calculations can be found in
Section \ref{1_scn:sepopt}). Moreover, it is possible to derandomize the sampling procedure to obtain a deterministic
algorithm (Section~\ref{1_sscn:finalptas}).

Of course, we are still in the realm of wishful thinking, because in order to sample from the optimal solution, we
would need to know it, which is precisely what we set out to do. However, the number of qubits we  wish to sample is
only \emph{logarithmic} in the input size. Thus, to simulate the sampling procedure, we can pick a random subset of
$O(\log n)$ qubits, and simply \emph{iterate} through all
possible assignments on them (with an appropriate $\delta$-net over the density matrices, which incurs a small
additional error) in polynomial time! Our algorithm then runs the SDP for each iteration, and we are guaranteed that at
least one iteration will return a solution within $\eps n^2$ of the optimal one. Because the denseness assumption
guarantees that $\optprod$ is $\Omega(n^2)$, our additive approximation turns into a factor $(1-\eps)$-approximation,
as desired. All details, the runtime of the algorithm and error bounds for the general $k$-local case on qudits
are given in Section~\ref{1_scn:sepopt}. We remark that the approach above works analogously in the setting where the objective function involves minimization instead of maximization.

\paragraph{Previous and related work:}
We note that many heuristics have been developed in the physics community to
approximate properties of local Hamiltonian systems and this area is extremely important in the study of physical
systems~(e.g.~\cite{W92,W93,OR95,RO97,S05,PWKH98,CV09,O11}). Our focus here is, however, on \emph{rigorous} bounds (unlike a heuristic) on the
approximation guarantee of algorithms for the \emph{general} problem (we allow interactions of arbitrary types occurring on arbitrary graphs, in contrast to the more common approach of studying specific local Hamiltonian models with certain classes of allowed interactions). In this area, to our knowledge, few results are known. In the setting of \emph{relative}-error approximation, as studied here, the first and only previous result we are aware of is that of Bansal, Bravyi and Terhal~\cite{BBT09}, who give a PTAS for a
special case of the local Hamiltonian problem, so called quantum Ising spin glasses, for the case where the instance
is on a planar graph and of bounded degree. Roughly, this PTAS is obtained by dividing the graph into constant
size chunks, which can be solved directly, and ignoring the constraints between chunks (this incurs an error
proportional to the number of such constraints, which is small because the graph is planar).
In the setting of \emph{absolute}-error approximation, in 1D models, rigorous results such as Hasting's 1D area law are known for gapped systems~\cite{Ha07} (where it is also shown that the ground state is well-approximated by a Matrix Product State~\cite{V03}), and rigorous approximation methods are known for 1D~\cite{AAI10,SC10} and for 2-local Hamiltonians on qubits where the two-qubit interaction strengths are weak~\cite{BDL08}. Finally, we remark that the use of a product state ansatz is closely related to the mean-field approximation or Hartree-Fock method in physics (see, e.g.~\cite{FV06}). %To our knowledge, we are the first to establish a bound on the approximation factor by optimizing over the set of product states.

\paragraph{Discussion and open questions:}

Our two results give approximations to the local Hamiltonian problem. Although at first glance, our approximation ratio of $1/d^{k-1}$ may appear an incremental improvement over the trivial random assignment strategy, there are three important notes that should be kept in mind: The first is that many classical NP-hard problems, such as MAX-3-SAT (a special case of MAX-$k$-CSP where each constraint is the disjunction (``OR") of $k$ variables or their negation), are \emph{approximation resistant}~(e.g.~\cite{H07,AM08}), meaning that unless P$=$NP, there do not even exist non-trivial approximation ratios beyond the random assignment strategy. For example, for MAX-3-SAT it is NP-hard to do better than the approximation ratio of $7/8$ achieved by random assignment~\cite{Ha97}. Thus, showing the existence of a non-trivial approximation ratio is typically a big step in the classical setting. Moreover, it could have been conceivable that for \klh, analogously to MAX-3-SAT, outperforming the random assignment strategy would have been \emph{QMA-hard}. Yet our results show that unless NP$=$QMA, this is not the case. The second important note that should be kept in mind is that our work considers the local Hamiltonian problem in its full generality by allowing arbitrary constraints on an arbitrary interaction graph. It could be (and is the case, for example, in~\cite{BBT09}) that for more restricted classes of local Hamiltonian models, better approximation ratios are achievable. Third, the currently \emph{best} approximation algorithm for MAX-$k$-CSP gives an approximation ratio of only about $0.44k/2^k$ for $k>2$~\cite{CMM07} (for $k=2$, one can achieve $0.874$~\cite{LLZ02}. See also the work of Raghavendra~\cite{R08}) and this is, moreover, essentially the best possible under a plausible complexity
theoretic conjecture (namely, the Unique Games Conjecture~\cite{K02})~\cite{T98,H05,ST06,AM08}. This is to be contrasted with our $2/2^k$-approximation ratio for the case of $d=2$ (i.e.\ qubit systems), which we show can be achieved by product state assignments for \emph{arbitrary} (i.e.\ even non-dense) \klh~instances (in the non-dense case, however, we do not show how to  \emph{efficiently} find a product state achieving this ratio). This raises the important open question: Is our approximation ratio tight?

Our product state approximation shows that approximating the local Hamiltonian problem to within $d^{-k+1}$
is in NP. It would be interesting to know if this approximation ratio could also be achieved in polynomial time. If
not, it might lead to an intriguing state of affairs where for low approximation ratios the problem is efficiently
solvable, for medium ratios it is in NP but not efficiently solvable, and for high ratios it is QMA-hard (assuming
a quantum PCP theorem exists). Further, as mentioned earlier, our work can be viewed as negative progress towards a quantum PCP theorem in that, by Theorem~\ref{1_thm:approxprodstate}, a quantum PCP theorem with hardness ratio $c\leq d^{-k+1}$ cannot exist unless {NP}$=${QMA}.

To obtain our results for the case of dense local Hamiltonians, we have introduced the exhaustive sampling technique of
Arora \emph{et al.}~\cite{AKK99} to the setting of low-degree semidefinite programs. We linearize such programs using
exhaustive sampling in combination with a careful analysis of the error coming from working with $\delta$-nets on
density matrices. We remark that it seems we cannot simply apply the results of~\cite{AKK99} for \emph{smooth
Polynomial Integer Programs} as a black-box to our setting. This is due to our
aforementioned need for a $\delta$-net, as well as the requirement that our assignment be a positive semidefinite
operator. We address the latter issue by extending the techniques of~\cite{AKK99} to the realm of positive semidefinite
programs by introducing the notion of ``degree-$k$ inner products'' over Hermitian operators to generalize the concept of degree-$k$ polynomials over real numbers, and performing the more complex analysis that ensues.
We hope that this technique will be of much wider applicability, particularly considering the growing use of semidefinite programs in numerous areas of quantum computing and information (e.g.~\cite{doherty04a,JJUW10,LMRS10}).

Another open question is whether similar ideas can be used to approximate other QMA-complete problems,
such as the Consistency problem~\cite{L06}. Moreover, can we obtain polynomial time algorithms without the denseness
assumption? And are there special cases of the local Hamiltonian problem for which there is a PTAS (other than for planar
Ising spin glasses \cite{BBT09})? Of course, we do not expect a PTAS for all instances of the local
Hamiltonian problem, as this would contradict known hardness results for special classical cases of the problem.
However, perhaps there exist other classes of physically relevant instances of the problem for which a PTAS does exist. Finally, can our scheme be extended to work with more general classes of quantum assignments than product states, such as Matrix Product States~\cite{V03}?

\paragraph{Organization of this chapter:} In Section~\ref{1_scn:prodratio}, we prove our result on product state approximations (Theorem~\ref{1_thm:approxguarantee} and the ensuing proof of Theorem~\ref{1_thm:approxprodstate}), show its tightness in the $2$-local case and provide the upper bound of $d^{-\lfloor k/2\rfloor}$ for the best possible
approximation by product states. Section~\ref{1_scn:sepopt} gives our polynomial time approximation algorithm and develops the
general sampling and SDP-based technique we use. It also shows that the dense local Hamiltonian problem remains
QMA-complete. As some of the proofs and notation of Section~\ref{1_scn:sepopt} are rather technical, we have deferred the full proofs of this section to Section~\ref{1_app:A} in order to facilitate reading.

%\paragraph{Notation:} We use $A\succeq 0$ to say operator $A$ is positive semidefinite, and denote by $L(\spa{X})$, $H(\spa{X})$, and $D(\spa{X})$ the sets of linear, Hermitian, and density operators acting on complex Euclidean space $\spa{X}$, respectively. We denote the Frobenius and operator norms of $A\in L(\spa{X})$ as $\fnorm{A}=\sqrt{\trace(A^\dagger A)}$ and $\snorm{A}=\max_{\ket{x}\in\spa{X}\mbox{ s.t. } \enorm{x}= 1}\enorm{A\ket{x}}$, respectively.

\section{Product states yield a $1/d^{k-1}$-approximation for qudits}\label{1_scn:prodratio}

We now show that product state assignments achieve a non-trivial approximation ratio for \klh, i.e.\ Theorem~\ref{1_thm:approxprodstate}. To do so, we first define the \emph{recursive Schmidt decomposition} (RSD, Definition~\ref{1_def:RSD}) of a state $\ket{\psi}\in(\complex^d)^{\otimes n}$, and for ease of exposition, the corresponding notion of a \emph{Schmidt cut} (Definition~\ref{1_def:SC}). We then state and prove the key to our approach, the \emph{Mixing Lemma} (Lemma~\ref{1_lem:mixinglemma}), which shows how to use the RSD to eliminate the entanglement across a particular Schmidt cut of $\ket{\psi}$ while maintaining the desired approximation ratio. Lemma~\ref{1_cor:mixinglemma} and Theorem~\ref{1_thm:approxguarantee} then expand on this by showing how to apply the Mixing Lemma to multiple Schmidt cuts. From Theorem~\ref{1_thm:approxguarantee}, a proof of Theorem~\ref{1_thm:approxprodstate} easily follows. We close with a discussion of the tightness of the approximation ratio given by Theorem~\ref{1_thm:approxprodstate}.

We first define the terms Recursive Schmidt Decomposition and Schmidt cut.

\begin{definition}[Recursive Schmidt Decomposition (RSD)]\label{1_def:RSD} \textup{Given a state $\ket{\psi}\in(\complex^d)^{\otimes n}$, we define its \emph{recursive} \emph{Schmidt} \emph{decomposition} as the expression obtained by recursively
applying the Schmidt decomposition on each qudit from $1$ to $n - 1$ inclusive. More formally, we define the RSD of $\ket{\psi}$ as follows:
\begin{itemize}
    \item (Base case) If $n=1$, then $\operatorname{RSD}(\ket{\psi})=\ket{\psi}$.
    \item (Recursive case) If $n>1$, then $\operatorname{RSD}(\ket{\psi})=\sum_{i=1}^d\alpha_i\ket{\psi_i}\otimes \operatorname{RSD}(\ket{\phi_i})$, where $\ket{\psi_i}\in\complex^d$, $\ket{\phi_i}\in(\complex^d)^{\otimes n-1}$, $\sum_{i=1}^d\alpha_i^2=1$, $\set{\ket{\psi_{i}}}$ is an orthonormal basis for the first qudit of $\ket{\psi}$, and $\set{\ket{\phi_i}}$ is a set of orthonormal vectors for the remaining $n-1$ qudits of $\ket{\psi}$.
\end{itemize}
%to obtain the RSD of $\ket{\psi}$, first take the Schmidt decomposition of $\ket{\psi}$ across the $\set{1}$ versus $\set{2,\ldots,n}$ split of qudits to obtain the expression
%\begin{equation}
%    \ket{\psi} = \sum_{i=1}^d\alpha_i\ket{\psi_i}\ket{\phi_i}_{2,\ldots ,n}\label{1_eqn:RSD0}
%\end{equation}
%where $\sum_{i=1}^d\alpha_i^2=1$, the set $\set{\ket{\psi_{i}}}$ is an orthonormal basis for qudit $\set{1}$, and $\set{\ket{\phi_i}_{2,\ldots , n}}$ is a set of orthonormal vectors for qudits $\set{2,\ldots,n}$. Now, recursively express each $\ket{\phi_i}_{2,\ldots , n}$ in Equation~(\ref{1_eqn:RSD0}) in terms of its Schmidt decomposition between qudit $\set{2}$ versus $\set{3,\ldots,n}$
%\begin{eqnarray}
%\ket{\psi_1}_{2,\ldots, n} &=& \alpha_3\ket{\psi_3}_{2}\ket{\psi_3}_{3,\ldots ,n} +\alpha_4\ket{\psi_4}_{2}\ket{\psi_4}_{3,\ldots,n}\label{1_eqn:RSD1}\\
%\ket{\psi_2}_{2,\ldots, n} &=& \alpha_5\ket{\psi_5}_{2}\ket{\psi_5}_{3,\ldots ,n} +\alpha_2\ket{\psi_6}_{2}\ket{\psi_6}_{3,\ldots,n}\label{1_eqn:RSD2},
%\end{eqnarray}
%and substitute Equations~(\ref{1_eqn:RSD1}) and (\ref{1_eqn:RSD2}) into Equation~(\ref{1_eqn:RSD0}) to obtain the expression
%\begin{eqnarray}
%    \ket{\psi} &=& \alpha_1\ket{\psi_1}_{1}\left(\alpha_3\ket{\psi_3}_{2}\ket{\psi_3}_{3,\ldots ,n} +\alpha_4\ket{\psi_4}_{2}\ket{\psi_4}_{3,\ldots,n}\right)+\\
%    && \alpha_2\ket{\psi_2}_{1}\left(\alpha_5\ket{\psi_5}_{2}\ket{\psi_5}_{3,\ldots ,n} +\alpha_2\ket{\psi_6}_{2}\ket{\psi_6}_{3,\ldots,n}\right).
%\end{eqnarray}
(This definition is relative to some fixed ordering of the qudits. The specific choice of ordering is unimportant in our scenario, as any decomposition output by such a process suffices to prove Theorem~\ref{1_thm:approxprodstate}.) For example, the RSD for $3$-qubit $\ket{\psi}$ is
\begin{equation}
    \ket{\psi} = \alpha_1 \ket{a_1}\otimes \left(\beta_1\ket{b_1}\ket{c_1}+\beta_2\ket{b_2}\ket{c_2}\right)+\alpha_2 \ket{a_2}\otimes (\beta^\prime_1\ket{{b^\prime}_1}\ket{{c^\prime}_1}+\beta^\prime_2\ket{{b^\prime}_2}\ket{{c^\prime}_2}),
\end{equation}
for $\alpha_1^2+\alpha_2^2=\beta_1^2+\beta_2^2={\beta^\prime_1}^2+{\beta^\prime}_2^2=1$, $\set{\ket{a_i}}_i$ an orthonormal basis for qubit $1$, $\set{\ket{b_i}}_i$ and $\set{\ket{{b^\prime}_i}}_i$ orthonormal bases for qubit $2$, and $\set{\ket{c_i}}_i$ and $\set{\ket{{c^\prime}_i}}_i$ orthonormal bases for qubit $3$.}
\end{definition}

\begin{definition}[Schmidt cut]\label{1_def:SC}
    \textup{For any $\ket{\psi}\in(\complex^d)^{\otimes n}$ with Schmidt decomposition $\ket{\psi}=\sum_{i=1}^{d} \alpha_i \ket{w_i}\ket{v_i}$, where $\alpha_i\in\reals$ with $\sum_i \alpha_i^2=1$, $\ket{w_i}\in\complex^d$ and $\ket{v_i}\in(\complex^d)^{\otimes n-1}$, and for any $\ket{\phi}\in(\complex^d)^{\otimes m}$, we refer to the expansion $\ket{\phi}\otimes\left(\sum_{i=1}^{d} \alpha_i \ket{w_i}\ket{v_i}\right)$ as the \emph{Schmidt cut} at qudit $m+1$. We say that a projector $\Pi$ \emph{crosses} this Schmidt cut if $\Pi$ acts on qudit $m+1$ and at least one qudit $i\in\set{m+2,\ldots,m+n}$.}
\end{definition}

The heart of our approach is the following Mixing Lemma, which provides, for \emph{any} assignment $\ket{\psi}\in(\complex^d)^{\otimes n}$, an explicit construction through which the entanglement across the first Schmidt cut of $\ket{\psi}$ can be eliminated, while maintaining at least a $(1/d)$-approximation ratio relative to the value $\ket{\psi}$ achieves against any local Hamiltonian $H\in \HH((\complex^d)^{\otimes n})$.

\begin{lemma}[Mixing Lemma]\label{1_lem:mixinglemma}
Given state $\ket{\psi}$ on $n$ qudits with Schmidt cut on qudit $1$ given by $\ket{\psi} = \sum_{i=1}^{d} \alpha_i \ket{w_i}\ket{v_i}$, where $\alpha_i\in\reals$ with $\sum_i \alpha_i^2=1$, $\ket{w_i}\in\complex^d$ and $\ket{v_i}\in(\complex^d)^{\otimes n-1}$, define $\rho := \sum_{i=1}^{d}\alpha_i^2\ketbra{w_i}{w_i}\otimes\ketbra{v_i}{v_i}$. Then, given projector $\Pi$ acting on some subset $\mathcal{S}$ of the qudits, if $\Pi$ crosses the Schmidt cut, then $\trace(\Pi\rho)\geq\frac{1}{d}\trace(\Pi\ketbra{\psi}{\psi})$. Otherwise, $\trace(\Pi\rho)=\trace(\Pi\ketbra{\psi}{\psi})$.
\end{lemma}
\begin{proof}
    Case $2$ follows easily by noting that the given Schmidt decomposition of $\ket{\psi}$ implies $\trace_1(\rho)=\trace_1(\ketbra{\psi}{\psi})$ and $\trace_{2,\ldots,n}(\rho)=\trace_{2,\ldots,n}(\ketbra{\psi}{\psi})$. To prove case $1$, we observe by straightforward expansion that
    \begin{equation}
        \trace(\Pi\ketbra{\psi}{\psi}) = \trace(\Pi\rho) + \sum_{i< j}\alpha_i\alpha_j\bra{w_i}\bra{v_i}\Pi\ket{w_j}\ket{v_j} + \alpha_i\alpha_j\bra{w_j}\bra{v_j}\Pi\ket{w_i}\ket{v_i}.\label{1_eqn:mixlem1}
    \end{equation}
    Then, by defining for each $i$ vector $\ket{a_{i}} := \alpha_i\Pi\ket{w_i}\ket{v_i}$, we have
    \begin{equation}
        \sum_{i< j}\alpha_i\alpha_j\bra{w_i}\bra{v_i}\Pi\ket{w_j}\ket{v_j} + \alpha_i\alpha_j\bra{w_j}\bra{v_j}\Pi\ket{w_i}\ket{v_i}= \sum_{i<j}\braket{a_{i}}{a_{j}}+\braket{a_{j}}{a_{i}},
    \end{equation}
    since $\Pi^2=\Pi$. Applying the fact that $\braket{a}{b} + \braket{b}{a} \leq \enorm{\ket{a}}^2 + \enorm{\ket{b}}^2$ for $\ket{a},\ket{b}\in(\complex^d)^{\otimes n}$ thus implies
    \begin{equation}
        \sum_{i<j}\braket{a_{i}}{a_{j}}+\braket{a_{j}}{a_{i}}\leq\sum_{i< j}\enorm{\ket{a_{i}}}^2 + \enorm{\ket{a_{j}}}^2
        =(d-1)\sum_{i}\alpha_i^2\bra{w_i}\bra{v_i}\Pi\ket{w_i}\ket{v_i}
        =(d-1)\trace(\Pi\rho),
    \end{equation}
    from which the claim follows.
\end{proof}

The following simple extension of Lemma~\ref{1_lem:mixinglemma} simplifies our proof of Theorem~\ref{1_thm:approxguarantee}.

\begin{corollary}\label{1_cor:mixinglemma}
    Define $\ket{\psi^\prime}:=\ket{\phi}\otimes \ket{\psi}$, where $\ket{\phi}\in(\complex^d)^{\otimes m}$ for $m>0$ and $\ket{\psi}$ is defined as in Lemma~\ref{1_lem:mixinglemma}, and let $\rho\in \DD(\complex^d)^{\otimes n}$ be obtained from $\ket{\psi}$ as in Lemma~\ref{1_lem:mixinglemma}. Then, for any projector $\Pi$ acting on a subset $\mathcal{S}$ of the qudits, if $\Pi$ crosses the Schmidt cut of $\ket{\psi^\prime}$ at qudit $m+1$, we have  $\trace(\Pi\ketbra{\phi}{\phi}\otimes\rho)\geq\frac{1}{d}\trace(\Pi\ketbra{\psi^\prime}{\psi^\prime})$. Otherwise, $\trace(\Pi\ketbra{\phi}{\phi}\otimes\rho)=\trace(\Pi\ketbra{\psi^\prime}{\psi^\prime})$.
\end{corollary}

\begin{proof}
     Immediate by applying the proof of Lemma~\ref{1_lem:mixinglemma} with the following modifications: (1) Define $\ket{a_{i}} := \alpha_i\Pi\ket{\phi}\ket{w_i}\ket{v_i}$, and (2) if $\mathcal{S}\subseteq\set{1,\ldots,m}\cup\set{m+2,\ldots,m+n}$ (i.e.\ this is one of two ways for $\Pi$ not to cross the cut --- the other way is for $\mathcal{S}\subseteq \set{1,\ldots,m+1}$), observe that by the same arguments as in Lemma~\ref{1_lem:mixinglemma} for case $2$ and the product structure between $\ket{\phi}$ and $\ket{\psi}$ in $\ket{\psi^\prime}$ that $\trace_{m+1}(\ketbra{\phi}{\phi}\otimes\rho)=\trace_{m+1}(\ketbra{\psi^\prime}{\psi^\prime})$.
\end{proof}

Lemma~\ref{1_lem:mixinglemma} shows that the state $\rho$ obtained by \emph{mixing} the $d$ Schmidt vectors of $\ket{\psi}$, as opposed to taking their \emph{superposition}, suffices to achieve a $(1/d)$-approximation across the first Schmidt cut. By iterating this argument over \emph{all} $n-1$ Schmidt cuts, we now prove that a mixture of all (product) states appearing in the RSD of $\ket{\psi}$ achieves an approximation ratio of $1/d^{k-1}$.
\begin{theorem}\label{1_thm:approxguarantee}
    For any $n$-qudit assignment $\ket{\psi}$ with RSD $\ket{\psi}=\sum_{i=1}^{d^{n-1}}\sqrt{p_i}\ket{\phi_i}$, where $\sum_i p_i = 1$ and $\set{\ket{\phi_i}}_{i=1}^{d^{n-1}}$ is a set of orthonormal product vectors in $(\complex^d)^{\otimes n}$, define $\rho := \sum_{i=1}^{d^{n-1}} p_i \ketbra{\phi_i}{\phi_i}$. Then, for any projector $\Pi$ acting on some subset $\mathcal{S}\subseteq\set{1,\ldots,n}$ of qudits with $\abs{\mathcal{S}}= k$, we have
%    \begin{equation}
$        \trace(\Pi\rho)\geq\frac{1}{d^{k-1}}\trace(\Pi\ketbra{\psi}{\psi})$.
%    \end{equation}
\end{theorem}
\begin{proof}
 Let $\Pi$ be a projector with $\abs{\mathcal{S}}= k$, and define $\ve{c}\in\set{0,1}^{n-1}$ such that $\ve{c}(j)=1$ iff $\Pi$ crosses the Schmidt cut at qudit $j$. For example, if $\Pi$ acts on qudits $\set{1,2}$, then $\ve{c}=(1,0,\ldots,0)$. Note that in general $\onorm{\ve{c}}=k-1$. Let $\ket{\psi_k}$ denote the expression obtained by taking the RSD of $\ket{\psi}$ up to the $k$th level of recursion for $1\leq k\leq n-1$, i.e. $\ket{\psi_k}$ can be written
  \begin{equation}
    \ket{\psi_k}=\sum_{i=1}^{d^k}\alpha_i \ket{\psi_i^1}\otimes\cdots\otimes\ket{\psi_i^k}\otimes\ket{\phi_i},
  \end{equation}
  where $\ket{\psi_i^j}\in\complex^d$ and $\ket{\phi_i}\in(\complex^d)^{\otimes n-k}$. (We assume $n\geq 2$, as otherwise the claim is vacuously true.) Corresponding to $\ket{\psi_k}$, define
  \begin{equation}\label{1_eqn:rhok}
    \rho^{(k)} := \sum_{i=1}^{d^k}\alpha_i^2\ketbra{\psi_i^1}{\psi_i^1}\otimes\cdots\otimes\ketbra{\psi_i^k}{\psi_i^k}\otimes\ketbra{\phi_i}{\phi_i}.
  \end{equation}
  Define $c_k:=\sum_{i=1}^k \ve{c}(i)$. To prove our claim, we show by induction that for all $1\leq k\leq n-1$, it holds that
\begin{equation}\label{1_eqn:hypothesis}
   \trace(\Pi\ketbra{\psi}{\psi})\leq d^{c_k} \trace(\Pi\rho^{(k)}).
\end{equation}
Note that the case $k=n-1$ is in particular the case we are interested in.

For the base case, let $k=1$. Consider first the Schmidt cut of $\ket{\psi}$ at qudit $1$, i.e.\ $\ket{\psi} = \sum_{i=1}^{d} \alpha_i \ket{\psi_i^1}\ket{\phi_i}$, for $\ket{\psi_i^1}\in\complex^d$ and $\ket{\phi_i}\in(\complex^d)^{\otimes n-1}$. Then, recalling that $\rho^{(1)} = \sum_{i=1}^{d}\alpha_i^2\ketbra{\psi_i^1}{\psi_i^1}\otimes\ketbra{\phi_i}{\phi_i}$, we have by Lemma~\ref{1_lem:mixinglemma} that
\begin{equation}
   \trace(\Pi\ketbra{\psi}{\psi})\leq d^{\ve{c}(1)} \trace(\Pi\rho^{(1)}),\label{1_eqn:recurse1}
\end{equation}
as desired.

For the inductive step, assume the inductive hypothesis holds for some $1\leq k\leq n-2$. We prove the claim holds for $k+1$. Note that by Equation~(\ref{1_eqn:hypothesis}), which holds due to the induction hypothesis for our specific value of $k$, it suffices to show that
\begin{equation}\label{1_eqn:goal}
    \trace(\Pi\rho^{(k)})\leq d^{\ve{c}(k+1)}\trace(\Pi\rho^{(k+1)}),
\end{equation}
since $d^{c_k+\ve{c}(k+1)}=d^{c_{k+1}}$.
To show this holds, consider the $i$th term in Equation~(\ref{1_eqn:rhok}), $\ketbra{\psi_i^1}{\psi_i^1}\otimes\cdots\otimes\ketbra{\psi_i^k}{\psi_i^k}\otimes\ketbra{\phi_i}{\phi_i}$, for arbitrary $1\leq i\leq d^k$. Observe this term satisfies the preconditions for Corollary~\ref{1_cor:mixinglemma} with $m=k$. Hence, via Corollary~\ref{1_cor:mixinglemma} there exists a state $\sigma_i$ acting on qudits $\set{k+1,\ldots,n}$ such that
\begin{equation}
\trace(\Pi\ketbra{\psi_i^1}{\psi_i^1}\otimes\cdots\otimes\ketbra{\psi_i^k}{\psi_i^k}\otimes\ketbra{\phi_i}{\phi_i})\leq d^{\ve{c}(k+1)}\trace(\Pi\ketbra{\psi_i^1}{\psi_i^1}\otimes\cdots\otimes\ketbra{\psi_i^k}{\psi_i^k}\otimes\sigma_i).
\end{equation}
Moreover, since $\sigma_i$ in Corollary~\ref{1_cor:mixinglemma} is obtained via the Mixing Lemma (Lemma~\ref{1_lem:mixinglemma}), by linearity we can express $\rho^{(k+1)}$ as
\begin{equation}
    \rho^{(k+1)}=\sum_{i=1}^{d^k}\alpha_i^2\ketbra{\psi_i^1}{\psi_i^1}\otimes\cdots\otimes\ketbra{\psi_i^k}{\psi_i^k}\otimes\sigma_i.
\end{equation}
We conclude by linearity that Equation~(\ref{1_eqn:goal}) holds, completing the proof.
%More generally, when considering the $p$th Schmidt cut, we apply Corollary~\ref{1_cor:mixinglemma} with $m=p-1$ to each of the at most $d^{p-1}$ terms appearing in the expansion of $\rho^{(p-1)}$. We continue iterating in this fashion until we have exhausted all $n-1$ Schmidt cuts, at which point the resulting mixture $\rho^{(n-1)}$ we are left with is in fact the $\rho$ from the statement of the claim (seen by noting that our procedure effectively iteratively computes the RSD of $\ket{\psi}$, mixing the Schmidt vectors it computes at each step). Moreover, due to the repeated application of Corollary~\ref{1_cor:mixinglemma}, we have
%\begin{equation}
%    \trace(\Pi\ketbra{\psi}{\psi})\leq d^{\onorm{\ve{c}}}\trace(\Pi\rho^{(n-1)}).
%\end{equation}
%Recalling that $\onorm{\ve{c}}=k-1$ completes the proof.
\end{proof}

With Theorem~\ref{1_thm:approxguarantee} in hand, we can now show Theorem~\ref{1_thm:approxprodstate}, i.e.\ that product states achieve approximation ratio $1/d^{k-1}$.

\begin{proof}(Theorem~\ref{1_thm:approxprodstate}) Simply apply Theorem~\ref{1_thm:approxguarantee} to each projector in the spectral decompositions of each (positive semidefinite) $H_i$ in our \klh~instance $H=\sum_i H_i$, and let $\ket{\psi}$ denote the optimal assignment for $H$. It is important to note that we can exploit Theorem~\ref{1_thm:approxguarantee} in this fashion due to the fact that the $\rho$ constructed by Theorem~\ref{1_thm:approxguarantee} is \emph{independent} of the projector $\Pi$ --- i.e.\ for any fixed $\ket{\psi}$ and $k$, the state $\rho$ provides the same approximation ratio against \emph{any} $k$-local projector $\Pi$ encountered in the spectral decompositions of the $H_i$. Finally, note that one can find a \emph{pure} product state achieving this approximation guarantee since $\rho$ is a convex mixture of pure product states.
\end{proof}

\paragraph{Upper bound of $d^{-\lfloor \frac{k}{2}\rfloor}$ for product state approximations.}
Is the result of Theorem~\ref{1_thm:approxprodstate} tight? In the case of MAX-$2$-local Hamiltonian on qudits, yes --- consider a single clause projecting onto the maximally entangled state $\frac{1}{\sqrt{d}}\sum_i \ket{ii}$, for which a product state achieves value at most $1/d$. On the other hand, for MAX-$3$-local Hamiltonian on qubits, the worst case clause for a $3$-qubit product state assignment is the projector onto the state $\ket{W}=\frac{1}{\sqrt{3}}(\ket{001}+\ket{010}+\ket{100})$~\cite{TWP09}. But here product states achieve value $4/9$~\cite{WG03}, implying the bound of $1/4$ from Theorem~\ref{1_thm:approxprodstate} is not tight.

An upper bound on the true optimal ratio of $8k^2/(2^{k})$ is implied by Theorem~2 of~\cite{GFE09} for the case where $d=2$ and $k\geq 11$. For general $d$ and $k$, a simple construction shows that the optimal ratio is upper bounded by $d^{-\lfloor\frac{k}{2}\rfloor}$. To see this, consider a single clause which
is the tensor product of maximally entangled bipartite states (for odd $k$, we assume the odd qudit out projects onto the identity). For example, for $n=4$, consider the clause
$\ketbra{\phi^+}{\phi^+}\otimes\ketbra{\phi^+}{\phi^+}$, where
$\ket{\phi^+}=\frac{1}{\sqrt{2}}(\ket{00}+\ket{11})$. The maximum value
a product state can attain is $1/4$, as claimed. In the qubit setting ($d=2$), one can further improve this construction for odd $k$ by replacing the term $\ketbra{\phi^+}{\phi^+}\otimes I$ on the last three qubits with $\ketbra{W}{W}$. For example, for $k=5$, setting our instance to be the clause $\ketbra{\phi^+}{\phi^+}\otimes\ketbra{W}{W}$ yields an upper bound of $(1/2)(4/9)=2/9<1/4=d^{-\lfloor\frac{k}{2}\rfloor}$ (where we again use the value $4/9$ for $\ket{W}$ from the previous paragraph). For general odd $k>1$, this improved bound generalizes to $2^{\frac{-k+7}{2}}/9$.

\section{Optimizing over the set of separable states}\label{1_scn:sepopt}
Section~\ref{1_scn:prodratio} showed that there always \emph{exists} a product state assignment achieving a certain non-trivial approximation ratio. In this section, we show how to efficiently \emph{find} such a product state. Our main theorem of this section is the following (Theorem~\ref{1_thm:ptas}), from which Theorem~\ref{1_thm:ptasdense} follows easily (see discussion at end of Section~\ref{1_sscn:finalptas}). As the proofs and full notation of this section are rather dense, we first discuss our results below using simplified notation and without proofs. Full proofs and technical details are deferred to Section~\ref{1_app:A}.

\begin{theorem}\label{1_thm:ptas}
    Let $H$ be an instance of \klh~acting on $n$ qudits, and let $\optprod$ denote the optimum value of $\trace(H\rho)$ over all \emph{product} states $\rho\in \DD((\complex^d)^{\otimes n})$. Then, for any fixed $\epsilon>0$, there exists a polynomial time (deterministic) algorithm which outputs $\rho_1\otimes\cdots\otimes\rho_n\in \DD((\complex^d)^{\otimes n})$ such that
$
        \trace(H\rho_1\otimes\cdots\otimes\rho_n)\geq \optprod -\epsilon n^k.
$
\end{theorem}

We first outline our approach by generalizing the discussion in Section~\ref{1_scn:intro}, introducing tools and notation we will require along the way. The optimal value $\optprod$ over product state assignments for any \klh~instance can be expressed as the following program, denoted $P_1$:
\begin{equation}
      \optprod\hspace{2mm} = \hspace{2mm}\max \hspace{2mm}\sum_{i_1,\ldots,i_k}^n\trace(H_{i_1,\ldots,i_k}\rho_{i_1}\otimes\cdots\otimes\rho_{i_k})
    \hspace{2mm}\mbox{s.t.}\hspace{2mm} \rho_i\succeq 0\hspace{2mm} \mbox{and}\hspace{2mm}\trace(\rho_i)=1\hspace{2mm}\mbox{for } 1\leq i\leq n.\label{1_eqn:obj}
\end{equation}

\noindent As done in Equation~(\ref{1_eqn:introdecomp}), we now recursively decompose our objective function as a sequence of nested sums. Let $\set{\sigma_i}_{i=1}^{d^2}$ be a Hermitian orthogonal basis for the set of Hermitian operators acting on $\complex^d$, such that $\trace(\sigma_i\sigma_j)=2\delta_{ij}$ (for $\delta_{ij}$ the Kroenecker delta). (See, e.g.~\cite{K03}, or Equations~(\ref{7_eqn:Ugenerators}), (\ref{7_eqn:Vgenerators}), and (\ref{7_eqn:Wgenerators}) for an explicit construction of such basis elements. We remark that there is nothing special about the normalization factor of $2$ in the term $2\delta_{ij}$ above; this value is simply consistent with the specific basis construction we have chosen to employ, which generalizes the Pauli basis for a qubit system.) Then, by rewriting each $H_{i_1,\ldots,i_k}$ in terms of $\set{\sigma_i}_{i=1}^{d^2}$, our objective function becomes
\begin{align}
    \sum_{i_k,\ldots,i_1}^n\trace\left[\left(\sum_{j_k,\ldots,j_1=1}^{d^2}r_{j_1,\ldots,j_k}^{i_1,\ldots,i_k}\sigma_{j_k}\otimes\cdots\otimes\sigma_{j_1}\right)\rho_{i_k}\otimes\cdots\otimes\rho_{i_1}\right] =\hspace{45mm}\nonumber\\
    \sum_{i_k,j_k}\trace(\sigma_{j_k}\rho_{i_k})\left[\sum_{i_{k-1},j_{k-1}}\trace(\sigma_{j_{k-1}}\rho_{i_{k-1}})\left[\cdots\left[\sum_{i_1}\trace\left(\left(\sum_{j_1}r_{j_1,\ldots,j_k}^{i_1,\ldots,i_k}\sigma_{j_1}\right)\rho_{i_1}\right)\right]\right]\right],\label{1_eqn:decomposition}
\end{align}
where each $\ve{r}^{i_1,\ldots,i_k}\in\reals^{d^2}$. We henceforth think of the objective function above as a ``degree-$k$ inner product'', i.e.\ as a sequence of $k$ nested sums involving inner products, in analogy to the degree-k polynomials of Reference~\cite{AKK99}. In this sense, a degree-$1$ inner product would refer to only the innermost sums over $i_1$ and $j_1$, and a degree-$k$ inner product would denote the entire expression in Equation~(\ref{1_eqn:decomposition}). More formally, we denote a degree-$b$ inner product for $1\leq b \leq k$ using map $t_b:\HH(\complex^d)^{\times n}\mapsto\reals$, defined such that
\begin{equation}\label{1_eqn:tb}
t_b(\rho_1,\ldots,\rho_n):=    \sum_{i_b,j_b}\trace(\sigma_{j_b}\rho_{i_b})\left[\cdots\left[\sum_{i_1}\trace\left(\left(\sum_{j_1}r_{j_1,\ldots,j_k}^{i_1,\ldots,i_k}\sigma_{j_1}\right)\rho_{i_1}\right)\right]\right].
\end{equation}
Note that $t_b$ implicitly depends on parameters $i_{b+1},\ldots, i_k$ and $j_{b+1},\ldots , j_k$. (See the beginning of Section~\ref{1_app:A} for more elaborate notation used in the proofs of the claims of Section~\ref{1_scn:sepopt}.)

Our approach is to ``linearize'' the objective function of $P_1$ using exhaustive sampling and recursion to estimate its degree-$(k-1)$ inner products. To do so, we require the Sampling Lemma.

\begin{lemma}[Sampling Lemma~\cite{AKK99}]\label{1_lem:sample}
    Let $(a_i)$ be a sequence of $n$ real numbers with $\abs{a_i}\leq M$ for all $i$, and let $f,g>0$. If we choose a multiset of $s=g\log n$ of the $a_i$ at random (with replacement), then their sum $q$ satisfies
$
        \sum_i a_i - nM\sqrt{\frac{f}{g}} \leq q\times\frac{n}{s}\leq \sum_i a_i + nM\sqrt{\frac{f}{g}}
$
    with probability at least $1-n^{-f}$.
\end{lemma}

\noindent The proof of Lemma~\ref{1_lem:sample} follows from a simple application of the H\"{o}ffding bound~\cite{H64}. To use the Sampling Lemma in conjunction with exhaustive sampling, we discretize the space of $1$-qudit density operators using a $\delta$-net $G\subseteq \HH(\complex^{d})$, such that for all $\rho\in \DD(\complex^{d})$, there exists $\sigma\in G$ such that $\fnorm{\rho-\sigma}\leq \delta$. We now show how to construct $G$.

To obtain $G$, we instead construct a $\delta$-net for a subset of $\HH(\complex^d)$ which \emph{contains} $\DD(\complex^d)$, namely the set $\mathcal{A}(\complex^d):=\set{A\in \HH(\complex^d)\mid \max_{i,j} \abs{A(i,j)}\leq 1}$. (Note: A net over $\mathcal{A}(\complex^d)$ may allow non-positive assignments for a qudit. See Section~\ref{1_sscn:finalptas} for why this is of no consequence.) Creating a $\delta$-net over $\mathcal{A}(\complex^d)$ is simple: we cast a $(\delta/d)$-net over the unit disk for each of the complex $d(d-1)/2$ matrix entries above the diagonal, and likewise over $[-1,1]$ for the entries on the diagonal. Letting $m$ and $n$ denote the minimum number of points required to create such $(\delta/d)$-nets for each of the diagonal and off-diagonal entries, respectively, we have that $\abs{G}=m^{\frac{d(d-1)}{2}}n^d$. For example, simple nets of size $m\approx d/\delta$ and $n\approx d^2/\delta^2$ can be obtained by placing a 1D and 2D grid over $[-1,1]$ and the length $2$ square in the complex plane centered at $(0,0)$, respectively, implying $\abs{G}\in O(1)$ when $d\in O(1)$. To show that $G$ is indeed a $\delta$-net, we now bound the Frobenius distance between arbitrary $\rho\in \DD(\complex^d)$ and the closest $\tilde{\rho}\in G$. (We use the Frobenius norm as it allows a simple analysis. Below, one could also consider the $l_\infty$ norm bound $\snorm{A}\leq \delta/d$, where in this context $\snorm{A}=\max_{ij}\abs{A(i,j)}$). Specifically, let $A:=\rho-\tilde{\rho}$. Then:
\begin{equation}
    \fnorm{A}=\sqrt{\trace(A^\dagger A)}=\sqrt{\sum_{ij}\abs{A(i,j)}^2}\leq\sqrt{\sum_{ij}(\delta/d)^2}=\frac{\delta}{d}(d)=\delta.
\end{equation}

Finally, we remark that our \emph{dense} assumption on \klh~instances is only necessary to convert the absolute error of Theorem~\ref{1_thm:ptas} to a relative one (this conversion is detailed in Section~\ref{1_sscn:finalptas}). A dense assumption is not needed to apply the Sampling Lemma: Specifically, observe that Lemma~\ref{1_lem:sample} assumes there are $n$ terms in the sum to be estimated, and that we are able to determine $s$ of them. Looking back at Equation~(\ref{1_eqn:introprogram}) and considering, say, qudit $i$, if we wish to use the Sampling Lemma to estimate the inner sum over neighbours $N(i)$ of $i$, we might run into a problem if $i$ does \emph{not} have $\Theta(n)$ neighbours. To circumvent this~\cite{AKK99}, observe that Lemma~\ref{1_lem:sample} only gives us an estimate to within $\pm \epsilon n$. Thus, if $N(i)\leq \epsilon n/10$ (say), then we do not use the Sampling Lemma, but rather let our estimate be simply $0$, which is guaranteed to fall within the desired error bounds (observe an estimate of $0$ does not necessarily work, on the other hand, if $N(i)$ is large (say $N(i)=n-1$), since typically $f/g<1$). Throughout the remainder of our discussion, we assume this cutoff principle is implicitly present when employing Lemma~\ref{1_lem:sample}.

The remaining sections of this chapter are organized as follows: In Section~\ref{1_sscn:recurseSample}, we show how to recursively estimate degree-$b$ inner products using the Sampling Lemma. We then use this estimation technique in Section~\ref{1_sscn:linearization} to linearize our optimization problem $P_1$. Section~\ref{1_sscn:finalptas} brings everything together by presenting and analyzing the complete approximation algorithm. All technical proofs are found in Section~\ref{1_app:A}.

\subsection{Estimating degree-$b$ inner products via sampling}\label{1_sscn:recurseSample}

Our recursive procedure, EVAL, for estimating a degree-$b$ inner product using the Sampling Lemma is stated as Algorithm~\ref{1_alg:estimate}. There are two sources of error we must analyze: the Sampling Lemma, and our $\delta$-net over $\complex^d$. We claim that EVAL estimates the degree-$b$ inner product $t_b(\rho_1,\ldots,\rho_n)$ to within additive error $\pm \epsilon_{b} n^{b}$, where $\epsilon_{b}$ is defined as follows. Set $\Delta := \sqrt{2}d(1+\delta)$, for $\delta$ from our $\delta$-net. Then,
\begin{equation}
    \epsilon_{b} := \cc \left(\sqrt{\frac{f}{g}}+\delta\right)\left(\frac{\Delta^{b}-1}{\Delta-1}\right).\label{1_eqn:alg1error}
\end{equation}
The following lemma formalizes this claim. We adopt the convention of~\cite{AKK99} and let $x\in y\pm z$ denote $x\in [y-z,y+z]$. Algorithm~\ref{1_alg:estimate} is our operator analogue of the algorithm \emph{Eval} in Section 3.3 of~\cite{AKK99}. %We sometimes abbreviate $t_b(\rho_1,\ldots,\rho_n)$ as $t_{b}$ to simplify notation.

\begin{figure}[t]
\noindent\rule{\linewidth}{0.3mm}
\begin{alg}\rm{EVAL( }$t_{b}$ , $S$ , $\set{\tilde{\rho}_i:i\in S}$\rm{ )}.\label{1_alg:estimate}
    \begin{itemize}
        \item Input:\hspace{3mm}(1) A degree-$b$ inner product $t_{b}:\HH(\complex^d)^{\times n}\mapsto\reals$ for $1\leq b\leq k$\\
        \mbox{\hspace{14mm}}(2) A subset $S\subseteq\set{1,\ldots,n}$ of size $\abs{S}=O(\log n)$\\
         \mbox{\hspace{14mm}}(3) Sample points $\set{\tilde{\rho}_i:i\in S}$ such that $\fnorm{\tilde{\rho}_i-\rho_i}\leq \delta$ for all $i\in S$
        \item Output: $x\in\reals$ such that $x\in t_b(\rho_1,\ldots,\rho_n) \pm \epsilon_b n^{b}$ (for $\epsilon_b$ defined in Equation~(\ref{1_eqn:alg1error})).
    \end{itemize}
    \begin{compactenum}
        \item (Base Case) If $b=1$, return $\frac{n}{\abs{S}}\sum_{i_1\in S}\trace\left(\left(\sum_{j_1=1}^{d^2}r_{j_1,\ldots,j_k}^{i_1,\ldots,i_k}\sigma_{j_1}\right)\rho_{i_1}\right)$. (Note this return value depends on $i_2,\ldots,i_k,j_2,\ldots,j_k$, which are assumed to have a fixed value in the current recursive call to EVAL.)
        \item (Recurse) \hspace{3mm}For all $i\in S$ and $j= 1\ldots d^2$, set $e_{ij}$ = EVAL($t_{b-1}^{ij},S,\set{\tilde{\rho}_i:i\in S})$, where $t_{b-1}^{ij}$ is the term to the right of $\trace(\sigma_{j_b}\rho_{i_b})$ in Equation~(\ref{1_eqn:tb}).
        \item Return $\frac{n}{\abs{S}}\sum_{i\in S}\left[\sum_{j=1}^{d^2}\trace(\sigma_{j}\tilde{\rho}_{i})e_{ij}\right]$.
    \end{compactenum}
\end{alg}
\noindent\rule{\linewidth}{0.3mm}
\end{figure}

\begin{lemma}\label{1_lem:evalbound}
    Let $t_{k}:\HH(\complex^k)^{\times n}\mapsto\reals$ be defined using set $\set{H_{i_1,\ldots,i_k}}\subseteq \HH((\complex^d)^{\otimes k})$ (as in Equation~(\ref{1_eqn:decomposition})). Let $S\subseteq \set{1,\ldots , n}$ such that $\abs{S}=g\log n$ have its elements chosen uniformly at random with replacement. Let $\rho_1,\ldots ,\rho_n \in \DD(\complex^d)$ be some assignment on all $n$ qudits, and $\set{\tilde{\rho}_i:i\in S}$ a set of elements in our $\delta$-net such that $\fnorm{\rho_i-\tilde{\rho}_i}\leq \delta$ for all $i\in S$. Then, for $1\leq b\leq k$, with probability at least $1-d^{2b}n^{b-f}$, we have
$
        \operatorname{EVAL}(t_{b},S,\set{\tilde{\rho}_i:i\in S})\in t_b(\rho_1,\ldots,\rho_n) \pm \epsilon_{b} n^{b},
$
    where $\epsilon_{b}$ is defined as in Equation~(\ref{1_eqn:alg1error}).
\end{lemma}

\subsection{Linearizing our optimization problem}\label{1_sscn:linearization}

Our procedure, LINEARIZE, for ``linearizing'' the objective function of $P_1$ using EVAL from Section~\ref{1_sscn:recurseSample} is stated as Algorithm~\ref{1_alg:linearize}. Algorithm~\ref{1_alg:linearize} takes as input $P_1$ and a set of sample points $\set{\tilde{\rho_i}}$, and outputs a semidefinite program (SDP) which we shall henceforth refer to as $P_2$. We remark that LINEARIZE is our version of the procedure \emph{Linearize} in Section~3.4 of~\cite{AKK99}, extended to the setting of operators and a more complex error structure. Although LINEARIZE is presented as linearizing an objective function here, the same techniques straightforwardly apply in linearizing constraints involving high-degree inner products.

\begin{figure}[t!]
\noindent\rule{\linewidth}{0.3mm}
\begin{alg}\rm{LINEARIZE( }$t_b$ , $\mathcal{N}$ , $S$, $\set{\tilde{\rho}_i:i\in S}$, $\epsilon$, $U$, $L$\rm{ )}.\label{1_alg:linearize}
    \begin{itemize}
        \item Input:\hspace{1mm}(1) A degree-$b$ inner product $t_{b}:\HH(\complex^d)^{\times n}\mapsto\reals$ for $1\leq b\leq k$.\\
               \mbox{\hspace{12mm}}(2) A set of linear constraints $\mathcal{N}$ (e.g.\ ``$\rho_i\succeq 0$'').\\
               \mbox{\hspace{12mm}}(3) A subset $S\subseteq\set{1,\ldots,n}$ of size $\abs{S}=O(\log n)$.\\
         \mbox{\hspace{12mm}}(4) \hspace{-2mm}Sample points $\set{\tilde{\rho}_i:i\in S}$ consistent with some feasible solution \mbox{\hspace{19mm}}$(\rho_1,\ldots,\rho_n)$ for $P_1$ such that $\fnorm{\tilde{\rho}_i-\rho_i}\leq \delta$ for all $i\in S$.  \\
               \mbox{\hspace{12mm}}(5) An error parameter $\epsilon>0$.\\
                \mbox{\hspace{12mm}}(6) (Optional) upper and lower bounds $U,L\in\reals$. If $U$ and $L$ are not provided, \mbox{\hspace{19mm}}we assume $U,L=\infty$.
        \item Output: (1) (Optional) A linear objective function $f:(\LL(\complex^d))^{ \times n}\rightarrow \reals$.\\
                \mbox{\hspace{16mm}}(2) An updated set of linear constraints, $\mathcal{N}$.
    \end{itemize}
    \begin{compactenum}
        \item (Base case) If $b=1$, then
        \begin{compactenum}
            \item (Trivial: Initial objective function was linear) If $U=L=\infty$, return [$t_b$, $\mathcal{N}$].
            \item (Reached bottom of recursion) Else, return [$\mathcal{N}\cup \set{``L\leq t_{b}(\rho_1,\ldots,\rho_n)\leq U"}$].
        \end{compactenum}
        \item (Recursive case) For $i=1\ldots n$ and $j=1\ldots d^2$ do
        \begin{compactenum}
            \item Set $e_{ij}:=\operatorname{EVAL}(t_{b-1}^{ij},S,\set{\tilde{\rho}_i:i\in S})$.
            \item Set $\epsilon^\prime := \epsilon - \cc\left(\sqrt{\frac{f}{g}}+\delta\right)\Delta^{b-1}$, for $\Delta$ defined in Equation~(\ref{1_eqn:alg1error}).
            \item Set $l_{ij}:=e_{ij}-\epsilon^\prime n^{b-1}$ and $u_{ij}:=e_{ij}+\epsilon^\prime n^{b-1}$.
            \item Call LINEARIZE($t_{b-1}^{ij},\mathcal{N},S,\set{\tilde{\rho}_i:i\in S}, \epsilon^\prime,u_{ij},l_{ij}$).
        \end{compactenum}
        \item (a) (Entire computation done) If $U=L=\infty$, return $\left[\sum_{ij}\trace(\sigma_{j}{\rho}_{i})e_{ij}, \mathcal{N}\right]$.\\
        (b) (Recursive call done) Else, return \\\mbox{\hspace{20mm}}$\left[\mathcal{N}\cup \set{``L-\epsilon^\prime d^2n^{b}\leq \sum_{ij}\trace(\sigma_{j}{\rho}_{i})e_{ij}\leq U+\epsilon^\prime d^2n^{b}"}\right]$.
    \end{compactenum}
\end{alg}
\noindent\rule{\linewidth}{0.3mm}
\end{figure}

We remark that the linear constraints output on each recursive call on line 3(b) of Algorithm~\ref{1_alg:linearize}~ensure the approximate consistency with our estimates from EVAL for any solution to $P_2$, as well as play a crucial role in bounding how good of an approximation $P_2$ yields to $P_1$.

To prove correctness of our final approximation algorithm, we require the following two important lemmas regarding $P_2$. The first shows that any feasible solution $(\rho_1,\ldots,\rho_n)$ for $P_1$ consistent with the sample set $\set{\tilde{\rho}_i:i\in S}$ fed into LINEARIZE is also a feasible solution for $P_2$ with high probability.

\begin{lemma}\label{1_lem:feasible}
    Let $t_{k}$, assignment $(\rho_1,\ldots,\rho_n)$, $S$, and $\set{\tilde{\rho}_i:i\in S}$ be defined as in Lemma \ref{1_lem:evalbound}. Then, for any $f,g>0$, calling LINEARIZE with parameters $t_{k}$, $\set{\tilde{\rho}_i:i\in S}$, and $\epsilon=\epsilon_k$ (for $\epsilon_k$ defined in Equation~(\ref{1_eqn:alg1error})) yields an SDP $P_2$ for which the assignment $\set{\rho_1,\ldots,\rho_n}$ is feasible with probability at least $1-d^{2k}n^{k-f}$.
\end{lemma}

The second lemma is a bound on how far the optimal solution of $P_2$ is from the optimal solution for $P_1$. We adopt the convention of~\cite{AKK99} and write $[x,y]\pm z$ to denote interval $[x-z,y+z]$.

\begin{lemma}\label{1_lem:guarantee}
    Let $\optprod$ be the optimal value for $P_1$, obtained by assignment $\rho^{\optprod}:=(\rho^{\operatorname{opt}}_1,\ldots,\rho^{\operatorname{opt}}_n)$. Let assignment $\set{\rho_i}_{i=1}^n=\set{\rho^{\operatorname{opt}}_i}_{i=1}^n$, $S$, and $\set{\tilde{\rho}_i:i\in S}$ be defined as in Lemma~\ref{1_lem:evalbound}. Let $P_2$ denote the SDP obtained by calling LINEARIZE with $S$, and denote by $\epsilon_m$ for $1\leq m \leq k$ the error parameter passed with map $t_m$ into a (possibly recursive) call to LINEARIZE. Then, letting $\optt$ denote the optimal value of $P_2$, we have with probability at least
     $1-d^{2k}n^{k-f}$ (for parameters set as in Lemma~\ref{1_lem:feasible}) that
%    \begin{equation}
 $       \optt \in \optprod \pm d(d+\sqrt{2})\left[\sum_{m=1}^{k-1}(\sqrt{2}d)^{k-1-m}\epsilon_m\right]n^{k}.$
%    \end{equation}
\end{lemma}

\subsection{The final algorithm}\label{1_sscn:finalptas}

We finally present our approximation algorithm, APPROXIMATE (Algorithm~\ref{1_alg:final}), in its entirety, which exploits our ability to linearize $P_1$ using LINEARIZE (Algorithm~\ref{1_alg:linearize}). This proves Theorem~\ref{1_thm:ptas}, which in turn implies Theorem~\ref{1_thm:ptasdense}. We first clarify a few points about APPROXIMATE, then analyze its runtime, and follow with further discussion, including the algorithm's derandomization and a proof that dense \klh~remains QMA-hard.

\begin{figure}[t]
\noindent\rule{\linewidth}{0.3mm}
\begin{alg}\rm {APPROXIMATE( }$H$ , $\epsilon${\rm~)}.\label{1_alg:final}
    \begin{itemize}
        \item Input: (1) A $k$-local Hamiltonian $H=\sum_{i_1, \ldots ,i_k} H_{i_1,\ldots,i_k}$ for each $H_{i_1,\ldots,i_k}\in \HH((\complex^d)^{\otimes k})$.\\
               \mbox{\hspace{13mm}}(2) An error parameter $\epsilon>0$.
        \item Output: A product assignment $\rho_1\otimes\cdots\otimes\rho_n$ that with probability at least $1/2$, has \mbox{\hspace{16mm}}value at least
              $\optprod-\epsilon n^k$, for $\optprod$ the optimal value for $H$ over all\\\mbox{\hspace{15mm}} product state assignments.
    \end{itemize}
    \begin{compactenum}
        \item Set $\epssdp:=\eps/10$.
        \item Define $h:\reals\rightarrow\reals$ such that for any error parameter $\epsilon$ input to LINEARIZE, $h(\eps)n^k$ is the absolute value of the bound on additive error given by Lemma~\ref{1_lem:guarantee}. Then, define $\epsilon^\prime$ implicitly so that $h(\epsilon^\prime) + \epssdp=\epsilon$ holds.
        \item Define constant $f$ such that $1-d^{2k}n^{k-f}>1/2$.
        \item Define constants $g$ and $\delta$ implicitly so that $\epsilon^\prime=\cc\left(\sqrt{\frac{f}{g}}+\delta\right)\left(\frac{\Delta^{k}-1}{\Delta-1}\right)$, for $\Delta$ defined in Equation~(\ref{1_eqn:alg1error}).
        \item Choose $g\log n$ indices $S\subseteq\set{1,\ldots,n}$ independently and uniformly at random.
        \item For each possible assignment $i$ from our $\delta$-net to the qudits in $S$:
        \begin{compactenum}
            \item Call LINEARIZE$(t_k,\set{P_1\mbox{'s constraints}},S,i, \epsilon^\prime)$ to obtain SDP $P_2^i$.
            \item Let $\alpha_i$ denote the value of $P_1$ obtained by substituting in the optimal solution of $P_2^i$.
        \end{compactenum}
        \item Return the assignment corresponding to the maximum over all $\alpha_i$.
    \end{compactenum}
\end{alg}
\noindent\rule{\linewidth}{0.3mm}
\end{figure}

We begin by explaining the rationale behind the constants in Algorithm~\ref{1_alg:final}. The constant $\epssdp$ is the additive error incurred when solving an SDP~\cite{GLS93}. We choose $\epsilon^\prime$ so that after running LINEARIZE and solving $P_2^i$, the total additive error is at most $\epsilon$, as desired. We choose $f$ to ensure the probability of success is at least $1/2$. Finally, we set $g$ large enough and $\delta$ (for our $\delta$-net) small enough to ensure that $\epsilon^\prime$ matches the error bounds for EVAL in Lemma~\ref{1_lem:evalbound}.

We now analyze the runtime of Algorithm~\ref{1_alg:final}. Let $\abs{G}$ denote the size of our $\delta$-net $G$ for a qudit. Then, for each of the $\abs{G}^{g\log n}$ iterations of line 6, we first take $O(n^{k-1})$ time to run LINEARIZE, outputting $O(n^{k-1})$ new linear constraints (seen via a simple inductive argument). We then solve SDP $P_2^i$, which can be done in time polynomial in $n$ and $\log(1/\epssdp)$ using the ellipsoid method~\cite{GLS93} (see, e.g.,~\cite{W09}). Let $r(n,\epssdp)$ denote the maximum runtime required to solve any of the $P_2^i$. Then, the overall runtime for Algorithm~\ref{1_alg:final} is $O(n^{g\log\abs{G}} (n^{k-1}+r(n,\epssdp)))$, which is polynomial in $n$ for $\epsilon,d,k\in O(1)$ (recall from Section~\ref{1_scn:sepopt} that $\abs{G}\in O((\frac{d}{\delta})^d)$, and that $\delta$ and $g$ are constant in our setting). Note that, due to the implicit dependence of $g$ on $\epsilon$, this runtime scales at least exponentially with varying $\epsilon$.

Before moving to further discussion, we make two remarks. First, one can efficiently convert the output of Algorithm~\ref{1_alg:final} to a \emph{pure} state with the same guarantee by adapting the standard classical \emph{method of conditional expectations}~\cite{V01}. To demonstrate, suppose $\set{\rho_i}$ is output by Algorithm~\ref{1_alg:final}. Then, set $\rho_1^\prime$ to be the eigenvector $\ketbra{\psi_j}{\psi_j}$ of $\rho_1$ for which the assignment $\ketbra{\psi_j}{\psi_j}\otimes\rho_2\otimes\cdots\otimes\rho_n$ performs best for $P_1$. (If the spectrum of $\rho_i$ is degenerate, begin by fixing an arbitrary choice of spectral decomposition for $\rho_i$.) Let our new assignment be $\rho_1^\prime\otimes\rho_2\otimes\cdots\otimes\rho_n$. Now repeat for each $\rho_i$ for $2\leq i \leq n$. The final state $\rho_1^\prime\otimes\cdots\otimes\rho_n^\prime$ is pure, and by convexity is guaranteed to perform as well as $\rho_1\otimes\cdots\otimes\rho_n$.

Second, recall from Section~\ref{1_scn:sepopt} that we constructed a $\delta$-net over a space larger than $\DD(\complex^d)$, allowing possibly non-positive assignments for a qudit. We now see that this is of no consequence, since regardless of which samples (positive or not) we use to derive our estimates with the Sampling Lemma, any feasible solution to $P_2^i$ in Algorithm~\ref{1_alg:final} is a valid assignment for $P_1$. Moreover, we know that for each optimal $\rho_i$ for $P_1$, there must be \emph{some} operator (positive or not) within distance $\delta$ in our net, ensuring our estimates obtained using the Sampling Lemma are within our error bounds.

\paragraph{Converting the absolute error of Algorithm~\ref{1_alg:final} into relative error.}
To convert the absolute error $\pm\epsilon n^k$ of Algorithm~\ref{1_alg:final} into a \emph{relative} error of $1-\epsilon^\prime$ for any $\epsilon^\prime$, define constant $c$ such that $c n^k$ is the value obtained for a \klh~instance by choosing the maximally mixed assignment $I/d^n$ (analogous to a classical random assignment). Since $I/d^n$ can be written as a mixture of computational basis states, we have $\optprod \geq c n^k$. It follows that by setting $\epsilon= c\epsilon^\prime$, Algorithm~\ref{1_alg:final} returns an assignment with value at least $\optprod-c\epsilon^\prime n^k\geq \optprod-\epsilon^\prime\optprod\geq \optprod(1-\epsilon^\prime)$, as desired.
\paragraph{Derandomizing Algorithm~\ref{1_alg:final}.}

The source of randomness in our algorithm is Lemma \ref{1_lem:sample}. By a standard argument in~\cite{AKK99} (see also~\cite{BR94,BGG93}), this randomness can be eliminated with only polynomial overhead. Specifically, we replace the random selection of $g\log n$ indices in the Sampling Lemma with the set of indices encountered on a random walk of length $O(g\log n)$ along a constant degree expander~\cite{G93}. Since the expander has constant degree, we can efficiently deterministically iterate through all $n^{O(g)}$ such walks, and since such a walk works with probability $1/n^{O(1)}$, at least one walk will work for all $\poly(n)$ sampling experiments we wish to run.

\paragraph{QMA-hardness of dense \klh.}
It is easy to see that (exact) MAX-$2$-local Hamiltonian remains QMA-hard for dense instances (a similar statement holds for MAX-$2$-SAT~\cite{AKK99}). For any MAX-$2$-local Hamiltonian instance with optimal value $\OPT$, we simply add $n$ qudits, between any two of which we place the constraint $\ketbra{00}{00}$ (no constraints are necessary between old and new qudits). Then, the new Hamiltonian has optimal value $\OPT+{n\choose{2}}$, making it dense, and the ability to solve this new instance implies the ability to solve the original one. The argument extends straightforwardly to \klh~for $k>2$.

\section{Further technical details and proofs}\label{1_app:A}

We now prove our claims in Section~\ref{1_scn:sepopt}. For this, we first require expanding on the notation we have set thus far.

\paragraph{Expanded Notation.} We now expand on our previous notation for analyzing Equation~(\ref{1_eqn:decomposition}) in order to facilitate proofs of the claims in Section~\ref{1_scn:sepopt}. First, to recursively analyze a clause $H_{i_1,\ldots,i_k}\subseteq \HH((\complex^{d})^{\otimes k})$, let $H_b\in \HH((\complex^{d})^{\otimes b})$ for any $1\leq b \leq k$ denote the action of $H_{i_1,\ldots,i_k}$ restricted to the first $b$ of its $k$ target qudits, i.e.\
\begin{equation}
H_b:=\sum^{d^2}_{j_b,\ldots,j_1=1}r_{j_1,\ldots,j_k}^{i_1,\ldots,i_k}\sigma_{j_{b}}\otimes\cdots\otimes\sigma_{j_{1}}.
\end{equation}
For example, $H_1=\sum^{d^2}_{j_1=1}r_{j_1,\ldots,j_k}^{i_1,\ldots,i_k}\sigma_{j_{1}}$ and  $H_k=H_{i_1,\ldots,i_k}$. Note that $H_b$ implicitly depends on variables ${i_1,\ldots,i_k,j_{b+1},\ldots, j_{k}}$. To reduce clutter, however, our notation does not explicitly denote this dependence unless necessary. Next, to recursively analyze a degree-$a$ inner product, we define $t_{a,b}:\HH(\complex^d)^{\times n}\mapsto\reals$ for any $0\leq a \leq k$ and $1\leq b \leq k$ such that
\begin{equation}
t_{a,b}(\rho_1,\ldots,\rho_n):=\sum^n_{i_{a},\ldots,i_1=1}\trace\left(H_b^{i_1,\ldots,i_k}\rho_{i_{b}}\otimes\cdots\otimes\rho_{i_1}\right)
\end{equation}
(where setting $a=0$ eliminates the sum over indices $i$). For example, $t_{k,k}$ is our full ``degree-$k$'' objective function in Equation~(\ref{1_eqn:obj}), and more generally, $t_{b,b}$ is the degree-b inner product in Equation~(\ref{1_eqn:decomposition}). Allowing different values for $a$ and $b$ greatly eases our technical analysis. We use the shorthand $t_b$ to denote $t_{b,b}$, and again only explicitly denote the dependence of $t_{a,b}$ on parameters $i_{a+1},\ldots, i_k$ and $j_{b+1},\ldots , j_k$ when necessary.

We now state and prove a technical lemma required for the remainder of our proofs here.

\begin{lemma}\label{1_lem:ubound}
    Let $\set{\rho_i}_{i=1}^n\subseteq \HH(\complex^d)$. For $\set{H_{i_1,\ldots,i_k}}\subseteq \HH(\complex^{d^k})$ any \klh~instance with decomposition for the $H_{i_1,\ldots,i_k}$ as given in Equation~(\ref{1_eqn:decomposition}), we have for any $0\leq a \leq k$ and $1\leq b \leq k$ that $\abs{t_{a,b}(\rho_1,\ldots,\rho_n)}\leq \left(\max_{i_{b},\ldots,i_1}\fnorm{\rho_{i_{b}}}\cdots\fnorm{\rho_{i_{1}}}\right)\cc n^{a}$.
\end{lemma}
\begin{proof}%[\textbf{Proof of Lemma~\ref{1_lem:ubound}}]
     By the triangle inequality and the H\"{o}lder inequality for Schatten $p$-norms (see Section~\ref{0_scn:linalg}), we have
    \begin{eqnarray}
        \abs{t_{a,b}}=\abs{\sum^n_{i_{a},\ldots,i_1=1}\trace\left(H_b\rho_{i_{b}}\otimes\cdots\otimes\rho_{i_1}\right)}&\leq& \sum^n_{i_{a},\ldots,i_1=1}\fnorm{H_b}\fnorm{\rho_{i_{b}}\otimes\cdots\otimes\rho_{i_1}}\\
        &\leq& \left(\max_{i_{b},\ldots,i_1}\fnorm{\rho_{i_{b}}}\cdots\fnorm{\rho_{i_{1}}}\right) \sum^n_{i_{a},\ldots,i_1=1}\fnorm{H_b},\nonumber
    \end{eqnarray}
    where we have used the fact that $\fnorm{A\otimes B}=\fnorm{A}\fnorm{B}$ for all $A,B\in \LL(\complex^d)$. If we can now show that $\fnorm{H_b}\leq \fnorm{H_k}$ for all $1\leq b \leq k$, then we would be done since we would have $\sum^n_{i_{a},\ldots,i_1}\fnorm{H_b}\leq \fnorm{H_{k}}n^{a}\leq \cc n^a$, where $\fnorm{H_{k}}\leq \cc$ since $\snorm{H_k}\leq 1$ by definition. Indeed, we claim that for any fixed $1\leq b\leq k$, we have $\fnorm{H_b}\leq 2^{\frac{b-k}{2}}\fnorm{H_{k}}$. To see this, note by straightforward expansion of the Frobenius norm and the fact that $\trace(\sigma_i\sigma_j)=2\delta_{ij}$ that
    \begin{equation}
        \fnorm{H_b} = \sqrt{\trace(H_b^2)}= 2^{\frac{b}{2}}\sqrt{\sum_{j_{b},\ldots,j_k}(r_{j_1,\ldots,j_k}^{i_1,\ldots,i_k})^2}\leq2^{\frac{b}{2}}\enorm{\ve{r}^{i_1,\ldots,i_k}}=2^{\frac{b-k}{2}}\left(2^{\frac{k}{2}}\enorm{\ve{r}^{i_1,\ldots,i_k}}\right),
    \end{equation}
    where $\ve{r}^{i_1,\ldots,i_k}$ is the coordinate vector of $H_{i_1,\ldots,i_k}$ from Equation~(\ref{1_eqn:decomposition}). Note, however, then for $b=k$, the inequality in the chain above is an equality, and so $\fnorm{H_{k}}=2^{\frac{k}{2}}\enorm{\ve{r}^{i_1,\ldots,i_k}}$. Substituting this into  the chain above completes the proof of our claim.
\end{proof}

We now prove our claims of Section~\ref{1_scn:sepopt}.

\begin{proof}[\textbf{Proof of Lemma~\ref{1_lem:evalbound}}]
    We first derive the error bound of $\epsilon_b$, and subsequently prove the probability bound. We follow~\cite{AKK99}, and proceed by induction on $b$. For the base case $b=1$, $\operatorname{EVAL}(H_{1},S,\set{\tilde{\rho}_i:i\in S})$ attempts to estimate
    \begin{equation}
        t_1(\rho_1,\ldots,\rho_n)=\sum_{i_1}\left[\sum_{j_1}r_{j_1,\ldots,j_k}^{i_1,\ldots,i_k}\trace(\sigma_{j_1}\rho_{i_1})\right]
    \end{equation}
    using our flawed sample points $\set{\tilde{\rho}_i:i\in S}$. To analyze the error of its output, assume first that our sample points are exact, i.e.\ $\tilde{\rho}_i=\rho_i$ for all $i\in S$. Then, by setting ``$a_i$'' in Lemma~\ref{1_lem:sample} to $t_{0,1}^{i_1}$ for $i=i_1$, and by using Lemma~\ref{1_lem:ubound} with parameters $a=0$ and $b=1$ to obtain upper bound $M=\cc$, we have by the Sampling Lemma that (with probability at least $1-n^{-f}$)
    \begin{equation}\label{1_eqn:basecaseerror}
        \frac{n}{\abs{S}}\sum_{i_1\in S}\left[\sum_{j_1}r_{j_1,\ldots,j_k}^{i_1,\ldots,i_k}\trace(\sigma_{j_1}\rho_{i_1})\right]\in t_1(\rho_1,\ldots,\rho_n)\pm \cc\sqrt{\frac{f}{g}}n.
    \end{equation}
    (Recall that the notation $x\in y\pm z$ means here $x\in [y-z,y+z]$.) This bound holds if we sum over exact sample points. If we instead sum over flawed sample points $\set{\tilde{\rho}_i:i\in S}$, the additional error is bounded by $\frac{n}{\abs{S}}$ times
    \begin{eqnarray}
        \abs{\sum_{i_1\in S}\left[\sum_{j_1}r_{j_1,\ldots,j_k}^{i_1,\ldots,i_k}\trace(\sigma_{j_1}(\rho_{i_1}-\tilde{\rho}_{i_1}))\right]}&\leq& \sum_{i_1\in S}\abs{\sum_{j_1}r_{j_1,\ldots,j_k}^{i_1,\ldots,i_k}\trace(\sigma_{j_1}(\rho_{i_1}-\tilde{\rho}_{i_1}))}\\&\leq& \sum_{i_1\in S}(\fnorm{\rho_{i_1}-\tilde{\rho}_{i_1}}\cc)\\&\leq& \cc\delta n,\label{1_eqn:sampleerror}
    \end{eqnarray}
    where the second inequality uses Lemma~\ref{1_lem:ubound} with parameters $a=0$ and $b=1$ and the promise of our $\delta$-net. We conclude for the base case that, as desired,
    \begin{eqnarray}
        \operatorname{EVAL}(H_{1},S,\set{\tilde{\rho}_i:i\in S})&=&\frac{n}{\abs{S}}\sum_{i_1\in S}\left[\sum_{j_1}r_{j_1,\ldots,j_k}^{i_1,\ldots,i_k}\trace(\sigma_{j_1}\tilde{\rho}_{i_1})\right]\\&\in& t_1(\rho_1,\ldots,\rho_n)\pm \cc\left(\sqrt{\frac{f}{g}}+\delta\right)n.
    \end{eqnarray}

    Assume now that the inductive hypothesis holds for $1\leq m \leq b-1$. We prove the claim for $m=b$. To do so, suppose first that the recursive calls on line 1(b) of Algorithm~\ref{1_alg:estimate} return the \emph{exact} values of $t_{b-1}^{ij}(\rho_1,\ldots,\rho_n)$, and that we have exact samples $\set{{\rho}_i:i\in S}$. Then, since by calling Lemma~\ref{1_lem:ubound} with $a=b-1$ we have $\abs{\sum_{j}\trace(\sigma_{j}\rho_{i})t_{b-1}^{ij}(\rho_1,\ldots,\rho_n)}\leq \cc n^{b-1}$, it follows by the Sampling Lemma that
    \begin{equation}
        \frac{n}{\abs{S}}\sum_{i\in S}\left[\sum_{j}\trace(\sigma_{j}\rho_{i})t_{b-1}^{ij}(\rho_1,\ldots,\rho_n)\right]\in         \sum_{i=1}^n\left[\sum_{j}\trace(\sigma_{j}\rho_{i})t_{b-1}^{ij}(\rho_1,\ldots,\rho_n)\right]\pm \cc\sqrt{\frac{f}{g}}n^{b}.\label{1_eqn:reccase1}
    \end{equation}
    To first adjust for using flawed samples, observe that an analogous calculation to Equation~(\ref{1_eqn:sampleerror}) yields $\abs{\frac{n}{\abs{S}}\sum_{i\in S}\left[\sum_{j}\trace(\sigma_{j}(\rho_{i}-\tilde{\rho}_{i}))\right]}\leq \cc\delta n^{b}$,
    where we have called Lemma~\ref{1_lem:ubound} with $a=b-1$. Thus, using flawed samples, the output of Algorithm~\ref{1_alg:estimate} satisfies
    \begin{equation}
        \frac{n}{\abs{S}}\sum_{i\in S}\left[\sum_{j}\trace(\sigma_{j}\tilde{\rho}_{i})t_{b-1}^{ij}\right]\in         \sum_{i=1}^n\left[\sum_{j}\trace(\sigma_{j}\rho_{i})t_{b-1}^{ij}\right]\pm \cc\left(\sqrt{\frac{f}{g}}+\delta\right)n^{b}.\label{1_eqn:exact}
    \end{equation}
    To next drop the assumption that our estimates $e_{ij}$ on line 1(b) are exact, apply the induction hypothesis to conclude that $e_{ij}\in t_{b-1}^{ij}(\rho_1,\ldots,\rho_n) \pm \epsilon_{b-1}n^{b-1}$. Then,
    \begin{eqnarray}
        \frac{n}{\abs{S}}\sum_{i\in S}\left[\sum_{j}\trace(\sigma_{j}\tilde{\rho}_{i})e_{ij}\right]&\in&
        \frac{n}{\abs{S}}\sum_{i\in S}\left[\sum_{j}\trace(\sigma_{j}\tilde{\rho}_{ij})\left(t_{b-1}^{ij} \pm \epsilon_{b-1}n^{b-1}\right)\right]\nonumber\\
        &\subseteq&
        \frac{n}{\abs{S}}\sum_{i\in S}\left[\sum_{j}\trace(\sigma_{j}\tilde{\rho}_{i})t_{b-1}^{ij}\right] \pm \frac{\epsilon_{b-1}n^{b}}{\abs{S}}\sum_{i\in S}\left[\sum_{j=1}^{d^2}\trace(\sigma_{j}\tilde{\rho}_{i}) \right]\nonumber\\
        &\subseteq&
        \frac{n}{\abs{S}}\sum_{i\in S}\left[\sum_{j}\trace(\sigma_{j}\tilde{\rho}_{i})t_{b-1}^{ij}\right] \pm \epsilon_{b-1}\sqrt{2}d(1+\delta)n^{b}\label{1_eqn:recurse2},
    \end{eqnarray}
    where the last statement follows since
    \begin{equation}
        \abs{\sum_{j=1}^{d^2}\trace(\sigma_{j}\tilde{\rho}_{i})}=\abs{\sum_{j=1}^{d^2}\trace\left(\sigma_{j}\left(\sum_{m=1}^{d^2} \tilde{r}_m \sigma_m\right)\right)}\leq 2\sum_{m=1}^{d^2} \abs{\tilde{r}_m}\leq 2d\enorm{\ve{\tilde{r}}}\leq\sqrt{2}d(1+\delta),\label{1_eqn:yetanotherequation}
    \end{equation}
    where $\ve{\tilde{r}}$ denotes the coordinate vector of $\tilde{\rho}_{i}$ with respect to basis $\set{\sigma_m}$, and we have used the facts that $\trace(\sigma_i\sigma_j)=2\delta_{ij}$, that $\onorm{\ve{x}}\leq \sqrt{d}\enorm{\ve{x}}$ for $\ve{x}\in\complex^d$, that $\fnorm{\tilde{\rho}_{i}}=\sqrt{2}\enorm{\ve{\tilde{r}}}$ for any $\tilde{\rho}_{i}\in \HH(\complex^d)$, and that $\fnorm{\tilde{\rho}_{i}}\leq 1+\delta$ (which follows from our $\delta$-net and the triangle inequality). Thus, recalling that $\Delta=\sqrt{2}d(1+\delta)$ and substituting Equation~(\ref{1_eqn:exact}) into Equation~(\ref{1_eqn:recurse2}), we have that
    \begin{equation}
        \frac{n}{\abs{S}}\sum_{i\in S}\left[\sum_{j}\trace(\sigma_{j}\tilde{\rho}_{i})e_{ij}\right]\in
        t_b(\rho_1,\ldots,\rho_n)\pm \left[\cc\left(\sqrt{\frac{f}{g}}+\delta\right)+\epsilon_{b-1}\Delta\right]n^{b}.
    \end{equation}
    We hence have the recurrence relation $\epsilon_b\leq \cc\left(\sqrt{\frac{f}{g}}+\delta\right)+\epsilon_{b-1}\Delta$, which when unrolled yields
    \begin{equation}
        \epsilon_b\leq \cc\left(\sqrt{\frac{f}{g}}+\delta\right)\sum_{m=0}^{b-1}\Delta^m=\cc\left(\sqrt{\frac{f}{g}}+\delta\right)\left(\frac{\Delta^{b}-1}{\Delta-1}\right),
    \end{equation}
    as desired. This concludes the proof of the error bound.

    To prove the probability bound, we show a stronger bound of $1-(\sum_{m=0}^{b-1}d^{2m} n^m)n^{-f}$ by induction on $b$. The base case $b=1$ follows directly from our application of the Sampling Lemma in Equation~(\ref{1_eqn:basecaseerror}). For the inductive step, define for brevity of notation $\gamma:=d^2 n$, and apply the induction hypothesis to line 1(b) of Algorithm~\ref{1_alg:estimate} to conclude that each of the $\gamma$ calls to EVAL fails will probability at most $(\sum_{m=0}^{b-2}\gamma^m)n^{-f}$. Then, by the union bound, the probability that at least one call fails is at most $(\sum_{m=1}^{b-1}\gamma^m)n^{-f}$. Similarly, since our application of the Sampling Lemma in line 2 of Algorithm~\ref{1_alg:estimate} fails with probability at most $n^{-f}$, we arrive at our claimed stronger bound of $1-\left(\sum_{m=0}^{b-1}\gamma^m\right)n^{-f}$, as desired.
\end{proof}

\begin{proof}[\textbf{Proof of Lemma~\ref{1_lem:feasible}}]
    We begin by observing that if one sets $\epsilon=\epsilon_k$, then the value of $\epsilon^\prime$ in line 2(b) of Algorithm~\ref{1_alg:linearize} is precisely $\epsilon_{k-1}$, and more generally, the $\epsilon$ passed into the recursive call of line 2(e) on $t_b$ for any $1\leq b\leq k$ is $\epsilon_b$. Now, focus on some recursive call on $t_b$ for $b>1$ (the case of $b=1$ is straightforward by Lemma~\ref{1_lem:evalbound}). If the estimates $e_{ij}$ in line 2(a) succeed, then by Lemma~\ref{1_lem:evalbound}, we know that $e_{ij}\in t_{b-1}^{ij}(\rho_1,\ldots,\rho_n)\pm \epsilon_{b-1}n^{b-1}$, implying $t_{b-1}^{ij}(\rho_1,\ldots,\rho_n)\in[l_{ij},u_{ij}]$. Now, $l_{ij}$ and $u_{ij}$ are only incorporated into linear constraints in recursive calls on $t_{b-1}^{ij}$, yielding constraints of the form
    \begin{equation}
        l_{i_bj_b}-\epsilon_{b-2} d^2n^{b-1}\leq \sum_{i_{b-1},j_{b-1}}\trace(\sigma_{j_{b-1}}{\rho}_{i_{b-1}})e_{i_{b-1}j_{b-1}}\leq u_{i_bj_b}+\epsilon_{b-2} d^2n^{b-1}.\label{1_eqn:lemma5_1}
    \end{equation}
    But $\set{\rho_1,\ldots,\rho_n}$ must now satisfy this constraint, since recall
    \begin{equation}
        t_{b-1}(\rho_1,\ldots,\rho_n)=\sum_{i_{b-1},j_{b-1}}\trace(\sigma_{j_{b-1}}{\rho}_{i_{b-1}})t_{b-2}^{i_{b-1}j_{b-1}}(\rho_1,\ldots,\rho_n),
    \end{equation}
    and there are $d^2n$ terms $e_{i_{b-1}j_{b-1}}$ in Equation~(\ref{1_eqn:lemma5_1}) each yielding an additional error of at most $\epsilon_{b-2}n^{b-2}$ (assuming EVAL succeeded on $t_{b-2}^{i_{b-1}j_{b-1}}$ in line 2(a)) above and beyond the bounds $t_{b-1}^{ij}(\rho_1,\ldots,\rho_n)\in[l_{ij},u_{ij}]$ we established above.

    We conclude that if, for \emph{all} $b$, $i$, and $j$, EVAL succeeds in producing estimates $e_{i_b}^{ij}$, then $\set{\rho_1,\ldots,\rho_n}$ is a feasible solution for $P_2$, as desired. The probability of this happening is, by the proof of Lemma~\ref{1_lem:evalbound}, at least $1-d^{2k}n^{k-f}$, since EVAL recursively estimates precisely the same terms during its execution\footnote{This holds even though on line 1 of Algorithm~\ref{1_alg:estimate}, we only estimate $d^2\abs{S}$ of the terms $e_{ij}$ (i.e.\ EVAL does not actually estimate \emph{all} terms in the recursive decomposition of $t_k$, as it does not need to) --- this is because in our analysis of the probability bound for Algorithm~\ref{1_alg:estimate}, we actually produced a looser bound by assuming all $n$ terms $e_{ij}$ are estimated.}.
\end{proof}

\begin{proof}[\textbf{Proof of Lemma~\ref{1_lem:guarantee}}]
    We begin by proving that for any recursive call to LINEARIZE on $t_b$ with valid upper and lower bounds $U$ and $L$ (i.e.\ $U,L\neq\infty$), respectively, we have for \emph{any} feasible solution $(\rho_1,\ldots,\rho_n)$ to $P_2$ that
    \begin{equation}\label{1_eqn:miniclaim}
        t_b(\rho_1,\ldots,\rho_n) \in [L,U] \pm d(d+\sqrt{2})\left[\sum_{m=1}^{b-1}(\sqrt{2}d)^{b-1-m}\epsilon_m\right]n^{b}.
    \end{equation}

    We prove this by induction on $b$, following~\cite{AKK99}. For base case $b=1$, the claim is trivial by line 1(b) of the algorithm. Now, assume by induction hypothesis that
    \begin{equation}
        t_{b-1}^{ij}(\rho_1,\ldots,\rho_n)\in [l_{ij},u_{ij}]\pm d(d+\sqrt{2})\left[\sum_{m=1}^{b-2}(\sqrt{2}d)^{b-2-m}\epsilon_m\right]n^{b-1}.
    \end{equation}
    By substituting the values of $l_{ij}$ and $u_{ij}$ from line 2(c), we have
    \begin{equation}
        t_{b-1}^{ij}(\rho_1,\ldots,\rho_n)\in e_{ij}\pm \left(d(d+\sqrt{2})\left[\sum_{m=1}^{b-2}(\sqrt{2}d)^{b-2-m}\epsilon_m\right]+\epsilon_{b-1}\right)n^{b-1}.
    \end{equation}
    We conclude that
    \begin{eqnarray}
        t_b(\rho_1,\ldots,\rho_n) &=& \sum_{ij}\trace(\sigma_{j}\rho_{i})t_{b-1}^{ij}(\rho_1,\ldots,\rho_n)\\
        &\subseteq&        \left[\sum_{ij}\trace(\sigma_{j}\rho_{i})e_{ij}\right]+ \\ &&\left(d(d+\sqrt{2})\left[\sum_{m=1}^{b-2}(\sqrt{2}d)^{b-2-m}\epsilon_m\right]+\epsilon_{b-1}\right)\left[\sum_{ij}\trace(\sigma_{j}\rho_{i})\right]n^{b-1}\\
        &\subseteq&        \left[\sum_{ij}\trace(\sigma_{j}\rho_{i})e_{ij}\right]+\nonumber\\&& \sqrt{2}d\left(d(d+\sqrt{2})\left[\sum_{m=1}^{b-2}(\sqrt{2}d)^{b-2-m}\epsilon_m\right]+\epsilon_{b-1}\right)n^{b}\label{1_eqn:miniclaim2}\\        &\subseteq&        \left[[L,U]\pm \epsilon_{b-1}d^2 n^b\right]+ \sqrt{2}d\left(d(d+\sqrt{2})\left[\sum_{m=1}^{b-2}(\sqrt{2}d)^{b-2-m}\epsilon_m\right]+\epsilon_{b-1}\right)n^{b}\nonumber\\    &\subseteq&        [L,U] \pm d(d+\sqrt{2})\left[\sum_{m=1}^{b-1}(\sqrt{2}d)^{b-1-m}\epsilon_m\right]n^{b},
        \end{eqnarray}
    where the third statement follows from a calculation similar to Equation~(\ref{1_eqn:yetanotherequation}), and the fourth statement from line 3(b) of Algorithm~\ref{1_alg:linearize}. This proves the claim of Equation~(\ref{1_eqn:miniclaim}).

    To complete the proof of Lemma~\ref{1_lem:guarantee}, observe that by Lemma~\ref{1_lem:feasible}, the assignment $\rho^{\operatorname{opt}}$ is feasible for $P_2$ with probability at least $1-d^{2k}n^{k-f}$. Thus, plugging $\rho^{\operatorname{opt}}$ into each of the $d^2n$ linear constraints produced by the recursive calls to LINEARIZE on each $t_{k-1}^{ij}$, we have by Equations~(\ref{1_eqn:miniclaim}) and~(\ref{1_eqn:miniclaim2}) that (with probability $1-d^{2k}n^{k-f}$) for $\optprod=t_k(\rho^{\operatorname{opt}})$,
    \begin{eqnarray}
        t_k(\rho^{\operatorname{opt}})&=&\sum_{ij}\trace\left(\sigma_{j}\rho_{i}^{\operatorname{opt}}\right)t_{k-1}^{ij}(\rho^{\operatorname{opt}})\\
        &\subseteq&\left[\sum_{ij}\trace(\sigma_{j}\rho_{i}^{\operatorname{opt}})e_{ij}\right]\pm \sqrt{2}d\left(d(d+\sqrt{2})\left[\sum_{m=1}^{k-2}(\sqrt{2}d)^{k-2-m}\epsilon_m\right]+\epsilon_{k-1}\right)n^{k}\nonumber\\        &\subseteq&\optt\pm d(d+\sqrt{2})\left[\sum_{m=1}^{k-1}(\sqrt{2}d)^{k-1-m}\epsilon_m\right]n^{k},
                \end{eqnarray}
    where the last statement follows since $\rho^{\operatorname{opt}}$ is not necessarily the optimal solution to $P_2$.
\end{proof}

\noindent \emph{Acknowledgements for this chapter.} We thank Jamie Sikora and Sarvagya Upadhyay for helpful feedback, and Yi-Kai Liu for interesting discussions. We wish to especially thank Oded Regev for many helpful comments and suggestions, and Richard Cleve for bringing our attention to the method of conditional expectations, and for stimulating discussions and support. 

%======================================================================
\chapter{Hardness of approximation for quantum problems}\label{chap:hardnessapprox}
\emph{This chapter is based on~\cite{GK12}:}\\

\vspace{-4mm}
\noindent S. Gharibian and J. Kempe. Hardness of approximation for quantum problems.
In \emph{Proceedings of 39th International Colloquium on Automata, Languages and Programming},
pages 387-398, 2012, DOI: 10.1007/978-3-642-31594-7, \copyright~2012 Springer, www.springerlink.com.

\vspace{3mm}
\noindent The polynomial hierarchy plays a central role in classical complexity theory. In this chapter, we define a quantum generalization of the polynomial hierarchy, and initiate its study. We show that not only are there natural complete problems for the second level of this quantum hierarchy, but that these problems are in fact hard to approximate. Using these techniques, we also obtain hardness of approximation for the class QCMA. Our approach is based on the use of dispersers, and is inspired by the classical results of Umans regarding hardness of approximation for the second level of the classical polynomial hierarchy~\cite{U99}. We close the chapter by showing that two variants of the local Hamiltonian problem with hybrid classical-quantum ground states are complete and hard to approximate for the second level of our quantum hierarchy, respectively.

\section{Introduction and results}\label{2_scn:intro}

Over the last decades, the Polynomial Hierarchy (PH)~\cite{MS72}, a natural generalization of the class NP, has been the focus of much study in classical computational complexity. Of particular interest is the second level of PH, denoted $\st$. Here, we say a problem is in $\st$ if it has an efficient verifier with the property that for any YES instance $x\in\set{0,1}^n$ of the problem, \emph{there exists} a polynomial length proof $y$ such that \emph{for all} polynomial length proofs $z$, the verifier accepts $x$, $y$ and $z$. Note that the \emph{alternation} from an existential quantifier over $y$ to a for-all quantifier over $z$ is crucial here -- keeping only the existential quantifier reduces us to NP.

It turns out that introducing such {alternating} quantifiers makes $\st$ a powerful class believed to be \emph{beyond} NP. For example, there exist natural and important problems known to be in $\st$ but not in NP. Such problems range from ``does the optimal assignment to a 3SAT instance satisfy \emph{exactly} $k$ clauses?'' to practically relevant problems related to circuit minimization, such as ``given a boolean formula $C$ in Disjunctive Normal Form (DNF), what is the smallest DNF formula $C'$ equivalent to $C$?'' (see, e.g.~\cite{U99}). The study of $\st$ has also led to a host of other fundamental theoretical results, such as the Karp-Lipton theorem, which states that $\class{NP}\not\subseteq \ppoly$ unless PH collapses to $\st$. $\st$ has even been used to prove that SAT cannot be solved simultaneously in linear time and logarithmic space~\cite{F97,FLMV05}. For these reasons, $\st$ and more generally PH have occupied a central role in classical complexity theoretic research.

Moving to the quantum setting, the study of quantum proof systems and a natural quantum generalization of NP, the class Quantum Merlin Arthur (QMA)~\cite{KSV02}, has been a very active area of research over the last decade. Recall from Section~\ref{0_sscn:classes} that a problem is in QMA if for any YES instance of the problem, there exists a polynomial size \emph{quantum} proof convincing a {quantum} verifier of this fact with high probability. With the notion of quantum proofs in mind, we thus ask the natural question: \emph{Can a {quantum} generalization of $\st$ be defined, and what types of problems might it contain and characterize?} Perhaps surprisingly, to date there are almost no known results in this direction.

\paragraph{Our results:} In this chapter, we introduce a quantum generalization of $\st$, which we call $\cqs$, and initiate its study. Our results include $\cqs$-completeness and $\cqs$-hardness of approximation for a number of new problems we define. Our techniques also yield hardness of approximation for the complexity class known as QCMA. We now describe these results in further detail.\\

\myparagraph{Hardness of approximation for $\cqs$} To begin, we informally define $\cqs$ (see Section~\ref{2_scn:def} for formal definitions).

\begin{definition}[$\cqs$ (informal)]
A problem $\Pi$ is in $\cqs$ if there exists an efficient quantum verifier satisfying the following property for any input $x\in\set{0,1}^n$:
\begin{itemize}
    \item If $x$ is a YES instance of $\Pi$, then {there exists} a \emph{classical} proof $y\in\set{0,1}^{\poly(n)}$ such that {for all} \emph{quantum} proofs $\ket{z}\in\B^{\otimes \poly(n)}$, the verifier accepts $x$, $y$ and $\ket{z}$ with high probability.
    \item If $x$ is a NO instance of $\Pi$, then for all \emph{classical} proofs $y\in\set{0,1}^{\poly(n)}$, there exists a \emph{quantum} proof $\ket{z}\in\B^{\otimes \poly(n)}$ such that the verifier rejects $x$, $y$ and $\ket{z}$ with high probability.
\end{itemize}
\end{definition}

\noindent (Recall here that $\B:=\complex^2$.) We believe this is a natural quantum generalization of $\st$. Here, the prefix $cq$ in $\cqs$ follows since the existential proof is classical, while the for-all proof is quantum. One can also consider variations of this scheme such as $\qqs$, $\qcs$, or $\ccs$ (with a quantum verifier), defined analogously. In this chapter, however, our focus is on $\cqs$, as it is the natural setting for the computational problems for which we wish to prove hardness of approximation. Note also that unlike for $\st$, the definition of $\cqs$ is bounded error -- this is due to the use of a quantum verifier for $\cqs$. This implies, for instance, that the quantum analogue of the classically non-trivial result $\BPP\subseteq\st$~\cite{S83,L83}, i.e.\ $\BQP\subseteq\cqs$, holds trivially. Finally, one can extend the definition of $\cqs$ to an entire hierarchy of quantum classes analogous to PH by adding further levels of alternating quantifiers, attaining presumably different classes depending on whether the quantifier at any particular level runs over classical or quantum proofs.

To next discuss hardness of approximation for $\cqs$, we recall two classical problems  crucial to our work here. First, in the NP-complete problem SET COVER, one is given a set of subsets $\set{S_i}$ whose union covers a ground set $U$, and we are asked for the smallest number of the $S_i$ whose union still covers $U$. If, however, the $S_i$ are represented \emph{succinctly} as the on-set\footnote{By \emph{on-set}, we mean the set of assignments which cause $\phi_i$ to be true.} of a $3$-DNF formula $\phi_i$, we obtain a more difficult problem known as \SSC~(SSC). SSC, along with a related problem \IRR~(IRR), are not just NP-hard, but are $\st$-complete (indeed, they are even $\st$-hard to approximate~\cite{U99}). SSC and IRR are defined as:

\begin{definition}[\SSC~(SSC)~\cite{U99}]
    Given a set $S=\set{\phi_i}$ of $3$-DNF formulae such that $\bigvee_{i\in S}\phi_i$ is a tautology, what is the size of the smallest $S'\subseteq S$ such that $\bigvee_{i\in S'}\phi_i$ a tautology?
\end{definition}

\begin{definition}[\IRR~(IRR)~\cite{U99}]
    Given a DNF formula $\phi=t_1\vee t_2\vee\cdots\vee t_n$, what is the size of the smallest $S\subseteq \set{t_i}_{i=1}^n$ such that $\phi\equiv\bigvee_{i\in S}t_i$?
\end{definition}

Our work introduces and studies quantum generalizations of SSC and IRR. In particular, analogous to the classically important task of circuit minimization, the quantum generalizations we define are arguably natural and related to what one might call ``Hamiltonian minimization'' -- given a sum of Hermitian operators $H=\sum_i H_i$, what is the smallest subset of terms $\set{H_i}$ whose sum approximately preserves certain spectral properties of $H$? We hope that such questions may be useful to physicists in a lab who wish to simulate the simplest Hamiltonian possible while retaining the desired characteristics of a complex Hamiltonian involving many interactions. We remark that at a high level, the connection to $\cqs$ for the task of Hamiltonian minimization is as follows: The classical existential proof encodes the subset of terms $\set{H_i}$, while the quantum for-all proof encodes complex unit vectors which achieve certain energies against $H$. The problem QUANTUM SUCCINCT SET COVER is now defined as follows.

\begin{definition} {\QSSC~(QSSC) (informal)}
    Given a set of local Hamiltonians $\set{H_i}$ such that $\sum_i H_i$ has smallest eigenvalue at least $\alpha$, what is the size of the smallest subset $S$ of the $H_i$ such that $\sum_{H_i\in S} H_i$ has smallest eigenvalue at least $\alpha$? Any subset satisfying this property is called a \emph{cover}.
\end{definition}

As defined in Section~\ref{0_sscn:LH}, a \emph{local Hamiltonian} is a sum of Hermitian operators, each of which acts non-trivially on at most $k\in\Theta(1)$ qubits. Intuitively, the goal in QSSC is to cover the entire Hilbert space using as few interaction terms $H_i$ as possible. Hence, we associate the notion of a ``cover'' with obtaining large eigenvalues, as opposed to small ones, making QSSC a direct quantum analogue of SSC. We remark that since SSC is a classical constraint satisfaction problem, we believe the language of \emph{quantum} constraint satisfaction, i.e.\ Hamiltonian constraints, is a natural avenue for defining QSSC. Our first result concerns QSSC, and is as follows.

\begin{theorem}\label{2_thm:QSSChard}
    QSSC is $\cqs$-complete, and moreover is $\cqs$-hard to approximate within $N^{1-\epsilon}$ for all $\epsilon>0$, where $N$ is the encoding size of the QSSC instance.
\end{theorem}

\noindent By \emph{hard to approximate}, we mean that any problem in $\cqs$ can be reduced to an instance of QSSC via a polynomial time mapping or Karp reduction such that the gap between the sizes of the optimal cover in the YES and NO cases scales as $N^{1-\epsilon}$. In other words, it is $\cqs$-hard to determine whether the smallest cover size of an arbitrary instance of QSSC is at most $g$ or at least $g'$ for $g'/g\in \Omega(N^{1-\epsilon})$ (where $g'\geq g$). We next define the problem QUANTUM IRREDUNDANT~(QIRR).

\begin{definition} {\QIRR~(QIRR) (informal)}
    Given a set of succinctly described orthogonal projection operators $\set{H_i}$ acting on $N$ qubits, and $\set{c_i\geq 0}\subseteq \reals$, define $H:=\sum_i c_iH_i$. Then, what is the size of the smallest subset $S\subseteq\set{H_i}$ such that for $H'=\sum_{H_i\in S}c_iH_i$, vectors achieving high and low energies against $H$ continue to obtain high and low energies against $H'$, respectively?
\end{definition}

\noindent Here, by a \emph{succinctly} described projector, we mean a possibly non-local operator which is the tensor product of $k$-local projectors for some $k\in\Theta(1)$. This non-local structure naturally generalizes IRR, where the DNF formula is allowed to be non-local. Our next result is the following.

\begin{theorem}\label{2_thm:QIRRhard}
    QIRR is $\cqs$-hard to approximate within $N^{\frac{1}{2}-\epsilon}$ for all $\epsilon>0$, where $N$ is the encoding size of the QIRR instance.
\end{theorem}

%We show two further results regarding hardness of approximation. First, via a simple application of the gap amplification technique of Umans~\cite{U99} and the improved disperser construction of Ta-Shma, Umans, and Zuckerman~\cite{TUZ07}, it turns out that the hardness ratios for QSSC and QIRR above can be improved to $N^{1-\epsilon}$ and $N^{\frac{1}{2}-\epsilon}$, respectively (see Corollary~\ref{2_cor:qsscampgap}).
\myparagraph{Hardness of approximation for QCMA} The techniques from above can also be used to show hardness of approximation for QCMA. Here, the class QCMA~\cite{AN02} is defined as $\cqs$ with the second (quantum) proof omitted, and can hence be thought of as the first level of our ``$cq$-hierarchy''. By defining the problem QUANTUM MONOTONE MINIMUM SATISFYING ASSIGNMENT (QMSA) (see Section~\ref{2_scn:QCMA}), we show:

\begin{theorem}\label{2_thm:QMSAhard}
        QMSA is QCMA-complete, and moreover is QCMA-hard to approximate within $N^{1-\epsilon}$ for all $\epsilon>0$, where $N$ is the encoding size of the QMSA instance.
\end{theorem}

\myparagraph{A canonical $\cqs$-complete problem} Our last results the canonical $\st$-complete problem $\ssat{i}$ and its generalization to the quantum setting. Specifically, given a boolean formula $\phi$, $\ssat{i}$ asks whether:
\begin{equation}
    \exists \ve{x}_1 \forall \ve{x}_2\exists \ve{x}_3 \cdots\forall \ve{x}_i {\rm ~~such~ that~~ } \phi(\ve{x}_1,\ve{x}_2,\ve{x}_3,\ldots,\ve{x}_i)=1.
\end{equation}
Here, we have assumed $i$ is even; for odd $i$, the last quantifier is a $\exists$. The terms $\ve{x}_j$ are vectors of boolean variables. For $i=2$, one can define a natural quantum generalization of this problem, denoted $\cqslh{k}$ and defined in Section~\ref{2_scn:anotherproblem}, using local Hamiltonians whose ground states are tensor products of a classical string and a quantum state. We show:

\begin{theorem}\label{2_thm:anothercomplete}
$\cqslh{3}$ is $\cqs$-complete.
\end{theorem}

\noindent Moreover, by defining an appropriate variant of $\cqslh{3}$, denoted $\cqslhmin$ and also defined in Section~\ref{2_scn:anotherproblem}, where the goal is to minimize the Hamming weight of the classical portion of the ground states mentioned above, we obtain the following result.

\begin{theorem}\label{2_thm:anothercomplete2}
    $\cqslhmin$ is $\cqs$-complete, and moreover is $\cqs$-hard to approximate within $N^{1-\epsilon}$ for any $\epsilon>0$, for $N$ the encoding size of the $\cqslhmin$ instance.
\end{theorem}

\textbf{Proof ideas:} Our proofs are inspired by the classical work of Umans~\cite{U99,HSUL02}, and are achieved in a few steps. First, we show a \emph{gap-introducing} reduction from an arbitrary $\cqs$ problem to a problem we call QUANTUM MONOTONE MINIMUM WEIGHT WORD (QMW) using \emph{dispersers}~(see e.g.,~\cite{SZ94,TUZ07}). We then show the following \emph{gap-preserving} reductions, where $\leq_K$ denotes a mapping or Karp reduction:
\begin{equation}
    \class{QMW} \leq_K \class{QSSC} \leq_K \class{QIRR}\enspace.
\end{equation}
This yields hardness ratios of $N^\epsilon$ for some $\epsilon>0$. To obtain the stronger results claimed in Section~\ref{2_scn:intro}, we finally apply the gap amplification of Umans~\cite{U99} and improved disperser construction of Ta-Shma, Umans, and Zuckerman~\cite{TUZ07}.

In the classical setting, Umans~\cite{U99,HSUL02} used dispersers to attain hardness of approximation results relative to $\st$ for the classical problems MMWW (the classical version of QMW), SSC and IRR. To extend his techniques to the quantum setting, the most involved aspects of our work are the gap-preserving reductions from QMW to QSSC to QIRR. Here, an intricate balancing act involving carefully defined local Hamiltonian terms is needed to construct operators with the spectral properties required for our reductions. To analyze the resulting sums of non-commuting Hamiltonians, we require heavier machinery, such as the specific structure of Kitaev's local Hamiltonian construction~\cite{KSV02}, the Projection Lemma of Kempe, Kitaev, and Regev~\cite{KKR06}, and the Geometric Lemma of Kitaev~\cite{KSV02}.

Finally, to show $\cqs$-completeness of $\cqslh{3}$, we study the interplay between proofs of a classical-quantum structure and Kempe and Regev's~\cite{KR03} $3$-local Hamiltonian construction. Specifically, a careful analysis reveals that any $\cqs$ verification circuit can be modified in such a way that fixing the value $c$ of its classical proof register leads to an \emph{effective} Hamiltonian $H_c$. We then study the spectrum of $H_c$ to achieve the desired result. Moving on to $\cqslhmin$, hardness of approximation is now attained by combining our reduction for $\cqslh{3}$ with the result that QMW is hard to approximate.

\paragraph{Previous and related work:}

%To the best of our knowledge, our work is the first to obtain hardness of approximation results for a quantum complexity class.
In terms of hardness of approximation, the related question of whether a \emph{quantum} PCP theorem holds is currently one of the biggest open problems in quantum complexity theory (see, e.g.,~\cite{A06,AALV09,A10,H12}). Regarding quantum generalizations of PH, the only previous work we are aware of is that of Yamakami~\cite{Y02}. However, the results of Yamakami are largely unrelated to ours (for example, complete problems are not studied), and the proposed definition of Reference~\cite{Y02} differs from ours in a number of ways: It is based on quantum Turing machines (whereas we work with quantum circuits), allows \emph{quantum} inputs (whereas here, like QMA, the input to a problem is a classical string), and considers quantum quantifiers at each level of the hierarchy (whereas in its full generality our scheme allows alternating between classical and quantum quantifiers between levels as desired).

\paragraph{Significance and open questions:}

The classical polynomial hierarchy plays an important role in classical complexity theory, both as a generalization of NP and as a proof tool in itself. It is hoped that the scheme we propose here for generalizing PH to the quantum setting will find similar applications in quantum complexity theory. Second, the problems we show to be $\cqs$-complete here are arguably natural, and in embodying a generalization of classical circuit minimization or optimization, may hopefully be related to practical scenarios in a lab. Further, although the alternation between classical and quantum quantifiers in $\cqs$ may \emph{a priori} seem odd, the notion of relating a classical proof to, say, subsets of local Hamiltonian terms, and the quantum proof to quantum states achieving certain energies is in itself quite natural, and in our opinion justifies the study of such a combination of quantifiers. Third, with respect to hardness of approximation, since whether a quantum PCP theorem holds remains a challenging open question, it is all the more interesting that one is able to prove hardness of approximation in a quantum setting here using an entirely different tool, namely that of dispersers. We remark that  dispersers and their two-sided analogues, extractors, have been used classically to amplify existing PCP inapproximability results~\cite{SZ94,Z96}. However, as far as we are aware, neither are known to directly yield PCP constructions.
%Third, our results are the first known hardness of approximation results for a quantum complexity class. Given that whether a quantum PCP theorem holds remains a challenging open question, it is all the more interesting that one is able to prove such results in a quantum setting using an entirely different tool, namely that of dispersers.

We leave a number of questions open: What other natural problems are complete for $\cqs$ or higher levels? Can we say anything non-trivial about the relationship between $\st$ and $\cqs$? How do the different classes $\cqs$, $\qcs$, $\qqs$, and $\ccs$ relate to each other? Where do the quantum hierarchies obtained by extending $\cqs$ to higher levels sit relative to known complexity classes? We hope the answers to such questions will help establish classes like $\cqs$ as fundamental concepts in the setting of quantum computational complexity.

\paragraph{Organization of this chapter:}
We begin in Section~\ref{2_scn:def} by formally defining the classes and problems studied in this chapter. In Section~\ref{2_scn:approxhardness}, we prove that QSSC and QIRR are hard to approximate for $\cqs$ within $N^\epsilon$; this is further improved in Section~\ref{2_scn:improvements}. Section~\ref{2_scn:QCMA} presents hardness of approximation results for QCMA. We close in Section~\ref{2_scn:anotherproblem} by showing $\cqs$-completeness of $\cqslh{3}$ and $\cqs$-hardness of approximation for $\cqslhmin$.

\section{Definitions}\label{2_scn:def}
We now define relevant classes and problems, and state lemmas which prove useful in our analysis. Throughout our discussion, recall that $\B := \complex^2$, and for a set $S$ of matrices over $\complex$, let $H_S:= \sum_{H_i\in S} H_i$.

We begin with a formal definition of $\cqs$. Recall that a promise problem is a pair $A=(\ayes,\ano)$ such that $\ayes,\ano\subseteq\set{0,1}^\ast$ and $\ayes \cap \ano = \emptyset$.
\begin{definition}[$\cqs$]\label{2_def:cqs}
    Let $A=(\ayes,\ano)$ be a promise problem. We say that $A\in\cqs$ if there exist polynomially bounded functions $t,c,q:\nats\mapsto\nats$, and a deterministic Turing machine $M$ acting as follows. For every $n$-bit input $x$, $M$ outputs in time $t(n)$ a description of a quantum circuit $V_x$ such that $V_x$ takes in a $c(n)$-bit proof $\ket{c}$, a $q(n)$-qubit proof $\ket{q}$, and outputs a single qubit. We say $V_x$ \emph{accepts} $\ket{c}\ket{q}$ if measuring its output qubit in the computational basis yields $1$. Then:
    \begin{itemize}
     \item Completeness: If $x\in \ayes$, then $\exists$ $\ket{c}$ such that $\forall$ $\ket{q}$, $V_x$ accepts $\ket{c}\ket{q}$ with probability $\geq2/3$.
     \item Soundness: If $x\in \ano$, then $\forall$ $\ket{c}$, $\exists$ $\ket{q}$ such that ${V_x}$ rejects $\ket{c}\ket{q}$ with probability $\geq2/3$.
    \end{itemize}
\end{definition}

\noindent Note that the completeness and soundness parameters can be amplified to values exponentially close to $1$. Specifically, we use the standard approach of repeating $V_x$ polynomially many times in parallel (see ``Error reduction for QMA'' in Section~\ref{0_sscn:classes}), except that we only need one copy of the classical register $\mathcal{C}$ for all parallel runs. For any value $c$ placed in $\mathcal{C}$, we think of it as being ``hardwired'' into $V_x$, thus obtaining a quantum verification circuit $V_{x,c}$, which we now apply in parallel to the many copies of the quantum proof. The standard weak error reduction analysis for QMA now applies (see, e.g.~\cite{AN02}).  Throughout this chapter, we refer to this as \emph{error reduction}.

We next define the terms $\XX$ circuit, monotone set, QMW, QSSC, and QIRR.

\begin{definition}[$\XX$ circuit]\label{2_def:cQMA}
    Let $n,m\in\nats^+$. A $\XX$ circuit $V$ is a quantum circuit receiving $n$ bits in an \emph{INPUT} register and $m$ qubits in a \emph{CHOICE} register, and outputting a single qubit $\ket{a}$. We say:
    \begin{itemize}
        \item $V$ \emph{accepts}  $x\in\set{0,1}^n$ in INPUT if for all $\ket{y}\in\B^{\otimes m}$ in CHOICE, measuring $\ket{a}$ in the computational basis yields $1$ with probability at least $2/3$.
        \item $V$ \emph{rejects}  $x\in\set{0,1}^n$ in INPUT if there exists a $\ket{y}\in\B^{\otimes m}$ in CHOICE such that measuring $\ket{a}$ in the computational basis yields $0$ with probability at least $2/3$.
    \end{itemize}
\end{definition}

\begin{definition}[Monotone set]
    A set $S\subseteq\set{0,1}^n$ is called \emph{monotone} if for any $x\in S$, any string obtained from $x$ by flipping one or more zeroes in $x$ to one is also in $S$.
\end{definition}

\begin{definition}[\QMMWW~(QMW)]\label{2_def:qmmww}
    Given a $\XX$ circuit $V$ accepting exactly a non-empty monotone set $S\subseteq\set{0,1}^n$, and integer thresholds $0\leq g\leq g'\leq n$, output:
    \begin{itemize}
        \item YES if there exists an $x\in\set{0,1}^n$ of Hamming weight at most $g$ accepted by $V$.
        \item NO if all $x\in\set{0,1}^n$ of Hamming weight at most $g'$ are rejected by $V$.
    \end{itemize}
\end{definition}

\noindent Note that clearly ${\rm QMW}\in\cqs$.

\begin{definition}[\QSSC~(QSSC)]\label{2_def:qssc}
    Let $S:=\set{H_i}$ be a set of $5$-local Hamiltonians $H_i$ acting on $N$ qubits such that $\sum_{H_i\in S} H_i\succeq \alpha I$ for $\alpha>0$. Then, given $\beta\in\reals$ such that $\alpha-\beta \geq 1$ and integer thresholds $0\leq g\leq g'$, output:
\begin{itemize}
    \item YES if there exists $S^\prime\subseteq S$ of cardinality at most $g$ such that $\sum_{H_i\in S^\prime}H_i\succeq\alpha I$.
    \item NO if for all $S^\prime\subseteq S$ of size at most $g'$, $\sum_{H_i\in S^\prime}H_i$ has an eigenvalue at most $\beta$.
\end{itemize}
    Any $S'$ satisfying the YES case is called a \emph{cover}.
\end{definition}

\noindent Note that requiring $\alpha-\beta\in\Omega(1)$ above is without loss of generality, as any instance of QSSC with gap $1/p(N)$ for $p$ a polynomially bounded function can be modified to obtain an equivalent instance with constant gap by multiplying each $H_i$ by $p(N)$~\cite{W09_2} (see Section~\ref{0_sscn:LH}).

\begin{definition}[\QIRR~(QIRR)]\label{2_def:qirr}
  Given $S:=\set{c_iH_i}$, where each $H_i$ acts on $N$ qubits and is a tensor product of $5$-local orthogonal projection operators and $c_i\geq 0$ are real. Then, given $\alpha,\beta\in\reals$ such that $\alpha-\beta\geq 1$, and integer thresholds $0\leq g\leq g'$, output:
\begin{itemize}
    \item YES if there exists $S'\subseteq S$ of cardinality at most $g$ such that for all $\ket{\psi}\in\B^{\otimes N}$:
         \begin{itemize}
            \item If $\trace(H_S\ketbra{\psi}{\psi})\geq \alpha$, then $\trace(H_{S'} \ketbra{\psi}{\psi})\geq \alpha$, and
            \item If $\trace(H_S\ketbra{\psi}{\psi})\leq \beta$, then $\trace(H_{S'} \ketbra{\psi}{\psi})\leq \beta$.
         \end{itemize}
    \item NO if for all $S'\subseteq S$ of cardinality at most $g'$, there exists a state $\ket{\psi}\in\B^{\otimes N}$ with $\trace(H_S\ketbra{\psi}{\psi})\geq \alpha$ and $\trace(H_{S'} \ketbra{\psi}{\psi})\leq \beta$.
\end{itemize}
\end{definition}

\noindent Roughly, QSSC asks how many local interaction terms in a local Hamiltonian one can discard while maintaining the value of the worst assignment. This is intended to mimic the idea of maintaining a tautology for a $3$-DNF formula in SSC classically. Analogous to the relationship between SSC and IRR, QIRR allows possibly non-local Hamiltonian terms so long as they have a succinct description (this generalizes the use of superconstant arity in IRR) and are projectors up to scalar multiplication (this generalizes the requirement that each term $t_i$ in IRR is an AND of variables). QIRR then asks how many interaction terms can be discarded in a sum of such Hamiltonian terms while ensuring that any assignment $\ket{\psi}$ achieves approximately the same value on both the original and modified Hamiltonians.

Next, the key tool enabling the creation of a gap in our reductions is a \emph{disperser} (see e.g.~\cite{SZ94,TUZ07}).

\begin{definition}[Disperser]\label{2_def:disperser}
    Let $G=(L,R,E)$ be a bipartite graph with $\abs{L}=2^n$, $\abs{R}=2^m$ and left-degree $2^d$. Then, $G$ is called a \emph{$(k,\epsilon)$-disperser} if, for any subset $L'\subseteq L$ of size $\abs{L'}\geq 2^k$, $L'$ has at least $(1-\epsilon)\abs{R}$ neighbors in $R$. Moreover, if for any pair $(v,i)$ for $v\in L$, one can compute the $i$th neighbor of $v$ in time polynomial in $n$, then the disperser is called \emph{explicit}.
\end{definition}

Finally, in this chapter we use the following useful known facts from local Hamiltonian complexity theory. To begin, we have two lemmas used to bound the eigenvalues of a pair of non-commuting operators. The first of these is the Geometric Lemma of Kitaev, which we stated as Lemma~\ref{l:pluslemma} in Section~\ref{0_sscn:5LH}. The second is the Projection Lemma, stated below.

\begin{lemma}[Kempe, Kitaev, Regev~\cite{KKR06}, Projection Lemma]\label{2_lem:projlemma}
    Let $Y=Y_1+Y_2$ act on Hilbert space $\mathcal{H}=\mathcal{S}+\mathcal{S}^\perp$ for Hamiltonians $Y_1$ and $Y_2$. Denote the zero eigenspace of $Y_2$ as $\mathcal{S}$, and assume the $Y_2$ eigenvectors in $\mathcal{S}^\perp$ have eigenvalue at least $J>2\snorm{Y_1}$. Then, for $\lambda(Y)$ the smallest eigenvalue of $Y$ and $Y|_\spa{S}:=\Pi_{\spa{S}}Y\Pi_{\spa{S}}$,
    \begin{equation}
        \lambda(Y_1|_{\mathcal{S}})-\frac{\snorm{Y_1}^2}{J-2\snorm{Y_1}}\leq \lambda(Y)\leq\lambda(Y_1|_{\spa{S}})\enspace.
    \end{equation}
\end{lemma}

%\begin{lemma}[Kitaev, Shen, Vyalyi~\cite{KSV02}, Geometric Lemma, Lemma 14.4]\label{l:pluslemma}
%    Let $A_1,A_2\succeq 0$, such that the minimum \emph{non-zero} eigenvalue of both operators is lower bounded by $v$. Assume that the null spaces $\nl$ and $\nll$ of $A_1$ and $A_2$, respectively, have trivial intersection, i.e.\ $\nl\cap\nll=\set{\vec{0}}$. Then
%    \begin{equation}
%        A_1+A_2 \succeq 2v\sin^2\frac{\alpha(\nl,\nll)}{2}I\enspace,
%    \end{equation}
%    where the \emph{angle} $\alpha(\spa{X},\spa{Y})$ between $\spa{X}$ and $\spa{Y}$ is defined over unit vectors $\ket{x}$ and $\ket{y}$ as $\cos\left[\angle(\spa{X},\spa{Y})\right] := \max_{\ket{x}\in\spa{X},\ket{y}\in\spa{Y}}\abs{\braket{x}{y}}$.
%\end{lemma}

We next briefly review the elements of Kitaev's circuit-to-Hamiltonian construction~\cite{KSV02} which play an important role in this chapter (see in Section~\ref{0_sscn:5LH} for an in-depth treatment). Given a $\cqs$ verification circuit $V=V_L\cdots V_1$ (where without loss of generality, each $V_i$ is a one- or two-qubit unitary) acting on $n$ proof bits (register $A$), $m$ proof qubits (register $B$), and $p$ ancilla qubits (register $C$), recall that this construction outputs a $5$-local Hamiltonian $H$ acting on $A\otimes B\otimes C\otimes D$, where $D$ is a {clock} register consisting of  $L$ qubits.  We then have $H:=\hin+\hout+\hprop+\hstab$, for \emph{penalty} terms as defined below:
\begin{eqnarray}
    \hin&:=&I_{A,B}\otimes
    \left(\sum_{i=1}^{p} \ketbra{1}{1}_{C_i}\right)\otimes \ketbra{0}{ 0}_D\\
    \hout&:=&I_A\otimes\ketbra{0}{0}_{B_1}\otimes
     I_C\otimes \ketbra{ L}{ L}_D\\
    \hprop &:=& \sum_{j=1}^{L} H_j {\rm,~where~ }H_j{\rm ~is ~defined~ as}\\
    && \hspace{-10mm}-\frac{1}{2}V_j\otimes\ketbra{ j}{ {j-1}}_D -\frac{1}{2}V_j^\dagger\otimes\ketbra{{j-1}}{ j}_D +\frac{1}{2}I\otimes(\ketbra{ j}{ j}+\ketbra{ {j-1}}{ {j-1}})_D\\
    \hstab&:=&I_{A,B,C}\otimes\sum_{i=1}^{L-1}\ketbra{01}{01}_{D_i,D_{i+1}}.
\end{eqnarray}
Above, the notation $A_i$ refers to the $i$th qubit of register $A$ (similarly for $B$, $C$, $D$). For any prospective proof $\ket{\psi}$ in $\trace(H\ketbra{\psi}{\psi})$, each penalty term has the following effect on the structure of $\ket{\psi}$: $\hin$ ensures that at time zero, the ancilla register is set to zero as it should be for $V$. $\hout$ ensures that at time step $L$ of $V$, measuring the output qubit causes acceptance with high probability. $\hprop$ forces all steps of $V$ appear in superposition in $\ket{\psi}$ with equal weights. Finally, note that for $\hin$, $\hout$, and $\hprop$ above, time $t$ in clock register $D$ is implicitly encoded in unary as $\ket{1^t0^{L-t}}$ (for $\hstab$ above, register $D$ is already explicitly written in unary); $\hstab$ is thus needed to prevent invalid encodings of time steps from appearing in $D$.

We use two important properties of this construction. First, the null space of $\hin+\hprop+\hstab$ is the space of \emph{history states}, which for arbitrary $\ket{\psi}_{A,B}$ are defined as
\begin{equation}\label{2_eqn:hist}
 \histstate:=\frac{1}{\sqrt{L+1}}\sum_{i=0}^L V_i\cdots V_1 \ket{\psi}_{A,B}\otimes\ket{0}_C\otimes\ket{i}_D.
\end{equation}
For $\cqs$ circuits $V$, it is convenient to define for $c\in\set{0,1}^n$ and $\ket{q}\in\B^{\otimes m}$ the shorthand $\histstatecq{c}{q}:=\histstate$ for $\ket{\psi}=\ket{c}\ket{q}$. The second important property of $H$ we use is that its spectrum is related to $V$ as follows.
\begin{lemma}[Kitaev~\cite{KSV02}]\label{2_lem:kitaev}
    The construction above maps $V$ to $(H,a,b)$ satisfying:
    \begin{itemize}
        \item If there exists a proof $\ket{\psi}$ accepted by $V$ with probability at least $1-\epsilon$, then $\histstate$ achieves $Tr(H\histstateketbra)\leq a$ for $a := \epsilon/(L+1)$.
        \item If $V$ rejects all proofs $\ket{\psi}$, then $H\succeq b I$ for $b\in\Omega\left(\frac{1-\sqrt{\epsilon}}{L^3}\right)$.
    \end{itemize}
\end{lemma}
\section{Hardness of approximation for $\cqs$}\label{2_scn:approxhardness}

We now show hardness of approximation for $\cqs$ for the problems QMW, QSSC, and QIRR. We begin with a gap-introducing reduction from an arbitary problem in $\cqs$ to QMW. We remind the reader that the hardness ratios obtained here are further strengthened in Section~\ref{2_scn:improvements}.

\begin{theorem}\label{2_thm:qmmwwHard}
    There exists a polynomial time reduction which, given an instance of an arbitrary $\cqs$ problem, outputs an instance of QMW with thresholds $g$ and $g'$ satisfying $g'/g\in \Theta(N^\epsilon)$ for some $\epsilon>0$, where $N$ is the encoding size of the QMW instance.
\end{theorem}
\begin{proof}
The reduction follows Theorem 1 of Umans~\cite{U99} closely; the points where we deviate from~\cite{U99} are explicitly noted. Let $\Pi$ be an instance of an arbitrary promise problem $A=(\ayes,\ano)$ in $\cqs$ with encoding size $n$, and whose verification circuit $V$ has a $c(n)$-bit existential proof register and a $q(n)$-qubit for-all proof register. We wish to map $\Pi$ to a \cqma~circuit $W$ for QMW such that $W$ accepts strings of small or large Hamming weight depending on whether $\Pi\in\ayes$ or $\Pi\in\ano$, respectively. To do so, we follow~\cite{U99} and construct an explicit $(k,1/2)$-disperser $G=(L,R,E)$ with left-degree $2^d$ using Reference~\cite{SZ94}, where $\abs{L}=2^{c(n)+1}$, $\abs{R}=2^{k+d-O(1)}$, and $k:=\gamma\log c(n)$ for $\gamma\in\Theta(1)$ to be set as needed. Note that the value of $d$ depends on the specific disperser construction used --- for the construction of ~\cite{SZ94}, we have $d=4k+O(\log n)$.
%define the value of d
Roughly, the idea of Umans is now to have $L$ correspond to assignments for the $c(n)$-bit classical register of $V$, and $R$ to assignments for the classical register of $W$ (in the setting of~\cite{U99}, note that $W$ is a classical circuit). We then \emph{encode} assignments from $L$ by instead choosing neighbor sets in $R$. By exploiting the properties of dispersers, one can ensure that the sizes of the neighbor sets in $R$ chosen vary widely between YES and NO cases for $\Pi$.

Specifically, imagine the vertices in $L$ are arranged into a complete binary tree whose $2^{c(n)}$ leaves denote the $2^{c(n)}$ possible assignments to $V$'s classical register. For convenience, we henceforth use $L$ to mean this tree. Now, let $x\in\set{0,1}^{c(n)}$ denote a leaf of $L$. Then, a subset of vertices $R'\subseteq R$ is said to \emph{encode} $x$ if it contains the union of the neighbor sets of all vertices in the unique path from the root of $L$ to $x$. Figure~\ref{2_fig:disp} illustrates this encoding scheme. How do the vertices of $R$ then relate to $W$? Each vertex $r\in R$ corresponds to an input bit of $W$ -- setting this $r$th bit to one means we ``choose'' vertex $r$.

\begin{figure}[t]\centering
  \includegraphics[height=5cm]{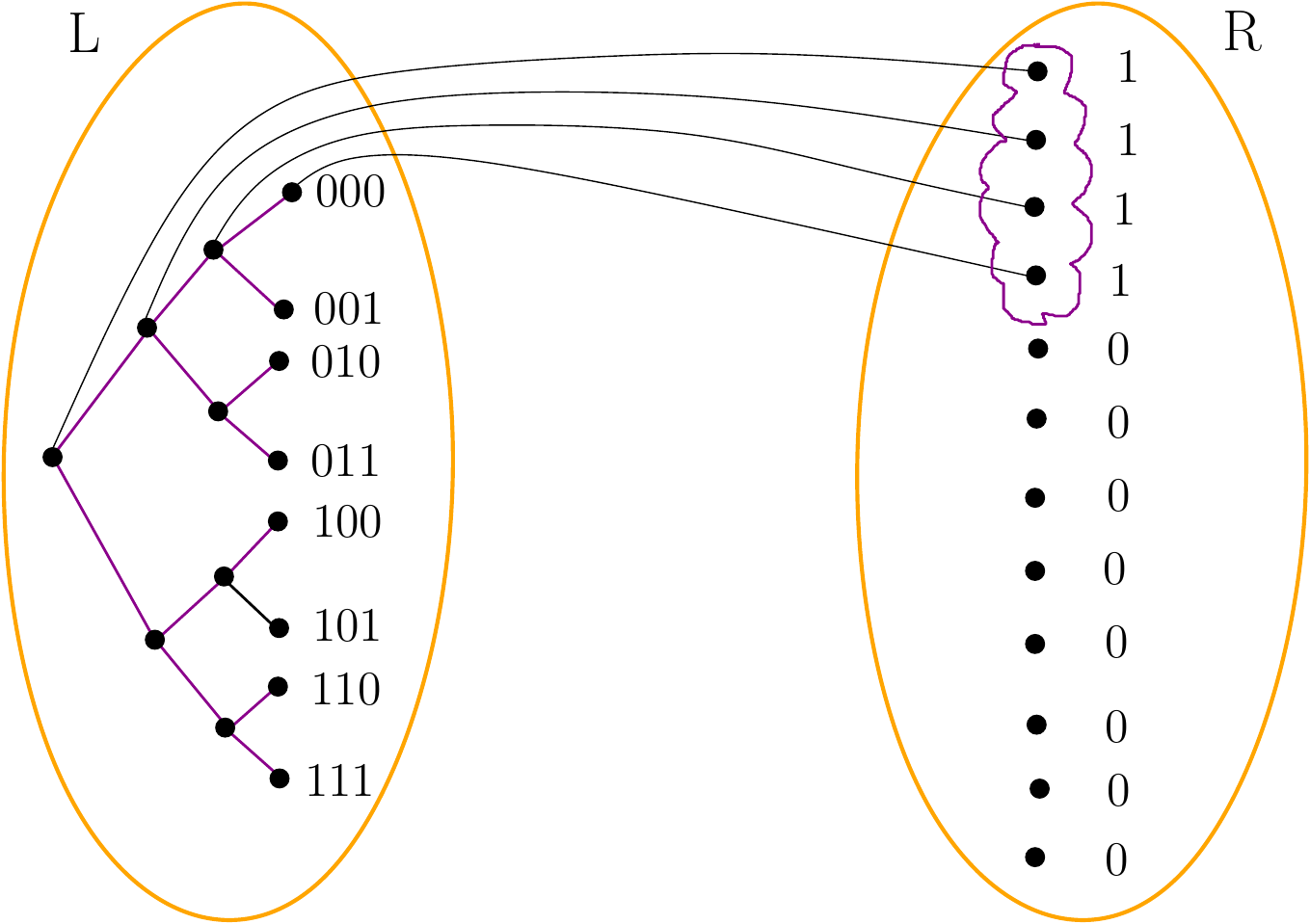}
  \caption{Here, the string $11110\cdots 0$ in $R$ encodes the string $000$ in $L$. (Note: This graph is not a disperser, but nevertheless illustrates the encoding scheme.)}\label{2_fig:disp}
\end{figure}

With the encoding scheme defined, we now construct the $\cqma$ circuit $W$. Given $y$ and $\ket{z}$ to its INPUT and CHOICE registers, respectively, $W$ acts as follows: (a) If $y$ corresponds to a subset $R_y\subseteq R$ such that $\abs{R_y}>\abs{R}/2$, then $W$ sets its output qubit to one. (b) If $\abs{R_y}\leq\abs{R}/2$, then $W$ first \emph{decodes} $R_y$ to obtain the set of leaves $L_y\subseteq L$. Roughly, it then outputs one if there exists $x\in L_y$ causing $\Pi$'s verification circuit $V$ to output one when fed the proofs $x$ and $\ket{z}$. These last two steps require further clarification, which we now provide.

First, given $R_y\subseteq R$, decoding it to obtain the set of leaves $L_y\subseteq L$ might \emph{a priori} require exponential time, as recall $\abs{L}=2^{c(n)+1}$. This, however, is precisely where dispersers play their part: Since we set $\epsilon=1/2$ in constructing our disperser, we know that for any $S\subseteq R$ with $\abs{S}\leq \abs{R}/2$, there are at most $2^k=c(n)^\gamma$ vertices in $L$ whose neighbor sets are completely contained in $S$. Thus, by starting at the root of $L$ and performing a breadth-first-search down the tree (where we prune any branches along which we encounter a vertex whose neighbor set is not contained in $R_y$, as by definition such vertices cannot encode any leaf $x$), we can efficiently decode $R_y$ to obtain $L_y$ while visiting only polynomially vertices in $L$. It remains to specify how $W$ checks whether there exists an $x\in L_y$ causing $V$ to accept, and here we must deviate from Umans' construction.

First, if $\abs{L_y}=1$, our task is straightforward -- simply run $V$ as a black box on proofs $x\in L_y$ and $\ket{z}$, and output the result. Then, $W$ outputs one with probability at least $2/3$ on input $y$ for all quantum proofs $\ket{z}$ if and only if $V$ also does so on proofs $x$ and $\ket{z}$. If , however, $\abs{L_y}>1$, a more involved construction of $W$ is necessary. Here, $W$ takes three inputs: a classical description of $V$, an $\abs{R}$-bit string $y$ to denote subsets in $R$, and a $2^kq(n)$-qubit proof $\ket{z}$. Then, for the $i$th candidate string $x_i\in L_y$, $W$ feeds $x_i$ and the $i$th block of $q(n)$ proof qubits of $\ket{z}$ into $V$. (If $\abs{L_y}<2^k$, we simply re-use values of $x\in L_y$ in the leftover parallel runs of $V$.) $W$ then coherently computes the OR of the output qubits of all parallel runs of $V$ and outputs this qubit as its answer.

Let us briefly justify why this works. For simplicity, assume the quantum proof to W can be written $\ket{z}=\ket{z_1}\otimes\cdots\otimes\ket{z_{2^k}}$; entangled proofs can be shown not to pose a problem via the same proof technique used in standard error reduction~\cite{AN02}. Now, if there exists an $x_i\in L_y$ causing $V$ to accept for all quantum proofs, then in the $i$th parallel run of $V$ in $W$ corresponding to $x_i$, $V$ outputs $1$ with probability at least $2/3$ on any $\ket{z_i}$, implying $W$ outputs $1$ with probability at least $2/3$. Conversely, if for all $x_i\in L_y$, there exists a quantum proof $\ket{z_i}$ rejected by $V$, then by standard error reduction for $V$ and the union bound, the state $\ket{z}=\ket{z_1}\otimes\cdots\otimes\ket{z_{2^k}}$ causes $W$ to output $1$ with probability at most $1/3$, as required.

Following Reference~\cite{U99} again, we now argue that $W$ accepts a \emph{non-empty monotone} set, and we analyze the hardness gap introduced by this reduction. The first of these is simple -- namely, $W$ accepts a set $R'\subseteq R$ if either $\abs{R}>\abs{R/2}$, in which case it also accepts any $R''\supseteq R'$, or if $R'$ encodes some $x\in L$ accepted by $V$, in which case any $R''\supseteq R'$ would also encode $x$ and hence be accepted. As for the gap, if $x\in L$ is an accepting assignment for $V$ when $\Pi\in\ayes$, then to encode $x$ using a subset of $R$ requires at most $c(n)2^d$ vertices in $R$, where recall $2^d$ is the left-degree of our disperser. On the other hand, if $\Pi\in\ano$, then the only way for $W$ to accept is to choose $R'\subseteq R$ with $\abs{R'}>\abs{R}/2\approx c(n)^\gamma2^d$. This yields a hardness ratio of $\Omega(c(n)^{\gamma-1})$. Since $W$'s encoding size $N$ is polynomial in $c(n)$, there exists some $\epsilon>0$ such that the ratio produced is of order $N^\epsilon$, as desired.
\end{proof}

We next show a gap-preserving reduction from QMW to QSSC. Its proof requires Lemmas~\ref{2_lem:y2bound} and~\ref{2_lem:ubound}, which are stated and proven subsequently.

\begin{theorem}\label{2_thm:qmmwwToQssc}
    QSSC is in $\cqs$. Further, there exists a polynomial time reduction which, given an instance of QMW with thresholds $f$ and $f'$, outputs an instance of QSSC with thresholds $g=f+2$ and $g'=f'+2$, respectively.
\end{theorem}
\begin{proof}
That QSSC is in $\cqs$ follows using Kitaev's verifier~\cite{KSV02} for putting $k$-local Hamiltonian in QMA. Specifically, we construct a $\cqs$ verification circuit for QSSC which takes a description $c$ of some subset of local Hamiltonians $S:=\set{H_i}$ in its classical register, and estimates the energy achieved by $\ket{q}$ in its quantum register against $H_S$ using Kitaev's approach, outputting zero or one according to whether the measured energy is above or below the desired thresholds.

To reduce QMW to QSSC, suppose we are given a $\cqma$ circuit $V$ accepting exactly a non-empty monotone set $T\subseteq\set{0,1}^n$ and threshold parameters $f$ and $f'$. We assume without loss of generality that $V$ is represented as a sequence of one and two qubit unitary gates $V_i$ such that $V=V_L\cdots V_1$. We also assume using standard error reduction that if $V$ accepts (rejects) input $x\in\set{0,1}^n$, then it outputs one (zero) with probability at least $1-\epsilon:=1-2^{-4(n+m)}$.

We now state our instance $(S,\alpha,\beta, g, g')$ of QSSC as follows. We first apply Kitaev's circuit-to-Hamiltonian construction from Section~\ref{2_scn:def} to $V$ to obtain a $3$-tuple $(H,a,b)$. Note that $H=\sum_{i=1}^r H_i$ with $r$ terms $0\preceq H_i \preceq I$. Then, set $\alpha := 1-(\zeta+1)\epsilon$, and $\zeta:=2(1+2^{2(n+m)})/(L+1)$. Define $\beta:=1-b$. Note that for large $n+m$, this yields $\alpha\geq 1-2^{-(n+m)}$ and $\beta\leq1-c(1-2^{-(n+m)})/{L^3}$ for some constant $c$. Further, define $g:=f+2$, $g':=f'+2$, and let $S$ consist of the elements (intuition to follow)
\begin{eqnarray}
     G_1 &:=& (L+1)\ketbra{0}{0}_{A_1}\otimes I_{B,C}\otimes \ketbra{0}{0}_D \\
     &\vdots&\nonumber\\
     G_n &:=& (L+1)\ketbra{0}{0}_{A_n}\otimes I_{B,C}\otimes \ketbra{0}{0}_D \\
     G_{n+1} &:=&  (\Delta+1)(\hin+\hprop+\hstab)\\
     G_{n+2} &:=& I - (\hin+\hprop+\hstab+\hout),\label{2_eqn:coverdef}
\end{eqnarray}
for $\Delta\geq 0$ to be chosen as required, and where $A_i$ denotes the $i$th qubit of register $A$. Intuitively, the terms in $S$ play the following roles: $G_{n+1}$ penalizes assignments which are not valid history states. $G_{n+2}$ penalizes valid history states accepted by $V$. Finally, the $G_i$ for $i\in[n]$ penalize valid history states rejected by $V$ (recall that $V$ accepts a monotone set, and so flipping a one to a zero in register $A$ may lead $V$ to reject). Thus, we cover the entire space. We now make this rigorous.

As required by Definition~\ref{2_def:qssc}, we begin by showing that $S$ itself is a cover, i.e.\ that $G_S\succeq \alpha I_{A,B,C,D}$. First, note that
\begin{equation}
    G_S = I + \sum_{i=1}^n G_i -\hout + \Delta(\hin+\hprop+\hstab).\label{2_eqn:1}
\end{equation}
It thus suffices to prove that for large enough $\Delta$,
\begin{equation}
    \Delta (\hin+\hprop+\hstab)+\left(\sum_{i=1}^n G_i\right)-\hout\succeq -(\zeta+1)\epsilon I.\label{2_eqn:keyeqn}
\end{equation}
To show this, we use Lemma~\ref{2_lem:projlemma}, the Projection Lemma, with
\begin{eqnarray}
    Y_1 := \left(\sum_{i=1}^n G_i\right)-\hout,\hspace{8mm}
    Y_2 := \Delta (\hin+\hprop+\hstab).
\end{eqnarray}
Intuitively, the Projection Lemma tells us that by increasing our weight $\Delta$, we can force the smallest eigenvalue of $Y_1+Y_2$ to be approximately the smallest eigenvalue of $Y_1$ restricted to the null space of $Y_2$. In our setting, this implies it suffices to study the smallest eigenvalue of $Y_1$ restricted to the space of all \emph{valid} history states, i.e.\ states of the form of Equation~(\ref{2_eqn:hist}). Let $\shist$ denote the space of valid history states; note $\shist$ is the null space of $\hin+\hprop+\hstab$. Then, in the notation of Lemma~\ref{2_lem:projlemma}, to lower bound $\lambda(Y_1|_{\shist})$, we invoke Lemma~\ref{2_lem:ubound} to instead upper bound the largest eigenvalue of $(-Y_1)|_{\shist}$. This yields $\lambda(Y_1|_{\shist})\geq -\zeta\epsilon$. Noting that $\snorm{Y_1}\leq n(L+1)+1$, and since by Lemma~\ref{2_lem:y2bound} the smallest non-zero eigenvalue of $Y_2$ scales as $\Omega(\Delta/L^3)$, it follows by Lemma~\ref{2_lem:projlemma} that by setting $\Delta\in\Omega(n^2L^5/\epsilon)$, we have $Y_1+Y_2\succeq -(\zeta+1)\epsilon I$, as desired. This completes the proof that $S$ is a cover.

We now show the desired reduction. Assume first that $V$ accepts a string $x$ of Hamming weight $k$, and let $T\subseteq [n]$ be such that $i\in T$ if and only if $x_i=1$. We claim there exists a cover $S'\subseteq S$ of size $\abs{S'}=k+2$ which consists of $G_{n+1}$, $G_{n+2}$, and the $k$ terms $G_i$ such that $i\in T$. To show this, following the proof above, the analogue of Equation~(\ref{2_eqn:keyeqn}) which we must prove is
\begin{equation}
    \Delta (\hin+\hprop+\hstab)+\left(\sum_{i\in T} G_i\right)-\hout\succeq -(\zeta+1)\epsilon I.\label{2_eqn:keyeqn2}
\end{equation}
First, applying Lemma~\ref{2_lem:ubound} again, we lower bound the smallest eigenvalue of
\begin{equation}
    Y'_1:=\left(\sum_{i\in T} G_i\right)-\hout
\end{equation}
restricted to $\shist$ by $-\zeta\epsilon$. Since $\snorm{Y'_1}\leq \snorm{Y_1}$ for $Y_1$ from the previous case of $T=[n]$, the value of $\Delta$ from before still suffices to apply Lemma~\ref{2_lem:projlemma} and conclude that Equation~(\ref{2_eqn:keyeqn2}) holds, as desired.

Conversely, suppose $V$ rejects any string $x$ of Hamming weight at most $k$. For any $S'\subseteq S$ with $\abs{S'}\leq k+2$, we claim that $G_{S'}$ has an eigenvalue at most $\beta$. To see this, note first that if  $G_{n+2}\not\in S'$, then the state $\histstatecq{1^n}{y}$ attains expected value zero against $G_{S'}$, where note $\beta\geq 0$. Similarly, if $G_{n+1}\not \in S'$, then the state $\ket{1^n}_{A,B,C}\otimes\ket{0}_D$ obtains expected value at most zero against $G$. We conclude that in order to refute the claim that $G$ has an eigenvalue at most $\beta$, we must have $G_{n+1},G_{n+2}\in S'$. This implies that $S'$ contains at most $k$ terms $G_i$ for $i\in [n]$. Then, consider the string $x$ which has ones precisely at these at most $k$ positions $i\in [n]$ corresponding to $G_i\in S'$. It follows that the state $\histstatecq{x}{y}$ lies in the null space of all terms in $S'$ with the possible exception of $G_{n+2}$. Moreover, since $V$ rejects all strings of Hamming weight at most $k$, there exists by the definition of a $\cqma$ circuit and Lemma~\ref{2_lem:kitaev} a $\ket{y}\in\B^{\otimes m}$ such that
\begin{equation}
\trace\left(G_{n+2}\histstatecqketbra{x}{y}\right)=1-\trace\left(H\histstatecqketbra{x}{y}\right)\leq 1-b=\beta,
\end{equation}
completing the proof.
\end{proof}

%SEV: Up to here, 3-LH version is OK. However, below we lift the ground space of hin + hprop. For 5-LH, this is easy since we know what that ground space looks like - it consists of all valid history states. For 3-LH, however, hprop is no longer PSD, and in fact does NOT map a valid history state to the zero vector (recall it results in a state with invalid counter states).
The following two lemmas are required for the proof of Theorem~\ref{2_thm:qmmwwToQssc}. Their statements and proofs assume the notation of Theorem~\ref{2_thm:qmmwwToQssc}.

\begin{lemma}\label{2_lem:y2bound}
    The smallest non-zero eigenvalue of $Y_2 = \Delta (\hin+\hprop+\hstab)$ scales as $\Omega(\Delta/L^3)$.
\end{lemma}
\begin{proof}
    We bound the smallest non-zero eigenvalue of $\hin+\hprop$; it is straightforward to show using the approach of Reference~\cite{KSV02} that the addition of $\hstab$ does not affect this lower bound (see Section~\ref{0_sscn:5LH}). Our proof idea here is to ``lift'' the null space of $\hin+\hprop$ so that the smallest non-zero eigenvalue of $\hin+\hprop$ becomes the smallest eigenvalue of the lifted operator, and then apply the Geometric Lemma (Lemma~\ref{l:pluslemma}) to lower bound the latter.

    To begin, recall that the null space of $\hin+\hprop$ consists of all valid history states
    \begin{equation}
        \histstate=\frac{1}{\sqrt{L+1}}\sum_{i=0}^L V_i\cdots V_1 \ket{\psi}_{A,B}\otimes\ket{0}_C\otimes\ket{i}_D,
    \end{equation}
    for any $\ket{\psi}_{A,B}$. (Since we omit $\hstab$ for now, we assume here that the clock register is represented in binary, i.e.\ there are no invalid clock states.) As done in Reference~\cite{KSV02} and Section~\ref{0_sscn:5LH}, our analysis is simplified by first applying the unitary change of basis $W=\sum_{j=0}^L V_1^\dagger\cdots V_j^\dagger\otimes\ketbra{j}{j}$, yielding
    \begin{eqnarray}
        W\histstate &=& \ket{\psi}_{A,B}\otimes\ket{0}_C\otimes\ket{\gamma}_D\\
        W\hin W^\dagger &=& \hin= I_{A,B}\otimes \left(\sum_{i=1}^p \ketbra{1}{1}_{C_i}\right)\otimes \ketbra{0}{0}_D\\
        W\hprop W^\dagger &=& I_{A,B}\otimes I_C\otimes E_D
    \end{eqnarray}
    where $\ket{\gamma}:=\left(\frac{1}{\sqrt{L+1}}\sum_{i=0}^L\ket{i}\right)$, and for some operator $E_D$ whose eigenvalues are given by $\lambda_k=1-\cos(\pi k / (L+1))$ for $0\leq k\leq L$ and whose unique zero-eigenvector is $\ket{\gamma}$.

As alluded to above, we now lift the null space of $W(\hin+\hprop)W^\dagger$. Letting $\hist$ denote the projector onto the space of valid history states $\histstate$, this is accomplished by defining
    \begin{eqnarray}
        A_1 &:=&  W(\hin + p\hist)W^\dagger\\
        A_2 &:=&  W(\hprop + 2\hist)W^\dagger.
    \end{eqnarray}
Note that $[\hin,\hist]=[\hprop,\hist]=0$, $\snorm{\hin}\leq p$ and $ \snorm{\hprop}\leq 2$. It thus remains to lower bound the smallest eigenvalue of $A_1+A_2$, for which we apply Lemma~\ref{l:pluslemma} (Geometric Lemma) to $A_1+A_2$ via the approach of Reference~\cite{KSV02}. For this, we require values for the parameters $v$ and $\alpha(\nl,\nll)$.

For $v$, note that since $A_1$ is a sum of commuting orthogonal projectors, its smallest non-zero eigenvalue is at least $1$ (assuming $p\geq 1$). Similarly, one infers from the spectrum of $E_D$ stated above that the smallest non-zero eigenvalue of $A_2$ scales as $\Omega(1/L^2)$. It follows that $v\in\Omega(1/L^2)$. As for $\alpha(\nl,\nll)$, note that the null spaces $\nl$ and $\nll$ can be written as
    \begin{eqnarray}
        \nl &=& \B^{\otimes(n+m)}_{A,B}\otimes\operatorname{span}(\ket{\psi}~:~\braket{\psi}{0\cdots0}=0)_C\otimes\operatorname{span}(\ket{1},\ldots,\ket{L})_D\oplus\label{2_eqn:spa1}\\
%        &&\B^{\otimes(n+m)}_{A,B}\otimes\ket{0\cdots 0}_C\otimes\left[\operatorname{span}(\ket{1},\ldots,\ket{L})\cap\operatorname{span}(\ket{\psi}~:~\braket{\psi}{\gamma}=0)\right]_D\label{2_eqn:spa2}\\
        &&\B^{\otimes(n+m)}_{A,B}\otimes\ket{0\cdots 0}_C\otimes\operatorname{span}(\ket{\psi}~:~\braket{\psi}{\gamma}=0)_D,\label{2_eqn:spa2}\\
        \nll &=&
        \B^{\otimes(n+m)}_{A,B}\otimes\operatorname{span}(\ket{\psi}~:~\braket{\psi}{0\cdots0}=0)_C\otimes\ket{\gamma}_D.
    \end{eqnarray}

\noindent Observe that $\nl \cap \nll = \set{\ve{0}}$, as required by Lemma~\ref{l:pluslemma}. Then, letting $\Pi_{\nl}$ denote the projector onto $\nl$, we analyze
\begin{equation}
\cos^2\alpha(\nl,\nll)=\max_{\text{unit }\ket{x}\in\nl,\ket{y}\in\nll}\abs{\braket{x}{y}}^2=\max_{\text{unit }\ket{y}\in\nll}\bra{y}\Pi_{\nl}\ket{y}=\max_{\text{unit }\ket{y}\in\nll}\bra{y}\Pi_1+\Pi_2\ket{y},
\end{equation}
where $\Pi_1$ and $\Pi_2$ project onto the spaces in Equations~(\ref{2_eqn:spa1}) and~(\ref{2_eqn:spa2}), respectively. As $\bra{y}\Pi_2\ket{y}=0$, we simply need to maximize $\bra{y}\Pi_1\ket{y}$, which is equivalent to maximizing $\abs{\braket{\psi}{\gamma'}}^2$ for any unit vector $\ket{\psi}$ in register $D$ and for unnormalized state $\ket{\gamma'}:=(\frac{1}{\sqrt{L+1}}\sum_{i=1}^L\ket{i})$.
%To see this, let $\ket{\phi_1}=\sum_i\alpha_i\ket{a_i}_{A,B,C}\ket{b_i}_D$ and $\ket{\phi_2}=\sum_j\beta_j\ket{a_j}_{A,B,C}\ket{\gamma}_D$ be arbitrary vectors in the spaces $\Pi_1$ and $\nll$, respectively. Then, $\braket{\phi_1}{\phi_2}=\bra{\gamma}_D(\sum_i\alpha_i\beta_i^\ast\ket{b_i}_D)$.
By the Cauchy-Schwarz inequality, this quantity is upper bounded by $L/(L+1)$. We thus obtain the bound $\cos\alpha(\nl,\nll)\leq \sqrt{L/(L+1)}$. Combining this with the identity $2\sin^2\frac{x}{2}=1-\cos x$ and the Maclaurin series expansion for $\sqrt{1+x}$ (where $\abs{x}\leq 1$) yields $2\sin^2\frac{\alpha\left(\nl,\nll\right)}{2}\geq\frac{1}{2(L+1)}$. Substituting into Lemma~\ref{l:pluslemma}, the desired result follows.
\end{proof}

\begin{lemma}\label{2_lem:ubound}
    Define $\hist:=\sum_{x\in\set{0,1}^n,y\in\set{0,1}^m}\histstatecqketbra{x}{y}$ as the projector onto $\shist$, let $\zeta:=2(1+2^{2(n+m)})/(L+1)$, and consider $T\subseteq[n]$. Then, if $V$ outputs one with probability at least $1-\epsilon$ for inputs $(x,\ket{y})$ with  $x\in\set{0,1}^n$ such that $x_i=1$ for all $i\in T$ and for all $m$-qubit $\ket{y}$, one has
\begin{equation}
    \hist\left[\hout-\sum_{i\in T} G_i \right]\hist \preceq \zeta\epsilon I.
\end{equation}
\end{lemma}
\begin{proof}
    Define $Z_1:=\hist(-\sum_{i\in T} G_i)\hist$ and $Z_2:=\hist \hout\hist$. Letting $z\in\set{0,1}^n$ denote the characteristic vector of $T$, i.e.\ the $i$th bit of $z$ is set to one if and only if $i\in T$, it follows that any state $\histstatecq{x}{y}$ is an eigenvector of $Z_1$ with eigenvalue $\braket{x}{z}-\abs{T}$. Hence, for example, $
        \trace\left(Z_1\histstatecqketbra{1^n}{y}\right)=0.
    $
    Further, since $V$ accepts a \emph{non-empty} monotone set, it must accept input $(1^n,\ket{y})$ with probability at least $1-\epsilon$, implying
    $
        \trace(Z_2\histstatecqketbra{1^n}{y})\leq \frac{\epsilon}{L+1}.
    $
    This yields an upper bound of
    \begin{equation}
        \trace((Z_1+Z_2)\histstatecqketbra{1^n}{y})\leq\frac{\epsilon}{L+1}
    \end{equation}
    in this simple case. We now show that deviating from $\histstatecq{1^n}{y}$ above cannot increase our expected value against $Z_1+Z_2$ by ``too much''.

    To do so, let $\ket{\phi}=\alpha_1\ket{\phi_1}+\alpha_2\ket{\phi_2}$ be an arbitrary valid history state where $\abs{\alpha_1}^2+\abs{\alpha_2}^2=1$, $\ket{\phi_1}$ is a (normalized) superposition of valid history states where each history state in the superposition has a string $x$ in register $A$ at time zero satisfying $x_i =1$ if $i\in T$, and where $\ket{\phi_2}$ is a valid history state in the space orthogonal to space of all possible states $\ket{\phi_1}$. We thus first have that
    \begin{equation}
        \trace\left(Z_1\ketbra{\phi}{\phi}\right)\leq0+\alpha_2^2\trace\left(Z_1\ketbra{\phi_2}{\phi_2}\right)\leq \alpha_2^2[(\abs{T}-1)-\abs{T}]\leq-\abs{\alpha_2}^2.
    \end{equation}
    %SEV: the terms with phi_1 disappear, and for phi_2, if there is 1 term in phi_2, it must differ in at least 1 position from z, and hence we get a nupper bound on the energy of (|T|-1)-|T| times alpha_2^2. To deal with more than 1 terms in phi_2, note that all cross terms disappear, since the H_i project onto time step 0 and ensure ith bit of both cross terms is 0, and after that we take the inner product of the cross terms, which of course goes to zero unless all bits in A and B registers agree.
    Moving on to $Z_2$, observe that straightforward expansion yields
    \begin{eqnarray}
        \trace(Z_2\ketbra{\phi}{\phi})&=&\abs{\alpha_1}^2\trace(Z_2\ketbra{\phi_1}{\phi_1})+\abs{\alpha_2}^2\trace(Z_2\ketbra{\phi_2}{\phi_2})\\
        &+& \alpha_1\alpha_2^*\trace(Z_2\ketbra{\phi_1}{\phi_2})+\alpha_1^*\alpha_2\trace(Z_2\ketbra{\phi_2}{\phi_1}).
    \end{eqnarray}
    To upper bound this quantity, we use the fact that $\braket{a}{b} + \braket{b}{a} \leq \braket{a}{a} + \braket{b}{b}$ for complex vectors $\ket{a}$ and $\ket{b}$. Namely, setting $\ket{a}:=\alpha_1\sqrt{Z_2}\ket{\phi_1}$ and $\ket{b}:=\alpha_2\sqrt{Z_2}\ket{\phi_2}$ yields
    \begin{eqnarray}
        \trace(Z_2\ketbra{\phi}{\phi})&\leq& 2\abs{\alpha_1}^2\trace(Z_2\ketbra{\phi_1}{\phi_1})+2\abs{\alpha_2}^2\trace(Z_2\ketbra{\phi_2}{\phi_2})\\
        &\leq& 2\abs{\alpha_1}^2\trace(Z_2\ketbra{\phi_1}{\phi_1})+ 2\abs{\alpha_2}^2\frac{1}{L+1}\label{2_eqn:last},
    \end{eqnarray}
    where the second inequality follows since $\snorm{Z_2}\leq 1/(L+1)$. Finally, in order to upper bound the term $\trace(Z_2\ketbra{\phi_1}{\phi_1})$ in Equation~(\ref{2_eqn:last}), observe that since by assumption $\trace(Z_2\histstatecqketbra{x}{y})\leq \frac{\epsilon}{L+1}$ for all $x$ with $x_i=1$ for $i\in T$, and since $\hout$ is a projector, it follows that the norm of $\hout \histstatecq{x}{y}$ is at most $\sqrt{\epsilon/(L+1)}$. Using the Cauchy-Schwarz inequality, this implies that each cross term in the expansion of $\trace(Z_2\ketbra{\phi_1}{\phi_1})$ can contribute a value of magnitude at most $\epsilon/(L+1)$. Since there are at most $2^{2(n+m)}$ such cross terms, and since the non-cross terms are weighted by a convex combination, we hence have the upper bound of $\trace(Z_2\ketbra{\phi_1}{\phi_1})\leq (1+2^{2(n+m)})\epsilon/(L+1)$.
    Combining these bounds, we have
    \begin{eqnarray}
        \trace((Z_1+Z_2)\ketbra{\phi}{\phi})&\leq& -\abs{\alpha_2}^2+\frac{2\abs{\alpha_1}^2(1+2^{2(n+m)})\epsilon}{L+1}+\frac{2\abs{\alpha_2}^2}{L+1}\\
        &=&\frac{2\abs{\alpha_1}^2(1+2^{2(n+m)})\epsilon+\abs{\alpha_2}^2(1-L)}{L+1}\\
        &\leq&\frac{2(1+2^{2(n+m)})}{L+1}\epsilon\\
        &=&\zeta\epsilon
    \end{eqnarray}
    where the second inequality holds when $L\geq 1$.
\end{proof}

Finally, we show that QIRR is $\cqs$-hard to approximate.

\begin{theorem}\label{2_thm:QSSCtoQIRR}
    There exists a polynomial time reduction which, given an instance of an arbitrary $\cqs$ problem $\Pi$, outputs an instance of QIRR with threshold parameters $h$ and $h'$ satisfying $h'/h\in \Theta(N^\epsilon)$ for some $\epsilon>0$, where $N$ is the encoding size of the QIRR instance.
\end{theorem}
\begin{proof}
We begin by applying Theorems~\ref{2_thm:qmmwwHard} and~\ref{2_thm:qmmwwToQssc} to reduce the instance of $\Pi$ to an instance $(S=\set{G_i}_{i=1}^{n+2},\alpha,\beta,g,g')$ of QSSC, and henceforth assume the terminology and definitions introduced in Theorem~\ref{2_thm:qmmwwToQssc}. Recall that any cover in this QSSC instance must include the terms $G_{n+1}$ and $G_{n+2}$. For ease of exposition, we first reduce this instance to QIRR with parameters $h=g+2r-3$ and $h'=g'+2r-3$, where recall $r$ is the number of terms in $H=\sum_{i=1}^r H_i$. This, however, does not suffice to obtain a hardness of approximation gap, as tracing through Theorems~\ref{2_thm:qmmwwHard} and~\ref{2_thm:qmmwwToQssc} yields $r\in\omega(g),\omega(g')$, implying $h'/h\rightarrow 1$ as the instance $\Pi$ in Theorem~\ref{2_thm:qmmwwHard} grows in size. We then slightly modify our reduction to improve the threshold parameters to $h=gr-1$ and $h'=g'r-1$, which yield the desired hardness of approximation gap.

We now state our instance $(T,\gamma,\delta,h,h')$ of QIRR, and follow with an intuitive explanation. For simplicity of exposition, we assume $r$ is a power of two, but our construction can be easily modified to handle the complementary case.
%If $r$ is not a power of 2, simply add up to $r$ terms G_i to the construction below of
%the form (\Delta+1) |1><1|\otimes I\otimes |i><i|. Then, the energy thresholds $\gamma$ and $\delta$ don't change, equation 16 is identical so first direction of proof goes through unchanged, and chaperone qubits mean these extra terms must always be chosen in T', so second direction of proof holds. Finally, the thresholds h and h' change by at most r, but this is inconsequential in terms of the gap.
We also label $H_r = \hout$. We now introduce three registers: a ``tag'' qubit register (denoted $A$), the space the original cover $\mathcal{S}$ acts on (denoted $B$), and $\log r$ ``chaperone'' qubits (denoted $C$). The Hamiltonian terms we define for QIRR, $T:=\set{F_i}_{i=1}^{n+2r-1}$, act on $A\otimes B\otimes C = \B\otimes\B^{\otimes (n+m+p+q)}\otimes\B^{\otimes \log r}$, and are defined as:
\begin{eqnarray}
    F_1 &:=& \ketbra{0}{0}_A\otimes (G_1)_B\otimes I_C\\
    &\vdots&\nonumber\\
    F_n &:=& \ketbra{0}{0}_A\otimes (G_n)_B\otimes I_C\\
    F_{n+1} &:=& (\Delta+1)\left[\ketbra{0}{0}_A\otimes (H_1)_B\otimes I_C+\ketbra{1}{1}_A\otimes I_B\otimes\ketbra{0}{0}_C\right]\\
    &\vdots&\nonumber\\
    F_{n+r-1} &:=&  (\Delta+1)\left[\ketbra{0}{0}_A\otimes (H_{r-1})_B\otimes I_C+\ketbra{1}{1}_A\otimes I_B \otimes\ketbra{r-2}{r-2}_C\right]\nonumber\\
    F_{n+r} &:=& \ketbra{0}{0}_A\otimes (I-H_1)_B\otimes I_C + \ketbra{1}{1}_A\otimes I_B \otimes \ketbra{r-1}{r-1}_C\\
    &\vdots&\nonumber\\
    F_{n+2r-1} &:=& \ketbra{0}{0}_A\otimes (I-H_r)_B\otimes I_C + \ketbra{1}{1}_A\otimes I_B \otimes\ketbra{r-1}{r-1}_C.
\end{eqnarray}
We set $\gamma:=\alpha+r-1$, $\delta:=\beta+r-1$, $h:=g+2r-3$, and $h':=g'+2r-3$. Note that each $F_j$ is a projection up to scalar multiplication, as required. We now provide the intuition behind the construction. QIRR is stated in terms of projectors $F_j$ (up to scalar multiplication), whereas QSSC is stated in terms of Hermitian operators $G_i$. Hence, in order to move from the latter to the former, a natural idea is to treat each local Hamiltonian term in the sums comprising $G_{n+1}$ and $G_{n+2}$ as distinct terms $F_{n+1},\ldots,F_{n+r-1}$ and $F_{n+r},\ldots,F_{n+2r-1}$, respectively. The problem with this approach is that in order to rigorously argue that the gap between thresholds $g$ and $g'$ for QSSC is preserved when defining thresholds $h$ and $h'$ for QIRR, we would like, for example, that \emph{all} terms $F_j$ making up $G_{n+1}$ are chosen together in any candidate cover $T'\subseteq T$. To address this issue, we introduce the chaperone qubits, which ensure that any candidate $T'$ plays by these rules. In particular, we can make sure that all terms $F_{n+1},\ldots,F_{n+2r-1}$ are chosen in any $T'$, allowing us to rigorously apply our knowledge of the spectra of $G_{n+1}$ and $G_{n+2}$ to the analysis of $F_T$ versus $F_{T'}$.

We now show that if there exists a cover $S^\prime\subseteq S$ for QSSC of size $v$, then there exists a $T'\subseteq T$ such that $\abs{T'}=v+2r-3$ satisfying the conditions for a YES instance of QIRR. Namely, let
\begin{equation}
    T' = \set{F_i}_{i\in [n] \text{ and } G_i\in S'}\cup\set{F_{n+1},\ldots,F_{n+2r-1}}.\label{2_eqn:succint}
\end{equation}
Note that it suffices to show that $F_{T'}\succeq \gamma I$ (since if $F_{T'}\succeq \gamma I$, then $F_{T}\succeq \gamma I$ as well). To show this, observe first that we can write $F_{T'}=K_1+K_2$, for $K_1$ and $K_2$ defined as:
\begin{eqnarray}
    K_1 &:=& \ketbra{0}{0}_A\otimes\left(\sum_{i\in [n]\text{ and }G_i\in S'} G_i + (\Delta+1)\sum_{i=1}^{r-1}H_i + \sum_{i=1}^{r}(I-H_i)\right)_B\otimes I_C  \\&=&\ketbra{0}{0}_A\otimes\left(G_{S'}+(r-1)I\right)_B\otimes I_C\label{2_eqn:K1}\\
    K_2 &:=& \ketbra{1}{1}_A\otimes I_B\otimes \left((\Delta+1)\left(\sum_{i=0}^{r-2}\ketbra{i}{i}\right)+r\ketbra{r-1}{r-1}\right)_C\\ &=&\ketbra{1}{1}_A\otimes I_B\otimes\left(rI + (\Delta+1-r)\sum_{i=0}^{r-2}\ketbra{i}{i}\right)_C\label{2_eqn:K2},
\end{eqnarray}
where we can assume without loss of generality that $\Delta\geq r-1$. Let $\ket{\phi}=a_0\ket{0}_A\ket{\phi_0}_{BC}+a_1\ket{1}_A\ket{\phi_1}_{BC}$ be an arbitrary state acting on this space with $\abs{a_0}^2+\abs{a_1}^2=1$ and for some unit vectors $\ket{\phi_0}_{BC}$ and $\ket{\phi_1}_{BC}$. Then
\begin{eqnarray}
    \trace(F_{T'}\ketbra{\phi}{\phi})&=&\trace(K_1\ketbra{\phi}{\phi})+\trace(K_2\ketbra{\phi}{\phi})\\
    &=&\abs{a_0}^2\trace(K_1\ketbra{0}{0}\otimes\ketbra{\phi_0}{\phi_0})+\abs{a_1}^2\trace(K_2\ketbra{1}{1}\otimes\ketbra{\phi_1}{\phi_1})\\
    &\geq&\abs{a_0}^2(\alpha + r-1)+\abs{a_1}^2r\\
    &\geq&\gamma,
\end{eqnarray}
where the first inequality follows since $\trace(X_{AB}I_A\otimes Y_B)=\trace(\trace_A(X_{AB})Y_B)$ and since $G_{S'}$ is a cover by assumption, and the second inequality since $0\leq\alpha\leq 1$. We conclude that $H_{T'}\succeq \gamma I$, as desired.

We now prove the other direction, namely that if there does not exist a cover $S^\prime\subseteq S$ for QSSC of size $v$, then all subsets $T'\subseteq T$ of size $\abs{T'}=v+2r-3$ satisfy the conditions for a NO instance of QIRR. To see this, note first that any candidate $T'$ must include the terms $F_i$ for $n+1\leq i\leq n+r-1$. This is because if, for example, $F_{n+1}\not\in T'$, then vector $\ket{\phi}:=\ket{1}_A\ket{\psi}_B\ket{0}_C$ obtains expected value $\Delta+1\geq \gamma$ against $F_T$, but $\ket{\phi}$ is orthogonal to $F_{T'}$. A similar argument holds for the terms $F_i$ with indices $n+r\leq i\leq n+2r-1$, since state $\ket{\phi}:=\ket{1}_A\ket{\psi}_B\ket{r-1}_C$ obtains expected value $r\geq \gamma$ against $F_T$, but obtains value at most $r-1\leq \delta$ against $F_{T'}$ if there exists an $i\in [n+r,n+2r-1]$ such that $i\not\in T'$. Thus, for any candidate $T'$ of size $v+2r-3$, this leaves $v-2$ terms to be chosen from $\set{F_1,\ldots,F_n}$. If we now restrict ourselves to states of the form $\ket{0}_A\ket{\psi}_{BC}$, we find that we are reduced to the same argument in the NO direction of Theorem~\ref{2_thm:qmmwwToQssc} -- namely, as $S$ is a cover and any $S'\subseteq S$ of size $v$ is not a cover, there must exist a state $\ket{\phi}:=\ket{0}_A\ket{\psi}_{BC}$ such that
\begin{equation}\label{2_eqn:qirrsound1}
    \trace(\ketbra{\phi}{\phi}F_T)=\trace\left[\trace_C(\ketbra{\psi}{\psi})(G_S + (r-1)I)\right]\geq \alpha + (r-1)\geq \gamma,
\end{equation}
whereas
\begin{equation}\label{2_eqn:qirrsound2}
    \trace(\ketbra{\phi}{\phi}F_{T'})=\trace\left[\trace_C(\ketbra{\psi}{\psi})(G_{S'} + (r-1)I)\right]\leq \beta + (r-1)=\delta.
\end{equation}
This concludes the reduction from QSSC to QIRR with parameters $h=g+2r-3$ and $h'=g'+2r-3$.

To obtain improved parameters $h=gr-1$ and $h'=g'r-1$, we modify the construction above as follows (intuition to follow): The terms $F_i$ for $n+1\leq n+2r-1$ from the old construction remain unchanged. For $i\in[n]$, we replace each $F_i := \ketbra{0}{0}_A\otimes (G_i)_B\otimes I_C$ with the $r$ distinct terms:
\begin{eqnarray}
&F_{i,1}& := \ketbra{0}{0}_A\otimes (G_i)_B\otimes \ketbra{0}{0}_C,\\
    &F_{i,2}& := \ketbra{0}{0}_A\otimes (G_i)_B\otimes \ketbra{1}{1}_C,\\
    &\vdots&\nonumber\\
    &F_{i,r}& := \ketbra{0}{0}_A\otimes (G_i)_B\otimes \ketbra{r-1}{r-1}_C.
\end{eqnarray}
Thus, the total number of terms in our QIRR instance increases from $n+2r-1$ to $r(n+2)-1$. Intuitively, we have used the chaperone qubits to split each $F_i$ into $r$ terms $F_{i,j}$, such that if in the old construction we chose $F_i\in T'$, then in the new construction we must place all $r$ terms $F_{i,j}$ in $T'$ in order for the new $F_{T'}$ to maintain its desired spectrum. Thus, whereas the old construction chose $g-2$ terms $F_i$ to place in $T'$, the new construction chooses $r(g-2)$ terms $F_{i,j}$ to place in $T'$, yielding the desired thresholds $h=gr-1$ and $h'=g'r-1$.

The completeness and soundness proofs now follow similarly to the previous case. Namely, given a cover $S'\subseteq S$ for QSSC of size $v$, the set $T'\subseteq T$ with $\abs{T'}=vr-1$ we choose is
\begin{equation}
    T' = \set{F_{i,j}}_{i\in [n] \text{ and } G_i\in S',j \in [r]}\cup\set{F_{n+1},\ldots,F_{n+2r-1}}.
\end{equation}
Since $F_{T'}$ in this new reduction is precisely $F_{T'}$ in the old reduction, the remainder of this direction proceeds identically. Conversely, if there does not exist a cover $S'\subseteq S$ for QSSC of size $v$, we similarly first argue that $F_i$ for $n+1\leq i\leq n + 2r-1$ must be chosen in any candidate $T'\subseteq T$ of size $\abs{T'}=vr-1$, leaving $r(v-2)$ terms to be chosen from $\set{F_{1,1},\ldots, F_{n,r}}$. This implies that for any such $T'$, there must exist a $j\in [r]$ such that the number of terms $F_{i,j}$ in $T'$ is at most $v-2$. Since no cover of size $v$ exists for our QSSC instance, we conclude there exists an appropriate choice of $\ket{\phi}:=\ket{0}_A\ket{\psi}_B\ket{j}_C$ such that Equations~(\ref{2_eqn:qirrsound1}) and~(\ref{2_eqn:qirrsound2}) still hold.
\end{proof}

\section{Improvements to hardness gaps}\label{2_scn:improvements}

We now improve the hardness gaps of Theorems~\ref{2_thm:qmmwwHard},~\ref{2_thm:qmmwwToQssc}, and~\ref{2_thm:QSSCtoQIRR} to obtain the results claimed in Theorems~\ref{2_thm:QSSChard} and~\ref{2_thm:QIRRhard}. The key idea is to use the fact that the gap for QMW from Theorem~\ref{2_thm:qmmwwHard} can be amplified by composing the cQMA circuit $W$ with itself. The results here adapt Section 5 of~\cite{U99} in a simple manner to the quantum setting.

Specifically, assume for the moment that the output qubit of $W$ is actually a classical bit, i.e.\ that the output qubit is given \emph{after} being measured in the computational basis. Then, one can recursively define $W^1:=W$ and $W^t$ as $W^{t-1}$ with $n$ independent copies of $W$ at each of its $n$ INPUT bits. (Note that entanglement between quantum proofs for different copies of $W$ does not affect the soundness of $W^t$, as each $W$ outputs a classical bit, and no quantum proofs are reused.) Now, such a recursive composition of $W$ can easily be made well-defined even if $W$'s output qubit is a superposition of $\ket{0}$ and $\ket{1}$ using the principle of deferred measurement~\cite{NC00} -- namely, without loss of generality, we can assume $W$ first copies its $n$ classical INPUT bits to an ancilla, and henceforth acts only on its CHOICE and ancilla registers. Thus, the output qubit of each copy of $W$ in $W^t$ is effectively used only as a classical control in the remainder of the circuit, and so the measurement of all output qubits can be deferred to the end of $W^t$. Finally, since we can assume using standard error reduction that the completeness and soundness error of $W$ scale as $2^{-n}$, it follows by the union bound that with probability exponentially close to $1$, all the $W$ circuits comprising $W^t$ output the correct answer. In other words, with high probability, one can think of $W^t$ as a composition of zero-error circuits $W$ (where zero-error means zero completeness and soundness error). With this viewpoint, the proof of Lemma 3 of Reference~\cite{U99} directly yields the following result in the quantum setting.

\begin{lemma}\label{2_lem:umansamp}
        If W is a cQMA circuit accepting exactly a monotone set, it follows that:
    \begin{enumerate}
        \item $\abs{W^t}\leq n^t \abs{W}$, where $\abs{W}$ denotes the size of $W$,
        \item $W$ accepts an input of Hamming weight $k$ if and only if $W^t$ accepts an input of weight $k^t$,
        \item $W^t$ accepts exactly a monotone set.
    \end{enumerate}
\end{lemma}

To improve the hardness gap of Theorem~\ref{2_thm:qmmwwHard}, we now simply replace the cQMA circuit $W$ constructed in the proof of Theorem~\ref{2_thm:qmmwwHard} with $W^t$ for an appropriate choice of $t$. The details and resulting analysis follow identically to the proof of Theorem 4 of Reference~\cite{U99}, which combined with the improved disperser construction of Reference~\cite{TUZ07} (see Theorem 7.2 therein) yields:

\begin{theorem}\label{2_thm:qmwampgap}
    QMW is $\cqs$-hard to approximate with gap $N^{1-\epsilon}$ for any $\epsilon>0$, for $N$ the encoding size of the QMW instance.
\end{theorem}

 Using this as the starting point in our reduction chain to QSSC and QIRR, a closer analysis of the proofs of Theorems~\ref{2_thm:QSSChard} and~\ref{2_thm:QIRRhard} now yields:
\begin{corollary}\label{2_cor:qsscampgap}
    QSSC and QIRR are $\cqs$-hard to approximate with gaps $N^{1-\epsilon}$ and $N^{\frac{1}{2}-\epsilon}$ for any $\epsilon>0$, respectively, and where $N$ is the encoding size of the respective QSSC and QIRR instances.
\end{corollary}

%Specifically, defining a \emph{co-ND} circuit as a special case of a cQMA circuit which has zero-error and whose CHOICE register is also classical, we first have in the classical setting:
%
%%\begin{definition}[\cite{U99}]
%%    For C a co-ND circuit, let $C^1:=C$, and recursively define $C^{l}$ as $C^{l-1}$ with $N$ independent copies of $C$ at each of its $n$ INPUT bits.
%%\end{definition}
%\begin{lemma}[\cite{U99}]\label{2_lem:umansamp}
%        For C a co-ND circuit, let $C^1:=C$, and recursively define $C^{t}$ as $C^{t-1}$ with $n$ independent copies of $C$ at each of its $n$ INPUT bits. Then, if $C$ accepts exactly a monotone set, one has:
%    \begin{enumerate}
%        \item $\abs{C^t}\leq n^t \abs{C}$, where $\abs{C}$ denote the size of $C$,
%        \item $C$ accepts an input of Hamming weight $k$ if and only if $C^t$ accepts an input of Hamming weight $k^t$,
%        \item $C^t$ accepts exactly a monotone set.
%    \end{enumerate}
%\end{lemma}
%
%As the proof of Lemma~\ref{2_lem:umansamp} does not  is not difficult to see that the same properties hold for $C^l$ f

\section{Hardness of approximation for QCMA}\label{2_scn:QCMA}

%~\cite{AN02,JW06,A06,AK07,Beigi08,ABBS08,WY08,JKNN11}

We now briefly remark that the approach of Theorems~\ref{2_thm:qmmwwHard} and~\ref{2_thm:qmwampgap} can be adapted to show hardness of approximation for {QCMA}. Our result is a straightforward extension of Umans' classical result~\cite{U99} showing NP-hardness of approximation for the problem MONOTONE MINIMUM SATISFYING ASSIGNMENT.

Specifically, define the problem QUANTUM MONOTONE MINIMUM SATISFYING ASSIGNMENT (QMSA) analogously to QMW, except with the definition of a cQMA circuit $V$ modified to drop the second (quantum) proof, i.e.\ $V$ now only takes one input register comprised of $n$ classical bits. (For example, Definition~\ref{2_def:cQMA} is modified to say that $V$ \emph{accepts}  $x\in\set{0,1}^n$ in INPUT if measuring $\ket{a}$ in the computational basis yields $1$ with probability at least $2/3$.) Then, it is straightforward to re-run the proofs of Theorems~\ref{2_thm:qmmwwHard} and~\ref{2_thm:qmwampgap} without the existence of a second quantum proof register, leading to Theorem~\ref{2_thm:QMSAhard}.

\section{A canonical $\cqs$-complete problem}\label{2_scn:anotherproblem}

In this section, we first show that a quantum generalization of the canonical $\st$-complete problem $\ssat{2}$, denoted $\cqslh{k}$, is $\cqs$-complete. We then observe that a similar proof yields $\cqs$-hardness of approximation for an appropriately defined variant of $\cqslh{k}$.

\begin{definition}[$\cqslh{k}$]\label{2_def:cqslh}
    Given a $3$-local Hamiltonian $H$ acting on $N= n+m$ qubits, and $a,b\in\reals$ such that $a\leq b$ for $b-a \geq 1$, output:
\begin{itemize}
    \item YES if $\exists$ $x\in\set{0,1}^n$ such that $\forall$ $\ket{y}\in\B^{\otimes m}$, $\trace(H \ketbra{x}{x}\otimes\ketbra{y}{y})\geq b$.
    \item NO if $\forall$ $x\in\set{0,1}^n$, $\exists$ $\ket{y}\in\B^{\otimes m}$ such that $\trace(H \ketbra{x}{x}\otimes\ketbra{y}{y})\leq a$.
\end{itemize}
\end{definition}

\begin{theorem}
$\cqslh{3}$ is $\cqs$-complete.
\end{theorem}
\begin{proof}
That $\cqslh{3}\in\cqs$ follows from Kitaev's verifier for placing $k$-local Hamiltonian in $\class{QMA}$~\cite{KSV02} (see Section~\ref{0_sscn:5LH}). As for $\cqs$-hardness, for simplicity we show the result for the case of $\cqslh{5}$ defined with $5$-local Hamiltonians. The proof for the $3$-local case follows identically by instead substituting the $3$-local circuit-to-Hamiltonian construction of Reference~\cite{KR03} below (this is possible because our proof does not exploit the structure of the clock register or $\hstab$).

To see that any instance $\Pi$ of a problem in $\cqs$ reduces to an instance of $\cqslh{5}$, let $V''$ denote the $\cqs$ verification circuit for $\Pi$. Recall that $V''$ acts on a classical proof register $A$, a quantum proof register $B$, and an ancilla register $C$. We begin by modifying $V''$ to obtain a new equivalent circuit $V'$ which first copies the (classical) contents of $A$ to its ancilla register $C$, and henceforth acts on this copied proof in $C$ throughout the verification. This ensures the contents of $A$ remain unchanged during the verification. Next, we modify $V'$ to obtain $V$ by concatenating to its end a Pauli $X$ on the output qubit; this swaps the cases in which $V'$ accepts and rejects, respectively. This is necessary because if $\ket{c}\otimes\ket{q}$ is accepted by $V'$, then $\histstatecq{c}{q}$ obtains low energy against Kitaev's Hamiltonian, whereas in our YES instance here we require high energy. Finally, we apply Kitaev's circuit-to-Hamiltonian construction from Section~\ref{2_scn:def} on $V$ to obtain a $5$-local Hamiltonian $H$.

Suppose now that we have a YES instance of $\Pi$, i.e.\ there exists bit string $\ket{c}$ such that for all quantum states $\ket{q}$, the circuit $V''$ accepts proof $\ket{c}\otimes\ket{q}$ with probability at least $1-\epsilon$ (and hence $V$ rejects $\ket{c}\otimes\ket{q}$ with probability at least $1-\epsilon$). We show that for all $\ket{\psi}_{B,C,D}$, the state $\ket{c}_A\otimes\ket{\psi}_{B,C,D}$ attains expectation value at least $b$ against $H$, for $b$ from Lemma~\ref{2_lem:kitaev}. In other words, letting $\Pi_c:=(\ketbra{c}{c}_A\otimes I_{B,C,D})$, we claim
\begin{eqnarray}
    \bra{c}\otimes\bra{\psi}H\ket{c}\otimes\ket{\psi}
    =\bra{c}\otimes\bra{\psi}\Pi_cH\Pi_c\ket{c}\otimes\ket{\psi}
    \geq b.
\end{eqnarray}
To see this, observe first that
\begin{eqnarray}
    \Pi_c \hin \Pi_c &=& \ketbra{c}{c}_A\otimes I_B\otimes\left(\sum_{i=1}^{p} \ketbra{1}{1}_{C_i}\right)\otimes \ketbra{0}{0}_D=:\ketbra{c}{c}_A\otimes\hin',\\
    \Pi_c \hout \Pi_c &=& \ketbra{c}{c}_A\otimes\ketbra{0}{0}_{B_1}\otimes
     I_C\otimes \ketbra{L}{L}_D=:\ketbra{c}{c}_A\otimes\hout',\\
    \Pi_c\hstab\Pi_c&=&\ketbra{c}{c}_A\otimes I_{B,C}\otimes\sum_{i=1}^{L-1}\ketbra{01}{01}_{D_i,D_{i+1}}=:\ketbra{c}{c}_A\otimes\hstab'.
\end{eqnarray}
As for $\Pi_c\hprop\Pi_c$, recall that the verification circuit $V$ consists of two phases: The \emph{copy} phase, consisting of $n$ CNOT gates copying the contents of $A$ to $C$, and the \emph{verification} phase, consisting of the remaining $L-n$ gates of $V$. In other words, we can write
\begin{equation}
    \hprop=\sum_{j=1}^{n}H_j + \sum_{j=n+1}^L H_j,
\end{equation}
where $\sum_{j=1}^{n}H_j$ corresponds to the copy phase and $\sum_{j=n+1}^L H_j$ to the verification phase. Since during the verification phase, $V$ does not act on $A$, we have for all $j> n$ that
\begin{eqnarray}
    \Pi_c H_j \Pi_c &=& \ketbra{c}{c}_A\otimes\left[-\frac{1}{2}(V_j)_{B,C}\otimes\ketbra{ j}{{j-1}}_D -\frac{1}{2}(V_j^\dagger)_{B,C}\otimes\ketbra{{j-1}}{{j}}_D +\right.\\ &&\left.\hspace{22mm}\frac{1}{2}I_{B,C}\otimes(\ketbra{{j}}{{j}}+\ketbra{{j-1}}{{j-1}})_D\right]\\
    &=:&\ketbra{c}{c}_A\otimes H_{j}'.
\end{eqnarray}
As for the copy phase, let $\ketbra{i}{i}\otimes I$ act on $\B\otimes\B$ for $i\in\set{0,1}$. Then, observe that
\begin{equation}
    (\ketbra{i}{i}\otimes I) \operatorname{CNOT}(\ketbra{i}{i}\otimes I)=\ketbra{i}{i}\otimes X^i,
\end{equation}
where $X$ is the Pauli $X$ operator and $X^i=X$ if $i=1$ and $X^i=I$ otherwise. This implies that for any step $j\leq n$, i.e.\ where $V$ applies a CNOT gate with qubit $A_j$ as control and $C_j$ as target, and letting $c_j$ denote the $j$th bit of $c$, we have
\begin{eqnarray}
    \Pi_c H_{j} \Pi_c &=& \ketbra{c}{c}_A\otimes\left[-\frac{1}{2}X_{C_j}^{c_j}\otimes\ketbra{{j}}{{j-1}}_D -\frac{1}{2}X_{C_j}^{c_j}\otimes\ketbra{{j-1}}{{j}}_D +\right.\\ &&\left.\hspace{22mm}\frac{1}{2}I\otimes(\ketbra{{j}}{{j}}+\ketbra{{j-1}}{{j-1}})_D\right]\\
    &=:&\ketbra{c}{c}_A\otimes H_{j}'(c),
\end{eqnarray}
where the notation $H_{j}'(c)$ means $H_{j}'$ is a function of $c$. Letting $\hprop'(c):=\sum_{i=1}^{n}H'_j+\sum_{i=n+1}^{L}H'_j(c)$ and $H(c):=\hin'+\hout'+\hstab'+\hprop'(c)$, we thus have that
\begin{equation}
\bra{c}\otimes\bra{\psi} H\ket{c}\otimes\ket{\psi}=\bra{\psi}H(c)\ket{\psi}.
 \end{equation}
 It thus suffices to show that
$
    H(c)\succeq b I.
$

To see this, we return to the circuit $V$, and think of $V$ not as accepting classical input $c$, but rather as corresponding to a set of circuits $\set{V_c}$, where each $V_c$ is just $V$ with $c$ hard-wired into register $A$. In particular, at time step $0\leq j\leq n$, $V_c$ applies $X^{c_j}$ to qubit $C_j$. Taking this interpretation, we observe that for any string $c$, plugging $V_c$ into Kitaev's circuit-to-Hamiltonian yields precisely the Hamiltonian $H(c)$. Thus, since by assumption for our particular choice of $c$, $V''$ accepts $\ket{c}\otimes\ket{q}$ for all quantum proofs $\ket{q}$, it follows that $V_c$ rejects all $\ket{q}$ with probability at least $1-\epsilon$. Hence, Lemma~\ref{2_lem:kitaev} implies $H(c)\succeq b I$, as desired.

The converse direction proceeds similarly. Namely, suppose we have a NO instance of $\Pi$, i.e.\ for all bit strings $\ket{c}$, there exists a quantum proof $\ket{q}$ such that $V''$ rejects $\ket{c}\otimes\ket{q}$ with probability at least $1-\epsilon$. Then, we wish to show that for all $c$, there exists a $\ket{\psi_{c}}$ such that $\bra{\psi_{c}}H(c)\ket{\psi_{c}}\leq a$, for $a$ from Lemma~\ref{2_lem:kitaev}. To show this, fix an arbitrary $c$. Since there exists a $\ket{q}$ such that $V_c$ accepts $\ket{q}$ with probability at least $1-\epsilon$, it follows that the history state $\ket{\psi_{c}}:=\sum_{i=0}^LV_i\cdots V_1 \ket{q}_B\otimes \ket{0\cdots 0}_C \otimes \ket{ i}_D$ indeed satisfies
\begin{equation}
    \bra{\psi_{c}}H(c)\ket{\psi_{c}}=\bra{\psi_{c}}\hin'+\hout'+\hstab'+\hprop'(c)\ket{\psi_{c}}\leq 0+a+0+0= a.
\end{equation}
\end{proof}

Note that the proof of Theorem~\ref{2_thm:anothercomplete} has a special property --- the string $c$ fed into the classical proof register of $V''$ is mapped directly in our reduction to the candidate ground states $\ket{c}\ket{q}$ for $3$-local Hamiltonian $H$. This means, for example, that if there exists a $c$ with the desired properties for a YES instance of our starting $\cqs$ problem, then setting $x=c$ in Definition~\ref{2_def:cqslh} yields that the $\cqslh{3}$ instance we have mapped to is also a YES instance. It follows that applying the reduction in the proof of Theorem~\ref{2_thm:anothercomplete} to our hard-to-approximate instance of QMW from Theorem~\ref{2_thm:qmwampgap} directly yields Theorem~\ref{2_thm:anothercomplete2}, i.e.\ that the following variant of $\cqslh{3}$, which we call $\cqslhmin$, is $\cqs$-hard to approximate. Intuitively, $\cqslhmin$ is defined analogously to $\cqslh{3}$, except that here the goal is to minimize the Hamming weight of $x$.

\begin{definition}[$\cqslhmin$]\label{def:cqslh}
    Given a $3$-local Hamiltonian $H$ acting on $N= n+m$ qubits, $a,b\in\reals$ such that $a\leq b$ for $b-a \geq 1$, and integer thresholds $0\leq g\leq g'$, output:
\begin{itemize}
    \item YES if there exists $x\in\set{0,1}^n$ of Hamming weight at most $g$ such that for all $\ket{y}\in\B^{\otimes m}$, $\trace(H \ketbra{x}{x}\otimes\ketbra{y}{y})\geq b$.
    \item NO if for all $x\in\set{0,1}^n$ of Hamming weight at most $g'$, there exists $\ket{y}\in\B^{\otimes m}$ such that $\trace(H \ketbra{x}{x}\otimes\ketbra{y}{y})\leq a$.
\end{itemize}
\end{definition}

\noindent \emph{Acknowledgements for this chapter.} We thank Richard Cleve, Ashwin Nayak, Sarvagya Upadhyay, and John Watrous for interesting discussions, and especially Oded Regev for many helpful insights, including the suggestion to think about a quantum version of PH.

%\addcontentsline{toc}{part}{Quantum proof systems}
%======================================================================
\chapter{QMA variants with polynomially many provers}\label{chap:qmapoly}

\emph{This chapter is based on~\cite{GSU11}:}\\

\vspace{-4mm}
\noindent S. Gharibian, J.~Sikora, and S. Upadhyay. QMA variants with polynomially many provers. Available at arXiv.org e-Print quant-ph/1108.0617v1, 2011.

\vspace{3mm}
\noindent In this chapter, we study three variants of multi-prover quantum Merlin-Arthur proof systems. We first show that the class of problems that can be efficiently verified using polynomially many quantum proofs, each of logarithmic-size, is exactly \class{MQA} (also known as QCMA), the class of problems which can be efficiently verified via a classical proof and a quantum verifier. We then study the class $\class{BellQMA}(\poly)$, characterized by a verifier who first applies unentangled, nonadaptive measurements to each of the polynomially many proofs, followed by an arbitrary but efficient quantum verification circuit on the resulting measurement outcomes. We show that if the number of outcomes per nonadaptive measurement is a polynomially-bounded function, then the expressive power of the proof system is exactly \class{QMA}. Finally, we study a class equivalent to \class{QMA}($m$), denoted $\class{SepQMA}(m)$, where the verifier's measurement operator corresponding to outcome {\it accept} is a fully separable operator across the $m$ quantum proofs. Using cone programming duality, we give an alternate proof of a result of Harrow and Montanaro~\cite{HM10} that shows a perfect parallel repetition theorem for $\class{SepQMA}(m)$ for any $m$.

%------------------------------------------------------------------------------------------------------------------------------------------------------%
\section{Introduction and results}\label{3_scn:intro}
%------------------------------------------------------------------------------------------------------------------------------------------------------%

The study of classical proof systems has yielded some of the greatest achievements in theoretical computer science, from the Cook-Levin theorem~\cite{C72,L73}, which formally ushered in the age of NP verification systems and the now ubiquitous notion of NP-hardness, to the more modern PCP theorem~\cite{AS98,ALMSS98}, which led to significant advancements in our understanding of hardness of approximation. A natural generalization of the class \class{NP} to the quantum setting is the class quantum Merlin-Arthur (QMA)~\cite{KSV02}, where a computationally powerful but untrustworthy prover, Merlin, sends a \emph{quantum} proof to convince an efficient \emph{quantum} verifier, Arthur, that a given input string $x \in \{ 0,1 \}^n$ is a YES-instance for a specified promise problem. (See Definition~\ref{0_def:QMA} for a formal definition of QMA.) %More specifically, a QMA proof system for promise problem $A$ is characterized by the following properties:
%\begin{itemize}
%\item For every YES-instance $x$ of $A$, there exists a polynomial-size quantum proof convincing Arthur of this fact with high probability, with the smallest such success probability over all YES-instances called the \emph{completeness} of the protocol.
%\item For every NO-instance $x$ of $A$, for \emph{any} such purported quantum proof, Arthur rejects with high probability, with the maximum success probability over all NO-instances called the \emph{soundness} of the protocol.
%\end{itemize}
\noindent It is easy to see that QMA proof systems are at least as powerful as \class{NP}, since the ability to process and exchange quantum information does not prevent Arthur from choosing to act classically.

As discussed in Sections~\ref{0_sscn:classes} and~\ref{0_scn:qma}, much attention has been devoted to \class{QMA} over recent years. We now have a number of problems which are \emph{complete} for \class{QMA}, with the quantum analogue of classical constraint satisfaction,  the physically well-motivated $k$-local Hamiltonian problem%~\cite{KSV02,KR03,KKR06,OT05,AGIK09}
~\cite{KSV02}, being the canonical \class{QMA}-complete problem. Further, \class{QMA} is an extremely robust complexity class that satisfies strong error-reduction properties~\cite{MW05}. However, there still remain important open questions. One natural such question, which is the focus of this chapter, is: \emph{How does allowing multiple unentangled provers affect the expressive power of \class{QMA}?}

Specifically, unlike in the classical setting where allowing multiple proofs, each quantified by a distinct existential quantifier, is trivially equivalent to a single existentially quantified proof, whether the same logic holds in the quantum setting is a highly non-trivial open question due to the quantum phenomenon known as \emph{entanglement} (see Section~\ref{0_sscn:entanglement}). Intuitively, entanglement between multiple proofs can be used by cheating provers to correlate their proofs in a way stronger than possible classically. To this end, in this chapter, we are interested in studying the class $\class{QMA}(\poly)$~\cite{KMY03}, a.k.a. quantum Merlin-Arthur proof systems with polynomially many Merlins, where the verifier receives a polynomial number of quantum proofs which are promised to be unentangled with each other. Despite much effort, little is known (more details under \emph{Previous Work} below) about the structural properties of $\class{QMA}(\poly)$, except for the obvious containments $\class{QMA}\subseteq\class{QMA}(\poly)\subseteq\class{NEXP}$.

%An approach for understanding a complexity class is to consider how introducing variations to its definition changes its properties. In this chapter, we thus ask: \emph{How does allowing multiple unentangled provers affect the expressive power of \class{QMA}?} In particular, we are interested in variants of the class $\class{QMA}(\poly)$, a.k.a. quantum Merlin-Arthur proof systems with polynomially many Merlins, where the verifier receives a polynomial number of quantum proofs, which are promised to be unentangled with each other. Note that the \emph{classical} version of this class collapses trivially to \class{MA}, as the set of potential strategies of a single Merlin and the set of potential strategies of multiple Merlins coincide. This logic fails, however, in the quantum case, as a single Merlin simulating the action of multiple Merlins can try to cheat by entangling the multiple proofs. Despite much effort, very little is known (more details under \emph{Previous Work} below) about the structural properties of  $\class{QMA}(\poly)$, except for the obvious containments $\class{QMA}\subseteq\class{QMA}(\poly)\subseteq\class{NEXP}$.

\paragraph{Our Results:} We show the following three results regarding variants of $\class{QMA}(\poly)$.

\myparagraph{A complete characterization in the logarithmic-size message setting} Let the class $\class{QMA}_{\log}(\poly)$ denote the restriction of the class $\class{QMA}(\poly)$ to the setting where each prover's proof is at most a \emph{logarithmic} number of quantum bits, or \emph{qubits}. We show:

\begin{theorem}\label{3_thm:qmalog}
    $\class{QMA}_{\log}(\poly)=\class{MQA}$.
\end{theorem}

\noindent Here, recall from Chapter~\ref{chap:intro} that \class{MQA}, also known as QCMA, is defined as $\class{QMA}$ except Merlin's proof is a polynomial-size \emph{classical} string. Theorem~\ref{3_thm:qmalog} says that if each prover is restricted to sending short quantum proofs, then one can not only do away with multiple provers, but also of the need for \emph{quantum} proofs altogether.

%The significance of this result is as follows: Understanding the expressive power of $\class{QMA}(\poly)$, or its relationship with \class{QMA}, is currently one of the biggest challenges in quantum complexity theory. (It was recently shown that, in fact, $\class{QMA(\poly)}=\class{QMA(2)}$~\cite{HM10}, where \class{QMA(2)} is \class{QMA} with two unentangled provers.) Theorem~\ref{3_thm:qmalog} completely settles this question in the logarithmic-size message setting. Moreover, following Reference~\cite{MW05} it implies the following separation, where $\class{QMA}_{\log}$ is $\class{QMA}$ restricted to a logarithmic-size proof:
%
%\begin{corollary}\label{3_cor:qma_log_neq_qma_log_poly}
%    If $\class{BQP} \neq \class{MQA}$, then $\class{QMA}_{\log} \neq \class{QMA}_{\log}(\poly)$.
%\end{corollary}
%
%\noindent This follows simply because $\class{QMA}_{\log}=\class{BQP}$~\cite{MW05}. Note that the assumption $\class{BQP} \neq \class{MQA}$ is reasonable, as otherwise quantum computers could efficiently solve $\class{NP}$-complete problems.

%\noindent Here, $\class{BQP}$ is the class of decision problems which can be solved on a quantum computer with a bounded probability of error. Corollary~\ref{3_cor:qma_log_neq_qma_log_poly} follows because $\class{QMA}_{\log}=\class{BQP}$~\cite{MW05}. Note that the assumption $\class{BQP} \neq \class{MQA}$ is reasonable, as otherwise quantum computers could efficiently solve $\class{NP}$-complete problems.

\myparagraph{Towards a non-trivial upper bound on $\class{BellQMA}(\poly)$} One possible approach to the question of $\class{QMA}\stackrel{?}{=}\class{QMA(\poly)}$ is to study \class{BellQMA}(\poly)~\cite{B08,ABDFS09,CD10}. \class{BellQMA}(\poly) is defined analogously to $\class{QMA}(\poly)$, except that before applying his verification circuit to the polynomially many unentangled quantum proofs, Arthur must measure each proof using a nonadaptive and unentangled (across all proofs) measurement (we call this \emph{Stage 1} of the verification). He then feeds the resulting \emph{classical} outcomes induced by these measurements into an efficient quantum circuit (we call this \emph{Stage 2}), which implements a two-outcome measurement operation corresponding to outcomes {\it accept} and {\it reject}.

The significance of $\class{BellQMA}(\poly)$ here is that if $\class{QMA} \neq \class{BellQMA}(\poly)$, then it follows that $\class{QMA}\neq\class{QMA}(\poly)$, since $\class{QMA}\subseteq\class{BellQMA}(\poly)\subseteq\class{QMA}(\poly)$. To this end, Brand\~{a}o has shown that for constant $m$,  $\class{QMA}=\class{BellQMA}(m)$~\cite{B08}. Where $\class{BellQMA}(\poly)$ lies, however, remains open. For example, the techniques used to show $\class{QMA}(2)=\class{QMA}(\poly)$ \cite{HM10} do not seem to yield an analogous result $\class{BellQMA}(2)=\class{BellQMA}(\poly)$ as they require entangled measurements (i.e.\ SWAP test measurements) across multiple proofs, which violate the definition of $\class{BellQMA}$.

To make progress on $\class{BellQMA}(\poly)$, we introduce the class $\class{BellQMA}[r,m]$, which is defined to be $\class{BellQMA}(m)$ with $m$ provers and the additional restriction that in Stage 1 above, the number of outcomes per proof in Arthur's nonadaptive measurements is upper bounded by $r$. We then show the following:

\begin{theorem}\label{3_thm:bellqma}
For any polynomially bounded functions $r,m:\nats \rightarrow \nats$, it holds that $\class{BellQMA}[r,m] \subseteq \class{QMA}$ (where the containment holds with equality when $r\geq 2$).
\end{theorem}

\noindent In other words, $\class{BellQMA}(\poly)$ cannot be used to show that $\class{QMA}\neq\class{QMA}(\poly)$ if the verifier in the $\class{BellQMA}(\poly)$ protocol is restricted to have a polynomially bounded number of measurement outcomes per proof in Stage 1. We remark that, in general, the number of such measurement outcomes can be exponential in the input length --- the restriction that $r$ be a polynomially bounded function is crucial for the proof of Theorem~\ref{3_thm:bellqma}. For this reason, our result complements, rather than subsumes Brand\~{a}o's result~\cite{B08}. In other words, in our notation, Brand\~{a}o has shown that $\class{BellQMA}[\exp,\rm{const}] = \class{QMA}$, and we show $\class{BellQMA}[\poly,\poly] = \class{QMA}$.

Note that we allow the second stage of the $\class{BellQMA}(\poly)$ verification procedure above to be {\it quantum}, as per the definition suggested by Chen and Drucker~\cite{CD10}, as opposed to {\it classical}, as studied by Brand\~{a}o~\cite{B08}. The conclusion of Theorem~\ref{3_thm:bellqma} holds even if the second stage of verification is completely classical.

Finally, it is worth noting that by combining Theorems~\ref{3_thm:qmalog} and~\ref{3_thm:bellqma}, we conclude that in the setting of $\class{BellQMA}(\poly)$, if $\class{MQA} \neq \class{QMA}$, then having the Merlins send logarithmic-size proofs without any restriction on the number of local measurement outcomes of Arthur in Stage 1 has less expressive power than sending polynomial-size proofs but restricting the number of outcomes, even though the number of measurement outcomes in Stage 1 per Merlin in both cases is the same, i.e.\ polynomial in the input length.

\myparagraph{Perfect parallel repetition for $\class{SepQMA}(m)$} A key question in designing proof systems is how to improve the completeness and soundness parameters of a verification protocol without increasing the required number of rounds of communication. A natural approach for doing so is to repeat the protocol multiple times in parallel. With \class{QMA}, however, this raises the concern that Merlin might try to cheat by entangling his proofs across these parallel runs. If, though, \emph{perfect parallel repetition} holds, it means that for any input string $x$, if the verification procedure $V$ accepts with probability $p(|x|)$, then if we run $V$ $k$ times in parallel, the probability of accepting in all $k$ runs of $V$ is precisely $p(|x|)^k$. In other words, if perfect parallel repetition holds, there is no incentive for Merlin to cheat --- an honest proof which is a product state across all $k$ runs achieves the maximum success probability.

Our final contribution is an alternate proof of a perfect parallel repetition theorem for a class equivalent to \class{QMA}($m$), namely $\class{SepQMA}(m)$. The theorem was first proved in Harrow and Montanaro~\cite{HM10} in connection with an error reduction technique for $\class{QMA}(\poly)$. However, our proof is significantly different from theirs and uses the cone programming characterization of $\class{QMA}(\poly)$. Here, $\class{SepQMA}(m)$ is defined as \class{QMA}($m$) with the restriction that Arthur's measurement operator corresponding to acceptance is an unentangled, or \emph{separable}, operator across the $m$ unentangled proofs. We show:

\begin{theorem}
\label{3_thm:parrep} $\class{SepQMA}(m)$ admits perfect parallel repetition.
\end{theorem}

\noindent Our alternate proof of Theorem~\ref{3_thm:parrep} is significant in that, to the best of our knowledge, it is the first use of duality theory for a cone program \emph{other} than a semidefinite program to establish a parallel repetition result (note that cone programming generalizes semidefinite programming). We remark that semidefinite programs have been previously used to show perfect or strong parallel repetition theorems for various other models of (single or two-prover) quantum interactive proof systems~\cite{CSUU08,G09,KRT10}, and that the alternate proof of Theorem~\ref{3_thm:parrep} of Harrow and Montanaro is not based on semidefinite programming. Perfect parallel repetition for $\class{SepQMA}(m)$ in itself is interesting, as it has been used to show that error reduction is possible for $\class{QMA}(m)$ proof systems~\cite{HM10}.

\paragraph{Proof ideas and tools:} The proof of our first result, Theorem~\ref{3_thm:qmalog}, is simple, and is an application of the facts that (1) quantum states of a logarithmic number of qubits can be described to within inverse exponential precision using a polynomial number of classical bits, and conversely that (2) given such a classical description, a logarithmic-size quantum state can be efficiently prepared by a quantum circuit. Hence, roughly speaking, one can replace a polynomial number of logarithmic-size quantum proofs with a single polynomial size classical proof, thereby avoiding the danger of a cheating Merlin using entanglement. Although the proof is simple, one cannot hope for a better characterization using other techniques because the reverse containment, i.e.\ $\class{MQA}\subseteq \class{QMA}_{\log}(\poly)$, also holds using similar ideas.% That is, \class{MQA} can be simulated by polynomially-many logarithmic-size quantum proofs, where each Merlin is supposed to send a bit of the optimal classical proof.

More technically challenging is our second result, Theorem~\ref{3_thm:bellqma}. To show the non-trivial direction $\class{BellQMA}[\poly,\poly]\subseteq \class{QMA}$, we simulate an arbitrary $\class{BellQMA}[\poly,\poly]$ protocol by a \class{QMA} protocol using the following observation: {Although consolidating $m$ quantum proofs into a single quantum proof raises the possibility of cheating using entanglement, if Arthur is {also} sent an appropriate classical ``consistency-check'' string, then a dishonest Merlin can be caught with non-negligible probability.} Specifically, in our \class{QMA} protocol, we ask a single Merlin to send the $m$ quantum proofs of the original \class{BellQMA} protocol (denoted by a single state $\ket{\psi}$), accompanied by a  ``consistency-check'' string $\ve{p}$ which is a classical description of the probability distributions obtained as the output of Stage 1. One can think of this as having the \class{QMA} verifier \emph{delegate} Stage 1 of the \class{BellQMA} verification to Merlin. Arthur then performs a consistency check between $\ket{\psi}$ and $\ve{p}$ based on the premise that if Merlin is honest, then $\ve{p}$ should arise from running Stage 1 of the original verification on $\ket{\psi}$. If this check passes, then Arthur runs Stage 2 of the \class{BellQMA} verification on $\ve{p}$. If Merlin tries to cheat, however, we show that the check detects this with non-negligible probability, hence achieving the desired containment. Note that the accuracy of the consistency check crucially uses the fact that there are at most polynomially many outcomes to check for each local measurement of Stage 1.

Our last result, Theorem~\ref{3_thm:parrep}, is shown using duality theory for cone programs. In particular, we phrase the maximum acceptance probability of a (possibly cheating) prover for the two-fold repetition of a $\class{SepQMA}(m)$ verification protocol as a cone program. We then demonstrate a feasible solution for its dual yielding an upper bound on the maximum acceptance probability. The objective value of this dual solution is precisely the product of the optimum values of the two instances of the $\class{SepQMA}(m)$ verification protocols. We conclude that one of the optimal strategies of the provers is to be faithful in the following sense: Each prover elects not to entangle his/her two quantum proofs for the two instances of the $\class{SepQMA}(m)$ protocol and instead sends a tensor product of optimal proofs for both the instances.

\paragraph{Previous work.} The expressive power of multiple Merlins was first studied by Kobayashi, Matsumoto and Yamakami~\cite{KMY03}, who showed that $\class{QMA}(2) = \class{QMA}(\poly)$ if and only if the class of \class{QMA}(2) protocols with completeness $c$ and soundness $s$ (with at least inverse polynomial gap) is exactly equal to $\class{QMA}(2)$ protocols with completeness $2/3$ and soundness $1/3$. Recently, Harrow and Montanaro~\cite{HM10} demonstrated a {\it product state test}, wherein given two copies of a {\it pure} quantum state on multiple systems, the test distinguishes between the cases when the quantum state is a {\it fully} product state across all the systems or {\it far} from any such state. Using this test, they answered a few important questions regarding $\class{QMA}(\poly)$. In particular, they showed that
\begin{equation}
\class{QMA}(2) = \class{QMA}(\poly)
\end{equation}
and that error reduction is possible for such proof systems. Prior to their result, the answers to both the questions were known to be affirmative assuming a {\it weak} version of the Additivity Conjecture~\cite{ABDFS09}. One of the crucial properties of the product state test is that it can be converted into a $\class{QMA}(2)$ protocol, where Arthur's measurement operator corresponding to outcome {\it accept} is a separable operator across the two proofs. Harrow and Montanaro established a perfect parallel repetition theorem for such proof systems, a crucial step in obtaining exponentially small error probabilities.

Blier and Tapp initiated the study of \emph{logarithmic}-size unentangled quantum proofs~\cite{BT10}.
They showed that two unentangled quantum proofs suffice to show that a 3-coloring of an input graph exists, implying that $\class{NP}$ has \emph{succinct} unentangled quantum proofs. A drawback of their protocol is that although it has {\it perfect} completeness, its soundness is only inverse polynomially bounded away from $1$. Shortly after, Aaronson, Beigi, Drucker, Fefferman and Shor~\cite{ABDFS09} showed that satisfiability of any 3-SAT formula of size $n$ can be proven by $\widetilde{O}(\sqrt{n})$ unentangled quantum proofs of $O(\log n)$ qubits with perfect completeness and constant soundness (see also~\cite{CD10}). In a subsequent paper~\cite{Beigi08}, Beigi improved directly on Blier and Tapp's result~\cite{BT10} by showing that by sacrificing perfect completeness, one can show that \class{NP} has two logarithmic-size quantum proofs with a better gap between completeness and soundness probabilities than in~\cite{BT10} (see also Chiesa and Forbes~\cite{CF11} and Le Gall, Nakagawa, and Nishimura~\cite{LNN12} for related improvements which do not sacrifice perfect completeness).

Finally, one of the open questions raised in Reference~\cite{ABDFS09} concerns  the power of Arthur's verification procedure. In particular, the paper introduces two different classes of verification procedures, \class{BellQMA} and \class{LOCCQMA} verification. Roughly speaking, \class{LOCCQMA} verification corresponds to Arthur applying a measurement operation that can be implemented by Local Operations and Classical Communication (LOCC) (with respect to the partition induced by the multiple proofs). The authors raised the question of whether $\class{BellQMA}(\poly) = \class{QMA}$ or not. Brand\~{a}o~\cite{B08} showed that $\class{BellQMA}(m)$ is equal to \class{QMA} for constant $m$. In a recent development, Brand\~{a}o, Christandl and Yard~\cite{BCY11} showed that $\class{LOCCQMA}(m)$ is equal to \class{QMA} for constant $m$.

\paragraph{Open problems.} %In this chapter, we have studied three variants of multi-prover quantum Merlin-Arthur proof systems. We first showed that a system with polynomially many provers is indeed strictly more powerful than a single prover system if messages are restricted to be logarithmic in length, unless $\class{BQP}=\class{MQA}$. We next showed that polynomially many provers do not provide additional expressive power over a single prover in the setting where the verifier is restricted to first applying unentangled and non-adaptive measurements with at most a polynomial number of outcomes per proof. Both of these questions make steps towards understanding the major open question of whether $\class{QMA}$ with polynomially many provers is more powerful than $\class{QMA}$. Finally, we used cone programming duality to give an alternate proof of the fact that perfect parallel repetition holds whenever a $\class{QMA}$ verifier's POVM element corresponding to {\it accept} is a fully separable operator.

A natural open question concerning the results presented in this chapter is the relationship between $\class{BellQMA}(\poly)$ and \class{QMA}. We believe that understanding the complexity of $\class{BellQMA}$ protocols will shed new light on the bigger question pertaining to \class{QMA}(2) and \class{QMA}. Another avenue of interest is to find further applications of the cone programming characterization of multi-prover quantum Merlin-Arthur proof systems. One question concerning the parallel repetition result presented in this chapter is to investigate whether cone programming duality can be used to analyze the product state test in Reference~\cite{HM10}. Finally, it would be interesting to find other classes of $\class{QMA}(m)$ protocols that admit a perfect parallel repetition theorem.

\paragraph{Organization of this chapter.} We begin in Section~\ref{3_scn:preliminaries} with background and notation, defining relevant complexity classes in Section~\ref{3_sscn:complexityclasses}, and reviewing cone programming in Section~\ref{3_sscn:conic}. Theorems~\ref{3_thm:qmalog},~\ref{3_thm:bellqma}, and~\ref{3_thm:parrep} are proved in Sections~\ref{3_scn:qmalog},~\ref{3_scn:bellqma}, and~\ref{3_scn:parrep}, respectively.

%------------------------------------------------------------------------------------------------------------------------------------------------------%
\section{Preliminaries}\label{3_scn:preliminaries}
%------------------------------------------------------------------------------------------------------------------------------------------------------%

In this section, we state useful lemmas, and discuss relevant complexity classes and cone programming. Throughout the chapter, we use $\abs{x}$ to denote the length of string $x \in \{ 0,1 \}^\ast$. The standard Hilbert-Schmidt inner product of operators $A$ and $B$ is denoted $\ip{A}{B}:=\tr(A^\dagger B)$, where $A^\dagger$ denotes the adjoint of $A$.

First, a useful lemma in this chapter regarding the trace norm (which is a Schatten $p$-norm with $p=1$) is the following:

\begin{lemma}[\cite{W02}]\label{3_lem:tracenorm}
Let $\{\rho_1 \dots, \rho_k\} \subset \density{\X}$ and $\{\sigma_1, \dots, \sigma_k\} \subset \density{\X}$. Then for any Schatten p-norm,
\begin{equation}
\norm{\bigotimes_{i=1}^k \rho_i - \bigotimes_{i=1}^k \sigma_i }_p \le \sum_{i=1}^k \norm{\rho_i - \sigma_i}_p.
\end{equation}
\end{lemma}

Next, generalizing Definition~\ref{0_eqn:sep}, we say a (possibly unnormalized) operator $A \in \pos{\X_1 \otimes \dots \otimes \X_m}$ is {\it fully separable} (i.e.\ unentangled) if it can be written as
\begin{equation}
A = \sum_{i=1}^k P_{1}(i) \otimes \dots \otimes P_{m}(i),
\end{equation}
where $P_j(i) \in \pos{\X_j}$, for every $j \in [m]$ and $i \in [k]$. We denote the cone of fully separable operators as $\sep{\X_1,\X_2, \dots ,\X_m}$. In the setting of quantum information, one typically also has $\tr(A)=1$. It will be useful to note that the set of fully separable density operators is convex, compact, and has non-empty interior since it contains a ball around the normalized identity operator~\cite{GB02,GB03,GB05}.

%SEV: Need to discuss Solovay-Kitaev and any pure state can be represented using Nf(N) bits?

\subsection{Relevant complexity classes}\label{3_sscn:complexityclasses}

We now define the relevant complexity classes specific to this chapter. Recall that a promise problem $A = (A_{\yes}, A_{\no})$ is a partition of the set $\{0,1\}^*$ into three disjoint subsets: the set $A_{\yes}$ denotes the set of YES-instances of the problem, the set $A_{\no}$ denotes the set of NO-instances of the problem, and  $\{0,1\}^*\backslash (A_{\yes} \cup A_{\no})$ is the set of disallowed strings.

We begin by formally generalizing the definition of QMA (see Definition~\ref{0_def:QMA}) to the setting of $m$ unentangled provers.

\begin{definition}[$\class{QMA}(m)$]\label{3_def:qma-poly}
    A promise problem $A=(\ayes,\ano)$ is in $\class{QMA}(m)$ if there exist polynomials $p$, $q$ and a polynomial-time uniform family of quantum circuits $\set{Q_n}$, where $Q_n$ takes as input a string $x\in\Sigma^*$ with $\abs{x}=:n$, quantum proof $\ket{y}\in\sep{\X_1,\X_2, \dots ,\X_{m(n)}}$ where $ \X_i := (\complex^2)^{\otimes p(n)}$ for $i\in[m(n)]$, and $q(n)$ ancilla qubits in state $\ket{0}^{\otimes q(n)}$, such that:
    \begin{itemize}
    \item (Completeness) If $x\in\ayes$, then there exists a proof $\ket{y}\in\sep{\X_1,\X_2, \dots ,\X_{m(n)}}$ such that $Q_n$ accepts $(x,\ket{y})$ with probability at least $2/3$.
    \item (Soundness) If $x\in\ano$, then for all proofs $\ket{y}\in\sep{\X_1,\X_2, \dots ,\X_{m(n)}}$, $Q_n$ accepts $(x,\ket{y})$ with probability at most $1/3$.
    \end{itemize}

    \noindent The class $\class{QMA}(\poly)$ is defined as $\class{QMA}(\poly) := \bigcup_{m \in \poly} \class{QMA}(m)$.
\end{definition}

\noindent For clarity, note that $\ket{y}\in\sep{\X_1,\X_2, \dots ,\X_{m(n)}}$ must have the form $\ket{y}=\ket{y_1}\otimes\cdots\otimes\ket{y_{m(n)}}$ for $\ket{y_i}\in\X_i$. Hence, $\class{QMA}(m)$ can be thought of as having $m(n)$ unentangled provers. Note that like $\class{QMA}=\class{QMA}(1)$, the constants $2/3$ and $1/3$ above can be amplified to values exponentially close to $1$ and $0$, respectively, by having the verifier run the verification procedure polynomially times in parallel (this requires increasing the number of provers, however). Also, we will use the fact that corresponding to any $\class{QMA}(m)$ protocol is a two-outcome POVM (see Section~\ref{0_sscn:measurement}) consisting of operators $\set{C_{\rm accept},C_{\rm reject}}$, such that for any candidate proof $\ket{\psi}=\ket{\psi_1}\otimes \cdots\otimes\ket{\psi_{m(n)}}$, the probability of the verifier accepting (rejecting) is given by $\tr(C_{\rm accept}\ketbra{\psi}{\psi})$ ($\tr(C_{\rm reject}\ketbra{\psi}{\psi})$).

All complexity classes considered in this chapter are variants of $\class{QMA}(m)$ and satisfy the properties mentioned above in Definition~\ref{3_def:qma-poly}. The next two classes we define are:

\begin{mylist}{\parindent}
	\item [1.] \textbf{[$\class{QMA}_{\log}(\poly)$]} A subclass of $\class{QMA}(\poly)$ in which each Merlin's message to Arthur is $O(\log (|x|))$ qubits in length for input string $x$.
	\item [2.] \textbf{[\class{SepQMA}(\poly)]} A subclass of $\class{QMA}(\poly)$, wherein Arthur's measurement operator $C_{\rm accept}$ corresponding to outcome {\it accept} is a fully separable operator across the proofs.	
\end{mylist}

\noindent For clarity, we next give a more formal definition of the variant of $\class{BellQMA}$ we introduce, $\class{BellQMA}[r,m]$.

\begin{definition}[$\class{BellQMA}{[}r,m{]}$]
Let $r,m:\nats \rightarrow \nats$ be two functions. A promise problem $A = (A_{\yes}, A_{\no})$ is in class $\class{BellQMA}[r,m]$ if there exists a $\class{QMA}(m)$ verification protocol in which Arthur is restricted to act as follows.

\begin{mylist}{\parindent}
\item[1.]
Arthur performs a polynomial-time quantum computation on the input $x$ and generates a description of quantum circuits $V_1(x), \dots, V_m(x)$, one for each of the $m$ provers.

\item[2.]
(Stage 1) Arthur simultaneously measures all $m$ quantum proofs by applying $V_i(x)$ to the $i$-th quantum proof, where the action of $V_i(x)$ can be described by a unitary operator followed by measurement in the standard basis. The label of the $i$-th measurement outcome is stored as a classical string $y_i$ also identified as an element of $[r(|x|)]$.

\item[3.]
(Stage 2) Arthur runs an efficient quantum verification circuit on input $x$ and measurement outcomes $(y_1,\ldots,y_m)$ to decide whether to accept or reject.
\end{mylist}
\end{definition}

\noindent Note that the key distinction between $\class{BellQMA}[r, m]$ and $\class{BellQMA}(\poly)$ is that the former has the number of measurement outcomes in Stage 1 of the protocol bounded by $r(\abs{x})$, whereas the latter may allow exponentially many possible outcomes. Throughout this chapter, we use the notation $\class{BellQMA}[\poly, \poly]$ to denote
        \begin{equation}
        \class{BellQMA}[\poly, \poly] := \bigcup_{r \in \poly} \bigcup_{m \in \poly}\class{BellQMA}[r,m].
        \end{equation}
%We remark that, as in~\cite{CD10}, our $\class{BellQMA}$ protocols are allowed to use a \emph{quantum} verification circuit in Stage 2, whereas originally in References~\cite{B08,ABDFS09} only classical processing of measurement outcomes $\set{y_i}$ was allowed in order to emulate the notion of a \emph{Bell experiment} performed by Arthur. We remark that Theorem~\ref{3_thm:bellqma} holds even if Arthur is restricted to do classical processing on the measurement outcomes.

%------------------------------------------------------------------------------------------------------------------------------------------------------%
\subsection{Cone programming} \label{3_sscn:conic}
%------------------------------------------------------------------------------------------------------------------------------------------------------%

We now briefly review basic notions in conic optimization (or cone programming), which is a generalization of semidefinite optimization. The reader is referred to the text of Boyd and Vandenberghe~\cite{BV04} for further details.

To begin, recall that a set $K$ in an underlying Euclidean space is a \emph{cone} if $x \in K$ implies that $\lambda x \in K$ for all $\lambda \geq 0$. A cone $K$ is \emph{convex} if $x,y \in K$ implies that $x+y \in K$. \emph{Cone programs} are concerned with optimizing a linear function over the intersection of a convex cone and an affine space. It generalizes several well-studied models of optimization including semidefinite programming (where $K = \Pos (\X)$) and linear programming (where $K = \R_+^n$). In this chapter, we are primarily concerned with the cone of fully separable operators $K=\sep{\X_1,\X_2, \dots ,\X_m}$, which as stated in Section~\ref{3_scn:preliminaries} is a closed, convex cone with non-empty interior.

A cone program associates the following 4-tuple $(C, b, \calA, K)$ to an optimization problem, which we denote as the \emph{Primal} problem:
%\begin{center}
%  \begin{minipage}{1.5in}
%    \centerline{\underline{Primal problem (P)}}
%\begin{align}
%	\text{supremum:}\quad & \ip{X}{C}\\
%  \text{subject to:}\quad & \calA(X) = b,\\
%  & X \in K,
%\end{align}
%  \end{minipage}
%  \end{center}
    \begin{center}
%  \begin{minipage}{1.5in}
    \centerline{\underline{Primal problem (P)}}\vspace{-7mm}
    \begin{align}
			\text{supremum:}\quad & \ip{X}{C}\\
  		\text{subject to:}\quad & \calA(X) = b,\\
  		& X \in K,
		\end{align}
%  \end{minipage}
\end{center}
where $\calA: \Span(K) \to \R^{m}$ is a linear transformation, and $K$ lies in a real Euclidean space. (Note that the choice of inner product in $\ip{X}{C}$ depends on the Euclidean space $K$ lies in.) We say that the cone program is \emph{feasible} if $\{X: \calA(X) = b \} \cap K$ is non-empty and \emph{strictly feasible} if $\{X: \calA(X) = b \} \cap \interior(K)$ is non-empty, where $\interior(\cdot)$ denotes the interior of a set.

Next, associated with a cone $K$ is its dual cone $K^{\ast}$, defined as
\begin{equation}
    K^{\ast} = \left\{S: \ip{X}{S} \ge 0 \text{ for all } X \in K \right\}.
\end{equation}
Via the dual cone, for every Primal problem, one can define an associated \emph{Dual} problem as follows:
\begin{center}
%  \begin{minipage}{1.5in}
%    \centerline{\underline{Primal problem (P)}}\vspace{-7mm}
%    \begin{align}
%			\text{supremum:}\quad & \ip{X}{C}\\
%  		\text{subject to:}\quad & \calA(X) = b,\\
%  		& X \in K,
%		\end{align}
%  \end{minipage}
%  \hspace*{25mm}
%  \begin{minipage}{1.5in}
    \centerline{\underline{Dual problem (D)}}\vspace{-7mm}
		\begin{align}
			\text{infimum:}\quad & \ip{b}{y}\\
  		\text{subject to:}\quad & \calA^{\ast}(y) = C + S,\\
  		& S \in K^{\ast},
  	\end{align}	
%  \end{minipage}
\end{center}
where $\calA^\ast$ is the adjoint of $\calA$. We remark that so long as $K$ is closed (which is the case for the cone of fully separable operators), the roles of the Primal and Dual problems can be freely interchanged, since a convex cone $K$ is closed if and only if $K = K^{**}$.

The problems (P) and (D) obey the following special relationship.

\begin{lemma}[Weak Duality]
If $X$ is primal feasible and $(y,S)$ is dual feasible then
\begin{equation}
\ip{b}{y} - \ip{X}{C} = \ip{X}{S} \geq 0.
\end{equation}
\end{lemma}
\noindent In other words, let the optimal values of (P) and (D) be denoted $p^*$ and $d^*$, respectively. Then $p^*\leq d^*$. This raises the important question: Does $p^*=d^*$? In general, this is not the case. However, if indeed $p^*=d^*$, we say that \emph{strong duality} holds. Below we give a condition which, if satisfied, guarantees that strong duality holds.

\begin{theorem}[Strong Duality]\label{3_thm:stduality}
If \textup{(P)} is strictly feasible, then strong duality holds, i.e.\ $p^*=d^*$. In particular, this implies that if $p^*$ is finite, then both \text{(P)} and \textup{(D)} attain their optimal values, which coincide.
\end{theorem}

\noindent Note that when $K$ is a closed, convex cone, one can flip the roles of primal and dual problems in Theorem~\ref{3_thm:stduality}.

%We refer the reader to the work of Tun\c cel and Wolkowicz~\cite{TW08}
%and the references therein for more details on cone programming duality.

%------------------------------------------------------------------------------------------------------------------------------------------------------%
\section{Equivalence of \class{MQA} and $\class{QMA}_{\blog}(\poly)$}\label{3_scn:qmalog}
%------------------------------------------------------------------------------------------------------------------------------------------------------%

We now prove Theorem~\ref{3_thm:qmalog}, i.e.\ that $\class{MQA}=\class{QMA}_{\log}(\poly)$. We first show the direction $\class{MQA} \subseteq \class{QMA}_{\log}(\poly)$. Let $A = (A_{\yes},A_{\no})$ be a promise problem in \class{MQA} and let $x \in \{ 0,1 \}^n$ be the input string. Suppose the \class{MQA} prover sends an $m$-bit classical proof to the verifier, for polynomially bounded $m$. Then the following straightforward $\class{QMA}_{\log}(m)$ protocol achieves the desired containment:

\begin{mylist}{\parindent}
\item[1.]
\textbf{Embed classical bits into qubits.} Each (unentangled) prover $i\in[m]$ sends a single qubit $\ket{\psi_i}\in \complex^2$ to Arthur. If the $i$-th prover is honest, his/her qubit is the computational basis state corresponding to the $i$-th bit of the classical \class{MQA} proof.

\item[2.] \textbf{Make things classical again.} Arthur measures all proofs in the computational basis, obtaining a classical string $y \in \{0,1\}^{m}$.

\item[3.]
\textbf{Run MQA verification.} Arthur runs the \class{MQA} verification circuit on $x$ and $y$ and accepts if and only if acceptance occurs in the \class{MQA} verification.
\end{mylist}

\noindent The completeness property follows straightforwardly. The soundness property is also easy to observe. Note that Arthur runs the \class{MQA} verification on a classical string $y$ and hence he accepts the string with probability at most $1/3$.

To show the reverse containment, let $A = (A_{\yes},A_{\no})$ be a promise problem in class $\class{QMA}_{\log}(\poly)$ and let $x \in \{ 0,1 \}^n$ be the input string. Suppose we have a $\class{QMA}_{\log}(m)$ protocol for polynomially bounded $m$, where prover $i$ sends a $\lceil c\log n\rceil$-qubit state $\ket{\psi_i}$ for some constant $c>0$. Let $r(n) = 2^{\lceil c\log n \rceil} = O(n^c)$. The \class{MQA} protocol proceeds as follows:

\begin{mylist}{\parindent}
\item[1.]
\textbf{Describe proofs classically.} The prover sends $m$ classical registers represented by the tuple $(\mathsf{C}_1, \mathsf{C}_2, \dots, \mathsf{C}_m)$, each of length $2n\cdot r(n)$ to Arthur. If the prover is honest, register $\mathsf{C}_i$ contains a classical description of the $i$-th quantum proof $\ket{\psi_i}$ of the $\class{QMA}_{\log}(m)$ protocol.

\item[2.]
\textbf{State preparation.} Using the contents of register $\mathsf{C}_i$, for every choice of $i \in [m]$, Arthur prepares the state $\ket{\psi_i}$ by first determining a unitary $U_i$ such that $U_i\ket{0\ldots 0}=\ket{\psi_i}$, and then implementing $U_i$ with high precision using a finite set of approximately universal gates, obtaining states $\ket{\psi_i^\prime}$.

\item[3.]
\textbf{Run $\class{QMA}_{\log}(m)$ verification.} Arthur runs the $\class{QMA}_{\log}(m)$ verification circuit on $\ket{\psi_1^\prime}\otimes\cdots\otimes\ket{\psi_m^\prime}$ and accepts if and only if acceptance occurs in $\class{QMA}_{\log}(m)$ verification.
\end{mylist}

\noindent Observe that each classical register $\mathsf{C}_i$ is of size polynomial in $n$, implying the overall proof length is of polynomial size. In Step~1, the prover uses $n$ bits to represent the real and imaginary parts of each of the polynomially many entities ($r(n)$ entries) required to describe each $\ket{\psi_i}$. Let the unit vector described by register $\mathsf{C}_i$ be denoted $\ket{\psi_i}$. In Step~2, $U_i$ is easily found, as the unitary that maps $\ket{0\ldots0}$ to $\ket{\psi_i}$ is the inverse of the unitary that maps $\ket{\psi_i}$ to $\ket{0\ldots0}$. Next, $U_i$ can be efficiently decomposed into a product of $U_i'$ one- and two-qubit unitary gates (see Bernstein and Vazirani~\cite{BV97} for details, or Section~\ref{0_sscn:circuit} under ``Universal gate sets'' for a brief discussion) such that $\snorm{U_i-U_i'}$ is inverse exponentially small. Since Steps 1 and 2 can be performed to within inverse exponential error, we thus can ensure $\norm{\ket{\psi_i} - \ket{\psi_i^\prime}} \le \epsilon$ for all $i \in [m]$ and for inverse exponential $\epsilon>0$. By Lemma~\ref{3_lem:tracenorm}, it follows that the overall precision error is at most $m\epsilon$ for polynomial $m$, and thus the completeness and soundness of the protocol are bounded from below and above by $\frac{2}{3} - m\epsilon$ and $\frac{1}{3} + m\epsilon$, respectively.

%Alternatively, the containment $\class{QMA}_{\log}(\poly) \subseteq \class{MQA}$ can be shown using a slightly different protocol\footnote{This protocol was mentioned to us by Richard Cleve.}, where Merlin sends classical descriptions of the quantum circuits that generate the quantum proofs from $\ket{0 \ldots 0}$ instead of classical descriptions of the proofs.

%------------------------------------------------------------------------------------------------------------------------------------------------------%
\section{Equivalence of $\class{BellQMA}[\poly, \poly]$ and $\class{QMA}$}\label{3_scn:bellqma}
%------------------------------------------------------------------------------------------------------------------------------------------------------%

We now show Theorem~\ref{3_thm:bellqma}, i.e.\ that $\class{BellQMA}[r, m]=\class{QMA}$ for polynomially-bounded functions $r$ and $m$. For notational convenience, let $\Pi_j(i)$ denote Arthur's $i$-th POVM element in Stage 1 of the \class{BellQMA} verification protocol for the $j$-th prover (i.e.\ $\sum_{i=1}^r \Pi_j(i) =\I$), where we assume without loss of generality that the number of possible outcomes is exactly $r$ for each prover, and where $j\in[m]$ for $m$ the number of provers.

We proceed as follows. Let $A = (A_{\yes}, A_{\no})$ be a promise problem, and $x$ be an input string of length $n := |x|$. Note first that the containment $\class{QMA} \subseteq \class{BellQMA}[\poly, \poly]$ follows since, by definition, $\class{QMA} \subseteq \class{BellQMA}[2,1]$. For the reverse containment, suppose we have a $\class{BellQMA}[r, m]$ protocol for polynomially bounded functions $r,m:\nats \rightarrow \nats$ with completeness $2/3$ and soundness $1/3$. We show that this protocol can be simulated by a \class{QMA} protocol as follows.

Merlin's proof consists of two registers $(\mathsf{X}, \mathsf{Y})$, which should be thought of as the \emph{classical} and \emph{quantum} registers, respectively. Suppose optimal proofs for the $\class{BellQMA}[r, m]$ protocol for input $x$ are given by $\rho_j$ for $j\in [m]$. Then, in the quantum register $\mathsf{Y}$, an honest Merlin should send many copies of the state $\rho_j$. Specifically, $\mathsf{Y}$ is partitioned into $m$ registers $\mathsf{Y}_j$, one for each original prover, and each $\mathsf{Y}_j$ should contain $k$ copies of $\rho_j$, for $k$ a carefully chosen polynomial. In other words, $\mathsf{Y}$ should contain the state $[\rho_1^{\otimes k}]_{\mathsf{Y}_1}\otimes\cdots\otimes[\rho_m^{\otimes k}]_{\mathsf{Y}_m}$. We further view each $\mathsf{Y}_j$ as a block of registers $(\mathsf{Y}_j^1, \dots, \mathsf{Y}_j^k)$ where $\mathsf{Y}_j^l$ should contain the $l$-th copy of $\rho_j$.

In the classical register $\mathsf{X}$, Merlin sends the classical ``consistency check'' string alluded to in Section~\ref{3_scn:intro}. Specifically, an honest Merlin prepares a quantum state in the computational basis, which intuitively corresponds to a bit string describing the $m$ classical probability distributions Arthur induces upon applying the measurement operation corresponding to Stage~1 of the \class{BellQMA} verification to each of the optimal proofs $\rho_j$, respectively. More formally, we partition $\mathsf{X}$ into $mr$ registers $\mathsf{X}_j^i$ corresponding to each of the $j\in[m]$ provers and $i\in[r]$ POVM outcomes per prover. The content of $\mathsf{X}_j^i$ should be $p_j(i) := \ip{\Pi_j(i)}{\rho_j}$, truncated to $\alpha$ bits of precision ($\alpha$ polynomially bounded), such that $\sum_{i=1}^r p_j(i)=1$. For example, if the $j$-th prover's proof was the single qubit state $\rho_j=\ketbra{0}{0}$, with $\Pi_j(1)=\ketbra{0}{0}$ and $\Pi_j(2)=\ketbra{1}{1}$, then $\mathsf{X}_j = (1, 0)$.

Of course, Merlin may elect to be dishonest and choose not to send a proof of the above form to Arthur by, e.g., sending a quantum state which is entangled across the registers $(\mathsf{X}, \mathsf{Y})$. To catch this, our \class{QMA} protocol is defined as follows:

%\begin{center}
%\underline{\class{QMA} Protocol}
%\end{center}

\begin{mylist}\parindent
\item[1.] Merlin sends Arthur a quantum state in registers $(\mathsf{X},\mathsf{Y})$, for $\mathsf{X}$ and $\mathsf{Y}$ defined as above.

\item[2.]
\textbf{Force $\mathsf{X}$ to be classical.} Arthur measures register $\mathsf{X}$ in the computational basis and reads the measurement outcome. This forces $\mathsf{X}$ to essentially be a classical register of bits, and destroys any entanglement or correlations between $\mathsf{X}$ and $\mathsf{Y}$.

\item[3.]
\textbf{$\mathsf{X}$ should contain probability distributions.} Arthur checks whether the content of registers $\mathsf{X}_j$ form a probability distribution $p_j$. Arthur rejects if this is not the case.

\item[4.]
\textbf{Consistency check: Can the quantum states in $\mathsf{Y}$ reproduce the distributions in $\mathsf{X}$?} Arthur picks independently and uniformly at random, an index $j \in [m]$ and another index $i \in [r]$. He applies the measurement $\{ \Pi_j(i) \}_{i=1}^r$ separately to each register $\mathsf{Y}_j^1, \dots, \mathsf{Y}_j^k$, and counts the number of times outcome $i$ appears, which we denote henceforth as $n_j(i)$. Arthur rejects if
\begin{equation}
\left\arrowvert \frac{n_j(i)}{k} - p_j(i) \right\arrowvert \ge \frac{1}{p},
\end{equation}
for $p$ a carefully chosen polynomial.

\item[5.]
\textbf{Run Stage 2 of the \class{BellQMA} verification and repeat for error reduction.} For each prover $j$, Arthur samples an outcome from $[r]$ according to the distribution in $(\mathsf{X}_j^1, \dots, \mathsf{X}_j^r)$, and runs Stage 2 of the \class{BellQMA} verification on the resulting set of samples. He repeats this process independently a polynomial number of times $q$, and accepts if and only if the \class{BellQMA} procedure accepts on the majority of the runs.
\end{mylist}

Let us discuss the intuition behind the verification procedure above. The key is Step 4, where Arthur cross-checks that the classical distributions sent in $\mathsf{X}$ really can be obtained by measuring $m$ quantum proofs, which for an honest Merlin should be unentangled. In this sense, our protocol can alternatively be viewed as using \emph{quantum} proofs ($\mathsf{Y}$) to check validity of a \emph{classical} proof ($\mathsf{X}$). Intuitively, the reason why entanglement in $\mathsf{Y}$ does not help a dishonest Merlin in Step 3 is due to the local nature of Arthur's checks/measurements. Finally, once Arthur is satisfied that $\mathsf{X}$ contains valid distributions, he runs Step 5. We remark that repetition is used here in order to boost the probability of acceptance in the $x\in A_{\yes}$ case to exponentially close to $1$, which is required to separate it from the $x \in A_{\no}$ case, where the probability of catching a dishonest Merlin is only inverse polynomially bounded away from $1$. Once such a gap exists, standard error reduction techniques~\cite{KW00,MW05} (see Section~\ref{0_sscn:classes}) can be used to further improve completeness and soundness parameters.

To formally analyze completeness and soundness of the $\class{QMA}$ protocol, we assign the following values to the parameters, all of which are polynomial in $n$ in our setting:
\begin{equation}
q = 50n \qquad \text{and} \qquad p = 20mr \qquad \text{and} \qquad k = 5p^3\qquad \text{and} \qquad \alpha = 20nmr.
\end{equation}

\paragraph{Completeness.} Intuitively, when $x\in A_{\yes}$, Merlin passes Step 4 with probability exponentially close to $1$ since he has no incentive to cheat --- he can send an unentangled proof in Step 1 to Arthur corresponding to the optimal proofs $\rho_j$ in the \class{BellQMA} protocol, such that the expected value of $n_j(i)/k$ is indeed $p_j(i)$. Arthur's checks in Step 4 are then independent local trials, allowing a Chernoff bound to be applied. We then show that Merlin passes each run in Step 5 with constant probability, and applying the Chernoff bound a second time yields the desired completeness exponentially close to $1$ for the protocol.

To state this formally, suppose Merlin is honest and sends registers $(\mathsf{X}, \mathsf{Y})$ in the desired form, i.e., $\mathsf{X}_j^i$ contains $p_j(i)=\ip{\Pi_j(i)}{\rho_j}$ up to $\alpha$ bits of precision, and $\mathsf{Y}_j^l$ contains $\rho_j$. Then, the expected value of the random variable $n_j(i)$ is $\E[n_j(i)] = k \ip{\Pi_j(i)}{\rho_j}$, which is equal to $k\cdot p_j(i)$ up to the error incurred by representing $p_j(i)$ using $\alpha$ bits of precision. In other words,
\begin{equation}\label{3_eqn:precision}
    \abs{ \frac{\E[n_j(i)]}{k} - p_j(i)} < \frac{1}{2^\alpha} < \frac{1}{2p}.
\end{equation}
We can hence upper bound the probability of rejecting in Step 3 by
\begin{equation}
\Pr \left[ \left\arrowvert \frac{n_j(i)}{k} - p_j(i) \right \arrowvert \ge \frac{1}{p} \right] < \Pr \left[ \left\arrowvert \frac{n_j(i)}{k} - \frac{\E[n_j(i)]}{k} \right \arrowvert \ge \frac{1}{2p} \right]  \le 2\exp\left(-\frac{5p}{4}\right),
\end{equation}
where the first inequality follows from Equation~(\ref{3_eqn:precision}) and the second from the Chernoff bound. Thus, Merlin passes Step 4 with probability exponentially close to $1$.

We now turn to the final step. Since $x\in A_{\yes}$, we know that the optimal distributions, denoted $q_j := \left(\ip{\Pi_j(1)}{\rho_j}, \dots, \ip{\Pi_j(r)}{\rho_j}\right)$ for $j\in[m]$, obtained in Stage 1 of the original \class{BellQMA} protocol are now accepted in Stage 2 with probability at least $2/3$. However, in our case, Merlin was only able to specify each $q_j$ up to $\alpha$ bits of precision per entry as the distributions $p_j$. To analyze how this affects the probability of acceptance, let $P_j$ and $Q_j$ be diagonal operators with entries $P_j(i,i)=p_j(i)$ and $Q_j(i,i)=\ip{\Pi_j(i)}{\rho_j}$, respectively. Letting $C_{\rm accept}$ denote the POVM element corresponding to outcome {\it accept} in Stage 2 of the BellQMA protocol, we thus bound the change in acceptance probability by:
\begin{eqnarray}
    \left| \tr\left[C_{\rm accept} \left(\bigotimes_{j=1}^m P_j - \bigotimes_{j=1}^m Q_j\right)\right] \right| &\leq& \bigg\Vert \bigotimes_{j=1}^m P_j - \bigotimes_{j=1}^m Q_j \bigg\Vert_{\textup{tr}}\\
    &\le& \sum_{j=1}^m \trnorm{P_j - Q_j} \\
    &=& \sum_{j=1}^m \sum_{i=1}^r |p_j(i) - \ip{\Pi_j(i)}{\rho_j}|\\
    &\le& \frac{mr}{2^{20nmr}},
\end{eqnarray}
where the first inequality follows from the fact that $|\tr(AB)| \le \snorm{A} \cdot \trnorm{B}$ and the second inequality follows from Lemma~\ref{3_lem:tracenorm}. Therefore, the probability of success for each of the $q$ runs of the \class{BellQMA} protocol in Step 5 is at least
\begin{equation}
\left(\frac{2}{3} - \frac{mr}{2^{20nmr}}\right) > 0.6.
\end{equation}
Since each run is independent, applying the Chernoff bound yields that Arthur accepts Merlin's proof in Step 5 with probability at least $1 - 2\exp( -0.02q)$, as desired. There may be some error incurred in sampling, which can be assumed to be exponentially small so that the success probability of each run is still at least $0.6$.

\paragraph{Soundness.} We now prove that when $x\in A_{\no}$, a dishonest Merlin can win with probability at most inverse polynomially bounded away from $1$. To show this, we bound the probability of passing Step 4 by relating the quantity $p_j(i)$ to the expected value of $n_j(i)/k$, and then apply the Markov bound. The desired relationship follows by observing first that the expected value of $n_j(i)/k$ is precisely the probability of obtaining outcome $i$ when measuring proof $j$ of some (honest) unentangled strategy, followed by arguing that the distribution $p_j$ must hence be far from this latter (honest) distribution if Merlin is to pass Step 5 with probability at least $1/2$ (since $x\in A_{\no}$). Combining these facts, we find that Arthur detects a cheating Merlin with inverse polynomial probability in Step 4.

More formally, let the quantum register $\mathsf{Y}_j$ contain an arbitrary quantum state $\sigma_j$ whose reduced states in registers $\mathsf{Y}_j^l$ for $l\in[k]$ are given by $\sigma_j(l)$, and define
\begin{equation}
\xi_j := \frac{1}{k}\sum_{l=1}^k\sigma_j(l).
\end{equation}
By the linearity of expectation, the expected value of the random variable $n_j(i)/k$ is
\begin{equation}
\E\left[\frac{n_j(i)}{k}\right] = \frac{1}{k}\sum_{l=1}^k \ip{\Pi_j(i)}{\sigma_j(l)} = \ip{\Pi_j(i)}{\xi_j}.
\end{equation}
Our goal is to lower bound the expression
\begin{equation}
    \Pr \left[ \left\arrowvert \frac{n_j(i)}{k} - p_j(i) \right \arrowvert \ge \frac{1}{p} \right]. \label{3_eqn:prob}
\end{equation}
To achieve this, we first substitute $p_j(i)$ above with a quantity involving $\E[n_j(i)/k]$, and then apply the Markov bound.

To relate $\E[n_j(i)/k]$ to $p_j(i)$, we first remark that in order for Merlin to pass each run of Step 5 with probability exponentially close to $1$, he must send probability distributions $p_j$, which are accepted by Stage 2 of the \class{BellQMA} verification with probability at least $1/2$. Let
\begin{equation}
q_j(i):= \ip{\Pi_j(i)}{\xi_j}.
\end{equation}
Let us imagine a $\class{BellQMA}$ protocol where the $j$-th Merlin sends $\xi_j$ as his quantum proof. Since $x \in A_{\no}$, by the soundness property of the $\class{BellQMA}(m)$ proof system, the success probability of the Merlins is at most $1/3$. In other words, sampling outcomes from the probability distributions $(q_j(1), \ldots, q_j(r))$ and then running the second stage of the \class{BellQMA} verification will yield outcome {\it accept} with probability at most $1/3$. Also, observe that
\begin{equation}
\E\left[\frac{n_j(i)}{k}\right] = q_j(i).
\end{equation}
It follows that by letting $P_j$ and $Q_j$ be diagonal operators with the probability vectors $p_j$ and $q_j$ on their diagonals, respectively, and $C_{\rm accept}$ the POVM element corresponding to outcome {\it accept} in Stage 2 of the BellQMA protocol, we have
\begin{equation}
\frac{1}{10} < \abs{\tr\left[C_{\rm accept} \left(\bigotimes_{j=1}^m P_j - \bigotimes_{j=1}^m Q_j\right)\right]}\le \bigg\Vert \bigotimes_{j=1}^m P_j - \bigotimes_{j=1}^m Q_j \bigg\Vert_{\textup{tr}} \le \sum_{j=1}^m \trnorm{P_j - Q_j}.
\end{equation}
Here, the (loose) lower bound of $1/10$ comes from the following two observations. First, the distributions represented by the $Q_j$'s are derived from a \class{BellQMA} protocol and therefore achieve a success probability at most $1/3$ by the soundness property of the \class{BellQMA} verification. Second, the distributions represented by the $P_j$'s have to achieve a success probability strictly greater than $1/2$ per run to guarantee that Merlin wins Step 5 with probability exponentially close to $1$. Combining these two, we get that the difference between the success probabilities obtained by distributions $\set{P_j}$ and $\set{Q_j}$ should be at least $1/6$ modulo the error incurred due to finite precision when encoding the distributions $p_j$. The use of the constant $1/10$ overcompensates for this precision error.  Hence, there exists a $j$ such that
\begin{equation}
\trnorm{P_j - Q_j} = \sum_{i=1}^r |p_j(i) - q_j(i)| \ge \frac{1}{10m},
\end{equation}
implying the existence of an $i$ such that
\begin{equation}\label{3_eqn:relation}
    |p_j(i) - q_j(i)| \ge \frac{1}{10mr}.
\end{equation}
This is our desired relationship between $p_j(i)$ and $\E[n_j(i)/k]=q_j(i)$. Note that the probability of picking pair $(i,j)$ in Step 4 is $1/mr$.

We now substitute this relationship into Equation~(\ref{3_eqn:prob}) and apply the Markov bound. Specifically, choose $i$ and $j$ as in Equation~(\ref{3_eqn:relation}), and assume that $p_j(i) > \ip{\Pi_j(i)}{\xi_j}$. Then, we have
\begin{equation}
\Pr \left[ \left\arrowvert \frac{n_j(i)}{k} - p_j(i) \right\arrowvert < \frac{1}{p} \right] < \Pr \left[ \frac{n_j(i)}{k} - \E\left[\frac{n_j(i)}{k}\right] > \frac{1}{10mr} - \frac{1}{p} \right] \le 1-\frac{1}{2p}.
\end{equation}

\noindent The case of $p_j(i) < \ip{\Pi_j(i)}{\xi_j}$ is similar. We conclude that a dishonest Merlin is caught in Step 4 with probability at least $1/2p$. Therefore, the probability that Arthur proceeds to Step 5 is upper bounded by
\begin{equation}
\left(\frac{1}{mr}\right)\left(1 - \frac{1}{20mr}\right) + \left(1 - \frac{1}{mr}\right)(1) = 1 - \frac{1}{20m^2r^2},
\end{equation}
where the first term represents the case where Arthur selects the correct pair $(i,j)$ to check, and the second term the complementary case, in which we assume the cheating prover can win with probability $1$. Hence the overall success probability of Merlin is at most $1 - 1/20m^2r^2$.

Finally, as mentioned before, since $m$ and $r$ are polynomially bounded functions, we have that the completeness is exponentially close to $1$, while the soundness is bounded away from $1$ by an inverse polynomial. By known error reduction techniques for \class{QMA} protocols~\cite{KW00,MW05}, one can amplify the completeness and soundness errors to be exponentially close to 0. This proves our desired containment.

%------------------------------------------------------------------------------------------------------------------------------------------------------%

\section{Perfect parallel repetition for $\class{SepQMA}(\poly)$}\label{3_scn:parrep}

Using cone programming, we now show Theorem~\ref{3_thm:parrep}, i.e., that the class $\class{SepQMA}(m)$ admits perfect parallel repetition. Recall now that for $C$ the measurement operator corresponding to outcome {\it accept}, the maximum success probability of the Merlins in any $\class{QMA}(m)$ protocol can be written as the maximum of $\ip{\rho}{C}$, where $\rho$ is a density operator in the cone $\sep{\X_1, \dots, \X_m}$. This is a simple cone program and can be written as the following primal-dual pair:
\begin{center}
  \begin{minipage}{1.5in}
    \centerline{\underline{Primal problem ($\textup{P}$)}}\vspace{-7mm}
    \begin{align*}
			\text{max}\quad & \ip{\rho}{C} \\
  		\text{s.~t.}\quad & \tr(\rho) = 1,\\
  		& \rho \in \sep{\X_1, \dots, \X_m},
  	\end{align*}
  \end{minipage}
  \hspace*{15mm}
  \begin{minipage}{1.5in}
    \centerline{\underline{Dual problem ($\textup{D}$)}}\vspace{-7mm}
		\begin{align*}
			\text{min}\quad & t\\
  		\text{s.~t.}\quad & t\I_{\X} = C + W,\\
  		& W \in \Sep{(\X_1, \ldots, \X_m)}^\ast,
  	\end{align*}	
  \end{minipage}
\end{center}
where $\X$ denotes $\X_1 \otimes \cdots \otimes \X_m$, and $\Sep{(\X_1, \ldots, \X_m)}^\ast$ is the dual cone defined as
\begin{equation}
    \Sep{(\X_1, \ldots, \X_m)}^\ast := \left\{W: \ip{\rho}{W} \ge 0 \text{ for all } \rho \in \sep{\X_1, \dots, \X_m} \right\}.
\end{equation}
(Note that $\Sep{(\X_1, \ldots, \X_m)}^\ast$ contains the set of \emph{entanglement witnesses} in the theory of entanglement, see~\cite{HHHH09}.) Moreover, the use of ``maximum'' and ``minimum'' is justified in the above programs since $\overline{\rho} = \frac{\I_{\X}}{\dim(\X)}$ and $(\overline{t}, \overline{W}) = (2, 2\I_{\X} - C)$ are strictly feasible solutions for $(\textup{P})$ and $(\textup{D})$, respectively~\cite{GB02,GB03,GB05} (i.e.\ strong duality (Theorem~\ref{3_thm:stduality}) holds).

%For the remainder of the section, it will be convenient for us to distinguish two instances of $\class{SepQMA}(m)$ protocols as the \emph{first} and \emph{second} protocol, which we will subsequently run in parallel.
Given two protocols, the corresponding cone programs are completely specified by Arthur's POVM corresponding to outcome {\it accept} and the underlying cone:
\begin{equation}
(C_1, \sep{\X_1, \dots, \X_m}) \text{ and } (C_2, \sep{\Y_1, \dots , \Y_m}),
\end{equation}
while the parallel repetition protocol is specified by
$(C_1 \otimes C_2, \sep{\X_1 \otimes \Y_1, \dots, \X_m \otimes \Y_m})$.

To show Theorem~\ref{3_thm:parrep}, note first that if $\rho_1$ and $\rho_2$ are optimal solutions of the primal problems associated with the two individual protocols, then $\rho_1 \otimes \rho_2$ is a feasible solution of the primal problem associated with the parallel repetition protocol. Therefore the success probability of the parallel repetition is at least the product of the success probabilities of the individual protocols. We now show that \emph{no} other strategy for the prover can perform better than this honest strategy. To do so, we demonstrate a feasible solution for the dual problem associated with the parallel repetition protocol attaining the same objective value.

More formally, let $(t_1, W_1)$ and $(t_2, W_2)$ be respective dual optimal solutions corresponding to two protocols. We show that $(t_1 \cdot t_2, W)$ is a dual feasible solution corresponding to the two-fold repetition of protocols for some choice of $W \in \Sep{(\X_1 \otimes \Y_1, \dots, \X_m\otimes\Y_m)}^\ast$. To do so, we first require the following lemma.

\begin{lemma}\label{3_lem:dualconstruction}
For complex Euclidean spaces $\X_1, \ldots, \X_m,\Y_1, \ldots, \Y_m$:
\begin{itemize}
\item $\Sep{(\X_1, \ldots, \X_m)}^\ast \otimes \Sep(\Y_1, \ldots, \Y_m) \subseteq \Sep{(\X_1 \otimes \Y_1, \ldots, \X_m \otimes \Y_m)}^\ast$, and
\item $\Sep(\X_1, \ldots, \X_m) \otimes \Sep{(\Y_1, \ldots, \Y_m)}^\ast \subseteq \Sep{(\X_1 \otimes \Y_1, \ldots, \X_m \otimes \Y_m)}^\ast$.
\end{itemize}
\end{lemma}
\begin{proof}
We prove the first condition as the second is similar. Fix $W \in \Sep{(\X_1, \ldots, \X_m)}^\ast$ and $C \in \Sep(\Y_1, \ldots, \Y_m)$. Then for $S \in \Sep(\X_1 \otimes \Y_1, \ldots, \X_m \otimes \Y_m)$, we have
\begin{equation}
\ip{W \otimes C}{S} = \inner{W}{\tr_{\Y}\left[S(\I_{\X} \otimes C)\right]} \geq 0,
\end{equation}
if $\tr_{\Y}\left[S(\I_{\X} \otimes C)\right] \in \Sep(\X_1, \ldots, \X_m)$. %Therefore, it suffices to prove that $\tr_{\Y}\left[S(\I_{\X} \otimes C)\right] \in \Sep(\X_1, \ldots, \X_m)$.
To this end, let
\begin{equation}
S = \sum_{i=1}^k \bigotimes_{l=1}^m \rho_i(l) \qquad \text{ and } \qquad C = \sum_{j=1}^{k'} \bigotimes_{l=1}^m \sigma_j(l),
\end{equation}
where $\rho_i(l) \in \pos{\X_l \otimes \Y_l}$ and $\sigma_j(l) \in \pos{\Y_l}$ for all $i \in [k]$, $j \in [k']$, and $l \in [m]$. Now we can write $\tr_{\Y}\left[ S\left(\I_{\X} \otimes C \right) \right]$ as
\begin{eqnarray}
\tr_{\Y} \left[ \left( \sum_{i=1}^k \bigotimes_{l=1}^m \rho_i(l)\right) \left( \I_{\X} \otimes \sum_{j=1}^{k'} \bigotimes_{l=1}^m \sigma_j(l) \right) \right] & &= \\
& & \hspace{-5mm} \sum_{i=1}^k \sum_{j=1}^{k'} \bigotimes_{l=1}^m \tr_{\Y_l} \left[ \rho_i(l) \left(\I_{\X_l} \otimes \sigma_j(l) \right) \right].\nonumber
\end{eqnarray}
Hence, $\tr_{\Y}\left[ S\left(\I_{\X} \otimes C \right) \right]\in\sep{\X_1, \ldots, \X_m}$  since $\tr_{\Y_l} \left[ \rho_i(l) \left(\I_{\X_l} \otimes \sigma_j(l) \right) \right]$ is positive semidefinite for all $i,j,l$. The latter follows since for positive semidefinite $A_{\spa{X}\otimes\spa{Y}}$ and $B_{\spa{Y}}$,
\begin{equation}
    \trace_{\spa{Y}}(A_{\spa{X}\otimes\spa{Y}}I_{\spa{X}}\otimes B_{\spa{Y}})=\trace_{\spa{Y}}(I_{\spa{X}}\otimes B^{\frac{1}{2}}_{\spa{Y}}A_{\spa{X}\otimes\spa{Y}}I_{\spa{X}}\otimes B^{\frac{1}{2}}_{\spa{Y}})\succeq 0,
\end{equation}
which follows since $C^\dagger D C\succeq 0$ if $D\succeq 0$. This concludes the proof.
\end{proof}

\noindent We use Lemma~\ref{3_lem:dualconstruction} to construct two operators in $\Sep{(\X_1 \otimes \Y_1, \ldots, \X_m \otimes \Y_m)}^\ast$, the appropriate convex combination of which is the dual feasible solution we are seeking. Specifically, observe first that since for the two instances of the $\class{SepQMA}(m)$ protocol, we have $C_1 \in \sep{\X_1, \ldots, \X_m}$ and $C_2 \in \sep{\Y_1, \ldots, \Y_m}$, and since $\I_{\X}$ and $\I_{\Y}$ are fully separable operators, it follows that
\begin{equation}
s_1 \I_{\X} + C_1 \in \Sep(\X_1, \ldots, \X_m) \qquad \text{and} \qquad s_2 \I_{\Y} + C_2 \in \Sep(\Y_1, \ldots, \Y_m)
\end{equation}
for all $s_1, s_2 \geq 0$. Using Lemma~\ref{3_lem:dualconstruction}, we thus obtain operators
\begin{equation}\label{3_eqn:dualeq1}
(t_1 \I_{\X} - C_1) \otimes (t_2 \I_{\Y} + C_2) \in \Sep{(\X_1 \otimes \Y_1, \ldots, \X_m \otimes \Y_m)}^\ast
\end{equation}
and
\begin{equation} \label{3_eqn:dualeq2}
(t_1 \I_{\X} + C_1) \otimes (t_2 \I_{\Y} - C_2) \in \Sep{(\X_1 \otimes \Y_1, \ldots, \X_m \otimes \Y_m)}^\ast.
\end{equation}
Here we have used the fact that $W=t_1 \I_{\X} - C_1\in \Sep{(\X_1, \ldots, \X_m)}^\ast$ since $(t_1,W)$ is by assumption the optimal dual solution for the first protocol (and similarly for the second protocol). Since $\Sep{(\X_1 \otimes \Y_1, \ldots, \X_m \otimes \Y_m)}^\ast$ is a convex cone, it follows that the average of Equations~\eqref{3_eqn:dualeq1} and~\eqref{3_eqn:dualeq2} yields the desired operator
\begin{equation}
W := t_1 \cdot t_2 \, \I_{\X \otimes \Y} -  C_1 \otimes C_2 \in \Sep{(\X_1 \otimes \Y_1, \ldots, \X_m \otimes \Y_m)}^\ast.
\end{equation}
We conclude that $\left(t_1 \cdot t_2, W \right)$ is a feasible solution of the dual problem associated with parallel repetition of protocols with objective value $t_1 \cdot t_2$ as desired. This concludes the proof of Theorem~\ref{3_thm:parrep}.

%%------------------------------------------------------------------------------------------------------------------------------------------------------%
%\section*{Acknowledgements}
%%------------------------------------------------------------------------------------------------------------------------------------------------------%
%
\noindent \emph{Acknowledgements for this chapter.} We thank Richard Cleve, Tsuyoshi Ito, Iordanis Kerenidis, Ashwin Nayak, Oded Regev, and Levent Tun{\c c}el for insightful discussions. We also thank LIAFA, Paris for their hospitality, where part of this work was completed.

%\addcontentsline{toc}{part}{Non-classical quantum correlations}
%======================================================================
\chapter{Signatures of non-classicality in mixed-state quantum computation}\label{chap:dqc}

\emph{This chapter is based on~\cite{DG08}:}\\

\vspace{-4mm}
\noindent A. Datta and S. Gharibian. Signatures of nonclassicality in mixed-state quantum computation. \emph{Physical Review A}, 79:042325, 2009, DOI: 10.1103/PhysRevA.79.042325, \copyright~2009 American Physical Society, pra.aps.org.

\vspace{3mm}
\noindent In this chapter, we investigate signatures of non-classicality in quantum states,
in particular, those involved in the DQC1 model of mixed-state
quantum computation~\cite{kl98}.
To do so, we consider two known non-classicality criteria. The
first quantifies disturbance of a quantum state under locally
noneffective unitary operations (LNU), which are local unitaries
acting invariantly on a subsystem. The second quantifies
measurement induced disturbance (MID) in the eigenbasis of the
reduced density matrices. We study the role of both figures of
non-classicality in the exponential speedup of the DQC1 model and
compare them \textit{vis-a-vis} the interpretation provided in
terms of quantum discord. In particular, we prove that a non-zero
quantum discord implies a non-zero shift under LNUs. We also use
the MID measure to study the locking of classical correlations~\cite{dhlst04} using two mutually
unbiased bases (MUB). We find the MID measure to exactly
correspond to the number of locked bits of correlation.

\section{Introduction and results}

        A thorough understanding of classical and quantum correlations
underlies their successful exploitation in quantum information
science. Characterizing the relative roles and abilities of these two forms of
correlations in performing specific computational and information
processing tasks would be a valuable advance in the field.
Substantial progress in this direction has already been achieved.
The role of entangled states in quantum information processing and
computing is quite well studied. Jozsa and Linden \cite{jozsa03a}
showed that multipartite entanglement must grow unboundedly with
the problem size if a pure-state quantum computation is to attain
an exponential speedup over its classical counterpart. In the
context of information processing, Masanes has shown
\cite{masanes06a} that all bipartite entangled states can enhance
the teleporting power of some other state. In spite of these
successes, there are instances of quantum computations where the
quantum advantage cannot be attributed to entanglement. Meyer has
presented a quantum search algorithm that uses no entanglement
\cite{meyer00a}. Instances are also known of oracle based problems
that can be solved without entanglement, yet with certain
advantages over the best known classical algorithms
\cite{biham04a,kenigsberg06a}.

Given this scenario, it becomes a logical necessity to study the
essentialness of entanglement in quantum information science. %The
%oldest signature of quantum behavior has been nonlocality.
%Interestingly, it is well known that quantum nonlocality and
%entanglement are not equivalent notions
%\cite{bennett99c},\cite{methot07a}. Entanglement stems from the
%superposition principle, or the amplitude description of quantum
%mechanics. This description is, however, not one that uniquely
%defines quantum mechanics. Consequently, it should not be a
%surprise that entanglement cannot capture the whole power of
%quantum mechanics. This provides a significant motivation for
%studying alternative certificates of quantum behavior.
A realistic motivation is that provided by mixed-state
quantum computation. Pure states in a quantum computation
inevitably get mixed due to decoherence. One way to
address this issue would be to study the prospects of quantum
computational speedup with mixed states themselves \cite{asv00}.
NMR quantum computation provides a good scenario for this. As a
simplified model for this, Knill and Laflamme proposed the DQC1 or
the `power of one qubit' model \cite{kl98}. Though not believed to
be as powerful as a pure-state quantum computer, it is believed to
provide an exponential speedup over the best known classical
algorithm for estimating the normalized trace of a unitary matrix.
The DQC1 model was found to have a limited amount of (bipartite)
entanglement that does not increase with the system size.
Additionally, for certain parameter settings, there is no
distillable entanglement present whatsoever, and yet the model
retains its exponential advantage. In this latter case the state
has a positive partial transpose, and thus possesses, at most,
just bound entanglement \cite{datta05a}. Looking for a more
satisfactory explanation for the exponential speedup, the quantum
discord \cite{ollivier01a,henderson01a} was calculated, of
which the amount found was a constant fraction of the maximum
possible \cite{datta08a}, regardless of the parameter settings for
the model. In this chapter, we study two alternative methods of
studying the quantum behavior of quantum computational and information tasks.

\paragraph{Our results:} This chapter studies the non-classical correlations found in the DQC1 states for trace estimation, as well as those used in the locking of classical correlations~\cite{dhlst04}, with respect to two quantification schemes abbreviated as LNU and MID.

\vspace{2mm}
\noindent\textbf{1. Locally noneffective unitaries (LNU).} \emph{Locally noneffective unitary} operations (LNU) have previously been
studied with the aim of developing an entanglement detection
criterion~\cite{f06,gkb08} (see Section for a definition~\ref{4_scn:lc}). Here, we study whether LNU can be used to quantify non-classicality, motivated by the disturbance
of a quantum state under unitary operations.  Specifically, we employ LNU in analyzing the DQC1 model, which has previously been studied using the quantum discord. Thus, we compare these two certificates of non-classicality, with the aim of contrasting \emph{disturbance
under measurement} with \emph{disturbance under unitary operations}. We also study a mixed-state task in the setting of quantum communication known as \emph{locking}~\cite{dhlst04}, which uses two mutually unbiased bases (MUB) to \emph{lock} classical correlations in a quantum state. For both tasks, we find that LNU do not indicate a high level of correlations.

\vspace{2mm}
\noindent\textbf{2. Measurement-Induced Disturbance (MID).} We then study the DQC1 model using the
Measurement-Induced Disturbance (MID) measure~\cite{Luo08} in
Section~\ref{S:middqc}. Regarding the MID measure, in Reference~\cite{Luo08}, a preliminary analysis
of the DQC1 model was begun. Here, we extend this analysis to the
entire parameter range for the DQC1 model, including those which
limit the DQC1 state to being at most bound entangled. This latter
case is of particular interest due to the lack of distillable
entanglement. We also study the task of locking. For the latter, the value of the MID measure is exactly the number of locked bits of correlation in the state. %Considering the
%same construction with more than two MUBs, the MID measure
%portends superior locking abilities, though they must involve MUBs
%more general than those based on Latin squares and generalized
%Pauli matrices~\cite{BW07}.

\paragraph{Discussion.} With regards to the LNU distance, we find (Equation~(\ref{4_eqn:dqclc})) that
there is little non-classicality in the $n+1$ qubit DQC1 state.
This behavior is very similar to that of negativity~\cite{VW02} in the DQC1
model which was used to characterize its
entanglement~\cite{datta05a}. The crucial difference is that the
bipartite split chosen in Section \ref{S:lcdqc1} is separable, and
therefore exhibits no entanglement at all. As the LNU distance
vanishes exponentially quickly with growing $n$, one is
hard-pressed to relegate the role of the resource exponentially
speeding up the DQC1 model to it. Similarly, the LNU distance
suggests vanishing non-classicality in the case of locking of
classical correlations in quantum states.

We find the MID measure, on the other hand, to be considerably more
satisfactory. The zero-entanglement split in the DQC1 model is
shown to have a non-zero amount of non-classicality as per the MID
measure. The magnitude of this measure, as shown in
Figure~(\ref{measdisc}), is a constant fraction of its maximum
possible value. Further, the MID measure performs well in quantifying non-classicality in the scenario of locking classical correlations in quantum states.
Further studies in this direction
are required before a comprehensive conclusion can be reached.

\paragraph{Organization of chapter.} We begin in Section~\ref{4_scn:lc} by defining LNU. Section~\ref{S:lcdqc1} studies LNU in the DQC1 model. Section~\ref{S:discordandlc} shows a one-way relationship between LNU and the quantum discord. In Section~\ref{S:middqc}, we define the MID measure, and use it to study the DQC1 model in Section~\ref{4_sscn:MIDDQC1}. In Section~\ref{S:locking}, we study both MID and LNU in the context of locking.

\paragraph{Notation.} Throughout this chapter, for $\dens$ we denote the dimensions of $\A$ and $\B$ as $M$ and $N$, respectively. All designations of a density matrix without any subscripts refers to a bipartite state. For example, $\rho$ stands for $\rho_{AB}$.

\section{Locally noneffective unitary (LNU) operations }
\label{4_scn:lc}

We begin by introducing locally noneffective unitary operations
(LNU), first proposed under the name local \emph{cyclic}
operations~\cite{f06}. For this, consider a bipartite quantum
state $\ro\in\dens$, shared between $A$ and $B$ such that
$\rho_A=\tr_B(\ro)$ and $\rho_B=\tr_A(\ro)$. Suppose now that
Alice performs a local unitary $U_A$ that does not change her
subsystem, that is, $\rho_A = U_A \rho_A U^{\dag}_A$, or
equivalently
 \be
 \label{4_eqn:lccond}
 [\rho_A,U_A] =0.
 \ee
This action can, however, affect the state of the total system,
such that if we define $\rof:=(U_A \otimes {I}_B) \ro(U_A
\otimes {I}_B)^{\dag}$, it is possible that $\ro\neq\rof$.
Unitaries satisfying Equation~(\ref{4_eqn:lccond}) are called
LNU~\cite{f06}. To quantify the difference between $\ro$ and
$\rof$, we use %the distance
% \be \label{4_eqn:fuDistance}
%    \fu := \frac{1}{\sqrt{2}}\fnorm{\ro-\rof},
% \ee
 \ben
  \label{lcmeasure}
 \fum := \max_{\scriptsize\begin{array}{c}
 U_A:\\
  $[$\rho_A$,$U_A$]=0$ \\
\end{array}}
  \frac{1}{\sqrt{2}}\fnorm{\ro-\rof}
  = \max_{\scriptsize\begin{array}{c}
 U_A:\\
  $[$\rho_A$,$U_A$]=0$ \\
\end{array}} \sqrt{\tr(\ro^2)-\tr(\ro\rof)},
 \een
where $\fnorm{A}=\sqrt{\tr(A^\dg A)}$ denotes the Frobenius norm.
From the latter expression, it is clear that $0\leq \fum \leq 1$.

For any product state $\rho_{prod}:=\rho_A\otimes\rho_B$,
$\fuma{\rho_{prod}}=0$. Closed form expressions for $\fum$ are known for (pseudo)pure states and Werner states~\cite{gkb08}. As with the quantum discord, it is possible to have
$\fuma{\rho_{sep}}>0$ for certain separable states, implying
$\fum$ is not a non-locality measure. Recall that a separable state
$\rho_{sep}\in\dens$ is defined as one of the form
\begin{equation}
    \rho_{sep}:=\sum_k p_k\ketbra{a_k}{a_k}\otimes\ketbra{b_k}{b_k},
\end{equation}
where $\sum_k p_k=1$, and the $\ket{a_k}\in\mathcal{A}$ and
$\ket{b_k}\in\mathcal{B}$ are vectors of Euclidean norm $1$. For
two-qubit separable states, the maximum LNU distance attainable
is~\cite{f06}
\begin{equation}\label{4_eqn:ccbound}
\fuma{\rho_{sep}}\leq \frac{1}{\sqrt{2}}.
\end{equation}
%Similar bounds exist for dimensions up to $MN=9$~\cite{gkb08}, but
%are only known to be tight in the two-qubit case.

As an illustration, the maximum LNU distance for the two-qubit
isotropic state,
 \be
 \ro_{iso} = \frac{1-z}{4}I_4 +
 z\proj{\Psi},\;\;\;\;\;z\in [0,1]
  \ee
where $\ket{\Psi}=(\ket{00}+\ket{11})/\sqrt{2}$, is given by $\fuma{\ro_{iso}} = z$~\cite{gkb08}. By Equation~(\ref{4_eqn:ccbound}), we can conclude that the two-qubit
isotropic state is entangled for $z >1/\sqrt{2}$. The partial
transpose test, which in this case is necessary and sufficient,
shows that this state is actually entangled for all $z>1/3$,
showing that the LNU distance is weaker at detecting
entangled states than the former.

We remark that we have restricted our attention here to the case
where the LNU is applied to subsystem $A$ of $\ro$. Let us derive
%One can
%alternatively consider subsystem $B$ as the target subsystem.
a simple upper bound on $\fum$ which holds regardless
of which target subsystem we choose, and which proves useful throughout this chapter.

\begin{theorem}\label{4_thm:fuBound}
    For any $\ro\in\dens$,
    \be
        \fum\leq\sqrt{2\left(\tr(\rho^2)-\frac{1}{MN}\right)}.
    \ee
\end{theorem}
\begin{proof}
   Since $\fnorm{\rho-\frac{I}{MN}}$ is invariant under
unitary operations, we have via the triangle inequality that:
    \ben
        \fnorm{\rho-\rho_f}\leq \fnorm{\rho-\frac{I}{MN}} + \fnorm{\frac{I}{MN}-\rho_f}
        =2\fnorm{\rho-\frac{I}{MN}}  \nonumber
        =2\sqrt{\tr(\rho^2)-\frac{1}{MN}}
    \een
    Substituting this expression in Equation~(\ref{lcmeasure}) gives the desired result.
\end{proof}

Thus, if the purity of a state $\ro$ strictly decreases as a
function of the dimension, then $\fum\rightarrow0$ as
$MN\rightarrow\infty$.

\section{LNU in the DQC1 model}
\label{S:lcdqc1}

        We now study the non-classical features of the DQC1 model
of quantum computation, as quantified by $\fum$. The $n+1$ qubit
DQC1 state for given unitary $U_n\in\UU(\mathcal{B}^{\otimes n})$, as demonstrated in Figure~(\ref{4_fig:F:dqc1}), is given
by~\cite{datta05a}
% \begin{figure}
% \centerline{ \Qcircuit @C=.5em @R=-.5em {
%    & \lstick{\frac{1}{2}(I+\alpha Z)} & \gate{H} & \ctrl{1} & \meter & \push{\rule{0em}{4em}} \\
%    & & \qw & \multigate{4}{U_n} & \qw & \qw \\
%    & & \qw & \ghost{U_n} & \qw & \qw \\
%    \lstick{\mbox{$I_n/2^n$}} & & \qw & \ghost{U_n} & \qw & \qw \\
%    & & \qw & \ghost{U_n} & \qw & \qw \\
%    & & \qw & \ghost{U_n} & \qw & \qw \gategroup{2}{2}{6}{2}{.6em}{\{}
%}} \caption[The `power of one qubit' model]{The DQC1 circuit}
% \label{4_fig:F:dqc1}
%\end{figure}
%\begin{center}
\begin{figure}
\begin{center}
{\includegraphics[width=60mm,height=50mm,keepaspectratio,trim = 0mm 40mm 0mm 50mm, clip]{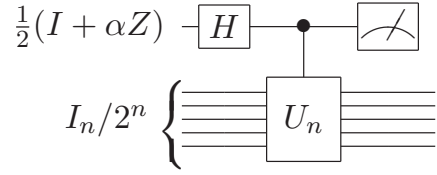}}
\caption[The `power of one qubit' model]{The DQC1 circuit} \label{4_fig:F:dqc1}
\end{center}
\end{figure}
%\end{center}
 \be
 \label{4_eqn:E:dqc1}
\dqc =\frac{1}{2^{n+1}}\left(
\begin{array}{cc}
  I_n & \alpha U^{\dg}_n \\
  \alpha U_n & I_n \\
\end{array}%
\right).
 \ee
We will consider the top qubit to be system $A$ on which our local
unitary acts and the remaining $n$ qubits as system $B$. The
reduced state is then
  \be
\rho_A = \tr_B(\dqc) = \frac{1}{2}\left(
\begin{array}{cc}
  1 & \alpha \tau^* \\
  \alpha \tau & 1 \\
\end{array}%
\right)
  \ee
with $\tau = \tr(U_n)/2^n$. %A single qubit $\mathrm{SU}(2)$
%operation on $A$ can be denoted as $U_A$.
For an arbitrary $\mathrm{SU}(2)$
unitary $U_A$ acting on $A$% from Equation~(\ref{su2})
, which we characterize as
\be
 \label{su2}
U_A=\left(
\begin{array}{cc}
  e^{i\phi}\cos\theta & e^{i\chi}\sin\theta \\
  -e^{-i\chi}\sin\theta   & e^{-i\phi}\cos\theta \\
\end{array}%
\right),
 \ee
the LNU condition of
Equation~(\ref{4_eqn:lccond}) requires that $\chi =
\frac{\pi}{2}-\arg(\tau)$ and either $\phi =0$ or $\theta=\pi/2$.
Both cases lead to the same final expression, so set
$\phi=0$. Via Equation~(\ref{lcmeasure}) and simple algebra, we hence have
% $$
%\tr(\rof\dqc) \!=\! \frac{1}{2^{n+1}}\left(1+\alpha^2 \cos^2\theta
%-\alpha^2\sin^2\theta \frac{\mathrm{Re}(\tr V)}{2^n}\right),
% $$
%where $V=e^{2i\chi}U^2_n$, which on further simplification gives
%via Equation~(\ref{lcmeasure}) a LNU distance of
 \begin{equation}
 \fua{\dqc}{\theta} =
 \frac{\alpha\sin\theta}{2^{(n+1)/2}}\sqrt{1-\frac{\mathrm{Re}(\tr(e^{-2i\arg\tau}U^2_n))}{2^n}}.
 \end{equation}
The now trivial maximization over all $\theta$ gives
 \ben
 \label{4_eqn:dqclc}
 \fumdqc = \frac{\alpha}{2^{(n+1)/2}}\sqrt{1-\frac{\mathrm{Re}(\tr(e^{-2i\arg\tau}U^2_n))}{2^n}}
            \leq  \frac{\alpha}{2^{n/2}}.
 \een
Here, we have used the rough estimate
$\mathrm{Re}(\tr(e^{2i\arg\tau}U^2_n))\geq -2^n$. For a two-qubit
pure state ($n=1, \alpha=1$), we thus have $\fumdqc \leq
1/\sqrt{2}$, which conforms with Equation~(\ref{4_eqn:ccbound}).
%We know
%that the quantity in the square root is, in general, not zero and
%so this state will typically have $\fumdqc>0$.
A typical instance
of the DQC1 circuit is provided by that of a random unitary $U_n$
in the DQC1 circuit of Figure~(\ref{4_fig:F:dqc1}). For such instances of
large enough Haar distributed unitaries, $\tr(U_n^2)$ is bounded
above by a constant with high probability ~\cite{diaconis03}.
Thus, the second term inside the square root in Equation~(\ref{4_eqn:dqclc})
is approximately zero, and
 \be
 \fumdqc \approx \frac{\alpha}{2^{(n+1)/2}}.
 \ee

This shows that the DQC1 state experiences very little disturbance
under LNU, and in fact this disturbance vanishes asymptotically as
$n$ grows. As discussed in the introduction, it would appear that
the quantum discord is better suited~\cite{datta08a} to
quantifying non-classicality in the DQC1 model. This, however,
raises the question of how the discord and LNU distance are
related, and whether the paradigms of `disturbance under
measurement' and `disturbance under unitary operations' lead to
differing notions of non-classicality. We explore these questions
in the following section.

Before closing, for completeness, we invoke Theorem~(\ref{4_thm:fuBound}) to show that the LNU distance is exponentially decreasing for \emph{any} other choice of bi-partitions $A$ and $B$ of the qubits in $\dqc$. In fact, since
% in calculating $\fumdqc$, one can
%alternatively choose to apply a LNU to subsystem $B$, or define
%different bi-partitions of $\dqc$ as the subsystems $A$ and $B$.
%Evaluating $\fumdqc$ directly in such cases unfortunately proves
%difficult. It turns out, however, that since
 \be\label{4_eqn:dqcMixedness}
    \tr(\dqc^2)=\frac{1+\alpha^2}{2^{n+1}},
 \ee
Theorem~(\ref{4_thm:fuBound}) immediately gives the same upper bound
of Equation~(\ref{4_eqn:dqclc}). %Thus, alternate bipartite splits
%cannot provide any significant increase in $\fumdqc$.

\section{Quantum discord \lowercase{\textit{vs}} LNU distance}
\label{S:discordandlc}

Motivated by the fact that both the quantum discord and the LNU
distance are aimed at capturing the non-classical features in a
quantum state via an induced disturbance, we seek an answer to the
question of whether one implies the other in any sense or not.
Here, we show that non-zero quantum discord implies a
non-zero LNU distance, but that the converse is not necessarily
true. We begin by recalling the definition of quantum discord.

%Given a quantum state $\rho\in\dens$, its quantum mutual
%information is defined as $\mathcal{I}(\rho) := S(\rho_A) +
%S(\rho_B) - S(\rho)$. The quantum mutual information can, however,
%also be defined in an inequivalent way as
% \be
%\mathcal{J}_{\set{\Pi_j^A}}(\rho) =
%S(\rho_B)-S\left(\rho_{B|\set{\Pi_j^A}}\right)
% \ee
%with
% $$
%    S(\rho_{B|\set{\Pi_j^A}}) = \sum_jp_jS\left((\Pi_j^A\otimes I^B)\rho(\Pi_j^A\otimes I^B)\Big/p_j\right),
% $$
%where $p_j=\tr(\Pi_j^A\otimes I^B\rho)$. Projective measurements
%on subsystem $A$ removes all non-classical correlations between $A$
%and $B$. The quantity $\mathcal{J}$ thus signifies a measure of
%classical correlations in the state $\rho$ \cite{henderson01a}. To
%ensure that it captures all classical correlations, we need to
%maximize $\mathcal{J}$ over the set of one dimensional projective
%measurements.

Given a quantum state $\rho\in\dens$, recall from Section~\ref{0_sscn:nonclassical} that the quantum discord~\cite{ollivier01a} is defined as
 \ben \label{4_eqn:discord_def}
    \discm :=
           S(\rho_A)-S(\rho)+\min_{\set{\Pi_j^A}}S\left(\rho_{B|\set{\Pi_j^A}}\right)
 \een
for $\set{\Pi_j^A}$ a rank-one projective measurement, and for
 \begin{equation}
    S\left(\rho_{B|\set{\Pi_j^A}}\right) := \sum_jp_jS\left((\Pi_j^A\otimes I^B)\rho(\Pi_j^A\otimes I^B)\Big/p_j\right),
 \end{equation}
where $p_j = \tr(\Pi_j^A\otimes I^B\rho)$. Intuitively, quantum discord captures purely quantum correlations in a quantum state. This is distinct from entanglement in the case
of mixed states. For pure states, quantum discord reduces to the
von Neumann entropy of the reduced density matrix, which is a
measure of entanglement. On the other hand, it is possible for
mixed separable states to have non-zero quantum discord. The main
theorem concerning the discord that we require here is the
following.
\begin{theorem}[Ollivier and Zurek~\cite{ollivier01a}]\label{4_thm:olzu}
        For $\rho\in\dens$, $\discm=0$ if and only if $\rho = \sum_j
(\Pi_j^A\otimes I^B)\rho(\Pi_j^A\otimes I^B)$, for some complete
set of rank one projectors $\set{\Pi_j^A}$.
\end{theorem}

We now show the following.
\begin{theorem}\label{4_thm:1}
    For $\rho\in\dens$, if $\discm>0$, then $\fum>0$.
\end{theorem}
\begin{proof}
    We begin by writing $\ro$ in Fano form~\cite{F83}, i.e.\
    \ben\label{4_eqn:fano}
        \ro = \frac{1}{MN}(I^A\otimes I^B + \bm{r}^A\cdot\bm{\sigma}^A\otimes{I^B}+
                    I^A\otimes\bm{r}^B\cdot\bm{\sigma}^B+\sum_{s=1}^{M^2-1}\sum_{t=1}^{N^2-1}T_{st}\sigma^A_s\otimes\sigma^B_t).
    \een
Here, $\bm{\sigma}^A$ denotes a $(M^2-1)$-component vector of
traceless orthogonal Hermitian basis elements (which
generalize the Pauli spin operators), $\bm{r}^A$ is the
$(M^2-1)$-dimensional Bloch vector for subsystem $A$ with
$r^A_s=\frac{M}{2}\tr(\rho_A\sigma^A_s)$, and $T$ is a real matrix
known as the correlation matrix with entries
$T_{st}=\frac{MN}{4}\tr(\sigma^A_s\otimes\sigma^B_t\ro)$. The
definitions for subsystem $B$ are analogous.

An explicit construction for the basis elements $\sigma_i$
for $M\geq 2$ is given as follows~\cite{he81}. Define
$\set{\sigma_i}_{i=1}^{M^2-1}= \set{U_{pq},V_{pq},W_{r}}$, such
that for $1\leq p<q\leq M$ and $1\leq r \leq M-1$, and
$\set{\ket{k}}_{k=1}^{M}$ some complete orthonormal basis for
$\A$:
    \begin{eqnarray}
        U_{pq}&=&\ket{p}\bra{q}+\ket{q}\bra{p}\label{4_eqn:Ugenerators}\\
        V_{pq}&=&-i\ket{p}\bra{q}+i\ket{q}\bra{p}\label{4_eqn:Vgenerators}\\
        W_{r} &=&\sqrt{\frac{2}{r(r+1)}}\left(\sum_{k=1}^{r}\ket{k}\bra{k}-r\ket{r+1}\bra{r+1}\right)\label{4_eqn:Wgenerators}.
    \end{eqnarray}
In our ensuing discussion, without loss of generality, we fix the choice of basis $\set{\ket{k}}_{k=1}^{M}$ above as the
eigenbasis of $\rho_A$. (Note that the set of orthonormal eigenvectors of
$\rho_A$ will not be unique if the eigenvalues of $\rho_A$ are
degenerate. Hence, we fix some choice of eigenbasis for $\rho_A$
as the ``canonical'' choice to be referred to throughout the rest
of our discussion.)

Assume now that $\discm>0$. Then, any choice of complete
measurement $\pmsmt$ must disturb $\ro$, i.e.\ by
Theorem~\ref{4_thm:olzu}, if we define \be
    \rof:= \sum_{j=1}^M(\Pi_j^A\otimes I)\ro(\Pi_j^A\otimes I),
\ee then $\rof\neq
\ro$.
Henceforth, when we discuss the action of $\pmsmt$ on $\rho_A$, we
are referring to the state $\sum_{j=1}^M\Pi_j^A\rho_A\Pi_j^A$.
Now, let $\pmsmt$ be a complete projective measurement onto the
eigenbasis of $\rho_A$. Then, $\pmsmt$ acts invariantly on
$\rho_A$, and thus must alter the last term in Equation
(\ref{4_eqn:fano}) to ensure $\rof\neq\ro$. To see this, recall that
one can write $\rho_A=\frac{1}{M}(I^A +
\bm{r}^A\cdot\bm{\sigma}^A)$, from which it follows that if
$\pmsmt$ acts invariantly on $\rho_A$, then it also acts
invariantly on $\bm{r}^A\cdot\bm{\sigma}^A$ from
Equation~(\ref{4_eqn:fano}). Since all basis elements
$\sigma^A_s\in\set{W_r}_r$ are diagonal, it follows that there
must exist some $T_{st}\neq 0$ such that
$\sigma^A_i\in\set{U_{pq},V_{pq}}_{pq}$. We now use this fact to
construct a LNU $U^A$ achieving $\fu>0$.

Define unitary $U^A$ as diagonal in the eigenbasis of $\rho_A$,
i.e.\ $U^A = \sum_{k=1}^{M}e^{i\theta_k}\ket{k}\bra{k}$, with
eigenvalues to be chosen as needed. Then, $[U^A,\rho_A]=0$ by
construction, and so $U^A\otimes I^B$ must alter $T$ through its
action on $\ro$ to ensure $\rof\neq \ro$. Focusing on the last
term from Equation~(\ref{4_eqn:fano}), we thus have:
    \ben\label{4_eqn:TshiftSetup}
\sum_{s=1}^{M^2-1}\sum^{N^2-1}_{t=1}T_{st}U^A\sigma^A_s{U^A}^\dg\otimes\sigma^B_t = \sum_{s=1}^{M^2-1}\sum^{N^2-1}_{t=1}T_{st} \Bigg(\sum_{m=1}^{M}\sum_{n=1}^{M}e^{i(\theta_m-\theta_n)}
        \bra{m}\sigma^A_s\ket{n}\ket{m}\bra{n}\Bigg)\otimes
        \sigma^B_t\nonumber
    \een
Analyzing each $\sigma^A_s$ case by case, we find, for
some $1\leq p<q\leq M$ or $1\leq r \leq M-1$:
    \ben
        \sum_{m=1}^{M}\sum_{n=1}^{M}e^{i(\theta_m-\theta_n)}\bra{m}\sigma_s\ket{n}\ket{m}\bra{n}=
        \begin{cases}
            \cos(\theta_p -\theta_q)U_{pq}-\sin(\theta_p -\theta_q)V_{pq}\text{\quad if $\sigma_s=U_{pq}$}\label{4_eqn:TshiftU}\\
            \sin(\theta_p -\theta_q)U_{pq}+\cos(\theta_p -\theta_q)V_{pq}\label{4_eqn:TshiftV}\text{\quad if $\sigma_s=V_{pq}$}\\
            W_r\hspace{53mm}\text{\quad if $\sigma_s=W_r$}\label{4_eqn:TshiftW}.
        \end{cases}
            \een
Denoting by $T^f$ the $T$ matrix for $\rof$, we have:
    \begin{eqnarray}
        T^f_{st}=
            \begin{cases}\label{4_eqn:Tfinal}
                \cos(\theta_p-\theta_q)T_{st}+\sin(\theta_p-\theta_q)T_{wt}\\
                    \hspace{28mm}\text{if $\sigma_s=U_{pq}$, where $\sigma_{w}=V_{pq}$}\\
                \cos(\theta_p-\theta_q)T_{st}-\sin(\theta_p-\theta_q)T_{wt}\\
                    \hspace{28mm}\text{if $\sigma_s=V_{pq}$, where $\sigma_{w}=U_{pq}$}\\
                T_{st}\hspace{23mm} \text{if $\sigma_s=W_r$}.\\
            \end{cases}
    \end{eqnarray}
Thus, if there exists an $s$ such that $T_{st}\neq0$ and
$\sigma^A_s\in\set{U_{pq},V_{pq}}_{pq}$, it follows that one can
easily choose appropriate eigenvalues $e^{i\theta_p}$ and
$e^{i\theta_q}$ for $U^A$ such that $T^f\neq T$, implying
$\fum>0$. By our argument above for $\discm>0$, such an $s$ does
in fact exist.
\end{proof}

%Observe that by the contrapositive of Theorem~\ref{4_thm:1}, it
%immediately follows that if $\fum=0$, then $\discm=0$. By our
%discussion in Section~\ref{4_scn:lc}, it follows that $\discm=0$ for
%any product state $\rho=\rho_A\otimes\rho_B$, as expected.

To show that the converse of Theorem~\ref{4_thm:1} does not
hold, we present an example of a zero discord state that has
non-zero LNU measure. Consider the two qubit separable state
 \begin{equation}
\ro = \frac{1}{2}\left(\frac{{I}_2 +
\bm{a}.\bm{\sigma}}{2}\otimes\frac{{I}_2 + \bm{b}.\bm{\sigma}}{2}
+ \frac{{I}_2 - \bm{a}.\bm{\sigma}}{2}\otimes\frac{{I}_2 -
\bm{b}.\bm{\sigma}}{2} \right),
 \end{equation}
where $\enorm{\bm{a}}=\enorm{\bm{b}}=1$. This state, by
construction, has zero discord for a single qubit measurement on
either $A$ or $B$. To see this, consider the projective
measurements
    \begin{equation}
    \set{\frac{{I}_2 \pm
\bm{a}.\bm{\sigma}}{2}}
   \end{equation}
on $A$. Let us now study the LNU distance for this state, with the
local unitary being applied to say $A$. Notice that $\rho_A =
\rho_ B = {I}_2/2$, and $\tr(\ro^2)=1/2$. The former implies that any local unitary on $\A$ can be chosen, as characterized by Equation~(\ref{su2}). Let us for
convenience parameterize $\bm{a} = (0,0,1)$ and $\bm{b} =
(\sin\gamma \cos \delta, \sin\gamma\sin \delta, \cos\gamma)$.
Then, some algebra leads to
 \be
 \tr(\ro\rof)=\frac{1}{2}\cos^2\theta,
 \ee
whose minimum is 0, whereby
 \be
 \fum=\frac{1}{\sqrt{2}}.
 \ee

We thus have an example of a class of separable, zero discord
states which demonstrates a non-zero shift under LNU. In fact, it
attains the maximum shift possible for two-qubit separable states.
Hence, if one wishes to define notions of non-classicality in
quantum states in terms of `disturbance under measurement' versus
`disturbance under unitary operations', and one chooses discord
and the LNU distance as canonical quantifiers of such effects,
respectively, then the resulting respective notions of
non-classicality are not equivalent. As we have shown in
Theorem~\ref{4_thm:1}, however, the quantum discord is a stronger
notion of non-classicality than the LNU criterion. %We remark that if one is to
%consider such maximization of disturbance under unitary
%operations, then LNU are indeed a possibly fair choice of
%canonical quantifiers --- otherwise, allowing \emph{arbitrary}
%unitary operations without the restriction of Equation~(\ref{4_eqn:lccond})
%would disturb even product states.

\section{Measuring correlations via measurement-induced disturbance}
\label{S:middqc}

        The measure we intend to use in this section was presented
by Luo in~\cite{Luo08}. It relies on the disturbance of a quantum
system under a generic measurement. In that sense, it is similar
in spirit to quantum discord, but not quite. In the case of
quantum discord, as per Equation~(\ref{4_eqn:discord_def}), one
maximizes over one-dimensional projective measurements on one of
the subsystems. For the measure used here, which we will call the
Measurement-Induced Disturbance (MID) measure, one performs
measurements on \emph{both} the subsystems, with the measurements
being given by projectors onto the eigenvectors of the reduced
subsystems. Then the MID measure of quantum correlations for a
quantum state $\rho\in\dens$ is given by~\cite{Luo08}
 \be
 \mathcal{M}(\ro) := \mathcal{I}(\ro) - \mathcal{I}(\mathcal{P}(\ro))
 \ee
where
 \be
 \label{4_eqn:E:midstate}
 \mathcal{P}(\ro):=\sum_{i=1}^M\sum_{j=1}^N(\Pi_i^A\otimes\Pi_j^B)
\ro(\Pi_i^A\otimes\Pi_j^B).
 \ee
Here $\{\Pi_i^A\},\{\Pi_j^B\}$ denote rank one projections onto
the eigenbases of $\rho_A$ and $\rho_B$, respectively, and
$\mathcal{I}(\sigma)$ is the quantum mutual information. The measurement induced by the local eigenvectors leaves the entropy of the reduced states invariant and
is, in a certain sense, the least disturbing. Actually, this
choice of measurement even leaves the reduced states
invariant~\cite{Luo08}. Interestingly, for pure states, both the quantum discord and the
MID measure reduce to the von Neumann entropy of the reduced
density matrix, which is a measure of bipartite entanglement. An advantage of the MID measure is that since no optimizations
are involved, it is much easier to calculate in
practice than the quantum discord or the LNU distance, which
involve optimizations over projective measurements and local
unitaries respectively. The corresponding disadvantage is that if the spectrum of either $\rho_A$ or $\rho_B$ is degenerate, there exist examples~\cite{WPM09} where the MID measure is not necessarily well-defined, as the choice of local eigenbases is no longer unique. In this case, the value of the MID measure should be interpreted moreso as a rough estimate or \emph{upper bound} on the non-classicality of a state. We remark that for this reason, it may be more reasonable to consider a quantity
\begin{equation}
    \mathcal{M}^*(\ro) := \mathcal{I}(\ro) - \max_{\{\Pi_i^A\},\{\Pi_j^B\}}\mathcal{I}(\mathcal{P}(\ro)),
\end{equation}
where $\{\Pi_i^A\},\{\Pi_j^B\}$ are again projections onto eigenbases of $\rho_A$ and $\rho_B$, respectively. (A quantity similar to $\mathcal{M}^*(\ro)$ was considered in~\cite{WPM09}, except the maximization there is over \emph{all} local POVMs. Also, note that it follows directly from the definition of $\mathcal{M}^*(\rho)$ that it is an upper bound on the \emph{distillable entanglement potential} of $\rho$ introduced in Chapter~\ref{chap:activation} (Equation~(\ref{5_eqn:DE})).) Computing $\mathcal{M^*(\ro)}$ is naturally much more difficult; we discuss $\mathcal{M^*(\ro)}$ in this section where appropriate in addition to our discussion of $\mathcal{M}(\ro)$.

\begin{figure}
\begin{center}
 \resizebox{9.5cm}{6cm}{\includegraphics{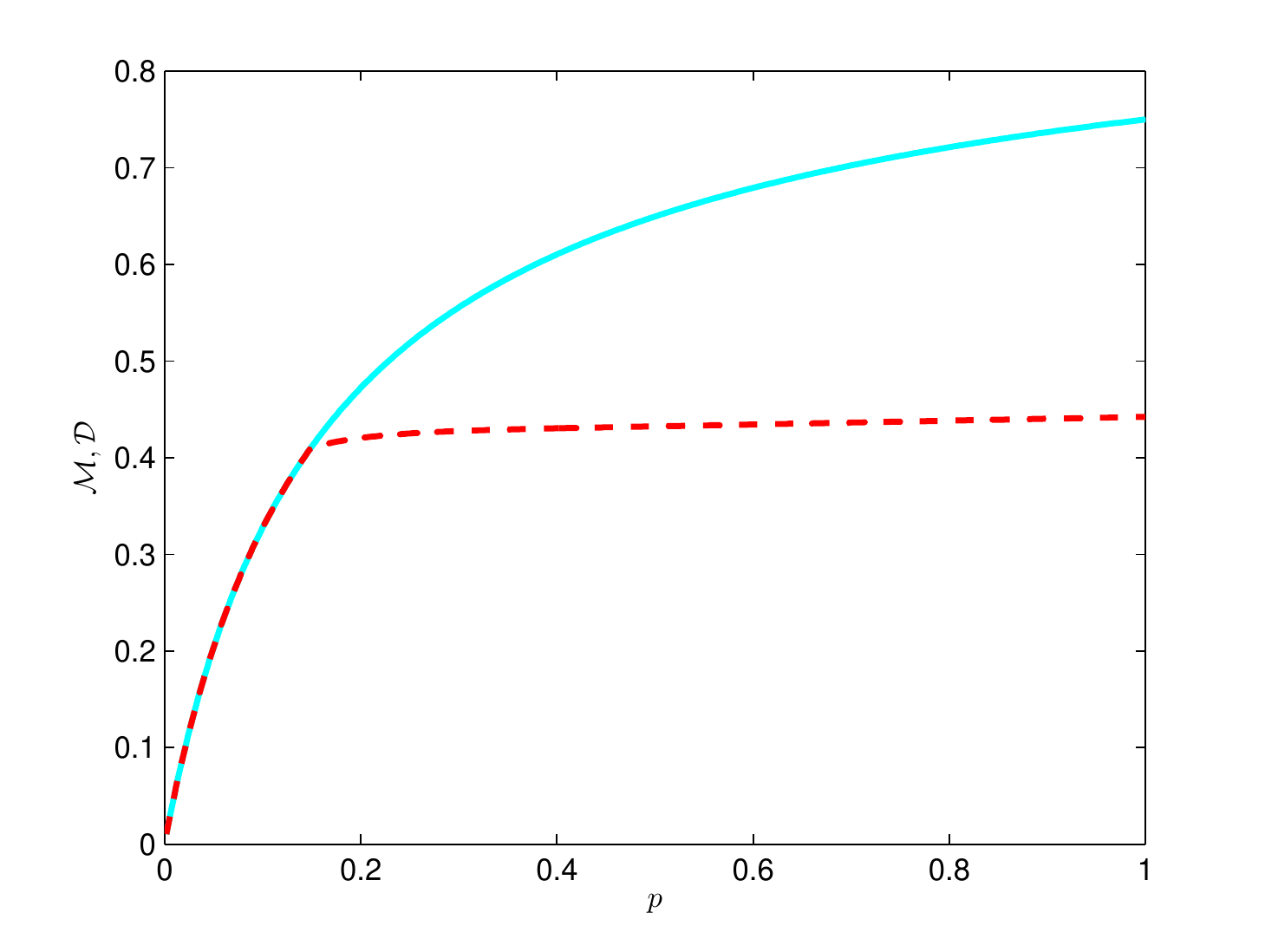}}
\caption{(Color online) The solid line is the MID measure
$\mathcal{M}$ for the $2\times 4$ Horodecki state from~\cite{H97_2}.
The dashed line is the quantum discord $\mathcal{D}$ for the same
state~\cite{dattathesis}. The kink in the latter curve occurs at
$p=1/7.$ We see here, as in the case of the DQC1 state, that the
MID measure is greater than or equal to the quantum discord.}
\label{Horod}
\end{center}
\end{figure}

                To demonstrate the MID measure on a non-trivial example, we first consider the well-known Horodecki
bound entangled state in $2\otimes 4$ dimensions~\cite{H97_2}. It is
bound entangled for all values of $0\leq p \leq 1$, and the state
is given as
 \be
\rho_H = \frac{1}{1+7p}\left(\!\!\!\!%
\begin{array}{cccccccc}
  \;p\;&\;0\; &\;0\;& \;0\; & \;0\; & \;p\; & \;0\; & \;0\; \\
  0 & p & 0  & 0 & 0 & 0 & p& 0 \\
  0 & 0 & p  & 0 & 0 & 0 & 0 & p \\
  0 & 0 & 0  & p & 0 & 0 & 0 & 0 \\
  0 & 0 & 0  & 0 & \frac{1+p}{2} & 0 & 0 & \frac{\sqrt{1-p^2}}{2} \\
  p & 0 & 0  & 0 & 0 & p & 0 & 0 \\
  0 & p & 0  & 0 & 0 & 0 & p & 0 \\
  0 & 0 & p  & 0 &\frac{\sqrt{1-p^2}}{2} & 0 & 0 & \frac{1+p}{2}  \\
\end{array}%
\!\!\!\!\right).
 \ee
The projectors onto the eigenvectors of the reduced density
matrices can be chosen as
 \ben
\{\Pi^A_1,\Pi^A_2\}&=&\left\{ \left(%
\begin{array}{cc}
  1 & 0 \\
  0 & 0 \\
\end{array}%
\right),
\left(%
\begin{array}{cc}
  0 & 0 \\
  0 & 1 \\
\end{array}%
\right)
\right\},\;\;\;\;\;\;\mbox{and}\\
\{\Pi^B_1,\cdots,\Pi^B_4\} &=&
\big\{\proj{\Psi^{+}},\proj{\Psi^{-}},
        \proj{\Phi^{+}},\proj{\Phi^{-}}
\big\}.
 \een
where $\ket{\Psi^{\pm}}=(\ket{1}\pm \ket{2})/\sqrt 2$ and
$\ket{\Phi^{\pm}}=(\ket{0}\pm \ket{3})/\sqrt 2,$ with
$\{\ket{0},\ket{1},\ket{2},\ket{3}\}$ forming the computational
basis for the second subsystem. Using these in Equation
(\ref{4_eqn:E:midstate}), we have
  \begin{equation}
\mathcal{P}(\rho_H)= \frac{1}{1+7p}\!\!\left(\!\!\!\!%
\begin{array}{cccccccc}
  \;p\;&\;0\; &\;0\;& \;0\; & \;0\; & \;0\; & \;0\; & \;0\; \\
  0 & p & 0  & 0 & 0 & 0 & 0& 0 \\
  0 & 0 & p  & 0 & 0 & 0 & 0 & 0 \\
  0 & 0 & 0  & p & 0 & 0 & 0 & 0 \\
  0 & 0 & 0  & 0 & \frac{1+p}{2} & 0 & 0 & \frac{\sqrt{1-p^2}}{2} \\
  0 & 0 & 0  & 0 & 0 & p & 0 & 0 \\
  0 & 0 & 0  & 0 & 0 & 0 & p & 0 \\
  0 & 0 & 0  & 0 &\frac{\sqrt{1-p^2}}{2} & 0 & 0 & \frac{1+p}{2}  \\
\end{array}
\!\!\!\!\right).
 \end{equation}
Note that this density matrix differs from $\rho_H$ in that
some of the off-diagonal terms $p$ have vanished. We have computed the MID
measure for $\rho_H$ as
$\mathcal{M}(\rho_H)=S(\mathcal{P}(\rho_H))-S(\rho_H)$ and
plotted it in Figure (\ref{Horod}). In the same figure, we also plot the
quantum discord for this state, when a measurement is made on the
two-dimensional subsystem~\cite{dattathesis}. As we see, there are non-classical
correlations in this state that are not distillable into maximally
entangled Bell pairs.

As a comparison, we remark that for $\rho_H$, $\mathcal{M}^*(\rho_H)$ behaves similarly to $\mathcal{M}(\rho_H)$. To see this, note that only $\rho_B$ has a degenerate eigenvalue, and this is on the space spanned by $\Pi^B_3$ and $\Pi^B_4$. Thus, in the minimization over local bases, one can more generally choose $\Pi^B_3$ and $\Pi^B_4$ to project onto an arbitrary basis for this space, $a \ket{0}+e^{i\theta}b\ket{3}$ and $b \ket{0}-e^{i\theta}a\ket{3}$ for $a,b,\theta\in\reals$, respectively. The eigenvalues of $\mathcal{P}(\rho_H)$ are then (up to normalization)
\begin{equation}
    \frac{1}{2}\left(p+1\pm\sqrt{1-p^2}\right), p, p, p\left(1\pm\abs{a}\abs{b}\sqrt{2(1-\cos(2\theta))}\right),    p\left(1\pm\abs{a}\abs{b}\sqrt{2(1-\cos(2\theta))}\right).
\end{equation}
In the expression $\mathcal{M}^*(\ro_H) = \min_{\{\Pi_i^A\},\{\Pi_j^B\}}S(\mathcal{P}(\ro_H))-S(\ro_H) $, the entropy $S(\mathcal{P}(\rho_H))$ is thus minimized by choosing $a=b=1/\sqrt{2}$ and $\theta=\pi/2$. A plot of the resulting value of $\mathcal{M}^*(\rho_H)$ is given in Figure~\ref{4_fig:Horod2}.
\begin{figure}
\begin{center}
{\includegraphics[width=70mm,trim = 0mm 40mm 0mm 60mm, clip]{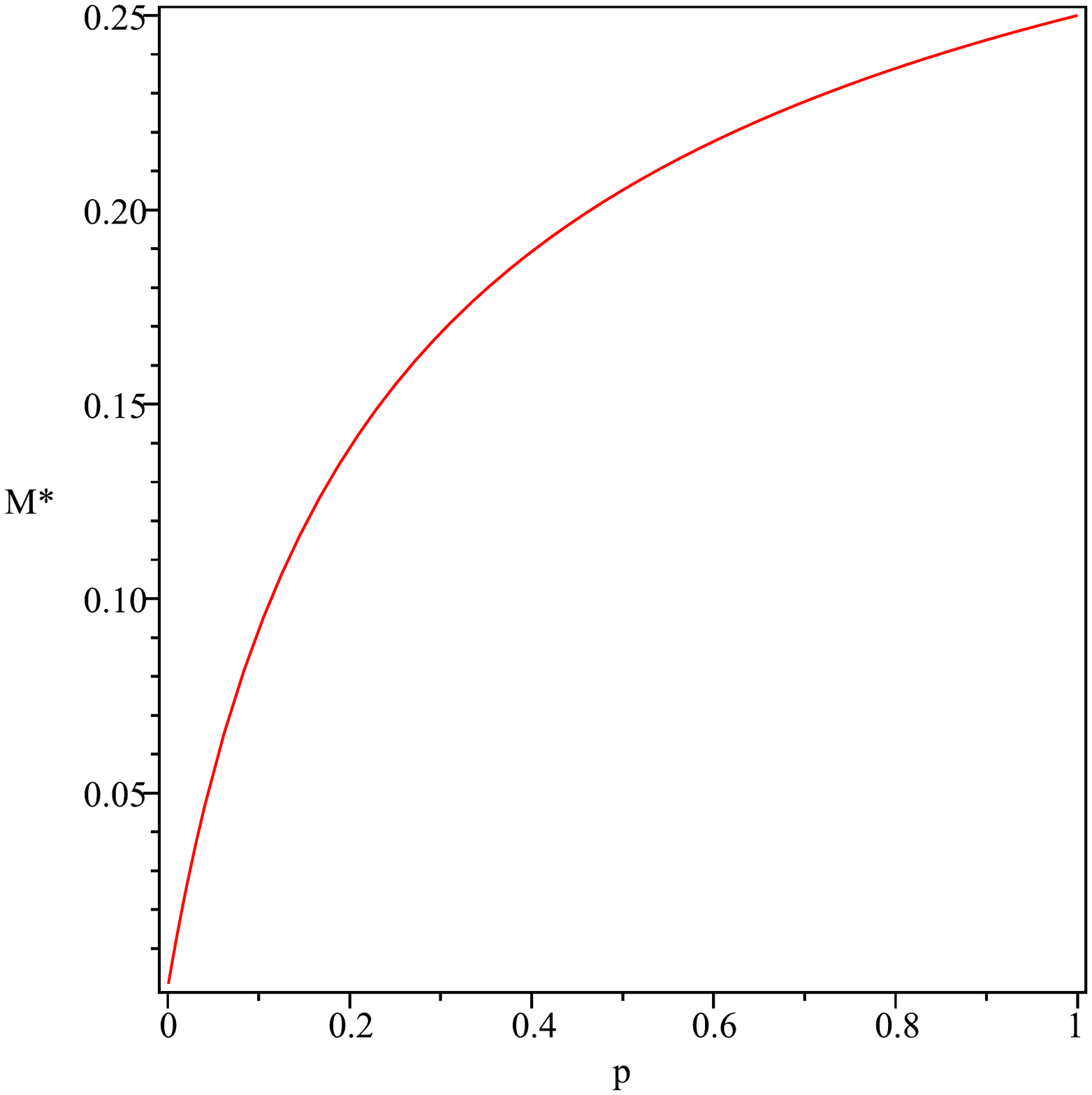}}
\caption{A plot of $\mathcal{M}^*$ for the $2\times 4$ Horodecki state from~\cite{H97_2}.}
\label{4_fig:Horod2}
\end{center}
\end{figure}

\subsection{MID measure in the DQC1 model}\label{4_sscn:MIDDQC1}

We now move on to calculate the MID measure in the DQC1 model. Our
analysis extends that of~\cite{Luo08}, where only the case of
$\alpha=1$ was considered. Considering $\alpha<1/2$ here will be
of particular interest, due to the lack of distillable
entanglement in the DQC1 state (in this regime, any bipartite split has a positive partial transpose). Consequently, we start with the
$(n+1)$-qubit DQC1 state, given by Equation~(\ref{4_eqn:E:dqc1}), wherefrom
 \be
 \rho_{A}= \frac{1}{2}\left(
\begin{array}{cc}
  1 & \alpha \tau^* \\
  \alpha \tau & 1 \\
\end{array}%
\right)\;\;\;\;\;\mbox{and}\;\;\;\; \rho_{B}= I_n/2^n,
   \ee
  where recall $\tau = \tr(U_n)/2^n$. The projectors onto $\rho_A$'s eigenvectors can be chosen as
    \begin{equation}
        \{\Pi_0^A,\Pi_1^A\}=\{\ketbra{\phi_0}{\phi_0},\ketbra{\phi_1}{\phi_1}\}
    \end{equation}
    for $\ket{\phi_0}:=(\ket{0}+e^{i\phi}\ket{1})/\sqrt{2}$ and $\ket{\phi_1}:=(\ket{0}-e^{i\phi}\ket{1})/\sqrt{2}$, respectively,
%   $$
%\{\Pi_1^A,\Pi_2^A\}= \left\{\frac{1}{2}\left(
%\begin{array}{cc}
%  1 & e^{-i\phi} \\
%  e^{i\phi} & 1 \\
%\end{array}%
%\right),\frac{1}{2}\left(
%\begin{array}{cc}
%  1 & -e^{-i\phi} \\
%  -e^{i\phi} & 1 \\
%\end{array}%
%\right)\right\}
%   $$
where $\tau=re^{i\phi}$ for $r=\abs{\tau}$. Similarly, set
 $\{\Pi_j^B\}=\{\ketbra{j}{j}\}$ for $\set{\ket{j}}_{j=1}^{2^n}$ the computational basis. Using this, we can calculate
 \ben
 \mathcal{P}(\dqc)&=&\sum_{j=1}^{2^n}\sum_{k=0}^1(\Pi_k^A\otimes\Pi_j^B)\dqc(\Pi_k^A\otimes\Pi_j^B) \\
                       &=&\frac{1}{2^{n+1}}\sum_{j=1}^{2^n}\sum_{k=0}^1\Pi_k^A\left(
                                                        \begin{array}{cc}
                                                        1 & \alpha \bra{j}U_n^\dagger\ket{j} \\
                                                      \alpha \bra{j}U_n\ket{j} & 1 \\
                                                        \end{array}%
                                                            \right)\Pi_k^A\otimes \ketbra{j}{j} .
                       %                        &=& \frac{1}{2^{n+1}}\left(
%                                                    \begin{array}{cc}
%                                                    I_n & \alpha D \\
%                                                   \alpha D^{\dg} & I_n \\
%                                                   \end{array}%
%                                                   \right)
 \een
Observing that
\begin{equation}
    \Pi_k^A\left(
              \begin{array}{cc}
              1 & \alpha \bra{j}U_n^\dagger\ket{j} \\
              \alpha \bra{j}U_n\ket{j} & 1 \\
              \end{array}%
              \right)\Pi_k^A=
    \left(1+(-1)^k\alpha\operatorname{Re}(\bra{j}U_n\ket{j}e^{-i\phi})\right)\ketbra{\phi_k}{\phi_k},
\end{equation}
we conclude that the spectrum of $\mathcal{P}(\dqc)$ is given by
 \be
{\bm \lambda}[\mathcal{P}(\dqc)]=\left\{\frac{1\pm \Delta_j}{2^{n+1}}\right\}
 \ee
for
\begin{equation}
    \Delta_j:=\alpha\operatorname{Re}\left(\bra{j}U_n\ket{j}e^{-i\phi}\right)
\end{equation}
and $j\in[2^n]$. Letting $\lambda_k$ denote the $k$th entry of ${\bm
\lambda}[\mathcal{P}(\dqc)]$, the von Neumann entropy of this
state is
 \ben
 S(\mathcal{P}(\dqc)) &=&-\sum_{k=1}^{2^{n+1}} \lambda_k\log(\lambda_k)\\
    &=& n+1 -\frac{1}{2^{n+1}}\sum_{j=1}^{2^n}\Bigg(\log(1-\Delta_j^2)
    +\Delta_j\log\left(\frac{1+\Delta_j}{1-\Delta_j}\right)\Bigg).
 \een
Now,
  \be
 S(\dqc)= n+H_2\left(\frac{1-\alpha}{2}\right),
  \ee
and since the entropies of the partial density matrices are invariant under the local measurements, we have
 \ben
 \label{mid}
\mathcal{M}_{DQC1} &=& \mathcal{I}(\dqc)-\mathcal{I}(\mathcal{P}(\dqc)) \\
            &=& S(\mathcal{P}(\dqc))-S(\dqc) \\
            &=& 1- H_2\left(\frac{1-\alpha}{2}\right) - \frac{1}{2^{n+1}}\sum_{j=1}^{2^n}\Bigg(\log(1-\Delta_j^2)
    +\Delta_j\log\left(\frac{1+\Delta_j}{1-\Delta_j}\right)\Bigg).\nonumber%\\
%    &\geq& 1- H_2\left(\frac{1-\alpha}{2}\right) - \frac{1}{2^{n+1}}\sum_{j=1}^{2^n}\Bigg( \Delta_j\log\left(\frac{1+\Delta_j}{1-\Delta_j}\right)\Bigg).
 \een
For any unitary $U_n$, which is known in any implementation of the DQC1 circuit, the
above quantity can be computed easily. Bounding this quantity more generally, however, is difficult. If, however, in the asymptotic limit of large $n$, $|\Delta_j|\rightarrow 0$ (as might intuitively be expected when $U_n$ is a Haar distributed random
unitary matrix, since then we might expect $\abs{u_{jj}} \sim 1/2^{n/2}$), then the whole
quantity within the summation in Equation (\ref{mid}) goes to zero.
%SEV
%Consider the method for constructing Haar random *orthogonal* (not unitary) matrices. Each entry is chosen at random according to standard normal dist., ie has mean 0 and variance 1. Then, central limit theorem implies (sum_i a_i)/2^n is distributed as N(0,1/2^n/2), which implies (sum_i a_i) ~ N(0,2^n/2). Since there are 2^n terms in the sum, this means each term is typically in the range (-2^-n/2,2^n/2), as claimed. Questions: (1) Does this still hold after Gram Schmidt applied? Intuition: Possibly, as if we choose two Haar distributed vectors v1 and v2, then their inner product is exponentially small and so subtracting one off from the other in Gram Schmidt process should not change v2 too much (v1 doesnt change in GS) Of course, here the rows are not quite Haar distributed (2) How about for unitary matrices, where entries are complex Gaussian random vars?
In this case,
 \be
 \label{4_eqn:E:midanal}
\mathcal{M}_{DQC1} \sim 1-H_2\left(\frac{1-\alpha}{2}\right).
 \ee
One fact immediately notable is that the above expression for the
MID measure is independent of $n$, for large $n$. The result for a
$n=5$ qubit Haar distributed random unitary matrix is shown in Figure
(\ref{measdisc}). As is evident, despite the approximations used
in the derivation of Equation (\ref{4_eqn:E:midanal}) the asymptotic
analytic expression matches the numerical result at $n=5$ quite
well. We remark that even for a version of $\mathcal{M}_{DQC1}$ where one minimizes over all local POVMs, the behavior one finds is quantitatively analogous to that of $\mathcal{M}_{DQC1}$ plotted in Figure~\ref{measdisc}~\cite{WPM09}.

\begin{figure}[t]
\begin{center}
{\includegraphics[width=80mm,keepaspectratio]{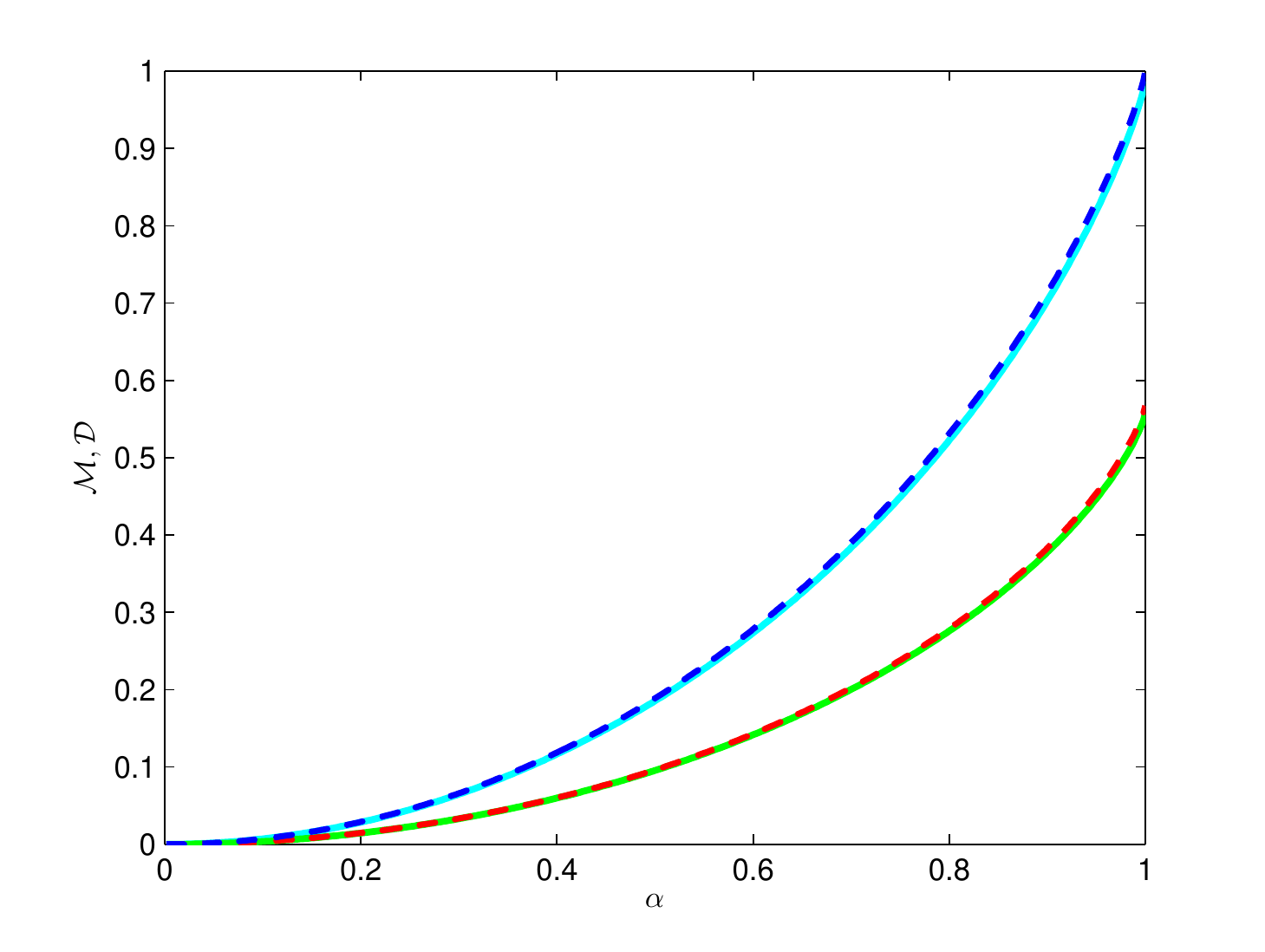}}
\caption{(Color online) The upper solid (cyan) line is the MID
measure $\mathcal{M}$ (Equation~(\ref{mid})) for the DQC1 circuit for
a $n=5$ qubit Haar distributed random unitary matrix. The upper
dashed (blue) line is the analytic expression for the MID measure
for certain DQC1 states from Equation (\ref{4_eqn:E:midanal}). The lower dashed (red) line shows the
discord $\mathcal{D}$ in the DQC1 circuit with the same unitary.
The lower solid (green) line shows the analytical expression of the
quantum discord from~\cite{datta08a}. All quantities are shown as
functions of the purity $\alpha$ of the control qubit. } \label{measdisc}
\end{center}
\end{figure}

    The MID measure for the DQC1 state across the bipartite split
separating the top qubit from the rest is non-zero for all
non-zero values of $\alpha$. Across this split, the DQC1
state is strictly separable~\cite{datta05a} and possesses no
entanglement. Hence, one might propose the MID measure as a
quantifier of the resource behind the quantum advantage in the
DQC1 model~\cite{Luo08}. Note that, as can be seen from Figure
(\ref{measdisc}), the behavior of the MID measure is qualitatively
quite similar to that of the quantum discord.

\subsection{Non-classical correlations in quantum communication}
\label{S:locking}

We now use the MID measure to study the locking of classical correlations in quantum states.
It has been shown~\cite{dhlst04} that there exist bipartite
quantum states which contain a large amount of locked classical
correlation which can be unlocked by a small amount of classical
communication. More precisely, there exist $(2n + 1)$-qubit states
for which the optimal classical mutual information between
measurement results on the subsystems can be increased from $n/2$
bits to $n$ bits via a single bit of classical communication.
Despite the impossibility of this feat classically, the states
used in the protocol are not entangled.

Here we use the MID measure to study this purely quantum
phenomenon. To do so, we evaluate the former on a generalization
of the state used in~\cite{dhlst04},
 \be
 \rho=\frac{1}{md}\sum_{k=1}^{d}\sum_{t=1}^{m}(\proj{k}\otimes\proj{t})_A\otimes(\proj{b_k^t})_B,
 \ee
where the set of $m$ orthonormal bases
$\set{\set{\ket{b_k^t}}_{k=1}^d}_{t=1}^{m}$ is mutually unbiased
(MUB), i.e.\ $\forall _{t\neq
t^{\prime},i,j}\abs{\braket{b_i^t}{b_j^{t^\prime}}}=1/\sqrt{d}$. As in
Reference~\cite{dhlst04}, when $d=2^n$ and $m=2$, the initial
correlations in this state amount to $n/2$ bits, and by Alice's
sending one bit (the bit $t$) to Bob, they end up with $n+1$
correlated bits. The state being separable, it has no
entanglement. Consequently, we cannot ascribe to entanglement the advantage
exhibited by this protocol.

To calculate the MID measure of this state, we need the reduced
states given by
\begin{equation}
\rho_A = \frac{I_{md}}{md},\;\;\;\;\rho_B =
\frac{I_d}{d}.
\end{equation}
  Choosing the local eigenvectors as the respective computational bases, we have that $\mathcal{P}(\rho)$
is simply the diagonal of $\rho.$ Thus,
 \be
{\bm
\lambda}[\mathcal{P}(\rho)]=\frac{1}{md}\big\{\!\underbrace{1,\cdots,1}_{d},\underbrace{1/d,\cdots,1/d}_{(m-1)d^2},\underbrace{0,0,\cdots,0}_{d(d-1)}\big\}
 \ee
whereby
 \be
S(\mathcal{P}(\rho)) = \log m +\left(2-\frac{1}{m}\right)\log d.
 \ee
The spectrum of $\rho$ is given by
\begin{equation}
{\bm\lambda}[\rho]=\frac{1}{md}\big\{\!\underbrace{1,1,\cdots,1}_{md},\underbrace{0,0,\cdots,0}_{md(d-1)}\big\}
\end{equation}
which leads to
 \be
S(\rho) = \log m+\log d.
 \ee
Finally, we have
 \be
\mathcal{M}(\rho) = S(\mathcal{P}(\rho)) - S(\rho) =
\left(1-\frac{1}{m}\right)\log d, \label{4_eqn:lockMID}
 \ee
which for $d=2^n$ and $m=2$ is the exactly equal to the gain
attained by this scheme. Moreover, once Bob receives Alice's bit,
the MID measure for their post-communication state drops to $0$,
the latter being diagonal in a local product basis. This suggests
the possibility that the MID measure quantifies the non-classical (yet
not entanglement-based) correlations in $\rho$ which were
initially locked. Moreover, we remark that for $d=2^n$ and $m=2$, we have $\mathcal{M}(\rho)=\mathcal{M}^*(\rho)$ --- this follows directly from the result~\cite{dhlst04} that the mutual information of any classical distribution induced via local measurements on $\rho$ is at most $(\log d)/2$.

A few remarks are in order. Equation~(\ref{4_eqn:lockMID}) might suggest that
a better locking effect may be possible for $m>2$. However, explicit
constructions to date using more than two MUBs have been unable to
achieve superior locking~\cite{BW07}, suggesting that the choice
of construction for the MUBs plays an important role. In contrast,
Equation~(\ref{4_eqn:lockMID}) holds irrespective of the specific choice
of MUBs. It is also known that if the bases above are constructed
using a large set of random unitaries chosen according to the Haar
measure, then the classical mutual information in $\rho$ between
Alice and Bob can be brought down to a
constant~\cite{HLSW04}. There is also numerical evidence (Appendix
of Reference~\cite{dhlst04}) that the dimension of the systems may
play a role in achieving better locking. Connections
between locking and non-classical correlations have since been discovered in References~\cite{WPM09,BACMW11}.

Finally, for completeness, we remark that
$\tr(\rho^2)=1/(md)$, and so by Theorem~\ref{4_thm:fuBound},
the LNU distance for $\rho$ is bounded by

\be
    \fum \leq \sqrt{\frac{2}{md}\left(1-\frac{1}{d}\right)}\leq \sqrt{\frac{2}{md}}.
\ee Thus, in contrast to the MID measure, the LNU distance once
again reveals vanishing non-classicality with growing $m$ or $d$.

%\section*{Acknowledgements}
%
\noindent \emph{Acknowledgements for this chapter.} We thank Carl Caves and Anil Shaji for numerous stimulating discussions, as well as an anonymous referee for raising certain points that
led to improvements in the paper this chapter is based on.
%AD was supported in part by the US Office of Naval
%Research (Grant No. N00014-07-1-0304) and also by EPSRC (Grant No.
%EP/C546237/1), EPSRC QIP-IRC and the EU Integrated Project (QAP).
%SG was partially supported by Canada's NSERC, CIAR and MITACS. We
%also thank the
%

%======================================================================
\chapter{Quantifying non-classicality with local unitary operations}\label{chap:localunitary}
\emph{This chapter is based on~\cite{G12}:}\\

\vspace{-4mm}
\noindent S. Gharibian. Quantifying non-classicality with local unitary operations.
Available at arXiv.org e-Print quant-ph/1202.1598v1, 2012.

\vspace{3mm}
\noindent In this chapter, we propose a measure of non-classical correlations in bipartite quantum states based on local unitary operations. We prove the measure is non-zero if and only if the quantum discord is non-zero; this is achieved via a new characterization of zero discord states in terms of the state's correlation matrix. Moreover, our scheme can be extended to ensure the same relationship holds even with a generalized version of quantum discord in which higher-rank projective measurements are allowed. We next derive a closed form expression for our scheme in the cases of Werner states and $(2\times N)$-dimensional systems. The latter reveals that for $(2\times N)$-dimensional states, our measure reduces to the geometric discord~\cite{DVB10}. A connection to the CHSH inequality is shown. We close with a characterization of all maximally non-classical, yet separable, $(2\times N)$-dimensional states of rank at most two (with respect to our measure).

%///////////////////////////////////////////////////////////////////////////////////////////////////////////////////
\section{Introduction and results}\label{7_scn:intro}
%///////////////////////////////////////////////////////////////////////////////////////////////////////////////////

One of the most intriguing aspects of quantum mechanics is quantum entanglement, which with the advent of quantum computing, was thrust into the limelight of quantum information theoretic research~\cite{HHHH09}. We now know that correlations in quantum states due to entanglement are necessary in order for \emph{pure-state} quantum computation to provide exponential speedups over its classical counterpart~\cite{jozsa03a}. With bipartite entanglement nowadays fairly well understood, however, attention has turned in recent years to a more general type of quantum correlation, dubbed simply \emph{non-classical correlations}. Unlike entanglement, such correlations \emph{can} be created via Local Operations and Classical Communication (LOCC), but nevertheless do not exist in the classical setting. Moreover, for certain \emph{mixed-state} quantum computational feats, the amount of entanglement present can be small or vanishing, such as in the DQC1 model of computing~\cite{kl98} and the locking of classical correlations~\cite{dhlst04}. In these settings, it is rather non-classical correlations which are the conjectured resource enabling such feats (see, e.g.~\cite{datta05a, datta08a, Luo08, DG08}). In fact, almost all quantum states possess non-classical correlations~\cite{ferraro}.

As a result, much attention has recently been devoted to the quantification of {non-classical correlations} (e.g.,~\cite{PhysRevA.71.062307,PhysRevA.72.032317,MPSVW10,groismanquantumness,PhysRevA.77.052101,Luo08,pianietal2008nolocalbrodcast,pianietal2009broadcastcopies,ADA,PhysRevA.82.052342,DVB10,streltsov2010,PGACHW11}, see~\cite{MBCPV11} for a survey, and Section~\ref{0_sscn:nonclassical} for a brief exposition). Here, we say a bipartite state $\rho$ acting on Hilbert space $\A\otimes \B$ is \emph{classically correlated} in $\A$ if and only if there exists an orthonormal basis $\set{\ket{a}}$ for $\spa{A}$ such that
\begin{equation}
    \rho = \sum_i p_i \ketbra{a_i}{a_i}\otimes\rho_i
\end{equation}
for $\set{p_i}$ a probability distribution and $\rho_i$ density operators. To quantify ``how far'' $\rho$ is from the form above, a number non-classicality measures, including perhaps the best-known such measure, the \emph{quantum discord}~\cite{ollivier01a,henderson01a}, ask the question of how drastically a bipartite quantum state is disturbed under local measurement on $\spa{A}$. In this chapter, we take a different approach to the problem. We ask: \emph{Can disturbance of a bipartite system under local unitary operations be used to quantify non-classical correlations?}

It turns out that not only is the answer to this question \emph{yes}, but that in fact for $(2\times N)$-dimensional systems, the measure we construct coincides with the \emph{geometric quantum discord}~\cite{DVB10}, a scheme based again on local measurements. Our measure is defined as follows. Given a bipartite quantum state $\rho$ and unitary $\UA$ acting on Hilbert spaces $\A\otimes\B$ and $\A$ with dimensions $MN$ and $M$, respectively, define
\begin{equation}\label{7_eqn:measure_U_def}
    D(\rho,\UA) := \frac{1}{\sqrt{2}}\fnorm{\rho - \left(\UA\otimes I_B\right) \rho\left( \UA^\dagger\otimes I_B\right)},
\end{equation}
where the Frobenius norm $\fnorm{A}=\sqrt{\trace{A^\dagger A}}$ is used due to its simple calculation. Then, consider the set of unitary operators whose eigenvalues are precisely the  $M$-th roots of unity, i.e.\ whose vector of eigenvalues equals $\ve{v}$ for $v_k = e^{2\pi ki/M}$ for $1\leq k \leq M$. (The corresponding eigenvectors can be chosen arbitrarily.) We call such operators \emph{Root-of-Unity} (RU) unitaries. They include, for example, the Pauli $X$, $Y$, and $Z$ matrices (see Section~\ref{0_sscn:evolution}). Then, letting $\rous$ denote the set of RU unitaries acting on $\A$, we define our measure as:
\begin{equation}\label{7_eqn:measure_def}
    D(\rho) := \min_{\UA \in\rous} D(\rho,\UA).
\end{equation}
Note that $0\leq D(\rho)\leq 1$ for all $\rho$ acting on $\A\otimes\B$.

\paragraph{Our results:} In this chapter, we show the following regarding $D(\rho)$.

\vspace{3mm}
\noindent\textbf{1. Closed form expressions.} Our first result is a closed-form expression for $D(\rho)$ for $(2\times N)$-dimensional systems (Theorem~\ref{7_thm:closedform}). This reveals that for $(2\times N)$-dimensional $\rho$, $D(\rho)$ coincides with the geometric discord of $\rho$. It also allows us to prove that, like the \emph{Fu distance}~\cite{f06,gkb08} (defined below in \emph{Previous Work}), if $D(\rho)>1/\sqrt{2}$, then $\rho$ violates the Clauser-Horne-Shimony-Holt (CHSH) inequality~\cite{CHSH69} (Corollary~\ref{7_cor:CHSH}). We also derive a closed form expression for $D(\rho)$ for Werner states, finding here that $D(\rho)$ in fact equals the Fu distance of $\rho$ (Theorem~\ref{7_thm:werner}).

\vspace{3mm}
\noindent\textbf{2. States achieving $D(\rho)=1$.} We next show that only pure maximally entangled states $\rho$ achieve the maximum value $D(\rho)=1$, as expected (Corollary~\ref{7_cor:max}).

\vspace{3mm}
\noindent\textbf{3. $D(\rho)$ is faithful.} We show that $D(\rho)$ is a \emph{faithful} non-classicality measure, i.e.\ it achieves a value of zero if and only if $\rho$ is classically correlated in $\spa{A}$ (Theorem~\ref{7_thm:discequiv}). To prove this, we first derive a new characterization of states with zero quantum discord based on the correlation matrix of $\rho$. We then show that the states achieving $D(\rho)=0$ can be characterized in the same way. More generally, by extending our scheme to allow the eigenvalues of $\UA$ to have multiplicity at most $k$, we prove a state is undisturbed under $\UA$ if and only if there exists a projective measurement on $\spa{A}$ of rank at most $k$ acting invariantly on the state (Theorem~\ref{7_thm:gendiscequiv}). This reproduces in a simple fashion a result of Reference~\cite{MAGGDI11} regarding entanglement quantification in the pure state setting. Based on this equivalence between disturbance under local unitary operations and local projective measurements, we propose a generalized definition of the quantum discord at the end of Section~\ref{7_scn:discord}.

\vspace{3mm}
\noindent\textbf{4. Maximally non-classical, yet separable states.} Finally, we characterize the set of maximally non-classical, yet separable, $(2\times N)$-dimensional $\rho$ of rank at most two, according to $D(\rho)$ (and hence according to the geometric discord) (Lemmas~\ref{7_lem:upperbound} and~\ref{7_lem:upperbound2}).

\paragraph{Previous work:} The Fu distance, defined as the \emph{maximization} of Equation~(\ref{7_eqn:measure_U_def}) over all $\UA$ such that $[\UA,\trace_{B}(\rho)]=0$, was defined in Reference~\cite{f06} and studied further in References~\cite{gkb08} and~\cite{DG08} with regards to quantifying entanglement and non-classicality. Despite its strengths, such as a closed form solution for two-qubit systems and Werner states, and a connection to the CHSH inequality, the distance has weaknesses: It can attain its maximum value even on non-maximally entangled pure states~\cite{gkb08}, and is not a faithful non-classicality measure~\cite{DG08}. Interestingly, our $D(\rho)$ eliminates these weaknesses while preserving the former strengths. Subsequent to the conception of our scheme, the present author learned that there has also been an excellent line of work studying (the square of) Equation~(\ref{7_eqn:measure_def}) in another setting --- that of \emph{pure state entanglement}. In Reference~\cite{GI07}, it was found that in $(2\times N)$ and $(3\times N)$ systems, $D(\ketbra{\psi}{\psi})^2$ coincides with the \emph{linear entropy of entanglement}. Reference~\cite{MAGGDI11} then showed that for arbitrary bipartite pure states, $D(\ketbra{\psi}{\psi})^2$ is a faithful entanglement monotone, and derived upper and lower bounds in terms of the linear entropy of entanglement. Finally, alternative characterizations of zero discord states have been given in~\cite{ollivier01a,DVB10,D10}. Maximally non-classical separable two-qubit states have been studied, for example, in~\cite{GPACH11,GA11}. For example, the set of such states found~\cite{GPACH11} with respect to the \emph{relative entropy of quantumness} matches our characterization for $D(\rho)$; we remark, however, that our analysis for $D(\rho)$ in this regard is more general than in~\cite{GPACH11} as it is based on a less restrictive ansatz. We remark that since the initial posting of the paper this chapter is based on, a related work by Streltsov \emph{et al.} has appeared~\cite{SGRBI12}.\\

\paragraph{Discussion and open questions:} Our results show that local unitary operations can indeed form the basis of a non-classicality measure with certain desirable properties. In particular, the scheme we consider is faithful, correctly identifies maximally non-classical states, and reveals interesting connections to a number of quantifiers of correlations, such as the Fu distance, the quantum discord, the geometric quantum discord, and the relative entropy of quantumness. As outlined above, the strengths of our scheme include a closed form for two-qubit states and Werner states, the former of which reveals a link between the paradigms of ``disturbance under local unitary operations'' and ``disturbance under local measurements'' by reducing to the geometric discord for two-qubit states. This link is further strengthened by the demonstration of connections to even generalized versions of the quantum discord.

We leave open the following questions. For what other interesting classes of quantum states can a closed form expression for $D(\rho)$ be found? Can a better intuitive understanding of the interplay between the notions of ``disturbance under local measurements'' and ``disturbance under local unitary operations'' be obtained in higher dimensions? We give an analytical characterization of all maximally non-classical rank-two $(2\times N)$-dimensional separable states --- we conjecture that higher rank two-qubit states, for example, achieve strictly smaller values of $D(\rho)$. Can this be proven rigorously and analytically? (We remark that a numerical proof for this conjecture was given in~\cite{GA11} for the geometric discord, for example.) What can the study of the generalized notion of quantum discord we define in Section~\ref{7_scn:discord}, $\delta_{\ve{v}}(\rho)$, tell us about non-classical correlations?

\paragraph{Organization of this chapter:} We begin in Section~\ref{7_scn:prelim} with necessary definitions and useful lemmas. Closed forms for $(2\times N)$-dimensional systems are given in Section~\ref{7_scn:twoqubitpure} and for Werner states in Section~\ref{7_scn:werner}. Section~\ref{7_scn:purearbdim} characterizes the set of states achieving $D(\rho)=1$. Section~\ref{7_scn:discord} shows that $D(\rho)$ is faithful. In Section~\ref{7_scn:maxnc}, we discuss maximally non-classical separable states.

\section{Preliminaries}\label{7_scn:prelim}
We begin by reviewing notation specific to this chapter, followed by relevant definitions and useful lemmas. Throughout this chapter, we use $\A$ and $\B$ to denote complex Euclidean spaces of dimensions $M$ and $N$, respectively. We define $\rho_A := \trace_{B}(\rho)$ and $\rho_B := \trace_{A}(\rho)$. The anti-commutator of $A$ and $B$ is $\set{A,B}=AB+BA$. The notation $\operatorname{diag}(\ve{v})$ for complex vector $\ve{v}$ denotes a diagonal matrix with $i$th diagonal entry $v_i$, and $\operatorname{span}(\set{\ve{v}_i})$ denotes the span of the set of vectors $\set{\ve{v}_i}$.

Moving to definitions, in this chapter we often decompose $\rho\in\dens$ in terms of a Hermitian basis for $\HERM$ (sometimes known as the Fano form~\cite{F83}):
\begin{eqnarray}\label{7_eqn:fano}
        \rho = &\frac{1}{MN}&(I^A\otimes I^B + \ve{r}^A\cdot\ve{\sigma}^A\otimes{I^B}+\hspace{10mm}\\&&
                    I^A\otimes\ve{r}^B\cdot\ve{\sigma}^B+\sum_{i=1}^{M^2-1}\sum_{j=1}^{N^2-1}T_{ij}\sigma^A_i\otimes\sigma^B_j).\nonumber
\end{eqnarray}
Here, $\ve{\sigma}^A$ is a $(M^2-1)$-component vector of
traceless orthogonal Hermitian basis elements $\sigma_i^A$ satisfying $\trace(\sigma_i^A\sigma_j^A)=2\delta_{ij}$, $\ve{r}^A\in\reals^{M^2-1}$
is the Bloch vector for subsystem $A$ with
$r^A_i=\frac{M}{2}\tr(\rho_A\sigma^A_i)$, and $T\in\reals^{(M^2-1)\times(N^2-1)}$ is the \emph{correlation matrix} with entries
$T_{ij}=\frac{MN}{4}\tr(\sigma^A_i\otimes\sigma^B_j\rho)$. For $M=2$, $\ve{r}_A$ satisfies $0\leq\enorm{\ve{r}_A}\leq 1$ with $\enorm{\ve{r}_A} = 1$ if and only if $\rho_A$ is pure. The
definitions for subsystem $B$ are analogous.

We now give a useful specific construction for the basis elements $\sigma_i^A$~\cite{he81}. Define
$\set{\sigma_i}_{i=1}^{M^2-1}= \set{U_{pq},V_{pq},W_{r}}$, such
that for $1\leq p<q\leq M$ and $1\leq r \leq M-1$, and
$\set{\ket{i}}_{i=1}^{M}$ some orthonormal basis for $\A$:
    \begin{eqnarray}
        U_{pq}&=&\ket{p}\bra{q}+\ket{q}\bra{p}\label{7_eqn:Ugenerators}\\
        V_{pq}&=&-i\ket{p}\bra{q}+i\ket{q}\bra{p}\label{7_eqn:Vgenerators}\\
        W_{r} &=&\sqrt{\frac{2}{r(r+1)}}\!\!\left(\sum_{k=1}^{r}\ket{k}\bra{k}-r\ket{r+1}\bra{r+1}\right).\label{7_eqn:Wgenerators}
    \end{eqnarray}
Note that when $M=2$, this construction yields the Pauli matrices $\ve{\sigma^A}=(X,Y,Z)$.

Regarding $D(\rho)$, defining $\rho_f := (\UA\otimes I_B) \rho(\UA^\dagger\otimes I_B)$, we often use the fact that Equation~(\ref{7_eqn:measure_def}) can be rewritten as:
\begin{equation}\label{7_eqn:measure_def2}
    D(\rho) = \min_{U_A\in\rous} \sqrt{\trace(\rho^2) -\trace(\rho\rho_f)}.
\end{equation}

Finally, we show a simple but important lemma.

\begin{lemma}\label{7_lem:invariantlocal}
    $D(\rho)$ is invariant under local unitary operations.
\end{lemma}
\begin{proof}
    Let $\rho' := (V_A\otimes V_B)\rho (V_A\otimes V_B)^\dagger$ for unitaries $V_A$, $V_B$. Then in Equation~(\ref{7_eqn:measure_def2}), $\trace(\rho'^2)=\trace(\rho^2)$, and $\trace(\rho'\rho'_f)$ becomes
    \begin{equation}
        \trace(\rho (V_A^\dagger\UA V_A\otimes I_B )\rho (V_A^\dagger\UA^\dagger V_A\otimes I_B )).
    \end{equation}
   Observe, however, that $V_A\UA V_A^\dagger$ is still an RU unitary, since we have simply changed basis. Hence, $D(\rho',\UA)=D(\rho,V_{A}^\dagger\UA V_{A})$, and since we are minimizing over all $\UA\in\rous$, the claim follows.
\end{proof}

%///////////////////////////////////////////////////////////////////////////////////////////////////////////////////
\section{$(2\times N)$-dimensional states}\label{7_scn:twoqubitpure}
%///////////////////////////////////////////////////////////////////////////////////////////////////////////////////
In this section, we study $D(\rho)$ for $\rho\in\denstn$, obtaining among other results a closed from expression for $D(\rho)$. To begin, note that any $U_A\in \rous$ must have the form
    \begin{equation}\label{7_eqn:2qU}
        U_A := \ketbra{c}{c}-\ketbra{d}{d}=2\ketbra{c}{c}-I_2,
    \end{equation}
    up to an irrelevant global phase which disappears upon application of $U_A$ to our system, and for some orthonormal basis $\set{\ket{c},\ket{d}}$ for $\complex^2$. Then, $D(\rho,\UA)$ can be rewritten as
    \begin{equation}\label{7_eqn:mx2start}
      2\sqrt{\trace[\rho^2 (\ketbra{c}{c}\otimes I) - \rho( \ketbra{c}{c}\otimes I)\rho (\ketbra{c}{c}\otimes I)]}.
    \end{equation}

We begin with a simple upper bound on $D(\rho)$.
\begin{theorem}\label{7_thm:upperbound1}
     For any $\rho\in \denstn$, one has
     \begin{equation}
        D(\rho)\leq 2\sqrt{\lambda_{\min}(\trace_{\B}(\rho^2))}.
     \end{equation}
\end{theorem}
\begin{proof}
    Starting with Equation~(\ref{7_eqn:mx2start}), by noting that $\trace[\rho( \ketbra{c}{c}\otimes I )\rho (\ketbra{c}{c}\otimes I)]\geq 0$ and using the fact that $\trace(\rho(C_A\otimes I_B))=\trace(\rho_A C_A)$, we have that $D(\rho)$ is at most
    \begin{eqnarray}
        \min_{\text{unit }\ket{c}\in\complex^2}2\sqrt{\trace[\trace_{\B}(\rho^2)\ketbra{c}{c}]}
               = 2\sqrt{\lambda_{\min}(\trace_{\B}(\rho^2))}.\qedhere
    \end{eqnarray}
\end{proof}

Theorem~\ref{7_thm:upperbound1} implies that for pure product $\ket{\psi}\in\complex^2\otimes\complex^N$, $D(\ketbra{\psi}{\psi})=0$, in agreement with the results in Reference~\cite{GI07}. By next exploiting the structure of $\rho$ further, we obtain a closed form expression for $D(\rho)$.

\begin{theorem}\label{7_thm:closedform}
    For any $\rho\in \denstn$, define $G:=\ve{r}^A(\ve{r}^A)^T + \frac{2}{N}TT^T$, for $T$ the correlation matrix of $\rho$. Then, $D(\rho)$ equals
%    \begin{equation}\label{7_eqn:2qd} \frac{1}{\sqrt{N}}\sqrt{\enorm{\ve{r}^A}^2+\frac{2}{N}\sum_{i,j=1}^3T_{ij}^2-\lambda_{\max}\left(\ve{r}^A(\ve{r}^A)^T + \frac{2}{N}TT^T\right)}.
%    \end{equation}
    \begin{equation}\label{7_eqn:2qd} \frac{1}{\sqrt{N}}\sqrt{\trace(G)-\lambda_{\max}(G)}=\frac{1}{\sqrt{N}}\sqrt{\lambda_2(G)+\lambda_3(G)}.
    \end{equation}
\end{theorem}
\begin{proof}
Define $P:=\ketbra{c}{c}$. Then, beginning with Equation~(\ref{7_eqn:mx2start}), by rewriting $\rho$ using Equation~(\ref{7_eqn:fano}) and applying the fact that the basis elements $\sigma_i$ are traceless, we obtain that $\trace(\rho^2 P\otimes I-\rho P\otimes I\rho P\otimes I)$ equals
\begin{equation}
    \frac{1}{4N}\trace(A_1-A_2+A_3-A_4),
\end{equation}
where
\begin{eqnarray}
    A_1 &:=& \left(\sum_i r_i^A{\sigma_i}^A\right)^2P,\quad\quad\quad
    A_2 := \left(\sum_i r_i^A{\sigma_i}^AP\right)^2\\
    A_3 &:=& \frac{1}{N}\left(\sum_{ij}T_{ij}\sigma^A_i\otimes\sigma^B_j\right)^2(P\otimes I)\\
    A_4&:=&\frac{1}{N}\left(\sum_{ij}T_{ij}\sigma^A_i\otimes\sigma^B_j\right)\left(\sum_{ij}T_{ij}P\sigma^A_iP\otimes \sigma^B_j\right).
\end{eqnarray}
Using the facts that $(\sigma_i^A)^2=I$, $\set{\sigma^A_i,\sigma^A_j}=0$ for $i\neq j$, $\trace(\sigma_i\sigma_j)=2\delta_{ij}$, and $\trace(P)=1$, we thus have
\begin{eqnarray}
    \trace(A_1) &=& \enorm{\ve{r}^A}^2,\quad\quad\quad     \trace(A_3) = \frac{2}{N}\sum_{ij} T_{ij}^2\\
    \trace(A_2)&=&\sum_{ij} r_i^Ar_j^A\bra{c}\sigma^A_i\ket{c}\bra{c}\sigma^A_j\ket{c}\\
    \trace(A_4)&=&\frac{2}{N}\sum_{ij} \left(\sum_k T_{ik}T_{jk}\right)\bra{c}\sigma^A_i\ket{c}\bra{c}\sigma^A_j\ket{c}.
\end{eqnarray}
Now, $\bra{c}\sigma^A_i\ket{c}$ can be thought of as the $i$th component of the Bloch vector $\ve{v}\in\reals^3$ of pure state $\ket{c}$ with $\enorm{\ve{v}}=1$, implying
\begin{equation}
    \trace(A_2+A_4) = \ve{v}^T\left[\ve{r}^A(\ve{r}^A)^T+\frac{2}{N}TT^T\right]\ve{v}.
\end{equation}
Plugging these values into Equation~(\ref{7_eqn:mx2start}), we conclude $D(\rho)$ equals
\begin{equation} \min_{\substack{\ve{v}\in\reals^3\\\enorm{\ve{v}}=1}}\frac{1}{\sqrt{N}}\sqrt{\enorm{\ve{r}^A}^2+ \frac{2}{N}\sum_{ij} T_{ij}^2-\trace(A_2+A_4)}.
\end{equation}
The claim now follows since for any symmetric $A\in\reals^{n\times n}$,
$\max_{\text{unit }\ve{v}\in\reals^n}\ve{v}^T A\ve{v}=\lambda_{\max}(A)$.
\end{proof}

The expression for $D(\rho)$ in Theorem~\ref{7_thm:closedform} matches that for the \emph{geometric discord}~\cite{DVB10,VR12}. Specifically, defining the latter as $\delta_g(\rho)=\min_{\sigma\in\Omega} \sqrt{2}\fnorm{\rho-\sigma}$, where $\Omega$ is the set of zero-discord states, we have for $(2\times N)$-dimensional $\rho$ that $D(\rho)=\delta_g(\rho)$. (Note: The original definition of Reference~\cite{DVB10} was more precisely $\delta_g(\rho)=\min_{\sigma\in\Omega} \fnorm{\rho-\sigma}^2$.)

We now discuss consequences of Theorem~\ref{7_thm:closedform}, beginning with a lower bound which proves useful later.

\begin{corollary}\label{7_cor:lower}
    For $\rho\in \denstn$, we have
    \begin{equation}\label{7_eqn:22qd} D(\rho)\geq \frac{\sqrt{2}}{N}\sqrt{\lambda_2(TT^T)+\lambda_3(TT^T)}.
    \end{equation}
    This holds with equality if $\ve{r}^A=0$, i.e.\ $\rho_A=\frac{I}{2}$.
\end{corollary}

\begin{proof}
    The first claim follows from the fact that:
    \begin{eqnarray}
        \lambda_{\max}\left(\ve{r}^A(\ve{r}^A)^T + \frac{2}{N}TT^T\right)\leq\enorm{\ve{r}^A}^2+ \frac{2}{N}\lambda_{\max}\left(TT^T\right).
    \end{eqnarray}
    The second claim follows by substitution into Equation~(\ref{7_eqn:2qd}).
\end{proof}

For example, for maximally entangled $\ket{\psi}=(\ket{00}+\ket{11})/\sqrt{2}$, for which $\ve{r}^B=\ve{0}$ and $T=\operatorname{diag}(1,-1,1)$, Corollary~\ref{7_cor:lower} yields $D(\ketbra{\psi}{\psi})= 1$, as desired. We also remark that Equation~(\ref{7_eqn:2qd}) can further be simplified for two-qubit states, since by Reference~\cite{HH96,HH96_2}, one can assume without loss of generality that $T$ is diagonal. This relies on the facts that (1) applying local unitary $V_1\otimes V_2$ to $\rho$ has the effect of mapping $T\mapsto O_1TO_2^\dagger$, $\ve{r}^A\mapsto O_1\ve{r}^A$, and $\ve{r}^B\mapsto O_2\ve{r}^B$ for some orthogonal rotation matrices $O_1$ and $O_2$, and (2) $D(\rho)$ is invariant under local unitaries by Lemma~\ref{7_lem:invariantlocal}.

Using Corollary~\ref{7_cor:lower}, we next obtain a connection to the CHSH inequality for two-qubit $\rho$. Defining $M(\rho) := \lambda_1(T^T T)+\lambda_2(T^TT)$, it is known that $\rho$ violates the CHSH inequality if and only if $M(\rho)>1$~\cite{HHH95}. We thus have:

\begin{corollary}\label{7_cor:CHSH}
    For $\rho\in \denstt$, if $D(\rho)>1/\sqrt{2}$, then $M(\rho)>1$. The converse does not hold.
\end{corollary}
\begin{proof}
    The first is immediate from Corollary~\ref{7_cor:lower} and the fact that $TT^T$ and $T^TT$ are cospectral (Theorem 1.3.20 of~\cite{HJ90}). The converse proceeds similarly to Theorem 7 of Reference~\cite{gkb08} --- namely, let $\ket{\psi}=a\ket{00}+b\ket{11}$ for real $a,b\geq 0$ and $a^2+b^2=1$. Then, for density operator $\ketbra{\psi}{\psi}$, we have $\ve{r}^B=(0,0,a^2-b^2)$ and $T=\operatorname{diag}(2ab,-2ab,1)$, implying $M(\ketbra{\psi}{\psi})>1$ for $a,b\neq 0$. In comparison, $D(\ketbra{\psi}{\psi})=2ab\leq1/\sqrt{2}$ when $a\leq \sqrt{\frac{1}{2}- \frac{1}{2\sqrt{2}}}$ or $a\geq \sqrt{\frac{1}{2}+ \frac{1}{2\sqrt{2}}}$.
\end{proof}
Interestingly, the exact same relationship as that in Corollary~\ref{7_cor:CHSH} was found between the Fu distance and the CHSH inequality in Reference~\cite{gkb08}.

%///////////////////////////////////////////////////////////////////////////////////////////////////////////////////
\section{Werner states}\label{7_scn:werner}
%///////////////////////////////////////////////////////////////////////////////////////////////////////////////////

We now derive a closed formula for $D(\rho)$ for Werner states $\rho\in \densdd$ where $d\geq 2$, which are defined as~\cite{W89}
\begin{equation} \rho:=\frac{2p}{d^2+d}P_{s} + \frac{2(1-p)}{d^2-d}P_a, \end{equation}
for $P_s:= (I+P)/2$ and $P_a:=(I-P)/2$ the projectors onto the symmetric and anti-symmetric subspaces, respectively, $P:=\sum_{i,j=1}^d\ketbra{i}{j}\otimes\ketbra{j}{i}$ the SWAP operator, and $0\leq p \leq 1$. Werner states are invariant under $U\otimes U$ for any unitary $U$, and are entangled if and only if $p< 1/2$.
\begin{theorem}\label{7_thm:werner}
    Let $\rho\in \densdd$ be a Werner state. Then
    \begin{equation}  D(\rho) = \frac{\abs{2pd-d-1}}{d^2-1}.
    \end{equation}
\end{theorem}
\begin{proof}
    As done in Theorem 3 of Reference~\cite{gkb08}, we first rewrite Equation~\ref{7_eqn:measure_def2} using the facts that $\trace(P)=d$, $\trace(P^2)=d^2$, and $\beta:=\trace(P(U_{A}\otimes I)P(U_{A}\otimes I)^\dagger)=\trace(U_{A})\trace(U_{A}^\dagger)$ to obtain that for any $U_{A}\in U(\A)$,
    \begin{equation}
        D(\rho, U_{A})=\frac{\sqrt{(2pd-d-1)^2(d^2-\beta)}}{d(d^2-1)}.
    \end{equation}
    Since $\trace(U_{A})=0$ for any $U_{A}\in\rous$, we have $\beta=0$ and the claim follows.
\end{proof}

Again, we find that this coincides exactly with the expression for the Fu distance for Werner states~\cite{gkb08}. Further, Theorem~\ref{7_thm:werner} implies that the quantum discord of Werner state $\rho$ is zero if and only if $p=(d+1)/2d$. This matches the results of Chitambar~\cite{C11}, who develops the following closed formula for the discord $\discm$ of Werner states:
\begin{eqnarray}
    \discm &=&\log (d+1) + (1-p)\log\frac{1-p}{d-1}+p\log\frac{p}{d+1} -\nonumber\\ &&\frac{2p}{d+1}\log p- \left(1-\frac{2p}{d+1}\right)\log \frac{d+1-2p}{2(d-1)}\label{7_eqn:discWerner}.
\end{eqnarray}
In Section~\ref{7_scn:discord}, we show that this is no coincidence --- it turns out that $D(\rho)= 0$ if and only if the discord of $\rho$ is zero for any $\rho$.

%///////////////////////////////////////////////////////////////////////////////////////////////////////////////////
\section{Pure states of arbitrary dimension}\label{7_scn:purearbdim}
%///////////////////////////////////////////////////////////////////////////////////////////////////////////////////

We now show that only pure maximally entangled states $\rho$ achieve $D(\rho)=1$. As mentioned in Section~\ref{7_scn:intro}, this is in contrast to the Fu distance~\cite{f06,gkb08}, whose maximal value is attained even for certain \emph{non-maximally} entangled $\ket{\psi}$. We remark that Theorem~\ref{7_thm:arbdimpure} below also follows from a more general non-trivial result that $D(\ketbra{\psi}{\psi})^2$ is tightly upper bounded by the linear entropy of entanglement of pure state $\ket{\psi}$~\cite{MAGGDI11}. However, our proof of Theorem~\ref{7_thm:arbdimpure} is much simpler and requires only elementary linear algebra.

%Assume without loss of generality that $M\leq N$, and write arbitrary $\ket{\psi}\in \A\otimes\B$ in its Schmidt decomposition as
%%\begin{equation}\label{7_eqn:arbdimpureschmidt}
%    $\ket{\psi}=\sum_{k=1}^{M}\alpha_k\ket{a_k}\otimes\ket{b_k}$,
%%\end{equation}
%where $\sum_k\alpha_k^2=1$ for $\alpha_k\in\reals$ and $\set{\ket{a_k}}$ and $\set{\ket{b_k}}$ are the Schmidt bases for $\A$ and $\B$, respectively. Then, Equation~(\ref{7_eqn:measure_def2}) can be rewritten as
%\begin{equation}\label{7_eqn:arbdimpureD}
%    D(\ketbra{\psi}{\psi}) = \min_{\UA\in\rous}\sqrt{1- \abs{\sum_{k=1}^{M}\alpha_k^2\bra{a_k}\UA\ket{a_k}}^2}.
%\end{equation}

To begin, assume without loss of generality that $M\leq N$, and let $\ket{\psi}\in\A\otimes\B$ be a pure quantum state with Schmidt decomposition $\ket{\psi}=\sum_{k=1}^{M}\alpha_k\ket{a_k}\otimes\ket{b_k}$, i.e.\ $\sum_k\alpha_k^2=1$ for $\alpha_k\in\reals$ and $\set{\ket{a_k}}$ and $\set{\ket{b_k}}$ the Schmidt bases for $\A$ and $\B$, respectively.

\begin{theorem}\label{7_thm:arbdimpure}
Let $\ket{\psi}\in\A\otimes\B$ with Schmidt decomposition as above. Then $D(\ketbra{\psi}{\psi})=1$ if and only if $\alpha_k=\frac{1}{\sqrt{M}}$ for all $1\leq k \leq M$ (i.e.\ $\ket{\psi}$ is maximally entangled).
\end{theorem}
\begin{proof}
We begin by rewriting Equation~(\ref{7_eqn:measure_def2}) as
\begin{equation}\label{7_eqn:arbdimpureD}
    D(\ketbra{\psi}{\psi}) = \min_{\UA\in\rous}\sqrt{1- \abs{\sum_{k=1}^{M}\alpha_k^2\bra{a_k}\UA\ket{a_k}}^2}.
\end{equation}
If $\ket{\psi}$ is maximally entangled, then $\alpha_k = 1/\sqrt{M}$ for all $1\leq k\leq M$. Then, since $\UA\in\rous$, Equation~(\ref{7_eqn:arbdimpureD}) yields
\begin{eqnarray}
    D(\ketbra{\psi}{\psi}) =% \min_{\UA\in\rous}\sqrt{1- \frac{1}{M^2}\abs{\sum_{k=1}^{M}\bra{a_k}\UA\ket{a_k}}^2}\\
     \min_{\UA\in\rous}\sqrt{1- \frac{1}{M^2}\abs{\trace(\UA)}^2}
    = 1.
\end{eqnarray}

For the converse, assume $D(\ketbra{\psi}{\psi})=1$. Then, by Equation~(\ref{7_eqn:arbdimpureD}), we must have that for all $\UA\in\rous$,
\begin{equation}\label{7_eqn:arbdim_pure_proof}
    \sum_{k=1}^{M}\alpha_k^2\bra{a_k}\UA\ket{a_k}=0.
\end{equation}
Thus, choosing $\UA$ as diagonal in basis $\set{\ket{a_k}}$, Equation~(\ref{7_eqn:arbdim_pure_proof}) says that $\ve{w}^T\pi\ve{v}=0$ for all permutations $\pi\in S_M$, where ${w}_k:=\alpha_k^2$ and ${v}_k := e^{2\pi ki/M}$. This can only hold, however, if all entries of $\ve{w}$ are the same, i.e.\ $\alpha_k=1/\sqrt{M}$ for all $1\leq k\leq M$, as desired.
\end{proof}

\begin{corollary}\label{7_cor:max}
    A quantum state $\rho\in\dens$ achieves $D(\rho)=1$ if and only if $\rho$ is pure and maximally entangled.
\end{corollary}
\begin{proof}
    Immediate from Theorem~\ref{7_thm:arbdimpure} and the $\trace(\rho^2)$ in Equation~(\ref{7_eqn:measure_def2}).
\end{proof}

%///////////////////////////////////////////////////////////////////////////////////////////////////////////////////
\section{Relationship to quantum discord}\label{7_scn:discord}
%///////////////////////////////////////////////////////////////////////////////////////////////////////////////////
We now show that for arbitrary $\rho\in\dens$, $D(\rho)$ is zero if and only if the quantum discord~\cite{ollivier01a,henderson01a} $\discm$ of $\rho$ is zero. (The discord $\discm$ was defined in Section~\ref{0_sscn:nonclassical}.)
%\begin{equation}\label{7_eqn:discord_def}
%    \discm := S(A)-S(A,B)+\min_{\proja}S(B|\proja),
%\end{equation}
%where $\proja$ corresponds to a complete measurement on subsystem $B$ consisting of rank $1$ projectors, $S(B)=-\tr(\rho_B\log(\rho_B))$ is the von Neumann entropy of $\rho_B$, similarly $S(A,B)=S(\rho)$, and
%\begin{equation}
%    S(B|\proja) = \sum_{j}p_jS\left(\frac{1}{p_j}\Pi_j^A\otimes I^B\rho \Pi_j^A\otimes I^B\right),
%\end{equation}
%where $p_j=\tr(\Pi_j^A\otimes I^B\rho)$.

The main fact we leverage about the discord here is the following.
\begin{theorem}[Ollivier and Zurek~\cite{ollivier01a}]\label{7_thm:olzu}
    For $\rho\in\dens$, $\discm=0$ if and only if
    \begin{equation}\label{7_eqn:disc_exp_proof_eq1}
        \rho = \sum_j \Pi_j^A\otimes I^B\rho \Pi_j^A\otimes I^B,
    \end{equation}
    for some complete set of rank $1$ projectors $\proja$.
\end{theorem}

We now prove the main result of this section. The first part of the proof involves a new characterization of the set of zero discord quantum states $\rho$ in terms of the basis elements $\sigma^A_i$ from the Fano form of $\rho$. Key to this characterization is the absence of non-diagonal $\sigma^A_i$ in the expansion of $\rho$. In the proofs below, we assume the basis elements $\sigma_i^A$ for $\spa{A}$ come from the set $\set{I, U_{pq}, V_{pq},W_{r}}_{p,q,r}^A$ from Section~\ref{7_scn:prelim} (analogously for $\spa{B}$).

\begin{theorem}\label{7_thm:discequiv}
    Let $\rho\in\dens$. Then $\discm=0$ if and only if there exists a local unitary $V^A$ such that
    \begin{equation}
        \trace\left(\left(V^A\otimes I^B\right)\rho\left({V^A}^\dagger\otimes I^B\right) \left(\sigma_i^A\otimes \sigma_j^B\right)\right)=0
    \end{equation}
    for all $\sigma_i^A\in\set{U_{pq},V_{pq}}^A$ and all $\sigma^B_j\in\set{I, U_{pq}, V_{pq},W_{r}}^B$. The same characterization holds for $D(\rho)=0$.
\end{theorem}
\begin{proof}
We prove the equivalent statement that $\delta(\rho)=0$ if and only if there exists an orthonormal basis $\set{\ket{k}}$ for $\A$ such that, for basis elements $\sigma_i^A$ constructed with respect to $\set{\ket{k}}$, we have $\trace(\rho (\sigma_i^A\otimes \sigma_j^B))=0$ for all $\sigma_i^A\in\set{U_{pq},V_{pq}}$ (and similarly for $D(\rho)=0$).

    Suppose $\discm=0$. Then by Theorem~\ref{7_thm:olzu}, there exists a complete set of rank 1 projectors $\proja$ such that Equation~(\ref{7_eqn:disc_exp_proof_eq1}) holds. Let $\set{\ket{k}}$ be the basis onto which $\proja$ projects, and define $\Phi(C):=\sum_j\Pi_j^AC\Pi_j^A$. By constructing the basis elements $\sigma_i^A$ in Equation~(\ref{7_eqn:fano}) using $\set{\ket{k}}$, we thus have
    \begin{eqnarray}\label{7_eqn:disc_exp_proof_eq2}
        \rho &=& \frac{1}{MN}\left[ I^A\otimes I^B + {I^A}\otimes\ve{r}^B\cdot\ve{\sigma}^B+\hspace{10mm}\right.\\&&
                    \left.\sum_{i=1}^{M^2-1}\Phi(\sigma^A_i)\otimes\left(r^A_iI^B+\sum_{j=1}^{N^2-1}T_{ij}\sigma^B_j\right) \right] .\nonumber
    \end{eqnarray}
    Now, for all $\sigma_i^A\in\set{W_r}$, we clearly have $\Phi(\sigma_i^A)=\sigma_i^A$. For $\sigma_i^A\in\set{U_{pq},V_{pq}}$, however, $\Phi(\sigma_i^A)=0$. Thus, in order for Equation~(\ref{7_eqn:disc_exp_proof_eq1}) to hold, we must have $r_i^A=T_{ij}=0$ for all basis elements $\sigma_i^A\in\set{U_{pq},V_{pq}}$, which by definition means $\trace(\rho(\sigma_i^A\otimes \sigma_j^B))=0$ for all $\sigma_i^A\in\set{U_{pq},V_{pq}}^A$, as desired. To show that this implies $D(\rho)=0$, construct $U^A\in\rous$ as diagonal in basis $\set{\ket{k}}$ and define $\Phi(C):=U^AC {U^A}^\dagger$. Then since in Equation~(\ref{7_eqn:disc_exp_proof_eq2}), we have $\Phi(\sigma_i^A)=\sigma_i^A$ for any $\sigma_i^A\in\set{I,W_r}$, the claim follows.

    To show the converse, assume $D(\rho,U^A)=0$ for some $U^A\in\rous$. Then, construct the basis elements $\sigma^A_i$ with respect to a diagonalizing basis $\set{\ket{k}}$ for $U^A$ and define $\Phi(C):=U^AC {U^A}^\dagger$. It follows that for any $p$ and $q$,
    \begin{eqnarray}
        \Phi(U_{pq})&=&e^{i(\theta_p-\theta_q)}\ketbra{p}{q}+e^{-i(\theta_p-\theta_q)}\ketbra{q}{p},\label{7_eqn:u1}\\
        \Phi(V_{pq})&=&-ie^{i(\theta_p-\theta_q)}\ketbra{p}{q}+ie^{-i(\theta_p-\theta_q)}\ketbra{q}{p}\label{7_eqn:u2}.
    \end{eqnarray}

    \noindent Consider now an arbitrary term $(c_{u}\sigma^A_u+c_v\sigma^A_v)\otimes\sigma^B_j$ from the Fano form of $\rho$ where $\sigma^A_u=U_{pq}$ and $\sigma^B_v=V_{pq}$ for some choice of $p$ and $q$. Since Equations~(\ref{7_eqn:u1}) and~(\ref{7_eqn:u2}) imply that $U^A$ can only map $U_{pq}$ to $V_{pq}$ and vice versa, it follows that in order for $D(\rho,U^A)=0$ to hold, we must have
$        \Phi(c_{u}\sigma^A_u+c_v\sigma^A_v)=c_{u}\sigma^A_u+c_v\sigma^A_v.
$
    This leads to the system of equations
    \begin{eqnarray}
        c_u-ic_v&=&e^{i(\theta_p-\theta_q)}(c_u-ic_v)\\
        c_u+ic_v&=&e^{-i(\theta_p-\theta_q)}(c_u+ic_v).
    \end{eqnarray}
    We conclude that if either $c_u\neq0$ or $c_v\neq0$, it must be that $\theta_p=\theta_q$ in order for $D(\rho)=0$ to hold. However, since all eigenvalues of $U^A$ are distinct by definition, this is impossible. Thus, $\trace(\rho (\sigma_i^A\otimes \sigma_j^B))=0$ for all $\sigma_i^A\in\set{U_{pq},V_{pq}}$, as desired. To see that this implies $\discm=0$, simply now choose $\proja$ as the projection onto $\set{\ket{k}}$. Then, defining $\Phi(C):=\sum_j\Pi_j^AC\Pi_j^A$ and applying the same arguments from the forward direction to Equation~(\ref{7_eqn:disc_exp_proof_eq2}), we conclude that $\rho$ is invariant under $\proja$. By Theorem~\ref{7_thm:olzu}, we have $\discm=0$, completing the proof.
\end{proof}

Theorem~\ref{7_thm:discequiv} shows that $D(\rho)$ defined in Equation~(\ref{7_eqn:measure_def}) is zero precisely for the set of states classically correlated in $\spa{A}$. In other words, unlike the Fu distance~\cite{DG08}, $D(\rho)$ is indeed a \emph{faithful} non-classicality measure. The proof of Theorem~\ref{7_thm:discequiv} does, however, have a curiosity --- the key property the proof relies on is that all $U^A\in\rous$ have non-degenerate spectra. Interestingly, this is the mixed-state analogue of the pure-state result of Reference~\cite{MAGGDI11}, where it was shown that a non-degenerate spectrum suffices to conclude $D(\ketbra{\psi}{\psi})$ is a faithful {entanglement monotone} for pure states $\ket{\psi}$. Specifically, Reference~\cite{MAGGDI11} shows that if in Equation~(\ref{7_eqn:measure_def}) we minimize over $U^A$ with eigenvalues of multiplicity at most $k$ (with at least one eigenvalue of multiplicity $k$), then $D(\ketbra{\psi}{\psi})=0$ if and only if $\ket{\psi}$ has Schmidt rank at most $k$. Could there be an analogue of this more general result in the mixed-state setting of non-classicality? It turns out the answer is \emph{yes}.

Let $\ve{v}\in\nats^M$ such that $\sum_{j=1}^M v_jj=M$. Then, consider an arbitrary (i.e.\ not necessarily RU) unitary $\UAv$ which has precisely $v_j$ distinct eigenvalues with multiplicity $j$. For example, $\UAv\in\rous$ has $\ve{v}=(M,0,\ldots,0)$ since it has $M$ distinct eigenvalues of multiplicity $1$. Similarly, if $\ve{v}=(0,0,\ldots,1)$, then $\UAv$ is just the identity (up to phase), and if $\ve{v}=(M-4,2,\ldots,0)$ then $\UAv$ has $M-4$ distinct eigenvalues of multiplicity $1$, and two distinct eigenvalues with multiplicity $2$ each. Now, corresponding to any $\UAv$ is a complete projective measurement $\projv$ which consists precisely of $v_j$ projectors of rank $j$. The correspondence is simple: Let $\lambda$ be an eigenvalue of $\UAv$ with multiplicity $j$, i.e.\ the projector $\Pi_\lambda$ onto its eigenspace has rank $j$. Then $\Pi_\lambda\in \projv$. It is easy to see that similarly, corresponding to any $\projv$ is a $\UAv$ (assuming we are not concerned with the precise eigenvalues of $\UAv$, as is this case here). We can now state the following.

\begin{theorem}\label{7_thm:gendiscequiv}
    Let $\rho\in\dens$ and $\ve{v}\in\nats^M$ such that $\sum_{j=1}^M v_jj=M$. Then, there exists a complete projective measurement $\projv$ such that
    \begin{equation}
        \rho = \sum_j \Pi_j^A\otimes I^B\rho \Pi_j^A\otimes I^B\label{7_eqn:invar}
    \end{equation}
    if and only if there exists a $\UAv\in \mathcal{U}(\spa{A})$ with $D(\rho,\UAv)=0$.
\end{theorem}
\begin{proof}
    The proof follows that of Theorem~\ref{7_thm:discequiv}, so we outline the differences. Here, $\UAv$ and $\projv$ will be related through the correspondence outlined above, and the basis elements $\sigma_i^A$ are constructed with respect to a diagonalizing basis $\set{\ket{k}}$ for $\UAv$ (which by definition also diagonalizes each $\Pi_j^A\in\projv$). For simplicity, we discuss the case of $\ve{v}=(M-2,1,0,\ldots,0)$; all other cases proceed analogously.

    Going in the forward direction, suppose $\Pi^A_j\in\projv$ projects onto $\spa{S}_{pq}:=\operatorname{span}(\ket{p},\ket{q})$. Then, in Equation~(\ref{7_eqn:disc_exp_proof_eq2}), $\Phi(\sigma_i^A)=\sigma_i^A$ for $\sigma_i^A=U_{pq}$ and $\sigma_i^A=V_{pq}$. In other words, now we can have $r_i^A\neq 0$ and $T_{ij}\neq 0$ (however, note we still have $r_{m\neq i}^A=0$ and $T_{m\neq i,j}=0$). Since $\UAv$ has a degenerate eigenvalue on $S_{pq}$, however, we have by Equations~(\ref{7_eqn:u1}) and~(\ref{7_eqn:u2}) that $\UAv$ acts invariantly on $\sigma_i^A$ as well (since $\theta_p=\theta_q$). The converse is similar; namely, suppose $\UAv$ has a degenerate eigenvalue on $\mathcal{S}_{pq}$. Then the projector onto the corresponding two-dimensional eigenspace $\Pi^A_j\in\projv$ is $\Pi^A_j=\ketbra{p}{p}+\ketbra{q}{q}$. It thus follows by the same argument as above that both $\UAv$ and $\Pi^A_j$ act invariantly on $U_{pq}$ and $V_{pq}$.
\end{proof}

From this general theorem, we can re-derive as a simple corollary the pure state result of Reference~\cite{MAGGDI11} mentioned earlier, which we rephrase in our terminology as follows.

\begin{corollary}
    Let $\ket{\psi}=\sum_{i=1}^r\alpha_i\ket{\psi^A_i}\ket{\psi^B_i}$ be the Schmidt decomposition of $\ket{\psi}\in\spa{A}\otimes\spa{B}$. Then, there exists $\UAv\in\mathcal{U}(\spa{A})$ with $v_k\geq 1$ (i.e.\ $\UAv$ has an eigenvalue of multiplicity $k$), $v_{k'>k}=0$ (all eigenvalues of $\UAv$ have multiplicity at most $k$), and $D(\ketbra{\psi}{\psi},\UAv)=0$ if and only if $k\geq r$.
\end{corollary}
\begin{proof}
    Suppose $k\geq r$. Then, by defining $\projv^k$ such that $v_k\geq 1$ and $v_{k'>k}=0$, one can choose a $\projv^k$ such that Equation~(\ref{7_eqn:invar}) holds for $\rho=\ketbra{\psi}{\psi}$ (i.e.\ simply project onto $\operatorname{span}(\set{\ket{\psi^A_i}})$). By Theorem~\ref{7_thm:gendiscequiv}, this implies there exists a $\UAv$ with $v_k\geq 1$ and $v_{k'>k}=0$ achieving $D(\ketbra{\psi}{\psi},\UAv)=0$. Conversely, if $k<r$, then clearly no such $\projv^k$ such that Equation~(\ref{7_eqn:invar}) holds exists. By Theorem~\ref{7_thm:gendiscequiv}, this implies that no $\UA$ with an eigenvalue of multiplicity at most $k$ and $D(\ketbra{\psi}{\psi},\UA)=0$ exists, as desired.
\end{proof}

We close this section with two final comments. First, given Theorem~\ref{7_thm:discequiv}, one might ask whether a stronger relationship between $D(\rho)$ and $\discm$ holds. For example, could it be that $D(\rho)\geq\discm$ for all $\rho$? This simplest type of relationship is ruled out easily via Theorem~\ref{7_thm:werner} and Equation~(\ref{7_eqn:discWerner}), since for $d=2$ and $p=2/3$, $D(\rho)=1/9 \geq \discm \approx 0.01614$, while for $d=50$ and $p=2/3$, $D(\rho)\approx 0.00627 \leq \discm\approx 0.07111$.

Second, note that Theorem~\ref{7_thm:gendiscequiv} reduces to Theorem~\ref{7_thm:discequiv} if we choose $\ve{v}=(M,0,\ldots,0)$. This suggests defining a \emph{generalized quantum discord}, denoted $\delta_{\ve{v}}(\rho)$, which is analogous to $\discm$, except that now we use the class of measurements $\projv$ in the definition of discord (see Equation~\ref{4_eqn:discord_def}). For example, $\delta_{(M,0,\ldots,0)}(\rho)=\discm$. We hope the study of $\delta_{\ve{v}}(\rho)$ would prove fruitful in its own right.

%///////////////////////////////////////////////////////////////////////////////////////////////////////////////////
\section{Maximally non-classical, yet separable, $(2\times N)$-dimensional states}\label{7_scn:maxnc}
%///////////////////////////////////////////////////////////////////////////////////////////////////////////////////
In this section, we characterize the set of maximally non-classical, yet separable, $(2\times N)$-dimensional states of rank at most $2$, as quantified by $D(\rho)$. To do so, consider separable state
\begin{equation}\label{7_eqn:sepstate}
    \rho = \sum_{i=1}^{n} p_i \ketbra{a_i}{a_i}\otimes\ketbra{b_i}{b_i},
\end{equation}
where $\sum_i p_i = 1$, $\ket{a_i}\in\complex^2$, $\ket{b_i}\in\complex^N$. Via simple algebraic manipulation, one then finds that $D(\rho,U_A)$ for any given $U_A\in\mathcal{U}(\spa{A})$ is given by
\begin{equation}\label{7_eqn:sepstateD} \sqrt{\sum_{i=1}^n\sum_{j=1}^np_ip_j\abs{\braket{b_i}{b_j}}^2(\abs{\braket{a_i}{a_j}}^2-\abs{\bra{a_i}U_A\ket{a_j}}^2)}.
\end{equation}
We begin by proving a simple but useful upper bound on $D(\rho)$ which depends solely on $n$.

\begin{lemma}\label{7_lem:upperbound}
    Let $\rho$ be a separable state as given by Equation~(\ref{7_eqn:sepstate}). Then $D(\rho)\leq 1-\max_i p_i \leq1-\frac{1}{n}$.
\end{lemma}
\begin{proof}
    Assume WLOG that $\max_i p_i = p_1$. Then $1/n \leq p_1\leq 1$. Choose any $U_A\in\mathcal{U}(\spa{A})$ such that $\ket{a_1}$ is an eigenvector of $U_A$. Then any term in the double sum of Equation~(\ref{7_eqn:sepstateD}) in which $\ket{a_1}$ appears vanishes. We can hence loosely upper bound the value of Equation~(\ref{7_eqn:sepstateD}) by
    $
        \sqrt{(\sum_{i\neq1,j\neq 1}p_ip_j)}=1-p_1.
    $
    Recalling that $p_1\geq1/n$ yields the desired bound.
\end{proof}

When $n=2$, i.e.\ when $\rho$ is rank at most two, observe from Lemma~\ref{7_lem:upperbound} that $D(\rho)\leq 1/2$, and this is attainable only when $p_1=p_2=1/2$. We now show that this bound can indeed be saturated, and characterize all states with $n=2$ that do so.

\begin{lemma}\label{7_lem:upperbound2}
    Let $\rho$ be a separable state as in Equation~(\ref{7_eqn:sepstate}) with $p_1=p_2=1/2$. Then $D(\rho)=1/2$ if and only if $\abs{\braket{a_1}{a_2}}=1/\sqrt{2}$ and $\braket{b_1}{b_2}=0$.
\end{lemma}
\begin{proof}
    Since by Lemma~\ref{7_lem:invariantlocal}, $D(\rho)$ is invariant under local unitaries, we can assume without loss of generality that $\ket{a_1}=\ket{0}$, $\ket{b_1}=\ket{0}$, $\ket{a_2}=\cos\frac{\beta}{2}\ket{0} + \sin\frac{\beta}{2}\ket{1}$ and $\ket{b_2}=\sum_{i=0}^{N-1}\alpha_i\ket{i}$ for $\beta\in[0,\pi]$ and $\alpha_i\in\reals$ with $\sum_i\alpha_i^2=1$, i.e.\ we can rotate the local states so as to eliminate relative phases. Further, since $\UA\in\rous$ in Equation~(\ref{7_eqn:sepstateD}), we can write $\UA=2\ketbra{u}{u}-I$ for some $\ket{u}=\cos\frac{\theta}{2}\ket{0} + e^{i\phi}\sin\frac{\theta}{2}\ket{1}$, where $\theta,\phi\in[0,2\pi)$. Via the latter, we can rewrite Equation~(\ref{7_eqn:sepstateD}) as:
    \begin{equation}\label{7_eqn:sepstateD2} \frac{1}{2}\sqrt{\sum_{i,j=1}^2\braket{b_i}{b_j}^2(\braket{a_i}{a_j}^2-\abs{\braket{a_i}{a_j} - 2\braket{a_i}{u}\braket{u}{a_j}}^2)}.
    \end{equation}
    Letting $\Delta$ denote the expression under the square root above, we have by substituting in our expressions for $\ket{a_1}$, $\ket{a_2}$, $\ket{b_1}$, $\ket{b_2}$, and $\ket{u}$ and algebraic manipulation that
    \begin{eqnarray}
        \Delta &=& \alpha_0^2\left[2\cos\beta
        \sin^2\theta - \sin\beta\sin(2\theta)\cos\phi\right]+\nonumber\\&& 1 + \sin^2\theta - (\cos\beta\cos\theta + \sin\beta\sin\theta\cos\phi)^2.\label{7_eqn:delta}
    \end{eqnarray}
     Our goal is to maximize $\Delta$ with respect to $\alpha_0$ and $\beta$ (which define $\rho$), and then minimize with respect to $\theta$ and $\phi$ (which define $\UA$). Observe now that choosing $\phi=\theta=0$ reduces Equation~(\ref{7_eqn:delta}) to $\Delta = 1-\cos^2\beta$. Hence, unless $\beta=\pi/2$ (i.e.\ $\abs{\braket{a_1}{a_2}}=1/\sqrt{2}$), we can always achieve $D(\rho)<1/2$. Thus, set $\beta=\pi/2$. Consider next $\phi=0$, and leave $\theta$ unassigned. Then, Equation~(\ref{7_eqn:delta}) reduces to $\Delta = 1 - \alpha_0^2\sin(2\theta)$, from which it is clear that unless $\alpha_0=0$ (i.e.\ $\braket{b_1}{b_2}=0$), we can always achieve $D(\rho)<1/2$. Plugging these values of $\alpha$ and $\beta$ into Equation~(\ref{7_eqn:delta}), we have $\Delta = 1 +\sin^2\theta\sin^2\phi$, from which the claim follows.
\end{proof}

For two-qubit $\rho$, we thus have that with respect to $D(\rho)$ and the geometric discord, the maximally non-classical two qubit states of rank at most two are, up to local unitaries,
\[
    \frac{1}{2}\ketbra{0}{0}\otimes\ketbra{0}{0}+\frac{1}{2}\ketbra{+}{+}\otimes\ketbra{1}{1},
\]
where $\ket{+}=(\ket{0}+\ket{1})/\sqrt{2}$. As mentioned earlier, this matches known results with respect to the relative entropy of quantumness~\cite{GPACH11}. However, the latter analysis is not as general as it begins by with the assumption that $\braket{b_1}{b_2}=0$, whereas we allow arbitrary $\ket{b_1},\ket{b_2}$. It would be interesting to know whether this analysis can be extended to arbitrary rank two-qubit states.\\

%%///////////////////////////////////////////////////////////////////////////////////////////////////////////////////
%\section*{Acknowledgements}
%%///////////////////////////////////////////////////////////////////////////////////////////////////////////////////
%
\noindent\emph{Acknowledgements for this chapter.} We thank Gerardo Adesso, Dagmar Bru\ss, Davide Girolami and Marco Piani for helpful discussions.

%\noindent \emph{Note:} After completion of this chapter, the author learned of independent work in preparation on a similar topic~\cite{SGRBI12}.

%======================================================================
\chapter{All non-classical correlations can be activated into distillable entanglement}\label{chap:activation}

\emph{This chapter is based on~\cite{PGACHW11}:}\\

\vspace{-4mm}
\noindent M.~Piani, S.~Gharibian, G.~Adesso, J.~Calsamiglia, P.~Horodecki and A.~Winter. All non-classical correlations can be activated into distillable entanglement.
\emph{Physical Review Letters}, 106:220403, 2011, DOI: 10.1103/PhysRevLett.106.220403, \copyright~2011 American Physical Society, prl.aps.org.

%%The correlations of multipartite quantum states have non-classical features that go beyond entanglement.
%In this chapter, we devise a protocol in which general non-classical multipartite correlations produce a physically relevant effect, leading to the creation of bipartite entanglement. In particular, we show that the relative entropy of quantumness, which measures all non-classical correlations
% among subsystems of a quantum system, is equivalent to and can be operationally interpreted as the minimum distillable entanglement generated between the system and local ancillae in our protocol.
%%All and only the non-classically correlated states  possess an entanglement potential and are thus useful for quantum information processing.
%We emphasize the key role of state mixedness in maximizing non-classicality: Mixed entangled states can be arbitrarily more non-classical than separable and pure entangled states.

\vspace{3mm}
\noindent In this chapter, we introduce a protocol through which general non-classical multipartite correlations can be mapped or ``activated'' into bipartite entanglement. In particular, we provide an operational interpretation for the measure of non-classicality known as the relative entropy of quantumness, showing that it quantifies the minimum distillable entanglement generated between the initial system and ancillae in our protocol. Moreover, we show the following surprising fact: That mixed entangled states can be arbitrarily more non-classical than separable and pure entangled states.

\section{Introduction and results}
The study of quantum correlations has traditionally  focused on entanglement~\cite{HHHH09}. In particular, it is generally believed that entanglement is a necessary resource for quantum computers to outperform their classical counterparts. Indeed, it has been shown that for the setting of \emph{pure}-state computation, the amount of entanglement present must grow with the system size for an exponential speed-up to occur~\cite{jozsa03a}. In the context of \emph{mixed}-state quantum information processing, however, there are surprising quantum computational and communication feats which are seemingly impossible to achieve with a classical computer, and yet can be attained with a quantum computer using little or no entanglement. Examples include the DQC1 model of computing~\cite{kl98} and the locking of classical correlations~\cite{dhlst04}; see Section~\ref{0_sscn:nonclassical} for a brief exposition. In the case of locking, for example, the task involved is impossible classically, and yet the quantum states used are \emph{separable}. This raises the question: What is the fundamental resource enabling such feats?

One plausible explanation is the presence in (generic~\cite{ferraro}) quantum states of non-classical correlations \emph{beyond} entanglement. Indeed, as outlined in Section~\ref{0_sscn:nonclassical}, much attention has recently been devoted to understanding and quantifying such correlations for this reason
\cite{ollivier01a,henderson01a,PhysRevA.71.062307,PhysRevA.72.032317,MPSVW10,groismanquantumness,PhysRevA.77.052101,Luo08,Bravyi2003,pianietal2008nolocalbrodcast,pianietal2009broadcastcopies,ferraro,ADA,PhysRevA.82.052342}.
In particular, the separable quantum states of the systems involved in DQC1 and the locking protocol have been shown to possess non-zero amounts of such correlations (see e.g.~\cite{datta08a,DG08}), as measured by the {\it quantum discord}
\cite{ollivier01a,henderson01a}. The latter strives to capture non-classical correlations beyond entanglement and has recently received operational interpretations in terms of the quantum state merging protocol~\cite{CABMPW10,MD10}, but is unfortunately not a \emph{faithful} measure (here, a \emph{faithful} measure achieves a non-zero value of zero if and only if a state is ``non-classical'').  A more accurate quantification of non-classical correlations is provided by the so-called {\em relative entropy of quantumness} (REQ) \cite{Bravyi2003,PhysRevA.71.062307,groismanquantumness,PhysRevA.77.052101,MPSVW10}, defined as the minimum distance, in terms of relative entropy, between a multipartite quantum state and the closest strictly classically correlated state (see Definition~\ref{5_def:classical}). Such a measure is faithful \cite{groismanquantumness}, symmetric under permutation of the subsystems, and enables a unified approach to the quantification of classical, separable and entangled correlations \cite{MPSVW10}. However, to date it still lacks an operational interpretation.

% More generally, the role of non-classical correlations in relation to quantum information tasks remains unclear.
% While we know that all entangled states are useful for  information processing \cite{masanes06a,PianiALL}, the fundamental question of whether the same holds for all non-classically correlated (separable) states has remained open.
%In this context,  we ask the general question:
%Is there a protocol by which general non-classical correlations produce a physically relevant effect that distinguishes them from purely classical ones?

More generally, in this chapter, we ask the following question: \emph{Is there a protocol by which general non-classical correlations produce a physically relevant effect that distinguishes them from purely classical ones?}

\noindent It turns out that the answer to the above question is not only \emph{yes}, but that among other results, the protocol we derive lends the desired operational interpretation to the REQ.

\paragraph{Our results:} In order to summarize our results, recall first from Equation~(\ref{eqn:NCdef1}) the definition of a \emph{strictly classically correlated} or \emph{classical} state in the bipartite setting. For completeness, we state the generalization of this definition to the multipartite setting below~\cite{pianietal2008nolocalbrodcast}.

\begin{definition}[Strictly classically correlated quantum state]
\label{5_def:classical}
Let $\rho\in\DD((\complex^d)^{\otimes n})$, i.e.\ $\rho$ acts on n $d$-dimensional systems. Let $\cBB_i:=\set{\ket{{\cB}_i(j)}}_{j=0}^{d-1}$ denote some orthonormal basis for $\complex^d$ for the $i$th system, and let ${\cBB}$ denote the orthonormal basis
\begin{equation}
\{\ket{\cB({k})}:=\ket{\cB_1(k_1)}\ket{\cB_2(k_2)}\cdots \ket{\cB_n(k_n)}\}
\end{equation}
for the entire space $(\complex^{d})^{\otimes n}$ formed by taking tensor products of all elements in bases $\set{{\cBB}_i}_{i=1}^n$. Here, ${k}:=k_1k_2\cdots k_n$ is a number written in base $d$.
%(e.g. $\ket{B_1(1)}\otimes \ket{B_2(1)}\otimes\cdots\otimes \ket{B_n(1)}\in C$).
We henceforth use the notation $\cBB$ to refer to such a \emph{local product basis}. Then, an $n$-qudit state $\rho$ is \emph{strictly classically correlated}, or \emph{classical}, if there exists a local product basis $\cBB$ with respect to which $\rho$ is diagonal.
\end{definition}
\noindent Recall that classical states correspond to the embedding of a multipartite classical probability distribution into the quantum formalism, and that states not of the form above are called \emph{non-classical}. We now summarize our results as follows.

\vspace{2mm}
\noindent\textbf{1. An ``activation'' protocol for non-classical correlations.} Our first result is a protocol through which non-classical correlations are mapped into entanglement. Roughly, given an input state $\rho\in\DD((\complex^d)^{\otimes n})$, the protocol first introduces an ancilla state $\ket{0\cdots 0}\in({\complex^d})^{\otimes n}$. We then show that $\rho$ is non-classically correlated if and only if applying local CNOT gates with system $i$ of $\rho$ as control and system $i$ of the ancilla as target always creates (distillable) entanglement across the system-ancilla split, \emph{even if} one adversarially applies local changes of basis to $\rho$ before applying the CNOT gates (Theorem~\ref{5_thm:quantumnessiff}).

We thus not only have a physical effect arising from non-classical correlations, as desired, but also an entire framework for designing non-classicality measures. Specifically, for each choice of entanglement measure one applies across the system-ancilla gap after the protocol is run, we have the potential for a new non-classicality measure for system $\rho$.

\noindent\textbf{2. Connections to non-classicality measures.} As mentioned above, by applying our favorite entanglement measure across the system-ancilla cut after our protocol is run, we have the potential for discovering new non-classicality measures for the initial system $\rho$. In this vein, we first find that applying the entanglement measure \emph{distillable entanglement}~\cite{reviewplenio}, we obtain a non-classicality measure we call the \emph{minimum distillable entanglement potential}, which turns out to equal the REQ (Corollary~\ref{5_cor:req}). We thus have an operational interpretation for the REQ. We also consider the \emph{negativity}~\cite{VW02} as an entanglement measure, obtaining various results of interest here (Section~\ref{5_sscn:neg}).

\noindent\textbf{3. Mixedness versus entanglement in non-classicality.} Our final result studies the minimum distillable entanglement potential (or equivalently, REQ). As might be expected, we first find that according to this non-classicality quantifier, pure entangled states are strictly ``more non-classical'' than separable states. However, perhaps surprisingly, we next show that in the asymptotic setting, (1) separable states can be as non-classical as pure entangled states (Theorem~\ref{5_prop:sep}), and (2) \emph{mixed} entangled states can be much more non-classical than pure entangled states (Theorem~\ref{5_prop:low-rank})! This suggests that non-classical correlations arise not just from the superposition principle of quantum mechanics, as is the case with (pure state) entanglement, but also due to the non-commutative nature of quantum physics. Our proofs here use ideas similar to known concentration of measure arguments~\cite{HLW:aspects,HLSW04}.
%
%We demonstrate a protocol which ``activates'' the non-classicality present in any multipartite quantum system given as input, leading to the creation of entanglement. Moreover, we find the minimum amount of distillable entanglement generated in the protocol is precisely equal to the REQ of the input state. We thus have that the REQ is \emph{both} an operational \emph{and} faithful non-classicality measure. We next prove that non-classical correlations are bounded for
%separable states and for pure entangled states in any dimension, while, perhaps surprisingly, these very limits can be exceeded by {\it mixed} entangled states.

\paragraph{Previous work.} We refer the reader to Section~\ref{0_sscn:nonclassical} for a brief introduction to non-classical correlations. With regards to this chapter, we remark that after completion of the paper this chapter is based on, we became aware of related results by Streltsov, Kampermann and Bru\ss~\cite{streltsov2010}. They show that the quantumness of correlations (as measured, for example, by the quantum discord) is also related to the minimum entanglement generated between system and apparatus in a partial measurement process. In light of those results, our findings can be understood also as dealing with the interplay between system-apparatus entanglement and non-classicality of correlations when realizing local measurements.

\paragraph{Discussion and open questions.} The study of general non-classical correlations is currently a burgeoning area, but in many ways such correlations are still not well-understood. Our activation protocol lends new insight into the nature of these correlations by furnishing them with a new operational meaning in terms of resources for \emph{entanglement generation}. One natural and interesting open question is whether the ideas behind the protocol could lead to novel applications in quantum computation and information.

Furthermore, our novel framework for non-classicality measures reduces the problem of non-classicality quantification to the more familiar setting of entanglement quantification, for which a multitude of tools for analysis are already known (see e.g.~\cite{HHHH09}). An open question here is what further known non-classicality quantification schemes can be obtained as arising through our framework?

Finally, that mixing can actually help \emph{surpass} the quantumness of pure-state entanglement, and that the latter can be asymptotically matched by fully separable states is, in our opinion, quite a surprising result. It would be good to better understand the non-commutative nature of states in a quantum mixture, both from the perspective of non-classical correlations, as well as with regard to computational and information theoretic feats.

\paragraph{Organization of chapter.}
%We begin in Section~\ref{5_scn:prelim} with notation specific to this chapter and background information.
In Section~\ref{5_scn:protocol}, we describe our activation protocol, and show how it yields a connection between entanglement and non-classical correlations. Section~\ref{5_scn:quantify} then exploits this connection further by introducing an entire family of non-classicality quantifiers, demonstrating along the way an operational interpretation of the REQ. In Section~\ref{5_scn:mixed}, we show two surprising results in systems of large local dimensions: That mixed separable states can be asymptotically as non-classical as pure maximally entangled states, and that mixed entangled states can be asymptotically twice as non-classical as pure maximally entangled states.

%SEV: distillable entanglement, entanglement monotone

\section{Preliminaries}\label{5_scn:prelim}
We now state notation and a lemma specific to this chapter. Regarding notation, given a local product basis $\cBB=\set{\ket{b(k)}}$ and multipartite quantum state $\rho$, we define
\begin{equation}
    \rho^{\wek{\cBB}}:=\sum_{\wek{k}}\proj{\wek{\cB}(\wek{k})}\rho \proj{\wek{\cB}(\wek{k})},
    \end{equation}
    and
\begin{equation}
    \rho_{\wek{k}\wek{l}}^\wek{\cBB}:=\bra{\wek{\cB}(\wek{k})}\rho \ket{\wek{\cB}(\wek{l})}.
\end{equation}

We next state Levy's Lemma, which is useful in Section~\ref{5_scn:mixed}. For this, we first define the \emph{Lipschitz constant} of a function $f$. Given function $f:X\mapsto Y$ for metric spaces $(X,d_{X})$ and $(Y,d_{Y})$, where $d_{X}$ and $d_{Y}$ are metrics on the sets $X$ and $Y$, respectively, we say that $f$ has Lipschitz constant $m\geq 0$ if the distance between any two input points in $X$ does not increase by more than $m$ after going through $f$. In other words, for all $x_1,x_2\in X$,
\begin{equation}
    d_Y(f(x_1),f(x_2))\leq m\cdot d_X(x_1,x_2).
\end{equation}

Then, for $\mathbb{S}^k$ the $k$-sphere and $\mathbb{E}(f)$ the expected value of function $f$, we can state the following useful Lemma, known as Levy's Lemma.

\begin{lemma}[Levy's Lemma, see e.g.~\cite{HLW:aspects}]\label{5_lem:levy}
    Let $f: \mathbb{S}^k\mapsto\reals$ be a function whose Lipschitz constant with respect to the Euclidean norm is $m\geq 0$. Let $x\in\mathbb{S}^k$ be chosen uniformly at random. Then, for some constant $c>0$,
    \begin{equation}
        \pr\left({f(x)-\mathbb{E}(f)} \gtrless \pm \alpha\right) \leq 2\exp \left(\frac{-c(k+1)\alpha^2}{m^2}\right).
    \end{equation}
\end{lemma}

\section{The activation protocol}\label{5_scn:protocol}
We now describe our protocol for the activation of non-classical correlations, which maps relatively ``not-well-understood'' non-classical correlations into ``more familiar'' bipartite entanglement, allowing one to employ tools from entanglement theory \cite{HHHH09} to study general non-classical correlations.
 %The scheme is somewhat inspired by the quantum optics setup of~\cite{john}, where  one attempts to quantify non-classicality of a single field mode (defined there as the state deviation from a mixture of coherent states) by reducing the problem to quantifying the two-mode \emph{entanglement} that can be generated from
%the field using linear optics, auxiliary classical (coherent) states, and ideal photodetectors.
%Similarly, we may expect that  mapping the (still not-well-understood)  non-classicality of multipartite \emph{correlations} into ``more familiar'' bipartite entanglement allows one to employ tools from entanglement theory \cite{HHHH09} to study ---e.g., interpret and quantify--- general non-classical correlations.
The protocol can be thought of as a game between an adversary and $n$ players, where the $n$ players together aim to generate an entangled state between a system $\wek{A}$ they control and an ancillary system $\wek{A'}$, and the adversary's goal is to thwart their efforts by locally rotating each subsystem of $\wek{A}$ before system and ancilla undergo a pre-defined interaction.

\begin{figure}[t]
\begin{center}
\includegraphics[width=8cm]{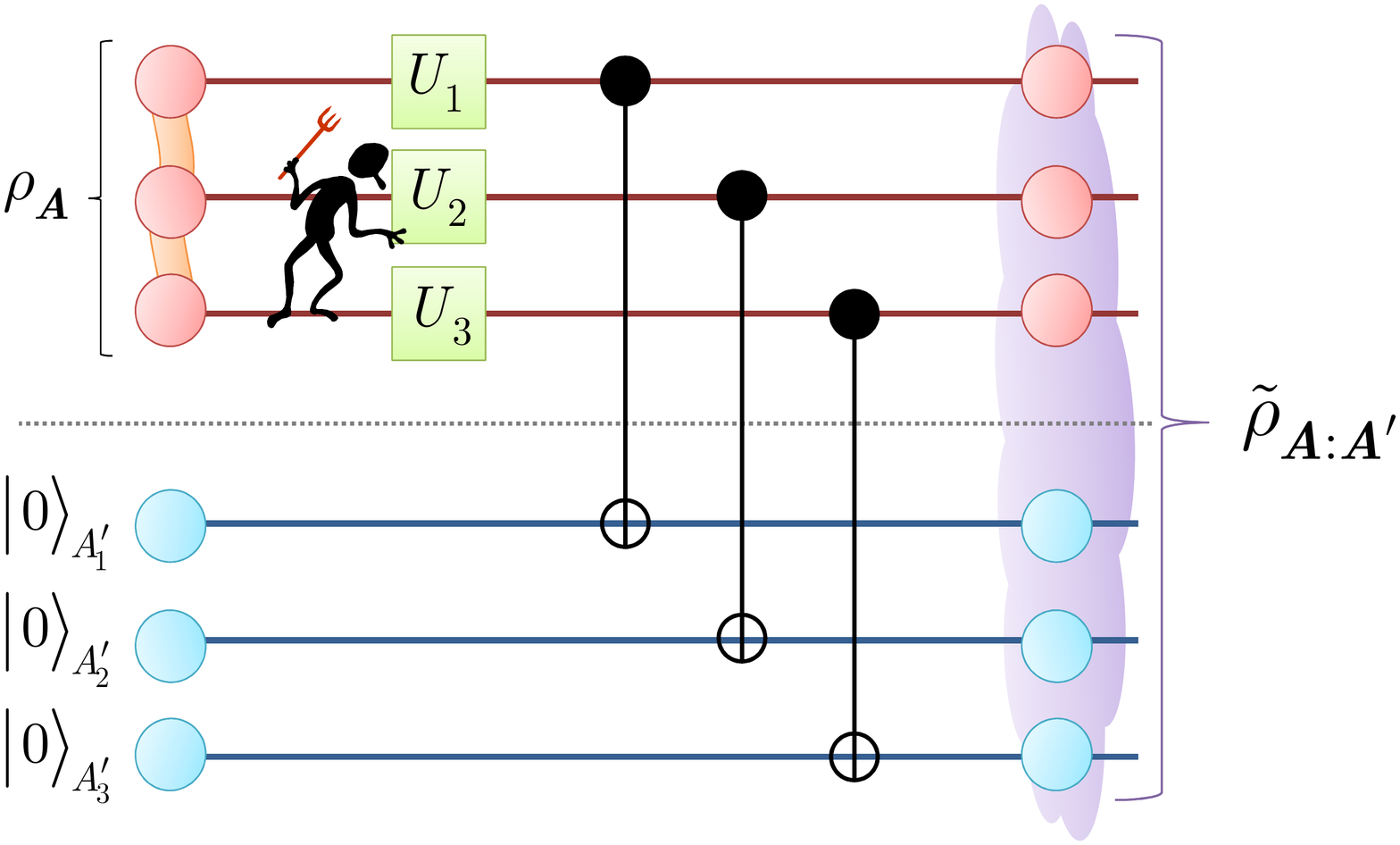}
\caption{(Color online) Scheme of the activation protocol for $n=3$.}\label{figact}
\end{center}
\end{figure}
More precisely, the  protocol proceeds as follows (see Figure~\ref{figact}).
We consider $n$ players ${\cal P}_i$, each controlling a system-ancilla pair of qudits $(A_i, A'_i)$.
We indicate by $\wek{A}$ the joint register  $A_1,\ldots,A_n$, henceforth called the ``system'', and by $\wek{A'}$ the joint register $A'_1,\ldots,A'_n$, henceforth called the ``ancilla''. The initial state of the total $2n$ qudits is a tensor product  $\rho_{\wek{A}\wek{A'}} = \rho_{\wek{A}} \otimes \ket{0}\bra{0}^{\otimes n}_{\wek{A'}}$.
For a given  $\rho_\wek{A}$, an adversary is first allowed to apply a local unitary $U_i$ of his choice to each $A_i$.
With the adversary's turn complete, each player ${\cal P}_i$ now  lets their subsystem $A_i$ (control qudit) interact with the corresponding ancillary party $A'_i$ (target qudit) via a CNOT gate,
whose action  on the computational basis states $\ket{j}\ket{j^\prime}$ of $\complex^d \otimes \complex^d$ is defined as
$\ket{j}\ket{j^\prime}\mapsto\ket{j}\ket{j^\prime\oplus j}$,
with $\oplus$ denoting addition modulo $d$. The final state of system plus ancilla is
\begin{equation}\label{5_eqn:finalstate}
\rho^f_{{A}{A'}} = V (\rho_{\wek{A}} \otimes \ket{0}\bra{0}^{\otimes n}_{\wek{A'}}) V^\dagger\,,
\end{equation}
where $V = {{CNOT}}_{\wek{A}\wek{A'}} (U_{\wek{A}} \otimes I_{\wek{A'}})$, $U_{\wek{A}} = \otimes_{i=1}^n U_i$ and ${{CNOT}}_{\wek{A}\wek{A'}}=\bigotimes_{i=1}^n CNOT_{A_iA_i'}$.
We ask: At the end of the protocol, have the $n$ players succeeded in generating bipartite entanglement across the split $\wek{A}:\wek{A'}$, and, if so, how much entanglement was created? It is natural to expect that the answer will depend on the initial state $\rho_{\wek {A}}$ of the $n$-qudit system. For simplicity of notation, in the remainder of this chapter, we shall take $\rho$ and $\rho^f$ to denote the states $\rho_A$ and $\rho^f_{AA'}$, respectively.

Although we cast the activation protocol as a game, from a more physical perspective our aim is to understand precisely how the nature and amount of correlations between the parts $A_i$ of the system $\wek{A}$ affects the entanglement that can be created with an ancilla $\wek{A'}$ via the paradigmatic entangling operation --- the CNOT; we are considering here the worst case scenario with respect to the choice of the control bases. We then find the following.

\begin{theorem}
\label{5_thm:quantumnessiff}
A state $\rho$ of an $n$-qudit system is classical if and only if there exists an adversarial choice of local unitaries $U_{\wek{A}}$ such that the state $\rho^f$ output by the activation protocol is separable across the system-ancilla (i.e.\ $A:A'$) split.
\end{theorem}
\begin{proof} The ``if'' part is trivial, as given a strictly classically correlated state, one can choose $U_{\wek{A}}$ to rotate the diagonalizing local product basis for $\rho=\sum_{i}p_{i}\proj{\wek{\cB}(i)}$ into the computational basis, so that applying the CNOTs in our protocol straightforwardly yields the separable state
\begin{equation}
\rho^f=\sum_{i}p_{i}\proj{i}_{\wek{A}}\otimes\proj{i}_{\wek{A'}}.
\end{equation}
As for the ``only if'' part, consider the separable decomposition
\begin{equation}
\rho^f=\sum_i q_i\proj{\psi_i}_{\wek{A}}\ot \proj{\phi_i}_{\wek{A'}},
\end{equation}
which exists by hypothesis for some choice of $U_{\wek{A}}$.
Since the transformation $V$ in Equation~\eqref{5_eqn:finalstate} is unitary, we must be able to write $\rho=\sum_i{q_i}\proj{v_i}_{\wek{A}}$ for some ensemble (not necessarily a spectral decomposition) $\set{q_i,\ket{v_i}}$ such that
\beq
\label{5_eqn:purecond}
V\ket{v_i}_{\wek{A}}\ket{0}_{\wek{A'}}=\ket{\psi_i}_{\wek{A}} \ket{\phi_i}_{\wek{A'}}.
\eeq
Letting $\set{\ket{j}}$ denote the computational basis, we now expand $\ket{v_i}$ in the basis $\set{U^\dagger_{\wek{A}}\ket{j}}$, such that
\begin{equation}
\ket{v_i}_{\wek{A}}=\sum_{j}\alpha_{ij}U^\dagger_{\wek{A}}\ket{j}_{\wek{A}},
\end{equation}
from which it follows that
\begin{equation}
V\ket{v_i}_{\wek{A}}\ket{0}_{\wek{A'}}=\sum_{j}\alpha_{ij}\ket{j}_{\wek{A}}\ket{j}_{\wek{A'}}.
\end{equation}
Combining this with Equation~\eqref{5_eqn:purecond}, we conclude that for all $i$, there must exist a $j$ such that $\alpha_{ij}=1$. Denote this value of $j$ as $j_i$, and note hence that
 \begin{equation}
    \ket{v_i}_{\wek{A}}\ket{0}_{\wek{A'}}=V^\dagger\left(\ket{j_i}_{\wek{A}}\ket{j_i}_{\wek{A'}}\right)=U_{\wek{A}}^\dagger\ket{j_i}_{\wek{A}}\ket{0}_{\wek{A'}}.
 \end{equation}
 We can now write
\begin{equation}
\rho
=\sum_i{q_i}\proj{v_i}_{\wek{A}}
=\sum_i{q_i}U_{\wek{A}}^\dagger\ketbra{j_i}{j_i}_{\wek{A}}U_{\wek{A}},
\end{equation}
which is a spectral decomposition for $\rho$ with respect to the computational basis up to local unitary $U_{\wek{A}}$, as desired.
\end{proof}

In other words, the system \emph{always} becomes (for any choice of $U_{\wek{A}}$) entangled with the ancilla as a result of the activation protocol, if and only if the input state of the system is non-classically correlated. This establishes a qualitative \emph{equivalence} between multipartite non-classical correlations among components of a quantum system, and bipartite entanglement between the system and an ancilla.

%To prove this, one can use arguments very similar to those of \cite{john}. We do not report explicitly such a proof because in the following we will actually be able to quantify the non-vanishing minimum entanglement generated if the state is not classical.
%On the other hand, it is easy to see that if the state is striclty classically correlated, we can choose the unitaries to make the factorized basis $\wek{\cB}$ in which the state is diagonal coincide with the control bases of the CNOT gates, so that entanglement is not created.

\section{Quantifying non-classicality}\label{5_scn:quantify}
We now exploit the spirit of Theorem~\ref{5_thm:quantumnessiff} further to \emph{quantify}, rather than simply detect, the presence of non-classical correlations in a quantum state. To do so, our approach is to apply entanglement measures across the $\wek{A}:\wek{A^\prime}$ split to study the amount of entanglement generated whenever $\wek{A}$ is initially in a non-classically correlated state. It is worth remarking here that this framework is general enough to possibly uncover a full zoology of non-classicality measures, as each choice of a different entanglement monotone~\cite{reviewplenio} we adopt (at the output) has the potential to lead to a unique non-classicality measure (for the input state), the  association being provided exactly by the activation protocol.

More precisely, let $E$ denote some entanglement measure of choice and $\rho^f$ the system-ancilla state at the end of the protocol as in Equation~(\ref{5_eqn:finalstate}), and define by \begin{equation}\label{5_eqn:QE}Q_E({\rho}):=\min_{U_{{\wek{A}}}}E_{\wek{A}:\wek{A'}}(\rho^f)\,
 \end{equation}
 the minimum entanglement generated across the $\wek{A}:\wek{A^\prime}$ split over all choices of adversarial local unitaries $U_\wek{A}$. We call $Q_E({\rho})$ the \emph{minimum entanglement potential} of ${\rho}$ with respect to $E$. As a consequence of Theorem~\ref{5_thm:quantumnessiff}, $Q_E$ is a  measure of non-classical correlations for arbitrary multipartite qudit states $\rho$, induced by the entanglement monotone $E$. In fact, the condition  $Q_E({\rho})=0$ perfectly characterizes the set of classically correlated states $\rho$ if $E$ is a \emph{faithful} entanglement measure (i.e.\ if $E$ vanishes only for separable states). However, even certain non-faithful entanglement measures can be plugged in to obtain a faithful measure of non-classical correlations.
%(analogously, a faithful non-classicality measure vanishes if and only if a state is strictly classically correlated).
The reason is that the output state $\rho^f$ has the so-called \emph{maximally correlated} form~\cite{R01} between $\wek{A}$ and $\wek{A'}$; namely,
\begin{equation}\label{5_eqn:maxcorr}
\rho^f=\sum_{{k}{l}}\rho_{{k}{l}}^\wek{\cBB}\ket{{k}}\bra{l}_{\wek{A}}\otimes \ket{{k}}\bra{{l}}_{\wek{A'}}
\end{equation}
with $\rho_{{k}{l}}^\wek{\cBB}=\bra{\wek{\cB}({k})}\rho \ket{\wek{\cB}({l})}$,
$\ket{{\cB}({k})}=U_\wek{A}^\dagger \ket{{k}}$ and $\ket{{k}}=\ket{k_1}\ket{k_2}\cdots\ket{k_n}$. We now exploit this observation in the next section.
%We hence can obtain an \emph{operational} measure of non-classical correlations, which unlike the quantum discord, is also faithful.

\subsection{Minimum distillable entanglement potential}

Let us consider the non-faithful but physically motivated distillable entanglement $E_{\textup{D}}$~\cite{reviewplenio} as a bipartite entanglement monotone (recall $E_{\textup{D}}$ is non-faithful as it vanishes on so-called bound entangled states). Note that the precise definition of $E_{\textup{D}}$ is not required here; rather we utilize results of~\cite{Hiroshima2004} linking $E_{\textup{D}}(\rho)$ to the relative entropy of entanglement. Specifically, we have the following.

\begin{theorem}\label{5_thm:distill}
    The minimum distillable entanglement potential $Q_{E_D}({\rho})$ equals
\beq
Q_{E_{\textup{D}}}(\rho)=\min_{{\cBB}}\Big(S(\rho^{\wek{\cBB}})-S(\rho)\Big),\label{5_eqn:DE}
\eeq
where the minimization is over the choice of local product bases $\wek{\cBB}$.
\end{theorem}
\begin{proof}
The claim follows by observing that for any choice of $\wek{\cBB}$, the $\wek{A}:\wek{A'}$ distillable entanglement of $\rho^f$ is equal to
\begin{equation}
E_{\textup{D}}(\rho^f)=S(\trace_{A'}(\rho^f))-S(\rho^f)=S({\rho}^\wek{\cBB})-S(\rho),
						%&=H(\{\rho^u_{i,i}\})-S(\rho_{\wek{A}})
\end{equation}
where $S(\sigma)=-\Tr( \sigma \log \sigma)$ is the von Neumann entropy of a state $\sigma$. In the first equality we used the results of~\cite{Hiroshima2004} about distillable entanglement for maximally correlated states --- for which it happens to coincide with the relative entropy of entanglement~\cite{PhysRevLett.78.2275,PhysRevA.57.1619}. The second equality is justified by the fact that $\rho^{\wek{{\cBB}}}$ is the state resulting from local projective measurements in the local bases $\wek{\cBB}$ on $\rho$ and is unitarily equivalent to $\trace_{A'}(\rho^f)$ (seen by considering Equation~(\ref{5_eqn:maxcorr})), while $\rho^f$ is obtained from $\rho$ via the activation protocol isometry, Equation~(\ref{5_eqn:finalstate}).
\end{proof}
This yields the following nice corollary regarding the REQ, which is defined as (see also Section~\ref{0_sscn:nonclassical})
\beq
\label{5_eqn:req}
Q(\rho)=\min_{\textrm{classical}~\sigma}S(\rho\|\sigma),
\eeq
for $S(\rho\|\sigma)=\trace(\rho \log \rho - \rho \log \sigma)$ the relative entropy and where the minimization is over all strictly classically correlated states $\sigma$.
\begin{corollary}\label{5_cor:req}
    The REQ of $\rho$ equals its minimum distillable entanglement potential, i.e.\ \begin{equation}Q(\rho)=Q_{E_{\textup{D}}}(\rho).\end{equation}
\end{corollary}
\begin{proof}
    The claim follows immediately from the result that, as proven for example in Theorem 2 of \cite{MPSVW10}, the REQ can alternatively be expressed as the expression in Equation~(\ref{5_eqn:DE}).
\end{proof}

This finding immediately provides a clear-cut \emph{operational interpretation} for the REQ, which therefore emerges as a natural, mathematically sound  and physically motivated measure of non-classical correlations in  quantum states of arbitrary-dimensional composite systems. The degree of non-classical correlations as quantified by the REQ, a measure whose original definition was purely geometric \cite{MPSVW10}, is quantitatively reinterpreted as the resource power of such correlations for the task of generating distillable entanglement with an ancilla in the worst case scenario.
Incidentally, since the REQ is faithful \cite{groismanquantumness}, this can be considered an alternate proof of Theorem~\ref{5_thm:quantumnessiff}.

Before closing this section, we prove a strict upper bound on the non-classicality of \emph{separable} bipartite quantum states with respect to $Q_{E_{\textup{D}}}$.

\begin{theorem}\label{5_thm:sepuppbnd}
  Consider bipartite separable state $\rho_{AB}=\sum_ip_i\proj{\alpha_i}\otimes\proj{\beta_i}$, for $\{p_i\}$ a probability distribution and $\set{\ket{\alpha_i}},\set{\ket{\beta_i}}\subseteq\complex^d$. Then, $Q(\rho_{AB}) < \log d$.
\end{theorem}
\begin{proof}
We have
\begin{eqnarray}
Q(\rho_{AB})&=&\min_{\wek{\cBB}}\Big(S(\rho_{AB}^{\wek{\cBB}})-S(\rho_{AB})\Big)\\
				&\leq& \min_{\wek{\cBB}}\Big(S(\rho_{AB}^{\wek{\cBB}})-S(\rho_{A})\Big)\\
				&=&\min_{\wek{\cBB}} \Big(S(\rho_{A}^{\cBB_A}) + \sum_i \bra{\cB_A(i)}\rho_A\ket{\cB_A(i)} S(\sigma^{\cBB_B}_{i})-S(\rho_{A})\Big)\\
				&\leq&\min_{\cBB_B} \sum_i p_i^A S(\sigma^{\cBB_B}_{i}),
\end{eqnarray}
where
\begin{equation}
\begin{split}
\sigma_i^{\cBB_B}:=\sum_j \frac{\bra{\cB_A(i)\cB_B(j)}\rho_{AB} \ket{\cB_A(i)\cB_B(j)}}{\bra{\cB_A(i)}\rho_A\ket{\cB_A(i)}}
\ket{\cB_B(j)}\bra{\cB_B(j)},
\end{split}
\end{equation}
where $\{p_i^A\}$ are the eigenvalues of $\rho_A$, the first inequality follows since for any separable state, $S(\rho_{AB})\geq \max\{S(\rho_A),S(\rho_B)\}$~\cite{nielsenkempe}, and the second inequality by choosing $\cBB_A$ as an eigenbasis of $\rho_A$ (yielding $S(\rho_{A}^{\cBB_A})=S(\rho_A)$).

Suppose now, for sake of contradiction, that this upper bound is equal to $\log d$. Then, it must be the case that $\sigma_i^{\cBB_B}$ is maximally mixed for all $i$, implying that $\rho_B$ is also maximally mixed. Reversing the role of $A$ and $B$, an analogous argument yields that $\rho_A$ must be maximally mixed as well. This means that the basis chosen in the second inequality is arbitrary, and we find that for the last line to be equal to $\log d$, it must be that $\bra{\cB_A(i)\cB_B(j)}\rho_{AB}\ket{\cB_A(i)\cB_B(j)}=1/d^2$ for all $\cB_A,\cB_B$ and all $i,j$. Thus, $\rho_{AB}=I/d^2$. However, this state is classical, and hence achieves $Q(\rho_{AB})=0$, yielding the desired contradiction.
\end{proof}

%We finally prove the following lemma, which was used above in the proof of Theorem~\ref{5_thm:sepuppbnd}.
%
%\begin{lemma}\label{5_lem:entropy}
%    For random variables $X$ and $Y$, and for $X\cap Y$ and $Y | (X=x)$ the random variables corresponding to events $(X=x) \cap (Y=y)$ and $(Y=y) | (X=x)$, respectively,
%    \begin{equation}
%        H(X\cap Y) = H(X) + \sum_x \pr(X=x)H(Y|X=x).
%    \end{equation}
%\end{lemma}
%\begin{proof}
%    The proof follows easily using the fact that $\pr(X\cap Y)=\pr(X)\pr (Y| X)$:
%    \begin{eqnarray}
%        H(X\cap Y) &=& -\sum_{xy}\pr(x\cap y)\log \pr(x\cap y)\\
%        &=& \left[-\sum_{x}\pr(x)\sum_y \pr(y|x)\log\pr(y|x)\right]+\\&&\left[-\sum_x\left(\sum_y\pr(y|x)\right)\pr(x)\log(\pr(x))\right]\\
%        &=& \sum_x \pr(x)H(Y|X=x) + H(X),
%    \end{eqnarray}
%    since $\sum_y\pr(y|x)=1$.
%\end{proof}

\subsection{Negativity of quantumness}\label{5_sscn:neg}
The next entanglement monotone we consider in our scheme is the \emph{Negativity}~\cite{VW02}. The latter is defined for a bipartite state $\rho_{AB}$ as $\cN(\rho_{AB}):=(\trnorm{{\rho}_{AB}^{T_{A}}}-1)/2$, for ${\rho}_{AB}^{T_{A}}$ the partially transposed state. Plugging $\cN$ into our framework, we obtain a non-classicality measure we call the \emph{negativity of quantumness}, $Q_\cN(\rho)$.

\begin{theorem}
    For the negativity of quantumness, $Q_\cN$, we have that
    \begin{equation}\label{5_eqn:negat0}
    Q_\cN(\rho)=\frac{1}{2}\min_{\wek{\cBB}} \sum_{\wek{i}\neq \wek{j}}|\rho^{\wek{\cBB}}_{\wek{i}\wek{j}}|.
    \end{equation}
\end{theorem}

\begin{proof}
Thanks to the maximally correlated form of the output of our protocol, by directly applying the definition of the partial transpose to $\rho^f$, we can calculate the eigenvalues of $(\rho^f)^{T_{A}}$ as $\rho^{\wek{\cB}}_{\wek{i}\wek{i}}$ for all $\wek{i}$, and $\pm|\rho^{\wek{\cB}}_{\wek{i}\wek{j}}|$ for $\wek{i}>\wek{j}$. Thus,
\begin{equation}\label{5_eqn:negat}
\mathcal{N}(\rho^f)=\frac{\trnorm{(\rho^f)^{T_{A}}}-1}{2}=\frac{\sum_{\wek{i}\neq\wek{j}}|\rho^{\wek{\cBB}}_{\wek{i}\wek{j}}|}{2}.
\end{equation}
\end{proof}
\noindent We remark that since, by definition, a non-classical state must have some non-vanishing off-diagonal terms $\rho^{\wek{\cBB}}_{\wek{i}\wek{j}}$ in any local product basis $\wek{\cBB}$, we thus obtain yet another proof of Theorem~\ref{5_thm:quantumnessiff}, i.e.\ that $\rho^f$ is entangled for any local rotation $U_\wek{A}$ if and only if ${\rho}$ is not classical.

Next, for the special case of \emph{pure} bipartite states $\ket{\psi}$, we find that $Q_\cN$ has a particularly simple form, in that it reduces to the negativity of $\ket{\psi}$.
\begin{corollary}\label{5_cor:neg}
    For rank one bipartite states $\ketbra{\psi}{\psi}$,
    \begin{equation}
        Q_\cN(\ketbra{\psi}{\psi})=\cN(\ketbra{\psi}{\psi}).
    \end{equation}
\end{corollary}
\begin{proof}
Note first that for $\ket{\psi}=\sum_i\alpha_i\ket{a_i}\ket{b_i}$ the Schmidt decomposition of $\ket{\psi}$, one has $\cN(\ketbra{\psi}{\psi})=\sum_{i\neq j}\alpha_i\alpha_j$. We now show that $Q_\cN(\ketbra{\psi}{\psi})$ matches this expression.

For any local product basis $\wek{\cBB}$, one can write
$    \ket{\psi}=\sum_{i,j=0}^{d-1}\alpha_{ij}\ket{\cB_1(i)}\ket{\cB_2(j)}
$. Then, letting $\rho=\ketbra{\psi}{\psi}$ and beginning from Equation~(\ref{5_eqn:maxcorr}),
%note that in Equation~(\ref{5_eqn:negat}),
%\begin{equation}
%\abs{\rho^{\wek{\cB}}_{\wek{ij},\wek{kl}}}=\bra{{\cB_1(k)}}\bra{{\cB_2(l)}}\rho_\wek{A} \ket{{\cB_1(i)}}\ket{{\cB_2(j)}}=\abs{\alpha_{ij} \alpha_{kl}^*}.
%\end{equation}
straightforwardly applying the definitions of the trace norm and partial transpose yields in Equation~(\ref{5_eqn:negat}) that
\begin{equation}\label{5_eqn:temp1}
    \trnorm{(\rho^f)^{T_{A}}}=\left(\sum_\wek{ij} \abs{\alpha_{ij}}\right)^2.
\end{equation}
Note here that the coefficients $\alpha_{ij}$ are specific to the choice of basis $\wek{\cBB}$ --- thus, our goal is to choose $\wek{\cBB}$ so as to minimize $\sum_\wek{ij} \abs{\alpha_{ij}}$. We claim that this minimizing basis is in fact just the tensor product of the local Schmidt bases for $\ket{\psi}$.

To see this, we use the $\operatorname{vec}$ mapping~\cite{W08_2}, which can be defined such that $\operatorname{vec}(\ket{a}\bra{b})=\ket{a}\ket{b}$ (and analogously, $\operatorname{vec}^{-1}(\ket{a}\ket{b})=\ket{a}\bra{b}$), and the $l_1$ norm, defined as $\norm{C}_{l_1}:=\sum_{ij}\abs{C_{ij}}$. Define now
\begin{equation}
    C:=\operatorname{vec}^{-1}\left(\sum_{ij=0}^{d-1}\alpha_{ij}\ket{\cB_1(i)}\ket{\cB_2(j)}\right)=\sum_{ij=0}^{d-1}\alpha_{ij}\ket{\cB_1(i)}\bra{\cB_2(j)}.
\end{equation}
Then, we have
\begin{equation}\label{5_eqn:temp2}
    \sum_\wek{ij} \abs{\alpha_{ij}}=\norm{C}_{l_1}\geq \trnorm{C}=\sum_i\sigma_i,
\end{equation}
for $\set{\sigma_i}$ the singular values of $C$. Here, the claim $\norm{C}_{l_1}\geq \trnorm{C}$ follows since
\begin{eqnarray}
    \trnorm{C} &=& \max_{0\preceq M\preceq I} \trace((2M-I)C)\\&\leq& \max_{0\preceq M\preceq I} \abs{\langle{\operatorname{vec}(2M-I)},{\operatorname{vec}(C)}\rangle}\\&\leq& \max_{0\preceq M\preceq I} \snorm{\operatorname{vec}(2M-I)}\norm{\operatorname{vec}(C)}_1\\&\leq&\norm{C}_{l_1},
\end{eqnarray}
where the second inequality follows from the H\"{o}lder inequality, the third inequality from the fact that $\snorm{\operatorname{vec}(2M-I)}$ is at most the spectral norm of ${2M-I}$, and where $\norm{v}_\infty := \max_i \abs{v_i}$ and $\norm{v}_1:=\sum_i \abs{v_i}$.

The final step is to observe that the $\sigma_i$ are in fact the Schmidt coefficients of $\ket{\psi}$, since if $C=\sum_i\sigma_i\ketbra{a_i}{b_i}$ is the singular value decomposition of $C$, then $\operatorname{vec}(C)=\sum_i\sigma_i\ket{a_i}\ket{b_i}$ is a Schmidt decomposition for $\ket{\psi}$. Since $\set{\ket{a_i}\otimes\ket{b_j}}$ is a valid local product basis $\wek{\cB}$ in which to expand $\ket{\psi}$, by combining Equations~\ref{5_eqn:temp1} and~\ref{5_eqn:temp2} the claim follows.
\end{proof}

We finally extend our analysis for pure states to the setting of \emph{pseudo-pure} states
\begin{equation}
    \rho(\psi,p):= \frac{(1-p)}{d^2}I+p\proj{\psi}
\end{equation}
where $0\leq p \leq 1$.

\begin{corollary}
    For pseudo-pure state $\rho(\psi,p)$, we have
    \begin{equation}
        Q_\cN(\rho)=p\cN(\ketbra{\psi}{\psi}).
        \end{equation}
\end{corollary}
\begin{proof}
    Follows immediately from Equation~(\ref{5_eqn:negat0}) and Corollary~\ref{5_cor:neg} by recalling that $\rho_{\wek{i}\wek{j}}^\wek{\cBB}=\bra{\wek{\cB}(\wek{i})}\rho \ket{\wek{\cB}(\wek{j})}$.
\end{proof}
\noindent Hence, as already observed in, for example, Reference~\cite{groismanquantumness}, $\rho(\psi,p)$ is non-classical as long as $p>0$ and $\psi$ is entangled.

\section{Non-classicality, mixedness, and entanglement}\label{5_scn:mixed}
Equipped with a faithful and operational measure of non-classical correlations, $Q_{E_{\textup{D}}}$, which we henceforth refer to as $Q$, we now investigate the interplay between non-classicality, entanglement and mixedness of general states  $\rho_{\wek{A}}$. For the sake of simplicity, from now on we restrict to the bipartite case $A_1=A$, $A_2=B$.  We begin by setting the stage with a few simple but general observations following from the definition of $Q$.

For \emph{pure} states $\rho_{AB} = \proj{\psi}$, $Q(\rho_{AB})$ reduces to the von Neumann entropy of entanglement $S(\rho_A)=S(\rho_B)$ \cite{Bravyi2003}, and is thus at most equal to  $\log d$. On the other hand, for arbitrary mixed $\rho_{AB}$, we have that $Q(\rho_{AB})$ is at most $2\log d$, since from Equation~(\ref{5_eqn:req}) one has $Q(\rho_{AB}) \le S(\rho_{AB}\|\rho_A \otimes \rho_B) = S(\rho_A)+S(\rho_B)-S(\rho_{AB}) \equiv I(\rho_{AB})$, where $I$ denotes the mutual information, a measure of \emph{total} correlations. From this and the results of~\cite{nielsenkempe}, one realizes that for a \emph{separable} state a bound $Q(\rho_{AB}^{\textup{sep}}) \leq\log d$ holds.
Now recall from Theorem~\ref{5_thm:sepuppbnd} that this inequality is always sharp for separable states, i.e.\ the bound $\log d$ cannot be exactly saturated for separable non-classical states, while it is instead trivially reached by pure maximally entangled states $\ket{\psi} = d^{-1/2} \sum_{j=0}^{d-1} \ket{j}\ket{j}$.

Surprisingly, what we now show is that as $d \rightarrow \infty$, this upper bound is in fact \emph{asymptotically} attained by separable states. More precisely, we show that there exist separable states such that $Q(\rho^{\textup{sep}}_{AB}) / \log d \rightarrow 1$ with growing $d$. Even more intriguingly, we can show that the upper bound on general mixed bipartite states $\rho_{AB}$ is also asymptotically tight; specifically, there exist families of \emph{mixed} states for which as $d \rightarrow \infty$, $Q(\rho_{AB}) / \log d\rightarrow 2$.

More formally, we prove the following two results, where $m := \lceil (\log d)^4 \rceil$.
\begin{theorem}
  \label{5_prop:sep}
  Define the random separable state:
   \begin{equation} \sigma_{AB} = \frac{1}{dm} \sum_{{i=1,\ldots,d}\atop{j=1,\ldots,m}}
                                    \proj{i}_{A} \ox \left(U_j\proj{i}U_j^\dagger\right)_{B},\end{equation}
  for unitaries $U_j$ drawn independently from the Haar measure.
  Then, with high probability, $Q(\sigma_{AB}) \geq \log d - O(\log \log d)$.
\end{theorem}

\begin{theorem}
  \label{5_prop:low-rank}
  For $C$ a system of dimension $m$, let
    $\rho_{AB} = \Tr_C \proj{\psi}_{ABC}$,
  where $\ket{\psi} \in \mathbb{C}^d \ox \mathbb{C}^d \ox \mathbb{C}^m$ is uniformly distributed
  (with probability induced by the Haar measure). Then, with high probability, $Q(\rho_{AB}) \geq 2\log d - O(\log \log d)$.
\end{theorem}

What these results tell us is that, first, there are separable states that asymptotically (in $d$) are as non-classical as the most non-classical pure state (which is the maximally entangled state); second, mixed entangled states can be much more non-classical (namely, twice as much) than pure entangled states. We remark that therefore \emph{both} entanglement \emph{and} mixedness are required to ``break the barrier'' of $\log d$. This goes against the intuition that entanglement by itself is the strongest form of non-classicality:
We demonstrate that mixedness also plays a prominent role and can make correlations maximally non-classical.

One possible explanation for these findings may be the following: Traditionally, the study of quantum correlations has focused on quantum entanglement, which arises from the superposition principle of quantum mechanics. Yet, there is another ``non-classical'' feature of quantum mechanics to be reckoned with; namely, that quantum systems can be in probabilistic mixtures of \emph{non-commuting} states. What Theorem~\ref{5_prop:sep} thus quantifies is the extent to which non-commutativity alone can give rise to non-classical correlations. When non-commutativity is then combined with the superposition principle, Theorem~\ref{5_prop:low-rank} tells us that the non-classical correlations generated are stronger than possible with either principle alone.

We close this section with the proofs of Theorems~\ref{5_prop:sep} and~\ref{5_prop:low-rank}.
\begin{proof}[Proof of Theorem~\ref{5_prop:sep}.]
   The claim will follow by showing that in Equation~(\ref{5_eqn:DE}), $S(\sigma_{AB}) \leq \log d + \log m$,
  whereas for $d$ sufficiently large and with high probability,
     $S\bigl(\sigma^{\wek{\cBB}}_{AB} \bigr) \geq 2\log d - \textup{const.}$ for all $\wek{\cBB}$. The first of these bounds is easy to prove --- it follows by observing that the rank of $\sigma_{AB}$ is at most $dm$.

As for the second bound, note first that
\begin{equation}
S\bigl(\sigma^{\wek{\cBB}}_{AB} \bigr)=S(\sigma^{\wek{\cBB}}_{A})+S(\sigma^{\wek{\cBB}}_{B})-I(\sigma^{\wek{\cBB}}_{AB})=2\log d-I(\sigma^{\wek{\cBB}}_{AB}),
\end{equation}
which follows since $\sigma^{\wek{\cBB}}_{A}=\sigma^{\wek{\cBB}}_{B}=I/d$. Hence, it suffices to show that $I(\sigma^{\wek{\cBB}}_{AB})\leq \textup{const.}$. To see this, note that $\sigma_{AB}$ is almost identical to the information locking states considered in~\cite{HLSW04} (see Theorem~V.1, Equation~(64)), which take the form
\begin{equation}
    \rho_{AB}=\frac{1}{dm} \sum_{{i=1,\ldots,d}\atop{j=1,\ldots,m}}
                                    \proj{ij}_{A} \ox \left(U_j\proj{i}U_j^\dagger\right)_{B}.
\end{equation}
Letting $x$ and $y$ denote random variables corresponding to the outcomes of local measurements $X$ and $Y$ on $A$ and $B$, respectively, define $I_c(\rho_{AB}):=\max_{X\otimes Y}I(x:y)$. Then, it is known that for large enough $d$ and with high probability over the choice of local unitaries $\set{U_j}$, $I_c(\rho_{AB})\leq const.$ (specifically, for our choice of $m$ here, set the parameter $\epsilon$ in Equation (66) of~\cite{HLSW04} to scale  as $1/\log d$). Observing that $\sigma_{AB}$ is attainable from $\rho_{AB}$ via a local operation (namely, we trace out the register containing label $j$ in $A$), and recalling that the mutual information is non-increasing under partial trace completes the proof.
%
%When $d$ is sufficiently large and with high probability -- for any (projective) measurements on $A$ and $B$ with classical outputs $x$ and $y$,
%respectively,
%\beq
%\label{5_eqn:boundMI}
%  I(x:y) \leq I(ij:y) \leq \text{const.}.
%\eeq
%This is because with our choice for $m$ ($n$ in ~\cite{HLSW04}), the parameter $\epsilon$ in the Equation (66) of~\cite{HLSW04} can be chosen to scale as $1/\log d$.
%%high probability -- for any choice of bases $\cB_A$ and $\cB_B on $A$ and $B$, respectively,
%%\begin{equation}
%%  I(A:B)_{\rho^{\cB_A\cB_B}} \leq I(ij:y) \leq \text{const.},
%%\end{equation}
%The bound \eqref{5_eqn:boundMI} must hold true also for our state, since all we do is remove $j$
%before the measurement.
%
%Note however, that $x$ and $y$ have maximal entropy $\log d$, since
%$\sigma_A = \sigma_B  =\II/d$ are both maximally mixed.
%That means that
%$S\bigl( \sigma_{AB}^\wek{\cB} \bigr) \geq 2\log d - \text{const.}$ as
%claimed.
\end{proof}

\begin{proof}[Proof of Theorem~\ref{5_prop:low-rank}.]
  The claim will follow by showing that, in Equation~(\ref{5_eqn:DE}), $S(\rho) \leq \log m$,
  whereas for $d$ sufficiently large and with high probability, $S\bigl( \rho^\wek{\cBB} \bigr) \geq 2\log d - \textup{const.}$ for all $\wek{\cBB}$. Again, the first of these bounds follows simply because the rank of $\rho$ is bounded by $m$.

Now, let $M$ and $N$ denote arbitrary complete von Neumann measurements on $A$ and $B$, respectively, such that
\begin{equation}
  M = \set{ M_x=\proj{m_x} }_{x=1}^{d}, \quad
  N = \set{ N_y=\proj{n_y} }_{y=1}^{d}.
\end{equation}
For such a measurement $M$, let $\Pi_M(\sigma)$ denote the completely positive trace-preserving linear map
\begin{equation}
  \Pi_M(\sigma) = \sum_{x=1}^d M_x \sigma M_x.
\end{equation}

%
%The random state considered here is analysed in great detail already
%in~\cite{HLW:aspects}, and we may refer to that paper for
%technical results.
To prove the desired bound of $S\bigl( \rho^\wek{\cBB} \bigr) \geq 2\log d - \textup{const.}$, we use the concentration of measure results of~\cite{HLW:aspects} (see also~\cite{HLSW04}). The intuition is as follows. We first consider a \emph{fixed} set of local measurement bases $M$ and $N$. Then, one can show that for the random state $\ket{\psi}$ and the corresponding state $\rho$ in the statement of the theorem, the expected value of $S(\Pi_M\otimes\Pi_N(\rho))$ is at least roughly $2\log d$. We then convert this into a high-probability statement using Levy's Lemma (Lemma~\ref{5_lem:levy}), which yields that with high probability, $S(\Pi_M\otimes\Pi_N(\rho))$ will indeed be close to its expected value. This was for a \emph{fixed} choice of local measurements $M$ and $N$ --- to extend this statement to \emph{all} such measurements, we use the union bound together with a net argument. Specifically, we cast a $\delta$-net over all choices of local measurements $M$ and $N$, and apply the union bound to conclude that for all measurements from this net, $S(\Pi_M\otimes\Pi_N(\rho))$ will still be close to its expected value with probability bounded away from $1$.

To begin, let $M$ and $N$ be a fixed choice of local measurement bases. We first lower bound the expected value of $S(\Pi_M\otimes\Pi_N(\rho))$ over random choices of $\ket{\psi}$ as follows:
\begin{equation}
  S\bigl( \Pi_M\otimes\Pi_N(\rho) \bigr) \geq S_2\bigl( \Pi_M\otimes\Pi_N(\rho)  \bigr)
                              =    -\log\sum_{x,y=1}^d \left[\Tr\left(\ketbra{\psi}{\psi}(M_x \ox N_y \ox \II)\right)\right]^2,
\end{equation}
where $S_2(\sigma)=-\log(\Tr (\sigma^2))$ is the quantum Renyi entropy of order 2, and the last equality follows since for rank one $M_x$ and $N_y$,
\begin{equation}
    \trace\left[(M_x \ox N_y \rho)^2\right]=\left[\trace(M_x \ox N_y \rho)\right]^2.
\end{equation}
Hence,
\begin{eqnarray}
  \EE_\psi \left[S\bigl(  \Pi_M\otimes\Pi_N(\rho) \bigr)\right]
   &\geq& -\log\EE_\psi\left[ \sum_{x,y=1}^d \left[ \Tr(\ketbra{\psi}{\psi}(M_x \ox N_y \ox \II)) \right]^2 \right]\\
       &=&    -\log\left( d^2 \EE_\psi \left[\bigl[ \Tr(\ketbra{\psi}{\psi}(\proj{0} \ox \proj{0} \ox \II)) \bigr]^2\right] \right)\label{5_eqn:temp3}
\end{eqnarray}
%\begin{equation}
%\begin{split}
%       =    -\log\Big( d^2 \Tr\Big( \frac{\II+F_{ABC:A'B'C'}}{d^2m(d^2m+1)}\proj{00}_{AA'} \ox \proj{00}_{BB'} \ox \II_{CC'} \Big) \Big) \\
%       &=     \log\frac{d^2m+1}{m+1}\\
%        &\geq  2\log d - \log\left(1+\frac{1}{m}\right),
%\end{split}
%\end{equation}
where the first statement follows by the convexity of $-\log$, and the second since the distribution of $\ket{\psi}$ is invariant under unitaries. Now,
\begin{eqnarray}
    &&d^2\EE_\psi\left[ \bigl[ \Tr(\ketbra{\psi}{\psi}_{ABC}(\proj{0}_A \ox \proj{0}_B \ox \II_C)) \bigr]^2\right]\\&=&    d^2\EE_\psi\left[ \Tr(\ketbra{\psi}{\psi}_{ABC}\otimes\ketbra{\psi}{\psi}_{A'B'C'}(\proj{00}_{AA'} \ox \proj{00}_{BB'} \ox \II_{CC'})) \right]\\
    &=&d^2\Tr(\EE_\psi\left[ \ketbra{\psi}{\psi}_{ABC}\otimes\ketbra{\psi}{\psi}_{A'B'C'}\right](\proj{00}_{AA'}  \ox \proj{00}_{BB'} \ox \II_{CC'}))\\
      &=&\frac{1}{m(d^2m+1)}\Tr(\left[ \II_{ABC:A'B'C'}+W_{ABC:A'B'C'}\right](\proj{00}_{AA'}  \ox \proj{00}_{BB'} \ox \II_{CC'}))\nonumber\\
      &=&\frac{m+1}{d^2m+1}
      \label{5_eqn:swap},
\end{eqnarray}
where the first equality uses the fact that $(\trace(AB))^2=\trace[(A\otimes A)(B\otimes B)]$, and the third equality follows since $\EE_\psi\left[ \ketbra{\psi}{\psi}\otimes\ketbra{\psi}{\psi}\right]$ is proportional to $I+W$ for $W=\sum_{ij}\ketbra{i}{j}\otimes\ketbra{j}{i}$ the swap gate. (One way to see the latter is to note that $\sigma=\EE_\psi\left[ \ketbra{\psi}{\psi}\otimes\ketbra{\psi}{\psi}\right]$ is invariant under $U\otimes U$ for any unitary $U$, and is hence a Werner state~\cite{W89}. Werner states, in turn, can be written as mixtures of the projectors onto the symmetric and antisymmetric spaces, $\Pi_{S}=(I+W)/2$ and $\Pi_{A}=(I-W)/2$, respectively. Observing that $\trace(\sigma\Pi_A)=0$ yields the claim.) Substituting Equation~(\ref{5_eqn:swap}) in Equation~(\ref{5_eqn:temp3}) thus yields
\begin{equation}
    \EE_\psi \left[S\bigl(  \Pi_M\otimes\Pi_N(\rho) \bigr)\right]
   \geq \log\frac{d^2m+1}{m+1}
        \geq  2\log d - \log\left(1+\frac{1}{m}\right)=:\Delta.
\end{equation}

With a lower bound on the expected value of $S(\Pi_M\otimes\Pi_N(\rho))$ in hand, we would now like to show that with high probability, $S(\Pi_M\otimes\Pi_N(\rho))$ indeed takes a value close to its expected value. To show this, we apply Levy's Lemma (Lemma~\ref{5_lem:levy}), which requires an upper bound on the Lipschitz constant of the entropy $S\bigl(\Pi_M\otimes\Pi_N(\rho) \bigr)$. The latter is given by the proof of Lemma III.2 of~\cite{HLW:aspects}, which demonstrates an upper bound on the constant of $\sqrt{8}\log d$. Thus, by Levy's Lemma:
\begin{equation}\label{5_eqn:levy}
%\begin{multline}
  \Pr\left\{ S\bigl( \Pi_M\otimes\Pi_N(\rho) \bigr)
                            < \Delta - \epsilon \right\}\\
        \leq \exp\left( -\frac{c\epsilon^2d^2m}{(\log d)^2}  \right),
%  \label{5_eqn:probability}
\end{equation}
for some constant $c>0$. This shows that the entropy is indeed large with high probability for a fixed choice of local measurement basis $M\otimes N$.

To extend this to \emph{all} local measurement bases $M$ and $N$, suppose we had a net of $T$ basis pairs $M^t$ and $N^t$
for $t\in[T]$ able to approximate the quantity $S\bigl( \Pi_{M}\otimes\Pi_{N}(\rho) \bigr)$ for \emph{any} local measurement $M\otimes N$ within precision $2\epsilon$. Then, by Equation~(\ref{5_eqn:levy}) and the union bound, we would have:
\begin{equation}
  \Pr\left\{ \exists t\mbox{ s.t. } S\bigl( \Pi_{M^t}\otimes\Pi_{N^t}(\rho) \bigr)
                            < \Delta - \epsilon \right\}\\
        \leq T \exp\left( -\frac{c\epsilon^2d^2m}{(\log d)^2}  \right),
%  \label{5_eqn:probability}
\end{equation}
implying
\beq
  \Pr\left\{ \exists M\otimes N\mbox{ s.t. } S\bigl( \Pi_{M}\otimes\Pi_{N}(\rho) \bigr)
                            < \Delta - 3\epsilon \right\}\\
        \leq T \exp\left( -\frac{c\epsilon^2d^2m}{(\log d)^2}  \right).
  \label{5_eqn:probability}
\eeq
Thus, if such a $2\epsilon$-net with small enough $T$ exists, then we are done. Indeed, Lemma~\ref{5_lem:net1} shows that such a $2\epsilon$-net exists with
\begin{equation}
  T \leq \left( \frac{c'd^{3/2}(\log d)^2}{\epsilon^2} \right)^{4d^2}
\end{equation}
for $c'\in\Theta(1)$. For this value of $T$, since we set $m = \lceil (\log d)^4 \rceil$, the probability on the right side of Equation~(\ref{5_eqn:probability}) is bounded away from $1$ for large enough $d$, completing the proof.
\end{proof}

In order to complete the proof of Theorem~\ref{5_prop:low-rank}, we finally show three lemmas required for the net argument above.

\begin{lemma}\label{5_lem:net1}
    For any constant $\epsilon>0$, there exists a set $S:=\set{M^t\otimes N^t}_{t=1}^T$of $T$ local measurement bases (where $M^t$ and $N^t$ are rank one von Neumann measurements each acting on $d$-dimensional spaces) with
    \begin{equation}
        T \leq \left( \frac{cd^{3/2}(\log d)^2}{\epsilon^2} \right)^{4d^2},
    \end{equation}
    (for $c>0$ a constant) such that for all $\rho\in \mathcal{D}(\complex^d\otimes\complex^d)$ and local measurements $M\otimes N$, there exists a $M^t\otimes N^t\in S$ such that
    \begin{equation}
        \abs{S\bigl( \Pi_{M}\otimes\Pi_{N}(\rho) \bigr)-S\bigl( \Pi_{M^t}\otimes\Pi_{N^t}(\rho) \bigr)}\leq 2\epsilon.
    \end{equation}
\end{lemma}
\begin{proof}
To construct the set $S$, we embed each local measurement basis into a unitary matrix, and then cast a net over unitary matrices. Specifically, recall that each local measurement is described by an orthonormal basis $M=\set{\ket{b_i}}_{i=1}^d$. Then, by arranging the vectors $\ket{b_i}$ as columns of a matrix, we obtain a $d\times d$ unitary matrix, denoted $U_M$, which rotates the standard basis to $\set{\ket{b_i}}_{i=1}^d$. By Lemma~\ref{5_lem:net2}, there exists a $\delta$-net (with respect to the spectral norm) for $\mathcal{U}(\complex^d)$ of size
\begin{equation}
  T_0 \leq \left( \frac{c'd^{3/2}}{\delta} \right)^{2d^2}.
\end{equation}
Thus, by setting $\delta = \epsilon^2/32(\log d)^2$ and picking $T_0$ elements $\{M^s\}$ and $T_0$ elements $\{N^t\}$, we obtain our set $S$ of local measurement bases with size $T=T_0^2$, as in the statement of our claim.

We now show that $S$ is a $2\epsilon$-net. Let $M$ and $N$ be arbitrary local measurements with corresponding unitaries $U_M$ and $U_N$. Then, there exist $U_{M^s}$ and $U_{N^t}$ in the net from Lemma~\ref{5_lem:net2} such that $\snorm{U_M-U_{M^s}}\leq \delta$ and $\snorm{U_N-U_{N^t}}\leq \delta$. Let $V$ denote the unitary mapping basis $M^s\otimes N^t$ to basis $M\otimes N$.

Now, if it were true that for all $\rho\in\DD(\complex^d\otimes\complex^d)$,
\begin{equation}\label{5_eqn:temp5}
    \trnorm{V\rho V^\dagger-\rho}\leq 4\delta,
\end{equation}
then our desired $2\epsilon$-net property would follow since
\begin{eqnarray}
  \abs{S\bigl( \Pi_{M}\otimes\Pi_{N}(\rho) \bigr)-S\bigl( \Pi_{M^t}\otimes\Pi_{N^t}(\rho) \bigr)}
         &=&    \left| S\bigl( \Pi_{M}\otimes\Pi_{N}(\rho) \bigr)-S\bigl( \Pi_{M}\otimes\Pi_{N}(V\rho V^\dagger) \bigr) \right| \nonumber\\
       &\leq& H(2\delta,1-2\delta) +2\delta \log d^2\\
       &\leq& (2\sqrt{2\delta}+2\delta) \log d^2\\
       &=& \epsilon+\frac{\epsilon^2}{8\log d}\\
        &\leq& 2\epsilon
\end{eqnarray}
for large enough $d$ (or alternatively for $0<\epsilon < 1$). Here, the second inequality follows from the estimate for the Shannon entropy $H(x,1-x) \leq 2\sqrt{x(1-x)}$, and the first inequality uses the Fannes-Audenaert inequality~\cite{Au06}, which states that for $\rho,\sigma\in \mathcal{D}(\complex^d)$ with $r:=\trnorm{\rho-\sigma}/2$,
\begin{equation}
    \abs{S(\rho)-S(\sigma)}\leq r\log(d-1)+H(r,1-r).
\end{equation}

Thus, it remains to show that Equation~(\ref{5_eqn:temp5}) indeed holds. To see this, note first that by Lemma~\ref{3_lem:tracenorm},
\begin{equation}
    \snorm{U_M\otimes U_N-U_{M^s}\otimes U_{N^t}}\leq \snorm{U_M-U_{M^s}}+\snorm{U_N- U_{N^t}}\leq 2\delta,
\end{equation}
where note the columns of $U_M\otimes U_N$ are now elements of the local product basis $M\otimes N$. We thus have
\begin{equation}\label{5_eqn:temp4}
    2\delta\geq\snorm{U_M\otimes U_N-U_{M^s}\otimes U_{N^t}}=\snorm{U_M\otimes U_N-V^\dagger(U_{M}\otimes U_{N})}=\snorm{I-V},
\end{equation}
where the last equality follows since the spectral norm is invariant under unitaries. But this implies
\begin{eqnarray}
    \trnorm{V\rho V^\dagger- \rho}&=&    \trnorm{V\rho V^\dagger -V\rho + V\rho - \rho}\\
    &\leq& \trnorm{\rho(V^\dagger-I)}+\trnorm{(V-I)\rho}\\
    &\leq& \trnorm{\rho}\snorm{V^\dagger-I}+\snorm{V-I}\trnorm{\rho}\\
    &\leq& 4\delta,
\end{eqnarray}
where the first inequality follows from the triangle inequality, the second from the fact that for Schatten $p$-norms, $\pnorm{ABC}\leq\snorm{A}\pnorm{B}\snorm{C}$ (see Section~\ref{0_scn:linalg}), and the third inequality from Equation~\ref{5_eqn:temp4}. This concludes the proof.

%Thus, we get directly
%\begin{equation}
%%\begin{multline*}
%  \Pr\left\{ \exists M \exists N\ S\bigl( (M\ox N)\rho \bigr)
%                            < 2\log d - \log\left(1+\frac{1}{m}\right) - 5\epsilon \right\}
%%   \begin{split}
%           \leq \left( \frac{c'd^{3/2}(\log d)^2}{\epsilon^2} \right)^{4d^2}
%                \exp\left( -c\frac{\epsilon^2}{32(\log d)^2} d^2m \right),
%%\end{split}
%%\end{multline*}
%\end{equation}
%and we're done.
\end{proof}

\begin{lemma}   \label{5_lem:net2}
    For any constant $\delta>0$, there exists a set $S:=\set{U_t}_{t=1}^{T_0}$ of unitaries $U_t\in \mathcal{U}(\complex^d)$ with
    \begin{equation}
        T_0 \leq \left( \frac{cd^{3/2}}{\delta} \right)^{2d^2},
    \end{equation}
    (for $c>0$ a constant) and such that for all $U\in \mathcal{U}(\complex^d)$, there exists a $U_t\in S$ satisfying $\snorm{U-U_t}\leq \delta$.
\end{lemma}
\begin{proof}
For any unitary $U\in \mathcal{U}(\complex^d)$, the idea is to replace the columns of $U$ with vectors taken from a net on the set of pure states in $\complex^d$. Of course, the resulting operator $U'$ is in general not unitary --- however, this can be corrected by an appropriate orthogonalization procedure inspired by the ``pretty good measurement''~\cite{HW94}, yielding a unitary $U''$ such that $\snorm{U-U''}\leq \delta$, as desired.

More specifically, by Lemma III.6 of~\cite{HLW:aspects}, there exists an $\epsilon$-net $N$ on the set of pure state vectors in $\CC^d$ (with respect to the Euclidean norm) such that $\abs{N}\leq(5/\epsilon)^{2d}$. Set $\epsilon = \frac{\delta}{6d^{3/2}}$. Now, for each column $\ket{b_i}$ of a given unitary $U$, we first find an $\epsilon$-close vector $\ket{b_i'}\in N$, and embed the latter as columns into a matrix $U'$. Of course, the vectors $\set{\ket{b_i'}}$ are not orthogonal in general, so $U'$ is not unitary. To correct this, define the operator $B = \sum_{i=1}^d \proj{b_i'}$ and let
\begin{equation}
  \ket{b_i''} = B^{-1/2} \ket{b_i'}.
\end{equation}
Note that if the $\ket{b_i'}$ are linearly independent, then $B$ is invertible, and moreover $\set{\ket{b_i''}}_{i=1}^d$ is an orthonormal basis since
\begin{equation}
\sum_{i=1}^d\ketbra{b_i''}{b_i''}=\sum_{i=1}^dB^{-1/2}\ketbra{b_i'}{b_i'}B^{-1/2}=B^{-1/2}BB^{-1/2}=I.
\end{equation}
Note that by Lemma~\ref{5_lem:LI} below, the $\ket{b_i'}$ are indeed linearly independent (for large enough $d$) for our choice of $\epsilon\in\Theta(d^{-3/2})$.

Now, in order to show that this construction constitutes a $\delta$-net, we must show that $\snorm{U-U''}\leq \delta$. To do so, we first bound $\enorm{\ket{b}-\ket{b''}}$ as
\begin{equation}
    \enorm{\ket{b}-\ket{b''}}\leq     \enorm{\ket{b}-\ket{b'}}+    \enorm{\ket{b'}-\ket{b''}}\leq \epsilon + \enorm{(I-B^{-\frac{1}{2}})\ket{b'}}\leq\epsilon + \snorm{I-B^{-\frac{1}{2}}},
\end{equation}
where the second inequality follows from our $\epsilon$-net, and the third inequality from the definition of the spectral norm. To bound this latter quantity, note that since
\begin{equation}
\snorm{\proj{b_i'}-\proj{b_i}}\leq\trnorm{ \proj{b_i'}-\proj{b_i}} \leq 2\norm{ \ket{b_i'}-\ket{b_i} }_2 \leq 2\epsilon,
\end{equation}
where the second inequality follows from Equation~(\ref{0_eqn:traceeuclid}), we have
\begin{equation}
  \snorm{ B- \II } =    \snorm{ \sum_{i=1}^d \left(\proj{b_i'}-\proj{b_i}\right) }
              \leq \sum_{i=1}^d \snorm{ \proj{b_i'}-\proj{b_i} } \leq 2d\epsilon,
\end{equation}
and consequently $\snorm{B^{-1/2} - \II} \leq \frac{2d\epsilon}{1-2d\epsilon}$. The latter can be seen by applying the definition of the spectral norm in terms of the singular values of its argument and showing that if $\abs{x-1}\leq y$ for $x\neq0$, then $\abs{\frac{1}{\sqrt{x}}-1}\leq y/(1-y)$. We conclude that
\begin{eqnarray}
    \enorm{\ket{b}-\ket{b''}}\leq \epsilon+ \frac{2d\epsilon}{1-2d\epsilon}\leq \frac{(1+2d)\epsilon}{1-2d\epsilon}\leq (4d+2)\epsilon,\label{5_eqn:twonorm}
\end{eqnarray}
where the last inequality holds when $2d\epsilon\leq 1/2$.

With this bound in hand, we can now upper bound $\snorm{U-U''}$ as
\begin{eqnarray}
    \snorm{U-U''}&=&\snorm{U^\dagger-(U'')^\dagger}\\
    &=&\max_{\ket{x}\in\complex^d\mbox{ s.t. } \enorm{x}= 1}\enorm{(U^\dagger-(U'')^\dagger)\ket{x}}\\
    &\leq&\sqrt{d}\max_{1\leq i\leq d}\abs{(\bra{b_i}-\bra{b_i''})\ket{x}}\\
    &\leq&\sqrt{d}(4d+2)\epsilon\\
    &\leq&\delta,
\end{eqnarray}
where the first inequality follows since $\enorm{x}\leq \sqrt{d}\snorm{x}$ for $\ket{x}\in\complex^d$, the second inequality follows from the Cauchy-Schwarz inequality and Equation~(\ref{5_eqn:twonorm}), and the third inequality from our definition of $\epsilon$, as desired.

It remains to bound the cardinality of our $\delta$-net: The number of different $U''$ in this net is at most
$T_0 \leq (5/{\epsilon})^{2d^2}=(30 d^{3/2}/\delta)^{2d^2}$, since for each of the $d$ columns of $U''$, we have at most $(5/\epsilon)^{2d}$ vectors in our pure state net $N$ to choose from.

\end{proof}

\begin{lemma}\label{5_lem:LI}
    Let $S=\set{\ket{b_i}}_{i=1}^d\subseteq\complex^d$ be an orthonormal basis for $\complex^d$, and let $S'=\set{\ket{b'_i}}_{i=1}^d\subseteq\complex^d$ satisfy $\enorm{\ket{b_i}-\ket{b_i'}}\leq \epsilon$ for all $i$. Then if $\epsilon\in o(1/\sqrt{d})$, $S'$ is a linearly independent set for large enough $d$.
\end{lemma}
\begin{proof}
    We proceed by contradiction. Assume $S'$ is a linearly dependent set, i.e.\ there exist coefficients $\alpha_i$, at least two of which are non-zero, such that
    \begin{equation}\label{5_eqn:temp6}
        0=\sum_i\alpha_i\ket{b'_i}=\sum_i\alpha_i(\ket{b_i}+\ket{\epsilon_i}).
    \end{equation}
    for $\ket{\epsilon_i}\in\complex^d$ with $\enorm{\ket{\epsilon_i}}\leq \epsilon$. It follows that $\enorm{\sum_i\alpha_i\ket{b_i}}^2=\enorm{\sum_i\alpha_i\ket{\epsilon_i}}^2$, implying
    \begin{equation}
        \sum_i\abs{\alpha_i}^2\leq\sum_{ij}\abs{\alpha_i}\abs{\alpha_j}\abs{\braket{\epsilon_i}{\epsilon_j}}\leq\epsilon^2\sum_{ij}\abs{\alpha_i}\abs{\alpha_j}=\epsilon^2\left(\sum_i\abs{\alpha_i}^2+\sum_{i\neq j}\abs{\alpha_i}\abs{\alpha_j}\right),
    \end{equation}
    where the second inequality follows from the Cauchy-Schwarz inequality. Thus,
    \begin{equation}
        (1-\epsilon^2)\sum_i\abs{\alpha_i}^2\leq \epsilon^2\left(\sum_{i\neq j}\abs{\alpha_i}\abs{\alpha_j}\right)\leq \epsilon^2\left(\sum_i\abs{\alpha_i}\right)^2\leq \epsilon^2 d\left(\sum_i\abs{\alpha_i}^2\right),
    \end{equation}
    where the third inequality follows since $\onorm{\ket{v}}\leq \sqrt{d}\enorm{\ket{v}}$ for any $\ket{v}\in\complex^d$. Since $S'$ is linearly dependent, $\sum_i\abs{\alpha_i}^2\neq 0$, and so dividing both end sides of the chain above by this quantity yields
    \begin{equation}
        1\leq \epsilon^2(d+1),
    \end{equation}
    which for $\epsilon\in o(1/\sqrt{d})$ yields a contradiction for large enough $d$.
\end{proof}
\noindent \emph{Acknowledgements for this chapter.} We thank  Fernando~Brand\~ao, Nicolas Brunner, Dagmar Bru\ss,  Hermann Kampermann, Debbie Leung and Alexander Streltsov for helpful discussions. %We acknowledge  support by NSERC, QuantumWorks, CIFAR, Ontario Centres of Excellence, the Spanish government (program FIS2008-01236/FIS) and the Catalan government (program 2009SGR-0985). This work was initiated when MP was a Lise Meitner Fellow of the Austrian Academy of Science (FWF) at the University of Innsbruck.

%======================================================================
\chapter{Characterizing quantumness via entanglement creation}\label{chap:entanglementswap}

\emph{This chapter is based on~\cite{GPACH11}:}\\

\vspace{-4mm}
\noindent S.~Gharibian, M.~Piani, G.~Adesso, J.~Calsamiglia and P.~Horodecki. Characterizing quantumness via entanglement creation.
\emph{International Journal of Quantum Information}, 9(7 \& 8):1701--1713, 2011, DOI: 10.1142/S0219749911008258, \copyright~2011 World Scientific Publishing Company, www.worldscientific.com/worldscinet/ijqi.

\vspace{3mm}
\noindent In Chapter~\ref{chap:entanglementswap}, we introduced an activation protocol which maps general non-classical (multipartite) correlations between given quantum systems into bipartite entanglement between the systems and an ancilla. Here, we study how this activation protocol can be used to entangle the starting systems themselves via entanglement swapping through a measurement on the ancilla. Furthermore, we bound the relative entropy of quantumness (a naturally arising measure of non-classicality in the scheme of Chapter~\ref{chap:entanglementswap}) for a special class of separable states, the so-called classical-quantum states. In particular, we fully characterize the classical-quantum two-qubit states that are maximally non-classical.

\section{Introduction and results}
\label{6_sec:intro}

In this chapter, we continue our study of non-classical correlations (see Section~\ref{0_sscn:nonclassical} for a brief survey). Specifically, recall that the non-classicality of correlations present in multipartite quantum states is not due solely to the presence of entanglement. Namely, there exist quantum states which are unentangled, but nevertheless exhibit traits that have \emph{no} counterpart in the classical world. Such traits include no-local broadcasting~\cite{pianietal2008nolocalbrodcast} and the locking of correlations~\cite{dhlst04,DG08} (see Section~\ref{0_sscn:nonclassical}). Much effort has been devoted in recent years to characterize and quantify the non-classicality --- or \emph{quantumness} --- of  correlations~\cite{ollivier01a,henderson01a,PhysRevA.71.062307,PhysRevA.72.032317,MPSVW10,groismanquantumness,PhysRevA.77.052101,Luo08,Bravyi2003,pianietal2008nolocalbrodcast,pianietal2009broadcastcopies,ferraro,ADA,PhysRevA.82.052342,streltsov2010} believed to be behind such feats.

In this context, we proposed an \emph{activation} protocol in Chapter~\ref{chap:activation} which maps general non-classical (multipartite) correlations between input systems into bipartite entanglement between the systems and an ancilla. This was accomplished by letting the ancilla and input systems interact via CNOT gates with the systems acting as controls (see Section~\ref{5_scn:protocol} for a formal description of the protocol). One advantange of this mapping is that it allows us to apply the tools and concepts of entanglement theory to the study of the quantumness of correlations. As an added bonus, the activation protocol, when considered in an adversarial context where the control bases are chosen so as to create the minimal amount of system-ancilla entanglement, provides an operational interpretation of the \emph{relative entropy of quantumness}~\cite{Bravyi2003,PhysRevA.71.062307,groismanquantumness,PhysRevA.77.052101,MPSVW10} as being the minimum distillable entanglement\cite{reviewplenio} \emph{necessarily} (i.e.\ in the worst case scenario) created between the input systems and the ancilla.

In this chapter, we continue our study of the activation protocol of Chapter~\ref{chap:activation}, and present two main contributions towards a better understanding of the quantumness of correlations.

\paragraph{Our Results:} Here, we show the following.

\vspace{2mm}
\noindent\textbf{1. Upper bounds on the non-classicality of separable states.}
We first give a non-trivial upper bound on the relative entropy of quantumness for a special class of separable states, the so-called \emph{classical-quantum states} (see Section~\ref{0_sscn:nonclassical}) (Lemma~\ref{6_lem:cqbound}). Using this, we then fully characterize the classical-quantum two-qubit states which are maximally non-classical with respect to the relative entropy of quantumness (Lemma~\ref{6_lem:maxCQ}).

\vspace{2mm}
\noindent\textbf{2. Entangling the input systems via entanglement swapping.}
The activation protocol of Chapter~\ref{chap:activation} demonstrates how to map non-classical correlations in an initial quantum system into entanglement between the system and an ancilla. However, one might prefer not to generate entanglement with an ancilla, but rather within the subsystems of the initial system itself.

We thus next study an approach for extending the activation protocol in order to entangle the input systems in such a manner as follows: We first run the original activation protocol (i.e.\ we let each system interact with an ancilla). Next, we try to ``swap''\cite{swapping} the entanglement created between the input systems and ancilla back into entanglement among the input systems by performing a measurement on the ancilla alone. Note that we assume a worst-case scenario in performing this mapping: We ask, does there \emph{exist} a choice of control bases for the activation protocol for which no entanglement can be created between the input systems with this approach, even if we allow post-selection after measuring the ancilla?

For this mapping, we derive conditions (Theorem~\ref{6_thm:suffswap}, Corollary~\ref{6_thm:iso}, discussion in Sections~\ref{6_scn:CQ} and~\ref{6_scn:QQ}) under which entanglement can or cannot be swapped back into the input system. In particular, we find that there exist non-classical states which, despite \emph{necessarily} leading to the creation of entanglement between systems and the ancilla in the activation protocol, may nevertheless fail to allow entanglement swapping back onto the initial system for a crafty choice of control bases.

\paragraph{Discussion and open questions.} In this chapter, we first find bounds on the non-classicality (as measured by the relative entropy of quantumness) of classical-quantum states, and we characterize the maximally non-classical two-qubit classical-quantum states. It would be interesting to find bounds on the non-classicality of general separable states: from Chapter~\ref{chap:activation} we know that, for example, a separable state of two qubits can never be as non-classical as a maximally entangled pure state, but at present we do not know how large the gap between the two is. We remark that the maximally non-classical two-qubit CQ states found here (with respect to the relative entropy of quantumness) in Lemmas~\ref{6_lem:cqbound} and~\ref{6_lem:maxCQ} match those found in Chapter~\ref{chap:localunitary} for the measure defined therein based on local unitary operations.

With respect to the swapping of the post-activation ancilla-system entanglement onto the original systems, we have both necessary conditions and sufficient conditions for the swapping to be possible in an adversarial scenario, but we lack conditions which are \emph{simultaneously} necessary and sufficient. In finding such conditions, we suspect it would be beneficial to study the problem which arises in our swapping scheme: when is it possible to make a state entangled by rescaling rows and columns as in Equation~\eqref{6_eqn:state}?

Finally, most of our results (e.g. Lemma~\ref{6_lem:cqbound}, Theorem~\ref{6_thm:suffswap}, and Corollary~\ref{6_thm:iso}) apply to higher dimensional systems. However, it would be nice to extend Lemma~\ref{6_lem:maxCQ} to this more general setting by characterizing the maximally non-classical classical-quantum states of higher dimension than qubits. Unfortunately, our approach here does not seem to apply in a straightforward manner to this setting, and further investigation is needed.

\paragraph{Organization of chapter.} We begin in Section~\ref{6_sec:definitions} with definitions and background information. In Section~\ref{6_sec:bounds}, we provide bounds on non-classicality for classical-quantum states, as measured by the relative entropy of quantumness. In Section~\ref{6_sec:swapping}, we present several results and observations regarding entangling input systems via the activation protocol and entanglement swapping.

\section{Preliminaries}
%\section{Preliminaries: Non-Classical Correlations and their Activation}
\label{6_sec:definitions}

Throughout this chapter, we continue to use the notation and definitions from Chapter~\ref{chap:activation}, which we briefly outline now. Recall from Definition~\ref{5_def:classical} that a \emph{strictly classically correlated} or \emph{classical} quantum state $\rho$ is one which diagonalizes in a local product basis $\cBB$.

We now outline the activation protocol of Chapter~\ref{chap:activation} (see Section~\ref{5_scn:protocol} for further details). Consider an arbitrary state $\rho\in \mathcal{D}((\complex^d)^{\otimes n})$ living in register $A$, where the $i$th local $d$-dimensional system lives in register $A_i$. We refer to $A$ as the \emph{system}. We further introduce a joint register ${A'}$ of $d$-dimensional registers $A'_1,\ldots,A'_n$ of $n$ ancilla qudit registers each initialized to the state $\ketbra{0}{0}_{A'}$, henceforth called the \emph{ancilla} (see Figure~\ref{figact}).
The initial state of the joint $A:A'$ system is thus $\rho_{\wek{A}\wek{A'}} = \rho_{\wek{A}} \otimes \ket{0}\bra{0}^{\otimes n}_{\wek{A'}}$.
For a given input $\rho_\wek{A}$, we first consider for each $i$ an adversarial application of a local unitary $U_i$ to each $A_i$ (i.e.\ this chooses the control basis for system $i$), and follow by applying one CNOT gate on each subsystem $A_i$ (control qudit) and the corresponding ancillary party $A'_i$ (target qudit). The final state of system plus ancilla at the end of this protocol is
\begin{equation}\label{6_eqn:finalstate}
\rho^f_{\wek{A}:\wek{A'}} = V (\rho_{\wek{A}} \otimes \ket{0}\bra{0}^{\otimes n}_{\wek{A'}}) V^\dagger\,,
\end{equation}
with $V = {{CNOT}}_{\wek{A}:\wek{A'}} (U_{\wek{A}} \otimes I_{\wek{A'}})$, $U_{\wek{A}} = \otimes_{i=1}^n U_i$, and ${{CNOT}}_{\wek{A}\wek{A'}}=\bigotimes_{i=1}^n CNOT_{A_iA_i'}$. Recall that by Theorem~\ref{5_thm:quantumnessiff}, the output $\rho^f_{AA'}$ is separable across the $A:A'$ split if and only if $\rho_A$ is classical. As done in Chapter~\ref{chap:activation}, we henceforth refer to $\rho_A$ and $\rho^f_{AA'}$ as $\rho$ and $\rho^f$ for simplicity, respectively.

It will be useful to also recall that $\rho^f$ can be written as
\beq
\label{6_eqn:maxcorr}
\rho^f=\sum_{\wek{i}\wek{j}}\rho_{\wek{i}\wek{j}}^\wek{\cBB}\ket{\wek{i}}\bra{\wek{j}}_{\wek{A}}\otimes \ket{\wek{i}}\bra{\wek{j}}_{\wek{A'}},
\eeq
where
\beq
\label{6_eqn:matrixelement}
\rho_{\wek{i}\wek{j}}^\wek{\cBB}:=\bra{\wek{\cB}(\wek{i})}\rho \ket{\wek{\cB}(\wek{j})},
\eeq
for $\ket{\wek{\cB}(\wek{i})}=U_\wek{A}^\dagger \ket{\wek{i}}$. In other words, $\rho^f$ is of the maximally correlated~\cite{R01} form in the $\wek{A}:\wek{A'}$ cut. Using this observation, we showed (Theorem~\ref{5_thm:distill}) that if one quantifies the minimum distillable entanglement generated across the $A:A'$ split in this protocol, the corresponding measure of non-classicality we obtain is given by
\begin{equation}\label{6_eqn:q}
Q(\rho)=\min_{{\cBB}}\Big(S(\rho^{\wek{\cBB}})-S(\rho)\Big),
\end{equation}
for $S(\rho)$ the von Neumann entropy, where the minimization is over all local product bases $\cBB$, and where $\rho^{\wek{\cBB}}:=\sum_{\wek{i}}\proj{\wek{\cB}(\wek{i})}\rho \proj{\wek{\cB}(\wek{i})}$. It turned out (Corollary~\ref{5_cor:req}) that $Q(\rho)$ in fact coincides with the measure of non-classicality known as the \emph{relative entropy of quantumness (REQ)}~\cite{Bravyi2003,PhysRevA.71.062307,groismanquantumness,PhysRevA.77.052101,MPSVW10}, bestowing the latter with an operational interpretation.

\section{Upper bounds for separable states}
\label{6_sec:bounds}

In Theorem~\ref{5_thm:sepuppbnd}, we showed that for a bipartite state $\rho_{AB}$ (where in the bipartite case we adopt the notational convention that $A_1=A$ and $A_2=B$), the quantity $Q(\rho_{AB})$ can achieve its maximum value of $\log d$ only for entangled states. We also showed that for increasing local dimension $d$, $Q(\rho_{AB})$ for certain separable $\rho_{AB}$ can asymptotically approach $\log d$. What can be said, however, in the non-asymptotic setting? In other words, for a fixed local dimension $d$, how non-classical can separable $\rho_{AB}$ be?

In this section, we first obtain a simple upper bound on $Q(\rho_{AB})$ for the subclass of separable states known as classical-quantum (CQ) states, that holds for arbitrary local dimensions. Recall from Section~\ref{0_sscn:nonclassical} that CQ states are those which can be written as $\rho_{AB} = \sum_{i=1}^{d_A} p_i \ketbra{i}{i}\otimes\rho_i$ for $\set{\ket{i}}_{i=1}^{d_A}$ an orthonormal basis, $\set{p_i}$ a probability distribution, and $d_A$ and $d_B$ the local dimensions of systems $A$ and $B$. We then completely characterize the set of maximally non-classical two-qubit CQ states with respect to the relative entropy of quantumness $Q(\rho_{AB})$, and show that such states achieve $Q(\rho_{AB})=1/2$.

We begin with our claimed upper bound, which holds even when the local dimensions of $A$ and $B$ differ.

\begin{lemma}\label{6_lem:cqbound}
For any CQ state $\rho_{AB} = \sum_{i=1}^{d_A} p_i \ketbra{i}{i}\otimes\rho_i$, one has
\begin{equation}
    Q (\rho_{AB})\leq \left(1-\frac{1}{d_A}\right)\log_2 d_B.
\end{equation}

\end{lemma}

\begin{proof}
We have
\begin{eqnarray}
    Q (\rho_{AB}) &=& \min_{\wek{\mathcal{B}}} S(\rho_{AB}^{\wek{\mathcal{B}}})-S(\rho_{AB})\\
        %&=& \min_{\set{\Pi^A_j},\set{\Pi^B_k}} S\left(\sum_{i=1}^{d_A}\sum_{j,k} p_i         \Pi^A_j\ketbra{i}{i}\Pi^A_j\otimes \Pi^B_k \rho_i \Pi^B_k \right) - S (\rho_{AB}) \\
        &=& \min_{\mathcal{B}_B} S\left(\sum_{i=1}^{d_A} p_i \ketbra{i}{i}\otimes\left(\sum_{{j=1}}^{d_B}\proj{{\cB}_B({j})}\rho_i \proj{{\cB}_B({j})} \right )\right)-S (\rho_{AB})\nonumber\\
        &=& \min_{\mathcal{B}_B} \left(H(p) + \sum_{i=1}^{d_A} p_i S \left(\rho^{\cBB_B}_{i} \right )\right) - \left(H(p) + \sum_{i=1}^{d_A} p_i S (\rho_i)\right)\\
        &=& \min_{\mathcal{B}_B}\sum_{i=1}^{d_A} p_i \left[S \left(\rho^{\cBB_B}_{i} \right ) -S (\rho_i)\right]\label{6_eqn:CQbound1},
\end{eqnarray}
where $H(p)$ denotes the Shannon entropy of the probability distribution $p=\set{p_i}_i$, and the second equality follows from choosing $\mathcal{B}_A$ to coincide with the basis $\set{\ket{i}}$. Let $p_m:=\max_i p_i$. Our strategy is to let $\mathcal{B}_B$ project onto an eigenbasis of $\rho_m$, yielding:
\begin{eqnarray}
    Q (\rho_{AB})
        &\leq& \sum_{i\neq m} p_i \left[S \left(\rho^{\cBB_B}_{i}\right ) -S (\rho_i)\right]\label{6_eqn:simplestrat1}\\
        &\leq& \sum_{i\neq m} p_i S \left(\rho^{\cBB_B}_{i} \right )\label{6_eqn:simplestrat2} \\
        &\leq& \left(1-\frac{1}{d_A}\right)\log d_B,
\end{eqnarray}
where the second inequality follows since $S(\rho_i)\geq 0$, and the third inequality follows since $p_m \geq 1/d_A$ and $S(\sigma_B)\leq \log d_B$ for any density operator $\sigma_B$.
\end{proof}

For a two-qubit CQ state $\rho_{AB}$, Lemma~\ref{6_lem:cqbound} implies $Q(\rho_{AB}) \leq 1/2$. We now show that this bound is tight by characterizing the set of CQ states attaining $Q (\rho_{AB}) = 1/2$.

\begin{lemma}\label{6_lem:maxCQ}
    Consider CQ state $\rho_{AB}\in \mathcal{D}(\complex^2\otimes\complex^2)$ such that $\rho_{AB} = \sum_{i=1}^{2} p_i \ketbra{i}{i}\otimes\rho_i$. Then $Q (\rho_{AB}) = 1/2$ if and only if $p_1=p_2=1/2$ and $\rho_1=\proj{\psi_1}$ and $\rho_2=\proj{\psi_2}$ for some $\ket{\psi_1},\ket{\psi_2}\in\complex^2$ such that $\abs{\braket{\psi_1}{\psi_2}}^2=1/2$.
\end{lemma}
\begin{proof}
That $\rho_{AB}$ with $p_1\neq 1/2$ implies $Q(\rho_{AB}) < 1/2$ follows immediately from Equation~\eqref{6_eqn:simplestrat2} and the fact that $0 \leq S(\sigma)\leq 1$ for any 1-qubit density operator $\sigma$. We thus henceforth assume $p_1=p_2=1/2$. That $\rho_1$ and $\rho_2$ must be pure now also follows analogously, for if, say, $\rho_1$ is mixed, then we simply choose $\mathcal{B}_B$ in Equation~\eqref{6_eqn:simplestrat1} to instead project onto an eigenbasis of $\rho_2$, and use the fact that $S(\rho_1)>0$ to achieve $Q(\rho_{AB}) < 1/2$. We thus henceforth assume $\rho_1=\proj{\psi_1}$ and $\rho_2=\proj{\psi_2}$ for some $\ket{\psi_1},\ket{\psi_2}\in\complex^2$. It remains to show that we must have $\abs{\braket{\psi_1}{\psi_2}}^2=1/2$.

Plugging $\rho_{AB}$ into Equation~\eqref{6_eqn:CQbound1} and noting that $S(\rho_1)=S(\rho_2)=0$, we have
\begin{eqnarray}
    Q (\rho_{AB}) &=&\frac{1}{2}\min_{\mathcal{B}_B} \left[S([\ketbra{\psi_1}{\psi_1}]^{\cBB_B}) + S([\ketbra{\psi_2}{\psi_2}]^{\cBB_B})\right]\\
    &=& \frac{1}{2}\min_{\mathcal{B}_B} \Big[H\left(\abs{\braket{{\cB}_B({0})}{\psi_1}}^2,\abs{\braket{{\cB}_B({1})}{\psi_1}}^2\right)
+H\left(\abs{\braket{{\cB}_B({0})}{\psi_2}}^2,\abs{\braket{{\cB}_B({1})}{\psi_2}}^2\right)\Big]\nonumber\\
    &=& \frac{1}{2}\min_{\ket{{\cB}_B({0})}} \Big[H\left(\abs{\braket{{\cB}_B({0})}{\psi_1}}^2,\abs{\braket{{\cB}_B({0})}{\psi_1^\perp}}^2\right)
+\nonumber \\ &&\hspace{16mm}H\left(\abs{\braket{{\cB}_B({0})}{\psi_2}}^2,\abs{\braket{{\cB}_B({0})}{\psi_2^\perp}}^2\right)\Big],\label{6_eqn:CQmin}
\end{eqnarray}
where $\braket{\psi_1}{\psi_1^\perp}=\braket{\psi_2}{\psi_2^\perp}=0$, and where the last equality follows since $\proj{{\cB}_B({j})}$ are rank-one projectors. Note that one can think of the last equality as effectively switching the roles of the measurement and the target state, so that the minimization can be thought of as being taken over all pure \emph{target} states $\ket{{\cB}_B({0})}$ with respect to \emph{measurements} in the bases $\mathcal{B}_1:=\set{\ket{\psi_1},\ket{\psi_1^\perp}}$ and $\mathcal{B}_2:=\set{\ket{\psi_2},\ket{\psi_2^\perp}}$. We can now plug Equation~\eqref{6_eqn:CQmin} into the well-known entropic uncertainty relation of Maassen and Uffink~\cite{Maassen1988,WW10}, which states that for classical distributions $P_c$ and $P_d$ obtained by measuring pure state $\ket{\psi}$ with respect to orthonormal bases $\mathcal{C}=\set{\ket{c}}$ and $\mathcal{D}=\set{\ket{d}}$, respectively, we have
\begin{equation}
    \frac{1}{2}(H(P_c)+H(P_d))\geq -\log f(\mathcal{C},\mathcal{D}),
\end{equation}
where $f(\mathcal{C},\mathcal{D}):=\max\set{\abs{\braket{c}{d}}\mid \ket{c}\in\mathcal{C}, \ket{d}\in\mathcal{D}}$. We thus obtain:
\begin{equation}
    Q(\rho_{AB}) \geq \max _{\ket{\phi_1}\in \mathcal{B}_1,\ket{\phi_2}\in \mathcal{B}_2}-\log\abs{\braket{\phi_1}{\phi_2}}.%=-\log\frac{1}{\sqrt{2}}=\frac{1}{2}.
\end{equation}
Note that this lower bound attains its maximum value of $1/2$ if $\mathcal{B}_1$ and $\mathcal{B}_2$ are mutually unbiased, i.e.\ when $\abs{\braket{\psi_1}{\psi_2}}^2=1/2$. On the other hand, suppose $\mathcal{B}_1$ and $\mathcal{B}_2$ are \emph{not} mutually unbiased, i.e.\ suppose without loss of generality that $\abs{\braket{\psi_1}{\psi_2}}^2 > 1/2$. Then choosing $\ket{{\cB}_B({0})}=\ket{\psi_1}$ in Equation~\eqref{6_eqn:CQmin} yields $Q (\rho_{AB})<1/2$. The claim follows.
\end{proof}

Combining Lemmas~\ref{6_lem:cqbound} and~\ref{6_lem:maxCQ}, we obtain a characterization of the set of two-qubit CQ states which are deemed maximally non-classical by $Q$. Such states include, for example, the CQ state
\begin{eqnarray}
    \rho &=& \frac{1}{2}\ketbra{0}{0}\otimes\ketbra{0}{0} + \frac{1}{2}\ketbra{1}{1}\otimes\ketbra{+}{+}\label{6_eqn:CQexample}
    =\frac{1}{2}\left(
         \begin{array}{cccc}
           1 & 0 & 0 & 0 \\
           0 & 0 & 0 & 0 \\
           0 & 0 & \frac{1}{2} & \frac{1}{2} \\
           0 & 0 & \frac{1}{2} & \frac{1}{2} \\
         \end{array}
       \right),
\end{eqnarray}
where $\ket{+}=(\ket{0}+\ket{1})/\sqrt{2}$.

\section{Swapping the ancilla-system entanglement onto the system}
\label{6_sec:swapping}

We now explore the possibility of generating entanglement in the \emph{original} system $AB$ by projecting the ancilla systems $\wek{A}'$ of the state $\rho^f$ of \eqref{6_eqn:finalstate} jointly onto an entangled pure state. In other words, we consider an entanglement swapping process~\cite{swapping} that maps the system-ancilla entanglement onto the systems $AB$. As we are only interested in knowing whether this is possible (rather than, say, in the probability of success), the filtering via a pure state is not restrictive and corresponds to the best possible strategy. Our results indicate that this feat is possible for some, but not all, separable non-classical states.

We begin by noting that thanks to the maximally-correlated form of $\rho^f$ (Equation~\eqref{6_eqn:maxcorr}), we have that the (unnormalized) final state of system $AB$ after projecting the ancilla system onto (normalized) state $\ket{\phi}=\sum_\wek{i}\alpha_{\wek{i}}\ket{\wek{i}}\in(\complex^d)^{\otimes n}$ is given by
\begin{equation}
    \rho_\phi
    :=
    \Tr_{\wek{A'}}(\rho^f\proj{\phi}_\wek{A'})
    =
    \sum_{\wek{i}\wek{j}}   \left[\rho_{\wek{ij}}^{\wek{\cBB}}\alpha_{\wek{i}}\alpha^*_{\wek{j}}\right]\ketbra{\wek{i}}{\wek{j}}\label{6_eqn:state},
\end{equation}
with $\rho_{\wek{ij}}^{\wek{\cBB}}$ defined in Equation~\eqref{6_eqn:matrixelement}.
%\beq
%\label{6_eqn:matrixelement}
%\rho_{\wek{i}\wek{j}}^\wek{\cB}=\bra{\wek{\cB}(\wek{i})}\rho_\wek{A} \ket{\wek{\cB}(\wek{j})},
%\eeq
Hence, the resulting (unnormalized) state $\rho_\phi$ is simply the Hadamard product of the original state (represented in the $\wek{\cBB}$ basis) and $\ketbra{\phi}{\phi}$ (represented in the computational basis), i.e.\ $\rho_\phi=\rho^{\cBB}\circ\ketbra{\phi}{\phi}$ for $\circ$ the Hadamard product defined such that $(C\circ D)(i,j):=C(i,j)D(i,j)$. %Since $\ket{\phi}$ is arbitrary, we can say that $\rho^f_{\wek{A}}$ is obtained by rescaling rows and columns (with the same---up to conjugation---rescaling factor for row and column $\wek{i}$) of the original state $\rho_\wek{A}$ represented in the basis $\wek{\cB}$.
%\begin{remark}
%\label{rem:original}
%Notice that $ \rho^f_{\wek{A}}$ can always be made equal (up to normalization) to the (locally rotated) original state $\rho_\wek{A}$ by choosing $\alpha_{\wek{i}}$ to be a constant independent of $\wek{i}$.
%\end{remark}

As previously mentioned, our goal is to answer the following question: For a given input $\rho$, is it true that for \emph{any} choice of starting local bases for the CNOT gates in the activation protocol, there exists a state $\ket{\phi}$ such that $\rho_\phi$ is entangled (across its constituent local $d$-dimensional systems)?

In Section~\ref{6_scn:iso} we provide a simple sufficient condition under which the generation of entanglement in the original system is always possible with an appropriate choice of $\ket{\phi}$, regardless of the choice of adversarial local unitary. We then observe that this condition holds for all pseudo-isotropic states as in Equation~\eqref{6_eqn:pseudoisotropic}, with $\ket{\psi}$ entangled and $p>0$.
In Sections~\ref{6_scn:CQ} and \ref{6_scn:QQ}, we  provide examples of classical-quantum (CQ) and quantum-quantum (QQ) separable states, respectively, for which entanglement in $AB$ cannot be generated in this fashion, i.e.\ there exists a choice of $U_{AB}$ that prevents the generation of entanglement in $AB$ via the swapping of system-ancilla entanglement, \emph{even if} there is necessarily entanglement between the system-ancilla cut after the activation protocol is run.%Sections~\ref{6_scn:CQ} and~\ref{6_scn:QQ}, respectively.
%This note focuses on the case of bipartite $\rho_\wek{A}=\rho_{AB}$.

\subsection{Sufficient condition for entanglement swapping}
\label{6_scn:iso}

%\begin{definition}[Isotropic State] For any integer $d\geq 2$, the isotropic state $\rhoi$ is an operator acting on $\complex^d\times \complex^d$ such that
%\begin{equation}
%    \rhoi := (1-p)\frac{I}{d^2} + p\ketbra{\phi^+}{\phi^+},
%\end{equation}
%where $I$ denotes the identity, $\phi^+ := \frac{1}{\sqrt{d}}\sum_j\ket{jj}$, and $-1/(d^2-1)\leq p\leq 1$. The isotropic state is separable for $p\leq 1/(d+1)$, and entangled otherwise.
%\end{definition}

We focus again on the bipartite case $A_1=A$, $A_2=B$. We have the following simple condition which ensures the swapping of entanglement is possible.
\begin{theorem}
\label{6_thm:suffswap}
If for any choice of local basis $\wek{\cBB}$, there exists a non-zero off-diagonal element of an off-diagonal block of $\rho_{AB}^{\wek{\cBB}}$, i.e.\ if for all $\wek{\cBB}=\cBB_A\cBB_B$ there exists a choice of $i\neq j$ and $k\neq l$ such that $\bra{\cB_A(i)\cB_B(k)}\rho_{AB}\ket{\cB_A(j)\cB_B(l)}\neq0$, then it is possible to swap entanglement back into the input systems (regardless of the choice of $\wek{\cBB}$), i.e.\ there exists a $\ket{\phi}$ such that $\rho_\phi$ is entangled.
\end{theorem}
\begin{proof}
   The strategy of the proof is to choose $\ket{\phi}$ so that the result of the Hadamard product in Equation~(\ref{6_eqn:state}) is non-positive under partial transposition (NPT)~\cite{peresPT,HorodeckiPT}. Fix any choice of local basis $\wek{\cBB}$. By assumption, we know there exist indices $i\neq j$ and $k\neq l$ such that $\bra{\cB_A(i)\cB_B(k)}\rho_{AB}\ket{\cB_A(j)\cB_B(l)}\neq0$. In order to ensure that $\rho_\phi$ is NPT, we thus choose $\ket{\phi}$ to single out these non-zero off-diagonal terms by setting
       \begin{equation}
        \ket{\phi} = \frac{1}{\sqrt{2}}(\ket{ik} + \ket{jl}).
    \end{equation}
    With this choice of $\ket{\phi}$, $\rho_\phi$ becomes a Hermitian matrix with only four non-zero entries, two of which lie on the diagonal at positions $\ketbra{i}{i}\otimes\ketbra{k}{k}$ and $\ketbra{j}{j}\otimes\ketbra{l}{l}$, and two of which lie at off-diagonal positions of off-diagonal blocks at $\ketbra{i}{j}\otimes\ketbra{k}{l}$ and $\ketbra{j}{i}\otimes\ketbra{l}{k}$ (i.e.\ the four entries form the four corners of a square). It follows that the partial transpose of $\rho_\phi$ is not positive.
\end{proof}

\begin{corollary}\label{6_thm:iso}
    For any
    \beq
\label{6_eqn:pseudoisotropic}
\rho(\psi,p)_{AB}:=(1-p)\frac{I_{AB}}{D}+p\proj{\psi}_{AB},
\eeq
with $I_{AB}/D$ the maximally mixed state for $AB$ and $D$ the dimension of $AB$, if $\ket{\psi}$ is entangled and $p>0$, then there exists a choice of $\ket{\phi}$ such that $\rho_\phi$ is entangled.
\end{corollary}
\begin{proof}
Since the maximally mixed component of~(\ref{6_eqn:pseudoisotropic}) is diagonal with respect to any choice of local bases, it suffices to argue that $|\psi\rangle$ satisfies the condition of Theorem 1. This easily follows from the fact $|\psi\rangle$ is entangled, and thus has, up to local unitaries, a Schmidt decomposition $\sum_{k=0}^{d_A-1} \sqrt{\lambda_k} \ket{k}\ket{k}$, with $\lambda_0\geq\lambda_1>0$.
%One easily checks by direct inspection that if the Schmidt rank of $\psi$ is strictly greater than one, then for all $p>0$ the condition of Theorem \ref{6_thm:suffswap} is satisfied.
\end{proof}

Corollary~\ref{6_thm:iso} shows that for any value of $p>0$, entanglement can be transferred to the original system for the pseudo-isotropic state $\rho(p,\psi)$ of Equation~\eqref{6_eqn:pseudoisotropic}, even for values of $p$ which correspond to \emph{separable} states (recall that for $p$ small enough, the state $\rho(p,\psi)$ is separable due to the existence of a separable ball around the maximally mixed state~\cite{zyczkowski1998sepball,GB02}). We remark that for all $p>0$ and entangled $\ket{\psi}$, $\rho(p,\psi)$ is known to be non-classical~\cite{groismanquantumness}, and so here the non-classicality of the starting state allows us to create entanglement in the original systems $AB$ by applying the activation protocol followed by our entanglement swapping procedure.

\subsection{Classical-quantum separable states}\label{6_scn:CQ}

In Section~\ref{6_scn:iso}, we demonstrated that for certain non-classically correlated states, entanglement can be mapped back into the original system after the activation protocol is run. Can this be achieved with \emph{any} type of non-classically correlated input $\rho$? We now show that the answer is no --- there exist separable non-classical $\rho$ such that, while entanglement is always generated in the activation protocol between systems and ancilla independently of the local unitaries $U_A$ and $U_B$, a proper adversarial choice of local unitaries $U_A$ and $U_B$ can nevertheless prevent entanglement from being mapped back to the system.

Consider the separable non-classical CQ state of Equation~\eqref{6_eqn:CQexample}.
%
%\begin{eqnarray}
%    \rho &=& \frac{1}{2}\ketbra{0}{0}\otimes\ketbra{0}{0} + \frac{1}{2}\ketbra{1}{1}\otimes\ketbra{+}{+}\label{6_eqn:CQexample}\\
%    &=&\frac{1}{2}\left(
%         \begin{array}{cccc}
%           1 & 0 & 0 & 0 \\
%           0 & 0 & 0 & 0 \\
%           0 & 0 & \frac{1}{2} & \frac{1}{2} \\
%           0 & 0 & \frac{1}{2} & \frac{1}{2} \\
%         \end{array}
%       \right).
%\end{eqnarray}
%where $\ket{+}=(\ket{0}+\ket{1})/\sqrt{2}$.
%
By Equation~(\ref{6_eqn:state}), note that when the adversarial local unitaries are chosen as $U_A=U_B=I$, we have
\begin{equation}
    \rho_\phi= \rho \circ \ketbra{\phi}{\phi}.
\end{equation}
Since $\rho$ is block diagonal, it hence follows that $\rho_\phi$ is block diagonal, since the Hadamard product cannot change this block diagonal structure regardless of the choice of $\ket{\phi}$. We conclude that there exists a choice of local bases (i.e the computational basis) with respect to which $\rho_\phi$ is always separable for all $\ket{\phi}$, i.e.\ it is not possible to project the (necessarily present) system-ancilae entanglement generated in the activation protocol back onto the system. In fact, this proof approach holds for \emph{any} CQ (or QC) state that is not strictly classically correlated, implying that for such states, there is a choice of local unitaries for which, even if entanglement is created between system and ancilla in the activation protocol, such entanglement cannot be swapped back into the input system.

\subsection{Quantum-quantum separable states}\label{6_scn:QQ}
Based on the results in Section~\ref{6_scn:CQ}, one might hope that entanglement generation in separable starting systems is possible if $\rho$ is not CQ nor QC (i.e.\ $\rho$ is what we might call \emph{QQ separable}). We provide a counterexample to this conjecture here --- namely, we show that there exist QQ separable states for which an adversarial choice of local bases in the activation protocol prevents the swapping of ancilla-system entanglement back into the input systems.

To do so, consider the separable QQ operator:
\begin{equation}
    \rho_{AB} = \frac{1}{2}\ketbra{0}{0}\otimes\ketbra{+}{+} + \frac{1}{2}\ketbra{+}{+}\otimes\ketbra{0}{0}
    =\frac{1}{4}\left(
                                                   \begin{array}{cccc}
                                                     2 & 1 & 1 & 0 \\
                                                     1 & 1 & 0 & 0 \\
                                                     1 & 0 & 1 & 0 \\
                                                     0 & 0 & 0 & 0 \\
                                                   \end{array}
                                                 \right).
\end{equation}
\noindent To prove our claim, as in Section~\ref{6_scn:CQ}, we choose local adversarial unitaries $U_A=U_B=I$ and show that $\rho^f= \rho_{AB} \circ \ketbra{\phi}{\phi}$ is separable for any choice of $\ket{\phi}$. The latter is shown by first deriving a condition under which the eigenvalues of Hermitian operators with a structure similar to $\rho$ remain invariant under partial transposition. We then show that $\rho$ fulfills this condition for any choice of $\ket{\phi}$, implying $\rho$ always remains separable, since the partial transpose is a necessary and sufficient condition for separability of two-qubit states~\cite{HorodeckiPT}.

\begin{lemma}\label{6_lem:invarSpec}
    Given any Hermitian operator $X$ acting on $\complex^2\otimes\complex^2$ with off-diagonal blocks which are diagonal, i.e.\
\begin{equation}
    X = \left(
      \begin{array}{cccc}
        a_{11} & a_{12} & a_{13} & 0 \\
        a_{12}^* & a_{22} & 0 & a_{24} \\
        a_{13}^* & 0 & a_{33} & a_{34} \\
        0 & a_{24}^* & a_{34}^* & a_{44} \\
      \end{array}
    \right),
\end{equation}
    if either $a_{12}a_{34}^\ast\in\reals$ or $a_{13}a_{24}^\ast\in\reals$, then the spectrum of $A$ is invariant under partial transposition.
\end{lemma}
\begin{proof}
    Let $p_X(\lambda)$ and $p_{X^\Gamma}(\lambda)$ denote the characteristic polynomials of $X$ and $X^\Gamma$, the partial transpose of $X$, respectively. Then
    \begin{equation}
    p_X(\lambda)-p_{X^\Gamma}(\lambda)=2\operatorname{Re}(a_{12}a_{34}^\ast a_{13} a_{24}^\ast-a_{12}^\ast a_{34} a_{13} a_{24}^\ast)=4 \operatorname{Im}( a_{13}^\ast a_{24} )\operatorname{Im}( a_{12} a_{34}^\ast),
    \end{equation}
    where $\operatorname{Re}(x)$ ($\operatorname{Im}(x))$ denotes the real (imaginary) part of $x$. The claim follows for $a_{12}a_{34}^\ast\in\reals$. An analogous calculation yields the $a_{13}a_{24}^\ast\in\reals$ case.
\end{proof}

With Lemma~\ref{6_lem:invarSpec} in hand, it is easy to see that $\rho^f$ has a positive partial transpose (and is hence separable) for all $\ket{\phi}$ --- specifically, we observe that $\rho$ satisfies the conditions of Lemma~\ref{6_lem:invarSpec} since $a_{12}a_{34}^\ast=(1/4)(0)=0$, and this in particular holds even after taking the Hadamard product with any $\ketbra{\phi}{\phi}$. Since $\rho$ is positive semidefinite, it thus follows from Lemma~\ref{6_lem:invarSpec} that $\rho^f$ must also be positive semidefinite under partial transposition and hence separable. Thus, there exist QQ separable states for which system-ancilla entanglement cannot be mapped back to the system.

%\subsection{More Questions?}

%Combining the results of Sections~\ref{6_scn:iso},~\ref{6_scn:CQ}, and~\ref{6_scn:QQ}, we see that

Theorem \ref{6_thm:suffswap} tells us that if a two-qubit state $\rho$ has off-diagonal terms on its off-diagonal blocks for any choice of local bases, then entanglement can be created among the systems via swapping. On the other hand, if $\rho$ is restricted to having off-diagonal blocks which are diagonal, as was seen with the CQ and QQ counterexamples considered in Sections~\ref{6_scn:CQ} and~\ref{6_scn:QQ}, then there are choices of local initial rotations such that entanglement generation among the systems is not necessarily possible (actually, in the CQ case, entanglement generation is not possible for \emph{any} choice of local initial rotations).

One could ask whether this ``diagonal off-diagonal'' block structure is sufficient to rule out the possibility of entanglement generation. The answer is negative. Consider the following (un-normalized) positive semidefinite operator which has diagonal off-diagonal blocks:
\begin{equation}
    \rho=\left(
      \begin{array}{cccc}
        \frac{3}{2} & i & 1 & 0 \\
        -i & \frac{3}{2} & 0 & i \\
        1 & 0 & \frac{3}{2} & 1 \\
        0 & -i & 1 & \frac{3}{2} \\
      \end{array}
    \right).
\end{equation}
%It turns out that the partial transposition of $\rho$ has a negative eigenvalue (observe that $\rho$ thus also necessarily violates the conditions of Lemma~\ref{6_lem:invarSpec}). Hence, despite the fact that $\rho$ has off-diagonal blocks which are diagonal, it is nevertheless entangled, implying entanglement transfer to the system is possible for any choice of local bases (indeed, as observed in Remark~\ref{rem:original}, the Hadamard product can be chosen to be trivial, so that the projection simply gives back (a locally rotated and unnormalized) $\rho_\wek{A}$).
It turns out that the partial transposition of $\rho$ has a negative eigenvalue (observe that $\rho$ thus also necessarily violates the conditions of Lemma~\ref{6_lem:invarSpec}). Hence, despite the fact that $\rho$ has off-diagonal blocks which are diagonal, it is nevertheless entangled, implying entanglement transfer to the system is possible for any choice of local bases: indeed, the Hadamard product can be chosen to be trivial, so that the projection simply gives back (a locally rotated and unnormalized) $\rho_\wek{A}$.

%\section*{Acknowledgments}
%
%\noindent \emph{Acknowledgements for this chapter.} We acknowledge  support by NSERC, QuantumWorks, CIFAR, Ontario Centres of Excellence, the Spanish MICINN through the Ram\'on y Cajal program, contract FIS2008-01236/FIS, project QOIT (CONSOLIDER2006-00019), by the Generalitat de Catalunya through
%CIRIT 2009SGR-0985, and IP Project QESSENCE.

%======================================================================
\chapter{Conclusion}\label{chap:conclusion}
%======================================================================

%As we have attempted to make each chapter in this thesis as self-contained as possible, we keep our conclusion here brief, and refer the reader to each respective chapter for relevant discussion and open questions.

In this thesis, we have studied three areas in quantum computation and information: the approximability of quantum problems, quantum proof systems, and non-classical correlations.  Our results in each of these areas are summarized as follows.

With respect to approximation, we have completed some of the first works in an area aiming to understand the computational complexity of efficiently and rigorously computing approximate solutions to problems which are complete for quantum complexity classes. In Chapter~\ref{chap:approx}, we demonstrated a polynomial time approximation algorithm for dense instances of the canonical QMA-complete problem, the local Hamiltonian problem. This required the derivation of a lower bound on the approximation ratio achievable by product state assignments, which as discussed in Chapter~\ref{chap:approx}, can be seen as negative progress towards a sought-after quantum PCP theorem. Among other open questions discussed therein, perhaps the most natural direction here is the pursuit of further new approximation algorithms for problems complete for QMA, the quantum generalization of NP. In Chapter~\ref{chap:hardnessapprox}, we then proceeded in the opposite direction by demonstrating hardness of approximation results for a new quantum complexity class we defined, $\cqs$. This class is an arguably natural generalization of $\st$, and the hard-to-approximate problems for $\cqs$ we considered are generalizations of classical covering problems obtained via the notion of \emph{quantum} constraint satisfaction (i.e.\ local Hamiltonian constraints). Aside from the obvious open questions here regarding further hardness of approximation results such as a quantum PCP theorem, we would be interested to know to what extent the class $\cqs$ itself may play an important role in quantum complexity theory, just as $\st$ has proven valuable in the classical setting.

With respect to quantum proof systems, in Chapter~\ref{chap:qmapoly} we focused on the question of whether multiple unentangled provers can be simulated by a single prover in the context of QMA proof systems. As the question of whether two provers are as good as one (i.e.\ is $\class{QMA}=\class{QMA}(2)$?) remains open despite much effort, we focused our attention on variants of QMA in which the verification protocol is suitably restricted. In this setting, we showed various results, including a collapse to QMA for a restricted variant of $\class{QMA}(\poly)$, and an alternate proof of a parallel repetition theorem for $\class{SepQMA}(2)$. Understanding the non-trivial ``power of unentanglement''~\cite{ABDFS09} between quantum provers remains a challenging and interesting direction of work.

Finally, with respect to non-classical correlations, in Chapter~\ref{chap:dqc}, we first motivated the study of such correlations by examining their role in the DQC1 trace estimation algorithm, as well as the quantum communication task of locking. Above all, understanding the precise role such correlations play in mixed-state quantum computing remains an important open question. In Chapter~\ref{chap:localunitary}, we then proposed a novel scheme for quantifying non-classical correlations based on a special class of local unitary operations. This raised the question as to how the notions of ``disturbance under measurement'' and ``disturbance under unitary operations'' differ in their characterization of non-classical correlations. Finally, Chapters~\ref{chap:activation} and~\ref{chap:entanglementswap} introduced and studied a protocol which ``activates'' non-classical correlations present in a multipartite quantum system into entanglement between the system and an ancilla. Aside from yielding a new framework through which new non-classicality measures can be discovered, our study here also revealed a surprising result: That mixedness in quantum states can play a very important role in giving rise to non-classical correlations, both for separable and entangled states. We would be interested to see how this framework may be further developed, and moreover whether the ideas behind it may prove useful in a quantum computational setting.

In conclusion, the field of quantum computation and information is, after over two decades of study, arguably no longer in its infancy. With a solid theoretical base and formalism in place, including the quantum circuit model, quantum complexity classes and proof systems, and foundations for quantum information theory, the field now covers a large number of areas of study, of which our focus here is but a small part. Yet, whether quantum computers will, at a practical level, indeed be the wave of the future, is in our opinion not yet entirely clear. What \emph{is} clear, however, is that no matter the outcome, the lessons learned through this line of work have already taught us much about the physical world around us. Indeed, the study of this field has united the physics and computer science communities towards a common ultimate goal: To probe the physical limits of nature and computing themselves. This in itself is no small feat. As it stands, information \emph{is} physical. We would not (and \emph{could not}) have it any other way.

%----------------------------------------------------------------------
% END MATERIAL
%----------------------------------------------------------------------

% B I B L I O G R A P H Y
% -----------------------

% The following statement selects the style to use for references.  It controls the sort order of the entries in the bibliography and also the formatting for the in-text labels.
\bibliographystyle{plain}
% This specifies the location of the file containing the bibliographic information.
% It assumes you're using BibTeX (if not, why not?).
\cleardoublepage % This is needed if the book class is used, to place the anchor in the correct page,
                 % because the bibliography will start on its own page.
                 % Use \clearpage instead if the document class uses the "oneside" argument
\phantomsection  % With hyperref package, enables hyperlinking from the table of contents to bibliography
% The following statement causes the title "References" to be used for the bibliography section:
\renewcommand*{\bibname}{References}

% Add the References to the Table of Contents
\addcontentsline{toc}{chapter}{\textbf{References}}

\bibliography{Sevag_Gharibian_Phd_Thesis_Bibliography_Abbrv,Sevag_Gharibian_Phd_Thesis_Bibliography}

\end{document}